\providecommand{\paperversion}{final}\providecommand{\paperversion}{draft}
\def\anonymoussubmission{1}
\RequirePackage{paperversions}

\ifblinded
\documentclass[acmsmall,review,screen,dvipsnames,anonymous,nonacm,nocopyright]{acmart}
\else
\documentclass[acmsmall,screen,dvipsnames,nonacm]{acmart}
\fi

\AtBeginDocument{%
  }

\usepackage{stmaryrd}
\usepackage{mathpartir}
\usepackage{scalerel}
\usepackage{olmacros}
\usepackage{algorithm}
\usepackage{algpseudocode}
\usepackage{lscape}
\usepackage{rotating}
\usepackage{subcaption}
\usepackage{subfiles}
\usepackage{cleveref}
\usepackage{aliascnt}
\usepackage{xcolor}
\usepackage[all, cmtip]{xy}
\usepackage{mdframed}
\usepackage{enumitem}
\usepackage{transparent}
\usepackage{checkpagelimit} \pagelimit{25}
\usepackage{titletoc}
\hypersetup{
    linktocpage=true
}
\ifshowcomments
  \usepackage[checknocomments]{authcomments}
\else
  \usepackage[disable]{authcomments}
\fi

\newcommenter{gary}{1,0,0}
\newcommenter{noam}{0,0,1}
\newcommenter{alex}{0,1,0}
\newcommenter{andrew}{1.0,0.8,0.7}    

\makeatletter
\def\@acmplainindent{0pt}
\def\@acmdefinitionindent{0pt}
\def\@proofindent{\noindent}
\makeatother

\setcopyright{acmlicensed}
\copyrightyear{2018}
\acmYear{2018}
\acmDOI{XXXXXXX.XXXXXXX}

\acmJournal{JACM}
\acmVolume{37}
\acmNumber{4}
\acmArticle{111}
\acmMonth{8}

\theoremstyle{acmplain}
\newtheorem{theorem}{Theorem}[section]

\newaliascnt{lemma}{theorem}
\newtheorem{lemma}[lemma]{Lemma}                \aliascntresetthe{lemma}
\newaliascnt{corollary}{theorem}
\newtheorem{corollary}[corollary]{Corollary}    \aliascntresetthe{corollary}
\newaliascnt{proposition}{theorem}
 \aliascntresetthe{proposition}
\newaliascnt{conjecture}{theorem}
 \aliascntresetthe{conjecture}

\theoremstyle{acmdefinition}
\newaliascnt{definition}{theorem}
\newtheorem{definition}[definition]{Definition} \aliascntresetthe{definition}
\newaliascnt{example}{theorem}
          \aliascntresetthe{example}

\newaliascnt{remark}{theorem}
             \aliascntresetthe{remark}

\theoremstyle{acmplain}
\numberwithin{equation}{section}

\crefname{theorem}{theorem}{theorems}              \Crefname{theorem}{Theorem}{Theorems}
\crefname{lemma}{lemma}{lemmas}                    \Crefname{lemma}{Lemma}{Lemmas}
\crefname{corollary}{corollary}{corollaries}       \Crefname{corollary}{Corollary}{Corollaries}
\crefname{proposition}{proposition}{propositions}  \Crefname{proposition}{Proposition}{Propositions}
\crefname{definition}{definition}{definitions}     \Crefname{definition}{Definition}{Definitions}
\crefname{example}{example}{examples}              \Crefname{example}{Example}{Examples}
\newcommand{\opcoltitle}{Oblivious Probabilistic Outcome Logic}
\providecommand{\shorttitle}{\opcoltitle}
\providecommand{\shortauthors}{}

\begin{document}

\title{\opcoltitle}
\subtitle{Verifying Probabilistic Programs with an Oblivious Adversary}


\author{Hanxi Chen}
\email{hc2296@cornell.edu}
\orcid{0009-0006-4486-7222}
\affiliation{%
  \institution{Cornell University}
  \city{Ithaca}
  \state{New York}
  \country{USA}
}

\author{Noam Zilberstein}
\email{noamzilberstein@me.com}
\orcid{0000-0001-6388-063X}
\affiliation{%
  \institution{New York University}
  \city{New York}
  \state{New York}
  \country{USA}
}

\author{Andrew C. Myers}
\email{andru@cs.cornell.edu}
\orcid{0000-0001-5819-7588}
\affiliation{%
  \institution{Cornell University}
  \city{Ithaca}
  \state{New York}
  \country{USA}
}

\author{Alexandra Silva}
\email{alexandra.silva@cornell.edu}
\orcid{0000-0001-5014-9784}
\affiliation{%
  \institution{Cornell University}
  \city{Ithaca}
  \state{New York}
  \country{USA}
}


\begin{abstract}
In the context of probabilistic programs, an \emph{oblivious adversary} resolves nondeterminism without seeing the outcomes of random draws.
Obliviousness is a common assumption in online algorithms and
distributed protocols, but the complex interaction between random
draws and adversarial choices makes it challenging to reason about 
correctness.
While there has been significant progress toward reasoning about programs that
combine randomization with nondeterminism, most of the work has focused on the
\emph{adaptive} model, whose omniscient view of program state is too powerful to
establish correctness for certain classes of programs.


We introduce \emph{Oblivious Probabilistic Outcome Logic} (\LogicName), a new logic for
reasoning about probabilistic programs with nondeterminism controlled by an
\emph{oblivious} adversary.
Building on Outcome Logic and Probabilistic Separation Logic, \LogicName
models adversarial choice as a resource and uses probabilistic independence
to ensure that random outcomes are hidden from the adversary.
The \LogicName proof system
provides expressive and compositional rules
for case analysis on both random and nondeterministic outcomes, and for
proving almost-sure termination.
Expressivity is tested through several case studies,
including a paging algorithm and a leader election protocol.
The \LogicName metatheory and case studies are mechanized in Lean 4.

\end{abstract}

\maketitle

\bibliographystyle{ACM-Reference-Format}


\section{Introduction}

Randomized algorithms achieve improved accuracy and performance compared to
deterministic ones, especially in sensitive settings like
distributed systems and operating systems.
In such settings, an adversary or
scheduler resolves nondeterministic choices such as choosing requests or scheduling processes.
Randomization makes the program flow unpredictable,
which prevents the adversary from picking the \emph{worst-case} option at every
step \cite{motwani1995randomized}.

The interaction between randomized computation and adversarial schedulers is very complex,
motivating the use of formal methods to establish correctness.
Recent program logics for randomized nondeterminism focus primarily on the \emph{adaptive}
model \cite{zilberstein2025demonic,zilberstein2026probabilistic,polaris}, where the adversary is fully omniscient
and can thus \emph{adapt} its strategy in real time as random draws are revealed.
Adaptive adversaries are 
mathematically convenient and yield well-studied
powerdomain semantics~\cite{jifeng1997probabilistic,mciver2005abstraction,varacca2002powerdomain}.

However, the competitive advantage of a randomized algorithm is stronger when
the adversary is \emph{oblivious}, meaning that its strategy is fixed ahead of the random draws \cite{ben-david1990power}.
Many randomized algorithms, including caching \cite{fiat1991competitive,raghavan1994memory} and
consensus \cite{aspnes2010modular,aspnes2012faster,chor1994wait-free} protocols,
rely on obliviousness to satisfy their intended guarantees.
The adaptive model is unrealistically strong; many threat models assume that
schedulers base their decisions on system-level metadata rather than internal random state,
and correctness relies on the weaker capabilities of the adversary.

On the other hand, attempting to leverage obliviousness surfaces subtle challenges.
Even with a fixed strategy, the adversary can exploit control flow to expose
random entropy, much like implicit
flows leak confidential data in information-flow control
\cite{denning1976lattice}.
A carefully chosen schedule can correlate adversarial choices with prior random draws,
which negates the competitive advantage of using randomization after a leak has occurred.
Care must be taken to design intuitive abstractions and sound
inference rules that handle entropy leakage in a first-class way.

We introduce Oblivious Probabilistic Outcome Logic (\LogicName), a program logic for reasoning about
\emph{randomized} programs with an \emph{oblivious} adversary.
The adversary's strategy is modeled by a predetermined tape of choices, which is
consulted to resolve
nondeterminism \cite{ben-david1990power,fatourou1997efficiency}.
The key insight is our use of \emph{probabilistic independence} to distinguish
private and leaked sources of randomness.
Each random computation branch records how
many tape entries it consumes, and a random choice remains private if it is probabilistically independent from the position in the tape.
For usability, \LogicName hides the schedule-index bookkeeping by tracking which
sources of entropy have been leaked,
yielding intuitive and expressive inference rules.


Taking inspiration
from \emph{probabilistic separation logics}
\cite{psl,li2023lilac,bao2025bluebell,zilberstein2026probabilistic}, \LogicName
extends compositional independence reasoning with 
\emph{schedule consumption} as a new resource.
Tracking schedule consumption allows
us to develop a more powerful \ruleref{Frame} rule
to prove that adversarial behavior is independent from prior random state.
Compared to prior work, which leverages obliviousness in a highly localized way
\cite{fan2025program}, \LogicName's separation logic approach supports compositional reasoning about
obliviousness in a global context, enabling verification of programs that
were previously out of scope,
 including an online paging algorithm and a leader election protocol. 
Our contributions are summarized below:

\paragraph{Denotational Semantics.}
In \Cref{sec:semantics}, we define a denotational model for probabilistic
programs with an oblivious adversary. Random computation is modeled using a probability distribution monad
and the scheduling tape is stored in a state monad.
We establish that the monad composition is well-behaved,
give fixed point semantics for unbounded loops, and show that the
resulting model is a refinement of the adaptive powerdomain semantics,
justifying the reuse of adaptive proof rules.

\paragraph{Schedule Consumption as a Resource.}
\Cref{sec:assertions} develops an assertion language for expressing properties
about probability distributions, including probabilistic independence. Going
beyond prior work, we encode schedule consumption as an Iris-style resource
\cite{iris1,iris}, and use it to rule out  leakage of entropy via side channels,
where subtle correlations can creep in.

\paragraph{Oblivious Probabilistic Outcome Logic.}

\Cref{sec:logic} presents \LogicName, a new logic for reasoning about
probabilistic programs with an oblivious adversary. The rules of \LogicName are
designed to be compositional,
which is challenging in cases where earlier
program steps leak entropy to the scheduler via side channels, thereby
changing the amount of information available in future steps. Schedule
consumption is therefore crucial to soundness, but \LogicName neatly hides the
details behind intuitive yet expressive inference rules. We also provide a rule
to prove almost sure termination.

\paragraph{Case Studies.}

In \Cref{sec:examples}, we validate \LogicName on representative examples where the distinction between
adaptive and oblivious scheduling is essential. Those examples come
from diverse domains, including online algorithms and consensus protocols.
Together, the case studies illustrate that \LogicName{} is powerful
enough to intuitively handle the kinds of patterns that arise when
designing algorithms to leverage the obliviousness of the adversary,
even those using unbounded loops.

\paragraph{Lean Mechanization.}
The full metatheory of \LogicName and case study proofs are
mechanized in Lean 4 \cite{moura2015lean,moura2021lean}.
All proofs were first written on paper (see appendix) and then mechanized in Lean with the help of Claude Code,
which was given the \LaTeX\ source of the proofs. The resulting artifact was then inspected and modified by hand for accuracy.

\medskip

\noindent We begin in \Cref{sec:overview} with the key subtleties and approach.
The technical content
(\Cref{sec:semantics,sec:assertions,sec:logic,sec:examples}) is followed by a
discussion of related work (\Cref{sec:related}) and next steps
(\Cref{sec:conclusion}).

\section{Overview}
\label{sec:overview}




Oblivious adversaries were first studied in the context of online algorithms,
which process inputs as they arrive.
\citet{ben-david1990power} showed that correctness for some randomized online
algorithms requires the adversary choosing the inputs to be
oblivious, meaning that the adversary cannot see random choices.
For example, consider the paging algorithm
in \Cref{fig:intro-paging}.
There are two memory pages, indexed by 0 and 1, and one of them is held in the
cache $c$. Each round, the adversary requests a page $r$ from the list
$[0,1]$.
If $r=c$, then no fetch is needed. But if $r \neq c$, then $r$ must be
fetched and the cache miss count $m$ is incremented. In that case, a random page
is also evicted from the cache.

\begin{figure}
\[
\begin{array}{l}
c \samp \bern{\tfrac12}\fatsemi i\coloneqq 0\fatsemi m \coloneqq 0\fatsemi \; \\
\whl{i < n}{} \\
  \quad i \coloneqq i+1\fatsemi r \gets [0,1] \fatsemi  \phantom{x} \\
  \quad \ift{r \neq c}{
  c \samp \bern{\tfrac12} \fatsemi m \coloneqq m+1}
\end{array}
\]
\caption{A simple online paging algorithm.}
\label{fig:intro-paging}
\end{figure}

An adaptive adversary could see the outcome of the random eviction
$c\samp\bern{\frac12}$, and choose its request accordingly to force a cache miss
on every round.
However, an oblivious adversary cannot see which page was evicted, so its next
request is \emph{probabilistically independent} from the cached page.
The probability of getting all misses is only $\frac{1}{2^n}$ against an
oblivious adversary whereas an adaptive one can force all requests to be
misses.
We will now give a glimpse of how Oblivious Probabilistic Outcome Logic (\LogicName) enables reasoning about programs like the one in \Cref{fig:intro-paging}.

\subsection{Semantics of Oblivious Adversaries}
\label{sec:overview-semantics}

Informally, \citet{ben-david1990power} defined an oblivious adversary as \emph{``one who must construct the request sequence in
advance (based only on the description of the online algorithm but before any
moves are made!).''}
To demonstrate the definition, let us consider two different compositions of
random and adversarial choices. Below on the left, a coin is first flipped and
then the adversary makes a move. On the right, the adversary picks first, before
the coin flip.
\begin{align}
{\color{NavyBlue}x\samp\bern{\tfrac12}} \fatsemi {\color{BrickRed}y \gets [0,1]}
\qquad\qquad
{\color{BrickRed}y \gets[0,1]} \fatsemi {\color{NavyBlue}x\samp\bern{\tfrac12}}
  \label{eq:adaptive-oblivious-ordering}
\end{align}
The two programs in (\ref{eq:adaptive-oblivious-ordering}) correspond to the following two execution
trees, where $\opnodeinl{p}$ nodes correspond to random choices and $\ndnode$
nodes correspond to adversarial choices.
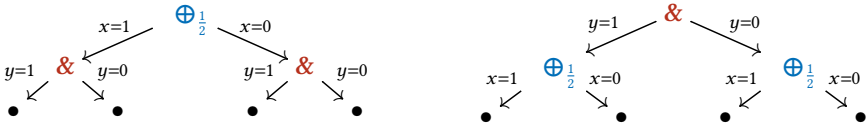
\begin{mathpar}
\xymatrix@R=.5em@C=.75em{
  &&& \opnode{\frac12}
    \ar[dll]_{x=1}
    \ar[drr]^{x=0}
  \\
  & \ndnode \ar[dl]_{y=1} \ar[dr]^{y=0}
  &&&& \ndnode \ar[dl]_{y=1} \ar[dr]^{y=0}
  \\
  \bullet &&\bullet &&\bullet && \bullet
}

\xymatrix@R=.5em@C=.75em{
  &&& \ndnode
    \ar[dll]_{y=1}
    \ar[drr]^{y=0}
  \\
  & \opnode{\frac12} \ar[dl]_{x=1} \ar[dr]^{x=0}
  &&&& \opnode{\frac12} \ar[dl]_{x=1} \ar[dr]^{x=0}
  \\
  \bullet &&\bullet &&\bullet && \bullet
}
\end{mathpar}
An adaptive adversary \emph{``makes the next request based on the algorithm's answers to previous ones''} \cite{ben-david1990power}.
The left tree (with two $\ndnode$ nodes) gives an adaptive adversary more power;
it can act differently for each outcome of $x$ and could, for example, force $x = y$ at the end of the execution. In the right tree, the adversary makes only one choice and the subsequent coin flips are probabilistically independent from that choice. Therefore, in the second program, $\PP[x=y] = \frac12$.

By contrast, 
the two programs in (\ref{eq:adaptive-oblivious-ordering}) are semantically equivalent in the oblivious adversary
model.
For example, if the
adversary predetermines 1 as its response to the first request, then $y$
is set to $1$ regardless of whether the
adversarial assignment occurs before or after the coin flip.
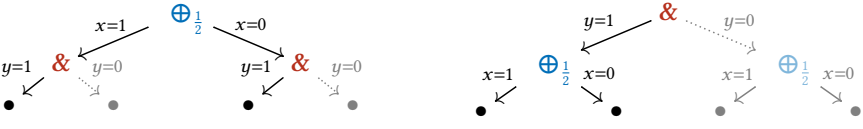
\begin{mathpar}
\xymatrix@R=.4em@C=.75em{
  &&& \opnode{\frac12}
    \ar[dll]_{x=1}
    \ar[drr]^{x=0}
  \\
  & \ndnode \ar[dl]_{y=1} \ar@{.>}@[gray][dr]^{\color{gray} y=0}
  &&&& \ndnode \ar[dl]_{y=1} \ar@{.>}@[gray][dr]^{\color{gray} y=0}
  \\
  \bullet && {\color{gray} \bullet} &&\bullet && {\color{gray}\bullet}
}

\xymatrix@R=.4em@C=.75em{
  &&& \ndnode
    \ar[dll]_{y=1}
    \ar@{.>}@[gray][drr]^{\color{gray}y=0}
  \\
  & \opnode{\frac12} \ar[dl]_{x=1} \ar[dr]^{x=0}
  &&&& {\transparent{0.5}\opnode{\frac12}} \ar@[gray][dl]_{\color{gray}x=1} \ar@[gray][dr]^{\color{gray}x=0}
  \\
  \bullet &&\bullet &&{\color{gray}\bullet} && {\color{gray}\bullet}
}
\end{mathpar}
Formally, we represent the adversary as an infinite tape or schedule
$s \in\schedule \triangleq \bbN \to \bbN$.
The tape is read starting at index $0$. When the current
index is $i$, the entry $s(i) \bmod n$
supplies an $n$-ary adversarial choice, and then the index
is incremented to $i + 1$.
Random computation and tape manipulation are encapsulated within 
$\calO(X)$, a composition of probability and state monads.
\[
  \calO(X) = {\color{BrickRed}\schedule} \to {\color{BrickRed}\bbN} \to {\color{NavyBlue}\D(}X \times {\color{BrickRed}\bbN}{\color{NavyBlue})}
\]
Above, the state of type $\bbN$ is simply the current index into the scheduling
tape, which
captures the key property of oblivious scheduling: the adversary cannot adapt the
tape to random outcomes, but adversarial choices still advance the current
index. Since different probabilistic branches may consume different numbers of
schedule entries, the denotation records both the memory output and the updated
scheduling index.
The adversarial choice operator $C_1 \nd C_2$ illustrates how the
current schedule entry selects the branch, and the selected branch continues at
the next index.
\[
  \de{C_1 \nd C_2}(\sigma)(s)(n) \triangleq \left\{
    \begin{array}{ll}
      \de{C_1}(\sigma)(s)(n+1) & \text{if}~ s(n) \bmod 2 = 0 \\
      \de{C_2}(\sigma)(s)(n+1) & \text{if}~ s(n) \bmod 2 = 1
    \end{array}
  \right.
\]
Thus, later adversarial choices remain independent of an earlier random branch
only when the branches arrive at the \emph{same scheduling index}, which implies reading
the same schedule entry.
In \Cref{sec:semantics}, we show that $\calO$ is a monad, and
a directed-complete partial order (DCPO), allowing construction of the denotational model $\de{C} \colon \Mem \to \calO(\Mem)$, where $\Mem \triangleq \Var\to\Val$.


\subsection{Obliviousness as Probabilistic Independence}
\label{sec:overview-sep}
Probabilistic independence is a powerful tool for \emph{compositional} reasoning
in randomized programs, as demonstrated in Probabilistic Separation Logic (\psl)
\cite{psl} and subsequent work
\cite{li2023lilac,bao2025bluebell,zilberstein2026probabilistic}.
As in standard separation logic \cite{sl,localreasoning}, the \emph{separating
  conjunction} $\varphi\sep\psi$ requires the variables of $\varphi$ and
$\psi$ to be disjoint. In probabilistic separation logics, separation
additionally asserts that the events described by
$\varphi$ and $\psi$ are probabilistically independent. For example, $x \sim \bern{\frac12} \sep y \sim\bern{\frac12}$ means that $x$ and
$y$ are distributed according to \emph{independent} Bernoulli
distributions, which allows us to reconstruct the joint distribution over $x$ and
$y$ through the product of individual probabilities:
\[
  \PP[x = 1 \land y=1]
  \quad=\quad
  \PP[x = 1]\cdot \PP[y=1]
  \quad=\quad
  \tfrac12\cdot\tfrac12
  \quad=\quad
  \tfrac14
\]
Specifications in probabilistic separation logics are given as triples $\triple\varphi{C}\psi$, which roughly mean
that if the program states are initially distributed according to $\varphi$ and
then the command $C$ is run, then the program states will be distributed
according to $\psi$ after the program execution.
Compositional reasoning is usually achieved via the \ruleref{Frame} rule,
which extends a local specification $\triple\varphi{C}\psi$ with an
independently distributed context $\vartheta$. In a naive form, the \ruleref{Frame} rule is:
\[
\inferrule{
  \triple{\varphi}C{\psi}
}{
  \triple{\varphi\sep\vartheta}C{\psi\sep\vartheta}
}{\rulename{Frame-Naive}}
\]
The \ruleref{Frame-Naive} rule is well behaved for purely probabilistic programs.
If $x$ is already distributed according to a Bernoulli distribution, and then $y$ is sampled,
we can use the \ruleref{Frame-Naive} rule to obtain the joint distribution
over $x$ and $y$.
\[
    \inferrule*[right=\ruleref{Frame-Naive}]{
      \triple{\sure{y \mapsto -}}{y \samp \bern{\tfrac12}}{y \sim \bern{\tfrac12}}
    }{
      \triple{x \sim \bern{\tfrac12} \sep \sure{y \mapsto-}}{y \samp \bern{\tfrac12}}{x \sim \bern{\tfrac12} \sep y \sim\bern{\tfrac12}}
    }
\]
However, Demonic Outcome Logic (\dol) and Probabilistic Concurrent Outcome Logic
(\pcol) showed that compositional probabilistic reasoning does not play nicely
with nondeterminism. Let $\varphi \nd \psi$ be an assertion stating an adversary
decides whether $\varphi$ or $\psi$ holds. Then in \dol and \pcol, which are
based on an adaptive adversary model, the following application of
\ruleref{Frame-Naive}
is unsound since the adversary may correlate its
choice of $y$ with $x$, as we saw in \Cref{sec:overview-semantics}.
\begin{equation}\label{eq:frame-naive}
    \inferrule*[right=\ruleref{Frame-Naive}]{
      \triple{\sure{y \mapsto -}}{y \gets [0,1]}{\sure{y\mapsto 0} \nd \sure{y \mapsto 1}}
    }{
      \triple{x \sim \bern{\tfrac12} \sep \sure{y \mapsto-}}{y \gets [0,1]}{x \sim \bern{\tfrac12} \sep (\sure{y\mapsto 0} \nd \sure{y \mapsto 1})}
    }
\end{equation}
Our first achievement in \LogicName was to recover a sound version of
\ruleref{Frame-Naive} for nondeterministic choices, such as the one above.
The key insight in (\ref{eq:frame-naive}) is that the tape index is \emph{probabilistically independent} from the variable $x$, and so the choice of $y$ is also independent from $x$, making the application sound.
Returning to the paging algorithm from \Cref{fig:intro-paging}, the \ruleref{Frame} rule
allows us to prove that the adversary cannot force a cache
miss. Since $c$ is uniformly random at the point where $r$ is adversarially
chosen, \ruleref{Frame} tells us that $c$ and $r$ are independent and therefore
the probability of a cache miss $r \neq c$ is exactly $\frac12$.
However, soundness of \ruleref{Frame} is not so simple, and requires us to carefully track correlations between variables and the tape index.

\subsection{A Hiccup: Leaking Random Bits Through Side Channels}
\label{sec:overview-leak}

The previous example may give the impression that the oblivious adversary is
toothless, since the program behaves like a purely probabilistic one once a schedule is fixed.
That impression is misleading; the schedule $s : \schedule$ is still chosen with \emph{full} knowledge
of the program structure, and a correctness claim must hold against \emph{every}
such schedule.
When control flow is randomized, different branches may consume different
numbers of schedule entries, so a nondeterministic choice may read from a
branch-dependent index. A carefully chosen schedule
can exploit the index offsets to correlate an adversarial choice with
an earlier random outcome.

Consider the program below on the left, which flips a fair coin into $x$, then
makes one adversarial choice for $y$ when $x = 0$ and then a second
 choice for $z$. Now, suppose that the schedule $s$ is such that $s(0) = 1$ and $s(1) = 0$ so that the program execution corresponds to the tree on the right.
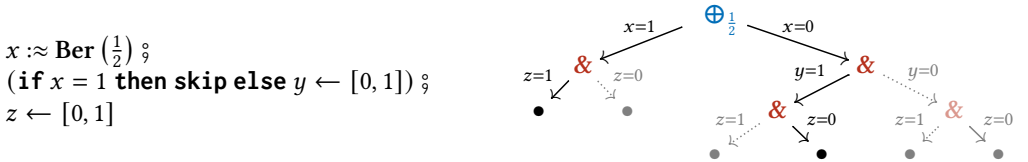
\begin{mathpar}
\begin{array}{l}
  x \samp \bern{\tfrac12}\fse \\
  (\iftf{x=1}{\skp}{y\gets [0,1]})\fse\\
  z \gets [0,1]
\end{array}

\vcenter{\vbox{\xymatrix@R=.45em@C=.45em{
  &&&& \opnode{\frac12}
    \ar[dlll]_{x=1}
    \ar[drrr]^{x=0}
  \\
  & \ndnode \ar[dl]_{z=1} \ar@{.>}@[gray][dr]^{\color{gray}z=0}
  &&&&&& \ndnode \ar[dll]_{y=1} \ar@{.>}@[gray][drr]^{\color{gray}y=0}
  \\
  \bullet && {\color{gray}\bullet}
  &&&\ndnode\ar@{.>}@[gray][dl]_{\color{gray}z=1} \ar[dr]^{z=0}
  &&&& {\transparent{0.5}\ndnode}\ar@{.>}@[gray][dl]_{\color{gray}z=1} \ar@[gray][dr]^{\color{gray}z=0}
  \\
  &&&&{\color{gray}\bullet} && \bullet && {\color{gray}\bullet} &&{\color{gray}\bullet}
}}}
\end{mathpar}
Along any execution path, the first adversarial choice goes left and the second goes right, making $x$ and $z$ perfectly correlated: if $x = 1$, then $z = 1$ and if $x=0$, then $z=0$. We therefore clearly cannot conclude that $x$ and $z$ are probabilistically independent, so an application of \ruleref{Frame} around a specification for $z \gets [0,1]$ is unsound, undoing the purported progress from \Cref{sec:overview-sep}.

The crux of the issue is that adversarial choice is an observable
effect (it increments the tape index), and
therefore information can leak to the adversary through the
program's control flow in precisely the same way that implicit flows
leak confidential data in information-flow control \cite{denning1976lattice}.
To control for leaks, \LogicName expresses not only the ways that variables are distributed, but also whether each source of randomness remains \emph{private} from the adversary. The new assertion forms below both state that $a$ and $b$ are distributed like fair coins, but whereas $a$ is unknown to the adversary, $b$ may have been leaked, so future adversarial decisions may be correlated with $b$.
\begin{mathpar}
  a \overset{\Priv}{\sim} \bern{\tfrac12}
  
  b \overset{\Leak}{\sim} \bern{\tfrac12}
\end{mathpar}
In a nondeterministic program, the \ruleref{Frame} rule can only augment the specification with information about private sources of randomness, precluding an unsound application of \ruleref{Frame} in the example above. Instead, that program can be handled via case analysis on $x$, similar to \dol and \pcol. Since \dol and \pcol effectively consider every source of randomness to be leaked, they are consistent with the rules of \LogicName, but \LogicName also provides increased flexibility for dealing with private sources of randomness.
Despite similarities to \dol and \pcol, the metatheory of \LogicName was designed from scratch to incorporate tape consumption as a first-class construct. Although tape indices must be tracked precisely, they are only surfaced through the logic to track leaked entropy sources, thereby retaining usability and enabling verification of new kinds of programs.


\section{Syntax and Semantics}
\label{sec:semantics}

We now describe the syntax and semantics of the command language used in \LogicName.
The syntax, shown in \Cref{fig:syntax}, includes standard imperative constructs,
such as no-ops, assignments,
sequential composition, conditionals, and while-loops, together with
probabilistic and adversarial actions. 

\begin{figure}
\begin{align*}
\Cmd \ni \cmd \Coloneqq
&~ \skp
  \mid a
  \mid \cmd_1 \fatsemi\cmd_2
  \mid \cmd_1 \oplus_\expr \cmd_2
  \mid \cmd_1 \nd \cmd_2
  \mid \cmdcase{\expr}{\cmd_1}{\cmd_2}
  \mid \cmdwhile{\expr}{\cmd}
  \\
\Act \ni a \Coloneqq
&~ \actassign{\exprvar}{\expr}
  \mid \exprvar \samp d(\expr)
  \mid \exprvar \gets \expr
  \\
  \Expr \ni \expr \Coloneqq
  &~ \exprval
    \mid \exprvar
    \mid e_1 + e_2
    \mid e_1 - e_2
    \mid e_1 \le e_2 \mid \cdots
    \tag{\textit{$x\in\Var$, $v \in \Val$}}
  \\
  \textsf{Distr} \ni d \Coloneqq
  &~ \bern{-} \mid \unif{-} 
\end{align*}
\caption{Syntax of program commands $\cmd \in \Cmd$.}
\label{fig:syntax}
\end{figure}

The command $\cmd_1 \oplus_e\cmd_2$ executes $\cmd_1$ with probability
determined by the expression $e$, and $\cmd_2$ with probability $1-e$. By
contrast, in $\cmd_1 \nd \cmd_2$, the adversary chooses whether to execute
$\cmd_1$ or $\cmd_2$. Similarly, $x \samp d(e)$ assigns to $x$ a random value
sampled from the
distribution $d(e)$, whereas $x \gets e$ assigns to $x$ an
adversarially chosen value from the list $e$. A distribution $d$ is either Bernoulli
$\bern{e}$, which returns 1 with probability $e$ and 0 with
probability $1-e$, or uniform $\unif{e}$, which assigns
equal probability to each element of the set $e$. For integer-valued
expressions $n$ and $m$, we let $\unif{n,m} = \unif{\{n, n+1, \ldots,
m\}}$.  Values $\exprval \in \Val$ include integers, rationals, and lists.
Expressions
$\expr \in \Expr$ range over values, variables, and arithmetic,
list, and boolean operations.

\subsection{Preliminaries}

We begin by establishing some preliminary concepts and definitions.

\paragraph{Partial Memories.}

Let $\Mem \triangleq \Var \partialto_\fin \Val$ be the set of \emph{memories}:
finite partial maps from variables to values.
For $\sigma \in \Mem$, let $\dom(\sigma)$ be its domain.
For a finite set of variables $V \subseteq_\fin \Var$, let
$\Mem[V] \triangleq \{\sigma \in \Mem \mid \dom(\sigma) \subseteq V\}$
be the set of memories whose domain is contained in $V$.
The memory order $\mathord\sqsubseteq \subseteq \Mem\times\Mem$ indicates that one memory contains more information than another:
$\sigma \sqsubseteq \tau 
$ iff $
   \dom(\sigma) \subseteq \dom(\tau)
$ and $\sigma(x) = \tau(x)$ for all $x \in \dom(\sigma)$.

Two memories $\sigma, \tau \in \Mem$ are disjoint $\sigma
\Perp \tau$ iff $\dom(\sigma) \cap \dom(\tau) = \emptyset$. When $\sigma \Perp
\tau$, their disjoint union is defined as $(\sigma \uplus \tau)(x) = \sigma(x)$ if $x \in \dom(\sigma)$ and $\tau(x)$ otherwise.
Given $\sigma \in \Mem$ and a finite set of variables
$S \subseteq_\fin \Var$, memory projection is defined as $(\ForgetMem{\sigma}{S})(x) = \sigma(x)$ if $x \in \dom(\sigma)\cap S$.
%

\paragraph{Discrete Probability Distributions.}
Discrete probability distributions $\mu \in \calD(X)$ over a countable set $X$ are mappings from elements
of $X$ to the interval $[0, 1] \subseteq \mathbb{R}$ such that $\sum_{x \in X} \mu(x) = 1$.
The program semantics uses distributions over $\Mem \cup \{\bot\}$, that is, memories $\sigma\in\Mem$ and nontermination $\bot$.
As a shorthand, we write $\calD_\bot(X) \triangleq \calD (X \cup \{\bot\})$.

The Dirac or point-mass distribution $\delta_x$ assigns probability 1 to $x$ and 0 to everything else.
The support of a distribution is the set of elements with nonzero probability $\supp(\mu) \triangleq \{ x \mid \mu(x) > 0 \}$.
For a collection of distributions $(\mu_i)_{i\in I}$ such that $\mu_i\in\calD(X)$ for all $i\in I$ and a discrete distribution $\xi\in\calD(I)$, let
$(\bigoplus_{i\sim\xi}\mu_i)(x) \triangleq\sum_{i\in\supp(\xi)}\xi(i)\cdot\mu_i(x)$ be the ordinary convex mixture.

\paragraph{Monads.} Monads encapsulate sequential composition of effectful computations \cite{moggi91,monads}. We say that $T$
is a monad if there is a \emph{Kleisli Triple} $\tuple{T, \eta, (\cdot)^\dagger}$ such that the unit $\eta \colon X \to T(X)$ and Kleisli composition $(\cdot)^\dagger \colon (X \to T(Y)) \to T(X) \to T(Y)$ operators obey the laws $\eta^\dagger = \mathsf{id}$, $f^\dagger\circ \eta = f$, and $f^\dagger \circ g^\dagger = (f^\dagger \circ g)^\dagger$.
Discrete distributions, $\D_\bot$ are a monad with the operations
$\eta(x) \triangleq \delta_x$ and $f^\dagger(\mu) \triangleq \bigoplus_{x \sim \mu} f_\bot(x)$ where
$f_\bot(x) \triangleq f(x)$ if $x \in X$ and $f_\bot(\bot) = \delta_\bot$.
%
Consider a semantics for probabilistic programs with type $\de{C}_\D \colon \Mem \to \D_\bot(\Mem)$. The $(\cdot)^\dagger$ operator will allow us to define the semantics of $C_1\fatsemi C_2$ compositionally.
Furthermore, the monad laws ensure that sequential composition is well-behaved: $C\fatsemi \skp \equiv \skp\fatsemi C \equiv C$ and $(C_1 \fatsemi C_2)\fatsemi C_3 \equiv C_1 \fatsemi (C_2\fatsemi C_3)$. 
For a more comprehensive introduction to monads and category theory, refer to \citet{awodey2006category}.

\paragraph{Domain Theory.} Domain Theory gives a toolset for defining infinitely iterating computations \cite{scott1970outline,scott1971towards}. Given a partially ordered set (poset) $\tuple{X, \le}$,
a subset $D\subseteq X$ is called \emph{directed} if every pair of elements $x,y\in D$ has an upper bound $z \in D$.
The supremum $\sup S$ is the least upper bound of the set $S$.
A poset $\tuple{X, \le}$ is a \emph{directed complete partial order} (DCPO) if every directed set $D$ has a supremum $\sup D \in X$. If $X$ also has a least element $\bot_X$, then $X$ is a \emph{pointed} DCPO. For any two DCPOs $\tuple{X, \le_X}$ and $\tuple{Y, \le_Y}$, $f \colon X \to Y$ is \emph{Scott continuous} if it is monotone and preserves suprema of directed sets
$
  \sup_{x \in D} f(x) = f(\sup D)
$.
The Kleene Fixed Point Theorem states that if $f \colon X \to X$ is a Scott continuous function on a pointed DCPO $\tuple{X, \le}$, then the least fixed point $\mathsf{lfp}(f)$ exists and is given by
$\sup_{n\in\mathbb N} f^n(\bot)$ where $f^0 = \mathsf{id}$ and $f^{n+1} = f\circ f^n$.
Refer to \citet{abramsky1995domain} for a more complete introduction to Domain Theory.

%

\subsection{The Oblivious Monad}

Building on the preliminaries from the previous section, we now present
our new computational domain $\calO(X)$, the \emph{oblivious monad}, for interpreting probabilistic
programs in the oblivious adversary model. It is constructed as a discrete
probability monad $\D_\bot$ transformed with a \emph{reader} monad to propagate the
immutable scheduler tape and a \emph{state} monad to track the
number of scheduler decisions that have been consumed.
Concretely, an oblivious computation $o \in \calO(X)$ takes a scheduler tape $s \in
\schedule$ and an index $n \in \bbN$ and returns a distribution over
pairs $(x, m)$ where $x \in X$ is an output value and $m \in \bbN$ is
the new scheduler index.
\begin{mathpar}
  \calO(X) \triangleq \schedule \to \mathbb{N} \to \calD_\bot(X \times \mathbb{N})
  
  \text{where}
  
  \schedule \triangleq \bbN \to \bbN
\end{mathpar}
We complete the monad structure for $\calO$ by defining the
operations. The monad unit $\eta \colon X \to \calO(X)$ corresponds to a no-op
computation: given a schedule $s$ and an index $n$, 
$\eta(x)$ is the
Dirac distribution concentrated at $(x,n)$.
The Kleisli composition $(\cdot)^\dagger \colon (X \to \calO(Y)) \to
\calO(X) \to \calO(Y)$ sequences an oblivious computation with a continuation.
For fixed schedule $s$, define
\[
  f_s^\bot(y) \triangleq
  \begin{cases}
    f(x)(s)(k) & \text{if } y=(x,k)\in X\times\bbN,\\
    \delta_\bot & \text{if } y=\bot .
  \end{cases}
\]
Then the monad operations are below, where $(\cdot)^\dagger_{\D_\bot}$ is the Kleisli composition of the
$\D_\bot$ monad.
\begin{mathpar}
\big(\eta(x)\big)(s)(n) \triangleq \delta_{(x, n)}

\big(f^\dagger(o)\big)(s)(n)
\triangleq
(f_s^\bot)^\dagger_{\D_\bot}(o(s)(n))
=
\textstyle\bigoplus_{y\sim o(s)(n)} f_s^\bot(y).
\end{mathpar}
Elements of the domain $\calO(X)$ are ordered by the pointwise extension of the typical distribution order $\mu \led \nu$ iff $\mu(x) \le \nu(x)$ for all $x\in X$. Note that $\mu\led\nu$ implies that $\nu(\bot) \le \mu(\bot)$; as a distribution becomes larger, the probability of nontermination \emph{decreases}.
\[
  o_1 \leo o_2
  \qquad\text{iff}\qquad
  \forall s \in \schedule.\quad \forall n \in \bbN.\quad
    o_1(s)(n) \led o_2(s)(n)
\]
The poset $\tuple{\D_\bot(X), \led}$ is a discrete variant of the probabilistic powerdomain \cite{jones1989probabilistic}, and is thus a DCPO where suprema of directed sets are given by $(\sup D)(x) \triangleq \sup_{\mu\in D}\mu(x)$ if $x\in X$ and $(\sup D)(\bot) \triangleq \inf_{\mu\in D}\mu(\bot)$, which maximizes the probability of all $x\in X$ and minimizes the probability of nontermination.
The pointwise extension of a pointed DCPO is also a
pointed DCPO~\cite[Proposition 2.1.18]{abramsky1995domain};
therefore $\tuple{\calO(X), \leo}$ is a pointed DCPO with bottom $\bot_\calO(s)(n) = \delta_\bot$.

We conclude by defining additional nondeterministic and probabilistic choice operators in $\calO$. Let $I \subseteq\bbN$ be a consecutive index set, \ie $I = \{0, 1, 2, \ldots, n\}$ if $I$ is finite and $I = \bbN$ if it is infinite.
The nondeterministic choice operator $\bignd_{i\in I} o_i$ picks an
element using the $n$th entry in the scheduling tape. If the set of
choices $I$ is infinite, then $s(n)$ is used directly; otherwise, it
is taken modulo $|I|$.
\[
  \Big(\bignd_{i \in I} o_i\Big)(s)(n) \triangleq
    \begin{cases}
      o_{(s(n) \bmod |I|)}(s)(n+1) & \text{if}~ |I| < \infty
      \\
      o_{s(n)}(s)(n+1) & \text{if}~ |I| = \infty
    \end{cases}
\]
To continue the computation, index $n$ is incremented so
that the next adversarial choice in $o_i$ uses the next entry in the
tape.
We also let $\bignd_{i=0}^n o_i = \bignd_{i\in\{0,\ldots, n\}} o_i$.
The ability to make countable choices gives $\calO$ an edge over other monads
for probabilistic nondeterminism
\cite{mciver2005abstraction,zilberstein2025demonic,tix2009semantic,keimel2017mixed}. Indeed,
\citet{apt1986countable} established that powerdomain approaches, which
represent nondeterminism with sets, can only support bounded
nondeterminism. Using a tape of choices avoids this obstruction at the semantic
level by using $s(n)$ directly for countable choices, with finite choices as a
special case.

Given a distribution $\mu \in \D(I)$ over a countable index set $I$, probabilistic choices $\bigoplus_{i \sim\mu} o_i$ are defined below as a convex combination of the random branches according to $\mu$. Unlike nondeterministic choices, probabilistic ones do not increment the scheduling index $n$, since they do not consume any adversary decisions. We also define binary versions of the above operators as syntactic sugar.
\begin{mathpar}
  \Big(\bigoplus_{i \sim \mu} o_i\Big)(s)(n) \triangleq \bigoplus_{i\sim \mu} o_i(s)(n)

o_0 \nd o_1 \triangleq \bignd_{i \in \{0, 1\}} o_i

o_1 \oplus_p o_0 \triangleq \bigoplus_{i \sim \bern{p}} o_i
\end{mathpar}
\subsection{Denotational Semantics}

We are now ready to define the semantics of the command language from \Cref{fig:syntax} in terms of the operations from the previous section. The semantics $\de{\cdot} \colon \Cmd \to \Mem \to \calO(\Mem)$ for commands is shown in \Cref{fig:cmd-denotational}. Given a program $C$ and an input memory $\sigma\in\Mem$, the semantics $\de{C}(\sigma)$ gives us an oblivious monad computation over output memories. By further supplying a schedule $s \in \schedule$ and a start index $n \in \bbN$, we get $\de{C}(\sigma)(s)(n) \in \D_\bot(\Mem \times \bbN)$, a distribution over output memories and scheduling indices. We now explain the semantics of each command.

\begin{figure}
  \centering
\begin{align*}
  \de{\cmdskip}(\sigma)
  &\triangleq \eta(\sigma)\\
  \de{\cmdseq{C_1}{C_2}}(\sigma)
  &\triangleq
    \de{C_2}^\dagger(\de{C_1}(\sigma))\\
  \de{\cmdnondet{C_1}{C_2}}(\sigma)
  &\triangleq
    \de{C_1}(\sigma) \nd \de{C_2}(\sigma)\\
  \de{\cmdchoice{e}{C_1}{C_2}}(\sigma)
  &\triangleq \de{C_1}(\sigma) \oplus_{\de{e}_\Exp(\sigma)}
    \de{C_2}(\sigma)\\
  \de{\cmdcase{e}{C_1}{C_2}}(\sigma)
  &\triangleq 
    \begin{cases}
      \de{C_1}(\sigma)
      &\text{if}\; \de{e}_\Exp(\sigma) = \Ttrue\\
      \de{C_2}(\sigma)
      &\text{if}\; \de{e}_\Exp(\sigma) = \Tfalse\\
    \end{cases}\\
  \de{x \coloneqq e}(\sigma)
    &\triangleq \eta(\sigma[x \coloneqq \de{e}_\Exp(\sigma)]) \\
  \de{x \samp d(\vec e)}(\sigma)
    &\triangleq \textstyle\bigoplus_{v \sim d(\de{\vec e}_\Exp(\sigma))} \eta(\sigma[x \coloneqq v]) \\
  \de{x \gets e}(\sigma)
    &\triangleq \textstyle\bignd_{i = 0}^{|\de{e}_\Exp(\sigma)|-1} \eta(\sigma[x \coloneqq \de{e}_\Exp(\sigma)[i]]) \\
  \de{\cmdwhile{e}{C}}(\sigma)
  &\triangleq \textsf{lfp}(\Phi_{\langle C, e \rangle})(\sigma)\\
    \text{where}\;\Phi_{\langle C, e \rangle}(f)(\sigma) &\triangleq
        \begin{cases}
        f^\dagger (\de{C}(\sigma))&\text{if}\;\de{e}_\Exp(\sigma) = \Ttrue\\
        \eta(\sigma)&\text{if}\;\de{e}_\Exp(\sigma) = \Tfalse
        \end{cases}
\end{align*}
  \caption{Denotational Semantics for programs
    $\de{\cdot} \colon \Cmd \to \Mem \to \calO(\Mem)$.
  }
  \label{fig:cmd-denotational}
\end{figure}

As is typical in monadic semantics, $\skp$ and $C_1\fatsemi C_2$ are interpreted using the monad unit and Kleisli composition, respectively. Nondeterministic choice $C_1 \nd C_2$ is interpreted using the binary nondeterminism operator in $\calO(\Mem)$, defined in the prior section. Similarly, probabilistic choice uses the binary probabilistic operator $- \oplus_p -$, where the parameter $p$ is given by the expression $e$. Expressions are interpreted using $\de{\cdot}_\Exp \colon \Exp \to \Mem \to \Val$, which is defined in the obvious way.
If statements use the truth of the guard $e$ to decide which branch to execute.

Deterministic assignment $x \coloneqq e$ updates $x$ to $\de{e}_\Exp(\sigma)$ in the input memory $\sigma$. Probabilistic sampling $x \samp d(\vec e)$ is defined by first interpreting the input parameter list $\vec e$ to $d$, and then performing a probabilistic choice over the resulting distribution and assigning $x$ to the value in each branch. Adversarial choice $x\gets e$ is defined similarly, but uses $\bignd$ to determine which index $i$ to take from the list $e$.
Finally, the semantics of a loop is the least fixed point of the characteristic
function $\Phi_{\langle C, e \rangle}$ from \Cref{fig:cmd-denotational}.
Letting $f = \de{\whl eC}$, we get the standard loop-unrolling equation:
\[
  \Phi_{\langle C, e \rangle}(\de{\whl eC}) = \de{\iftf e{C \fatsemi (\whl eC)}\skp}
\]
A fixed point of $\Phi_{\langle C, e \rangle}$
thus represents a loop semantics, since $\Phi_{\langle C, e
  \rangle}$ amounts to one loop unrolling.
To prove that such a fixed point exists, it suffices to show that $\Phi_{\langle C,
  e \rangle}$ is Scott continuous. Since $\leo$ is a pointwise extension of
$\led$, Scott continuity follows from Lemma C.9 of \citet{li2025total}.



\subsection{Connections to Adaptive Semantics}

The semantics in \Cref{fig:cmd-denotational} clarifies the relationship between \LogicName and prior logics for adaptive nondeterminism. In \dol and \pcol, nondeterminism is
modeled using the convex powerdomain $\C$, which consists of sets of distributions with extra constraints to
ensure that $\C$ is a monad and a DCPO \cite{zilberstein2025demonic,jifeng1997probabilistic,mciver2005abstraction}.
Whereas in $\C$, nondeterminism can be resolved differently across
random branches, $\calO$ ranges over fixed
schedules.

Thus, after erasing scheduler indices from final distributions, every
oblivious behavior is included in the adaptive semantics.
To show that result formally, let $\de{\cdot}_{\C} \colon \Cmd \to \Mem \to \C(\Mem)$ be the (adaptive) convex
 powerdomain semantics from Figure 1 of \citet{zilberstein2025demonic} and $\mathsf{fst} \colon X\times Y \to X$ be the first projection of a pair.
\begin{theorem}[Refinement of Semantics]
  \label{thm:oblivious-refines-convex-powerset}
  For every $C \in \Cmd$ and $\sigma \in \Mem$,
  \[
    \left\{
      (\eta \circ \mathsf{fst})^\dagger\big(\de{C}(\sigma)(s)(n)\big)
      \mid
      s \in \schedule,\ n \in \bbN
    \right\}
    \quad\subseteq\quad
    \de{C}_{\C}(\sigma).
  \] 
\end{theorem}
The inclusion in \Cref{thm:oblivious-refines-convex-powerset} shows that
adaptive proof rules are sound in the oblivious setting since every oblivious
behavior is contained in $\de{\cdot}_\C$. But the inclusion can be strict;
$\de{\cdot}_\C$ may introduce correlations that are impossible
under any fixed scheduling tape, allowing \LogicName{} to exploit a higher degree of independence.
For example, in the program $x \samp \bern{\frac12}\fatsemi y\gets[0,1]$
from \Cref{sec:overview-semantics}, an adaptive scheduler can force $x=y$,
whereas every fixed tape yields $\PP[x=y]=\frac12$.



\section{Oblivious Probabilistic Outcome Assertions}
\label{sec:assertions}

We now present an assertion language used to describe pre- and postconditions of
programs, inspired by \pcol \cite{zilberstein2026probabilistic} but
 extended to express when randomness has been leaked to the adversary.
A source of randomness is leaked if it might not be
probabilistically independent from the number of consumed scheduling
bits.

\subsection{Pure Separation Logic Assertions}

The following \emph{pure} assertions for describing deterministic state are the same as the ones used in \pcol and are inspired by standard separation logic \cite{sl,localreasoning}.
\begin{align*}
  P
  &::= \true \;|\;
    \false \;|\;
    P \land Q \;|\;
    P \lor Q \;|\;
    \exists X.\ P \;|\;
    P * Q \;|\;
    e \mapsto E \;|\;
    E_1 \asymp E_2
    && \big( \mathord\asymp \in \{ =, \le, \ldots \} \big)
    \\
  E
  &::= X \;|\;
    v \;|\;
    E_1 + E_2 \;|\;
    E_1 \cdot E_2 \;|\;\dots
\end{align*}
Pure assertions use \emph{logical expressions} $E \in \LExpr$, which
mirror standard expressions from \Cref{fig:syntax} but range over logical variables $X \in \LVar$ rather than program variables
$x \in \Var$.
The semantics of pure assertions is given in \Cref{app:assertions} but we cover some key points here.
Pure assertions are satisfied by a context $\Gamma :
\LVar \rightharpoonup \Val$,
which assigns values to logical variables,
and a memory $\sigma \in \Mem$. We write $\Gamma,\sigma\vDash P$ to
mean that $P$ holds under $\Gamma$ and $\sigma$.  Assertions $\true$, $\false$, $P\land Q$,
and $P\lor Q$ are standard. Existential quantification
$\exists X.\ P$ is satisfied when there exists a value $v \in \Val$ such that 
$\Gamma[X\coloneqq
v],\sigma\vDash P$.
A comparison $E_1 \asymp E_2$ holds when
$\de{E_1}_\LExpr(\Gamma)
\asymp \de{E_2}_\LExpr(\Gamma)$.

The remaining two assertion forms are inspired by separation logic. The
\emph{separating conjunction} $\Gamma,\sigma\vDash P\sep Q$ holds iff $\sigma$
can be decomposed into $\sigma_1 \uplus\sigma_2 \sqsubseteq \sigma$ such that
$\Gamma,\sigma_1\vDash P$ and $\Gamma,\sigma_2\vDash Q$. That definition
guarantees that $P$ and $Q$ describe \emph{disjoint} variables. The
\emph{points-to} predicate $\Gamma,\sigma\vDash e \mapsto E$ provides a way to
connect logical and physical state; it holds when $\de{e}_\Expr(\sigma) =
\de{E}_\LExpr(\Gamma)$.
We also define the following notation for the extensional semantics of pure assertions.
\begin{mathpar}
  \sem{P}_\Gamma^S
  \triangleq \{\sigma \in \Mem[S] \;|\; \Gamma, \sigma \vDash P\}

  \sem{P}_\Gamma
  \triangleq \{\sigma \in \Mem \;|\; \Gamma, \sigma \vDash P\}
\end{mathpar}
The following syntactic sugar expresses ownership of variables without mentioning their values.
\begin{mathpar}
  {e \mapsto -} ~\triangleq~ \exists X.\ e \mapsto X

  \own(e_1, \dots, e_n)
  ~\triangleq~
    {e_1 \mapsto -} \sep \cdots \sep {e_n \mapsto -}
\end{mathpar}

\subsection{Modeling Outcome Assertions with Probability Spaces}
\label{sec:probability-spaces}

Lilac first showed that modeling probabilistic assertions with probability
spaces enables reasoning about independence at varying levels of granularity,
and is thus advantageous even in discrete settings \cite{li2023lilac}.
As we saw in \Cref{sec:overview-sep}, assertions of the form $x \sim
\bern{\tfrac12}\sep y\sim\bern{\tfrac12}$ imply that $x$ and $y$ are
probabilistically independent. In that case, the probability of a joint event
factors into the product of its marginals
$\PP[x = 1 \land y = 1] = \PP[x = 1]\cdot \PP[y=1] =
\frac12\cdot\frac12 = \frac14$.
However, separating conjunctions $\varphi \sep \psi$ only require independence of the
events that $\varphi$ and $\psi$ describe, not their entire distributions.
For example, consider the following program.
\begin{equation}\label{eq:prob-space-example}
\triple{\sure{\own(x, y)}}{\big(x \coloneqq 0 \oplus_\frac12 x \gets [1] \big) \fatsemi y \gets[0, 1]}
{x\overset{\Leak}{\sim}\bern{\tfrac12}\sep\sure{y\in \{0, 1\}}}
\end{equation}
The assignments in the two probabilistic branches consume different numbers of
scheduler entries, so the value of $x$ is leaked through the scheduling index.
As a result, the later adversarial choice of $y$ may be correlated with $x$.
However, because probability spaces track probabilities of \emph{measurable
  events} rather than individual samples, we can coarsen the assertion about $y$ to the
probability-one event $\sure{y \in \{0, 1\}}$, and describe the output
distribution via a separating conjunction. Independence still holds because
$\PP[x = 1 \land y\in \{0, 1\}] = \frac12 = \PP[x=1] \cdot \PP[y \in\{0, 1\}]$.

The assertion semantics is formalized based on measure theory \cite{royden1968real,fremlin2001measure}.
Let $\Omega$ be a countable \emph{sample space}. An \emph{event} $E \subseteq \Omega$ is a set of samples.
A $\sigma$-algebra over $\Omega$ is a set of \emph{events} $\calF \subseteq 2^\Omega$ such that $\Omega \in \calF$ and $\calF$ is closed under complementation and countable unions. A \emph{probability measure} $\mu \colon \calF \to [0,1]$ is a function assigning probabilities to events such that $\mu(\Omega) = 1$ and $\mu$ obeys countable additivity: $\mu(\bigcup_{i\in I} E_i) = \sum_{i\in I}\mu(E_i)$ for any countable collection of pairwise disjoint events $(E_i)_{i\in I}$.
A probability space $\calP = \langle \Omega_\calP, \calF_\calP,
\mu_\calP\rangle$ consists of a sample space $\Omega_\calP$, a
$\sigma$-algebra $\calF_\calP \subseteq 2^{\Omega_\calP}$ as an
event space, and a probability measure $\mu_\calP \colon \calF_\calP \to
[0, 1]$.

We consider only complete probability spaces:
$\calF$ contains all subsets of $\mu$-null sets.
Given a probability space $\calP$ and an event $G \in \calF_\calP$
with $\mu_\calP(G)>0$, the conditional probability space is defined as
$\calP \mid G
\triangleq
\langle G, \calF_{\calP \mid G}, \mu_{\calP \mid G} \rangle$
where
$\calF_{\calP \mid G}
\triangleq \{H \in \calF_\calP \mid H \subseteq G\}$
and 
$\mu_{\calP \mid G}(H)
\triangleq \frac{\mu_\calP(H)}{\mu_\calP(G)}$ for any $H \in \calF_{\calP \mid G}$. If $\calP$ is complete, $\calP
\mid G$ is complete.

The semantic object underlying an outcome assertion is a \emph{resource}, which
packages the random memory state together with the scheduler count information
that the adversary may observe.
\begin{definition}[Resources]\label{def:resource} A resource $\calR$ is a 4-tuple, as shown below on the left.
\begin{mathpar}
\calR \triangleq \tuple{
  \calP_\calR,
  V_\calR,
  \memfn_\calR,
  \cntfn_\calR
}

\obsfn_\calR(\omega) \triangleq (\memfn_\calR(\omega), \cntfn_\calR(\omega))
\end{mathpar}
Here,
$\calP_\calR=\langle\Omega_\calR,\calF_\calR,\mu_\calR\rangle$
is a complete probability space over $\Omega_\calR \subseteq \bbN$;
$V_\calR\subseteq_\fin\Var$ is a finite memory footprint; and
$\memfn_\calR\colon\Omega_\calR\to\Mem[V_\calR]$ and
$\cntfn_\calR\colon\Omega_\calR\to\bbN$ assign a memory and a
scheduler count to each index. The count map $\cntfn_\calR$ is
measurable in the discrete $\sigma$-algebra, \ie $\cntfn_\calR^{-1}(n) \in \calF_{\calP_\calR}$ for all $n\in \mathbb{N}$. In addition, $\obsfn_\calR \colon \Omega_\calR \to\Mem[V_\calR]\times\bbN$ is the joint observation.
\end{definition}

To see how the event space controls which facts about observations are exposed, we return
to the triple (\ref{eq:prob-space-example}). The assertion $\sure{y\in\{0,1\}}$ can be
witnessed by a coarse resource specifying $y$.
\begin{mathpar}
\Omega = \{0, 1\}

\calF=\{\emptyset,\Omega\}

\begin{array}{r}
\mu(\emptyset) = 0
\\
\mu(\Omega) = 1
\end{array}

\begin{array}{c|cc}
  \omega & \memfn(\omega)(y) & \cntfn(\omega) \\
  \hline
  0 & 0 & 0 \\
  1 & 1 & 0
\end{array}
\end{mathpar}
According to $\calF$, $y\mapsto v$ is not measurable for any $v \in \Val$, but
the event $y\in\{0,1\}$ is measurable and has probability one.
Since $\cntfn_\calR$ is measurable, we can define its \emph{pushforward}
$\Proj{\calR}{\bbN}$ to
obtain a probability space over scheduler count, which is also a proper discrete
probability distribution
\[
  \Proj{\calR}{\bbN} \triangleq \tuple{\bbN, 2^\bbN, \mu}
  \qquad\text{where}\qquad
  \mu(E) \triangleq \mu_\calR\left( \cntfn_\calR^{-1}(E) \right) \;
  \text{for all}\; E \subseteq \bbN
\]
For any given resource $\calR$, the specific identifiers used in $\Omega_\calR$ are not important, and can thus be permuted.
We therefore identify resources up to the following isomorphism.
\begin{definition}[Resource Isomorphism]
Given two resources $\calR$ and $\calR'$ with $V_\calR = V_{\calR'}$, a \emph{resource isomorphism} is a bijection $\rho\colon\Omega_\calR\to\Omega_{\calR'}$ such that:
\begin{mathpar}
  \calF_{\calR'} = \{\rho(A)\mid A\in\calF_\calR\}

  \mu_{\calR'}(\rho(A)) = \mu_\calR(A)
  \quad(A\in\calF_\calR)
%
%
%

  \obsfn_{\calR'}\circ\rho = \obsfn_\calR
\end{mathpar}
We write $\calR\cong\calR'$ when such a $\rho$ exists; $\cong$ is an
equivalence relation, and throughout this paper we identify resources up to
index renaming $\cong$.
\end{definition}

\paragraph{Ordering and Refinement}
Consider $\calR$ and $\calR'$ such that $V_\calR=S\subseteq V_{\calR'}$.
The order $\calR \preceq \calR'$ means that $\calR'$ may contain more memory
information and more measurable events, while preserving scheduler-count
observations; $\calR$ is a coarser view of $\calR'$, as in the coarse
$y$-resource above.
We define $\calR\preceq\calR'$ iff there is a measurable map
$h:\Omega_{\calR'}\to\Omega_\calR$ such that for all
$A \in \calF_\calR$ and $\omega \in \Omega_{\calR'}$:
\begin{mathpar}
  \mu_\calR(A)
  = \mu_{\calR'}(h^{-1}(A))

  \memfn_\calR(h(\omega))
  \sqsubseteq \ForgetMem{\memfn_{\calR'}(\omega)}{S}

  \cntfn_\calR(h(\omega))
  = \cntfn_{\calR'}(\omega)
\end{mathpar}
As in intuitionistic separation logics, \LogicName{} permits weakening of
memory ownership: $\calR$ may forget memory information present in
$\calR'$. Scheduler counts track tape consumption, so the count condition is
strict (\ie $\calR \preceq \calR' \Rightarrow \forall n\in\bbN.\;
  \mu_\calR(\cntfn_\calR^{-1}(\{n\}))
  =
  \mu_{\calR'}(\cntfn_{\calR'}^{-1}(\{n\}))$).
Consider a resource $\calR$ with a proper discrete distribution
$\mu\in\D(\Mem\times\bbN)$.  The resource is a refinement of the distribution
(i.e., $\calR\preceq\mu$)\footnote{When the right-hand side is a partial
  distribution $\mu \in \D_\bot(\Mem \times \bbN)$, the judgment is defined only
when $\mu(\bot) = 0$.} iff there is a
coupling 
$\lambda\in\calD(\Omega_\calR\times(\Mem \times \bbN))$
such that:
\begin{mathpar}
  \forall A\in\calF_\calR.\ \lambda(A\times(\Mem \times \bbN)) = \mu_\calR(A)
    
  \forall  (\sigma, n)\in\Mem \times \bbN.\ \lambda(\Omega_\calR\times \{(\sigma, n)\}) = \mu(\sigma, n)

  \lambda(\omega,(\sigma,n))>0 \implies
    \memfn_\calR(\omega)\sqsubseteq\ForgetMem{\sigma}{V_\calR}
    \land
    \cntfn_\calR(\omega)=n .
\end{mathpar}
Equivalently, the coupling $\lambda$ is supported on pairs whose resource
observation agrees with the ordinary memory-count observation after projecting
memory to $V_\calR$.

\subsection{Outcome Assertions}
Outcome assertions describe the pre- and postconditions of
\LogicName triples.
The syntax is shown below, where $E \in \LExpr$ and $o\in\{\Priv,\Leak\}$ is the probabilistic
choice mode.
\begin{align*}
  \textsf{Prop} \ni \varphi
  &::= \top\;|\;
    \bot\;|\;
    \varphi \land \psi\;|\;
    \varphi \lor \psi \;|\;
    \exists X.\ \varphi\;|\;
    \bigoplus^o_{X \sim d(E)} \varphi \;|\;
    \bignd_{X \in E}\; \varphi \;|\;
    \sure{P}\;|\;
    \varphi \sep \psi
\end{align*}
The validity of outcome assertions (\Cref{fig:assertion-semantics}) is defined relative to an environment
$\Gamma:\LVar\to\Val$ and a resource
$\calR=\langle \calP_\calR,V_\calR,\memfn_\calR,\cntfn_\calR\rangle$
satisfying the completeness and count-measurability assumptions 
from \Cref{def:resource}.
Assertions $\top$, $\bot$, $\wedge$, $\vee$, and $\exists$ have the usual
semantics.

\begin{figure}
\begin{align*}
  \Gamma, \calR \vDash \top \quad
  &\text{always}\\
  \Gamma, \calR \vDash \bot \quad
  &\text{never}\\
  \Gamma, \calR \vDash \varphi \land \psi \quad
  &\text{iff}\quad \Gamma, \calR  \vDash \varphi \quad
    \text{and}\quad \Gamma, \calR  \vDash \psi\\
  \Gamma, \calR  \vDash \varphi \lor \psi \quad
  &\text{iff}\quad \Gamma, \calR  \vDash \varphi \quad
    \text{or}\quad \Gamma, \calR  \vDash \psi\\
  \Gamma, \calR \vDash \exists X.\ \varphi \quad
  &\text{iff}\quad \Gamma[X \coloneqq v],\calR\vDash\varphi \quad \text{for some} \quad v \in \Val\\
  \Gamma, \calR  \vDash \bigoplus^o_{X \sim d(E)} \varphi \quad
  &\text{iff}\quad
    \exists (\calR_v)_{v \in \supp(\mu)}.\;
    \bigoplus_{v \sim \mu} \calR_v \preceq \calR
    \quad \text{and} \quad
    \forall v\in\supp(\mu).\;
      \Gamma[X \mapsto v], \calR_v \vDash \varphi\\
  &\phantom{\text{iff}\quad}\text{and} \quad
    \Obs_o^\mu(\calR,(\calR_v)_{v\in\supp(\mu)})
    \quad\text{where} \quad
    \mu = d(\de{E}_{\LExp}(\Gamma))\\
  \Gamma, \calR  \vDash \bignd_{X \in E}\varphi \quad
  &\text{iff} \quad
    \Gamma, \calR  \vDash \bigoplus^{\Leak}_{X \sim \mu}\varphi
    \quad \text{for some} \quad \mu \in \calD(\de{E}_{\LExp}(\Gamma))
  \\
  \Gamma, \calR  \vDash \sure{P} \quad
  &\text{iff} \quad
    \memfn_\calR^{-1}(\sem{P}_\Gamma^{V_\calR}) \in \calF_\calR
    \quad \text{and} \quad
    \mu_\calR(\memfn_\calR^{-1}(\sem{P}_\Gamma^{V_\calR})) = 1\\
  \Gamma, \calR  \vDash \varphi \sep \psi \quad
  &\text{iff} \quad \exists \calR_1, \calR_2.\;
    \calR_1 \otimes \calR_2 \preceq \calR
    \quad \text{and} \quad \Gamma, \calR_1 \vDash \varphi
    \quad \text{and} \quad \Gamma, \calR_2 \vDash \psi
\end{align*}
\caption{Semantics of outcome assertions expressed as satisfaction
  relation. $\Gamma: \LVar \rightarrow \Val$ is a logical context and $\calR =
  \langle\calP_\calR,V_\calR,\memfn_\calR,\cntfn_\calR \rangle$ is a resource
  satisfying completeness and count-measurability assumptions.}
\label{fig:assertion-semantics}
\end{figure}

In logics like \pcol with an adaptive scheduler, any random event can influence later scheduler choices.  Since
 randomness is already visible to the scheduler, there is no
scheduler-privacy distinction to track, and probability spaces over $\Mem$
suffice to characterize program states.
In \LogicName, the adversary is weaker since its schedule is fixed ahead of the program execution.
As we saw in \Cref{sec:overview-leak}, that does not make all randomness private; the program structure may still reveal
information through branch-dependent scheduler consumption.  However, when the
number of scheduler choices is \emph{independent} of the sampled branch, the
fixed schedule cannot correlate future choices with that sample, thereby allowing the
logic to establish stronger independence properties.

To express that distinction, resources in \LogicName record observations not
only about memory, but also about scheduler counts.  Thus a resource $\calR$ describes
the full program state relevant to the semantics and assertions dictate
whether each random choice is observable through scheduler consumption
in addition to the distribution over concrete memory values.

We embed the tracking of schedule consumption into probabilistic outcome conjunctions,
adapting the outcome conjunction of
\pcol~\cite{zilberstein2026probabilistic} to the resource direct sum.
Intuitively, $\bigoplus^o_{X\sim d(E)}\varphi$ 
says that a value $X$ is sampled according to
$\mu=d(\de{E}_{\LExp}(\Gamma))$, and that the current resource can be refined by
a \emph{direct sum} of branch resources $(\calR_v)_{v\in\supp(\mu)}$, formally defined below.

Let $\xi \in \D(I)$ be a discrete distribution over the index set $I$.
Also let $(\calR_i)_{i\in \supp(\xi)}$ be a countably indexed collection of resources such that there exists a $V_\oplus\subseteq\Var$ with $V_{\calR_i} = V_\oplus$ for all $i\in \supp(\xi)$ and the index spaces are pairwise disjoint \ie $\Omega_{\calR_i}\cap\Omega_{\calR_j}=\emptyset$ for all $i \neq j$.
The weighted direct sum is the resource
$
\bigoplus_{i\sim\xi}\calR_i
\triangleq
\langle \bigoplus_{i\sim\xi}\calP_{\calR_i}, V_\oplus, \memfn_{\calR_\oplus},
\cntfn_{\calR_\oplus}\rangle$, where the weighted disjoint-sum probability space
$
  \bigoplus_{i\sim\xi}\calP_{\calR_i}
\triangleq
\left\langle
  \bigcup_{i\in\supp(\xi)}\Omega_i,
  \calF_\oplus,
  \mu_\oplus
\right\rangle$
is defined as:
\begin{mathpar}
\calF_\oplus
\triangleq
\{A 
\mid \forall i\in\supp(\xi).\; A\cap\Omega_i\in\calF_i\}

\mu_\oplus(A)
\triangleq
\textstyle\sum_{i\in\supp(\xi)}\xi(i)\cdot\mu_{\calR_i}(A\cap\Omega_{\calR_i})
\end{mathpar}
and the observation functions are inherited branchwise: if
$\omega\in\Omega_{\calR_i}$, then
$\memfn_{\calR_\oplus}(\omega)=\memfn_{\calR_i}(\omega)$ and
$\cntfn_{\calR_\oplus}(\omega)=\cntfn_{\calR_i}(\omega)$.
If each $\calP_{\calR_i}$ is complete, then $\bigoplus_{i\sim\xi}\calP_{\calR_i}$ is complete.

In each branch $i$, the assertion $\varphi$ holds under the extended environment
$\Gamma[X\mapsto i]$.
The mode $o\in\{\Priv,\Leak\}$ specifies the side condition imposed on this
branch decomposition.  Formally, define
\begin{mathpar}
  \Obs_\Priv^\mu(\calR,(\calR_v)_{v\in\supp(\mu)})
  \triangleq
    \forall v\in\supp(\mu).\;
      \Proj{\calR_v}{\bbN}=\Proj{\calR}{\bbN}
      
  \Obs_\Leak^\mu(-,-) \triangleq\top
\end{mathpar}
Thus, in private mode, every branch has the same scheduler count distribution, so the counts are independent from the
sampled value $X$.  In leaky mode, no count-marginal condition is imposed, so $X$ may be correlated with the scheduler index.
We order modes by the strength of their side conditions, writing
$\Leak\sqsubseteq\Priv$, since $\Leak$ is a weaker mode than $\Priv$.
The nondeterministic outcome conjunction $\bignd_{X\in E}$ is the distribution-free analogue of
probabilistic outcome conjunction.  It does not prescribe a distribution for the
bound variable; instead, it records that the resource can be explained by some
distribution $\mu \in \D(\de{E}_\LExpr(\Gamma))$ over the possible outcomes $\bigoplus^{\Leak}_{X\sim\mu}\varphi$.

The separating conjunction $\varphi\sep\psi$ is interpreted by the resource product $\otimes$, as in
\Cref{fig:assertion-semantics}. 
The canonical product $\calR_1\otimes\calR_2$ is the resource induced by
the independent product.  Let $\Omega_\otimes\subseteq\bbN$ be a fresh
index set and
$g:\Omega_{\calR_1}\times\Omega_{\calR_2}\to\Omega_\otimes$
be a bijection.  Let
$\calP_\otimes
\triangleq
\langle
\Omega_\otimes,\;
\calF_\otimes,\;
\mu_\otimes
\rangle$,
where
$\calF_\otimes
\triangleq
\{C\subseteq\Omega_\otimes \mid
g^{-1}(C)\in\calF_{\calR_1}\otimes\calF_{\calR_2}\}$
and
$\mu_\otimes(C)
\triangleq
(\mu_{\calR_1}\otimes\mu_{\calR_2})(g^{-1}(C))$.
Here $\calF_{\calR_1}\otimes\calF_{\calR_2}$ is the ordinary product
$\sigma$-algebra on $\Omega_{\calR_1}\times\Omega_{\calR_2}$
and $\mu_{\calR_1}\otimes\mu_{\calR_2}$ is the unique product
measure on
$\calF_{\calR_1}\otimes\calF_{\calR_2}$ satisfying the following equation
for all $A\in\calF_{\calR_1}$ and $B\in\calF_{\calR_2}$.
\begin{mathpar}
  \calF_{\calR_1}\otimes\calF_{\calR_2}
  \triangleq
  \sigma(\{A\times B \mid A\in\calF_{\calR_1},
  B\in\calF_{\calR_2}\})

  (\mu_{\calR_1}\otimes\mu_{\calR_2})(A\times B)
  =
  \mu_{\calR_1}(A)\cdot\mu_{\calR_2}(B)
\end{mathpar}
We define
$\calR_1\otimes\calR_2
\triangleq
\langle \calP_\otimes, V_{\calR_1}\cup V_{\calR_2},
\memfn_{\calR_1\otimes\calR_2}, \cntfn_{\calR_1\otimes\calR_2}\rangle$
where, for $g^{-1}(\omega)=(\omega_1,\omega_2)$,
\begin{align*}
  &\memfn_{\calR_1\otimes\calR_2}(\omega)
  \triangleq
    \memfn_{\calR_1}(\omega_1)\uplus\memfn_{\calR_2}(\omega_2)
  &\cntfn_{\calR_1\otimes\calR_2}(\omega)
  \triangleq
    \cntfn_{\calR_1}(\omega_1)+\cntfn_{\calR_2}(\omega_2).
\end{align*}
The definition of $\calR_1\otimes\calR_2$ uses the standard product space construction up to renaming of indices.
Rather than using $\Omega_{\calR_1}\times\Omega_{\calR_2}$ directly, we push the
ordinary product probability space forward along the fresh index bijection $g$.
The construction is similar to that of Amaryllis, where standard product measures are constructed over a set of indices and resources are combined pointwise \cite{lohse2026first}. 
A key difference is that Amaryllis restricts the outcome spaces to be finitely
supported, whereas our ambient index space permits countably supported discrete
probability spaces.
Thus $\varphi\sep\psi$ asserts probabilistic independence of the two resource
components, disjointness of their memory footprints, and additive composition of
their scheduler counts.

Unlike \pcol
\cite{zilberstein2026probabilistic}, which distinguishes strong
separation as independent product from weak separation as only memory disjointness,
\LogicName has only one separating conjunction.  The weak
form is used in \pcol to account for the adaptive adversary. But now, the schedule is
fixed independently of memory, and the remaining leakage channel is tracked by the $\Priv$ and $\Leak$ modes of $\bigoplus^o$.  Thus $\sep$ always denotes the
independent resource product.

Finally, $\sure{P}$ embeds pure memory assertions into outcome assertions:
$\Gamma,\calR\vDash\sure{P}$ iff the memory observation of $\calR$ satisfies
$P$ with probability one.  The pure layer is count-oblivious because we do not need to worry about correlations between scheduler counts and a probability 1 event.
We
also use the following syntactic sugar for common outcome assertions:
\[\textstyle
  \varphi\oplus^o_E\psi
  \;\triangleq\;
  \bigoplus^o_{X\sim\bern{E}}
  (\sure{X = 1} \sep \varphi) \lor (\sure{X = 0} \sep \psi)
  \qquad
  \varphi\oplus^o_{\ge E}\psi
  \;\triangleq\;
  \exists X.\ \left(\varphi\oplus^o_X\psi\right)\sep \sure{E \le X \le 1}
\]
\[\textstyle
  e\overset{o}{\sim}d(E)
  \;\triangleq\;
  \bigoplus^o_{X\sim d(E)}\sure{e\mapsto X}
\]
As a default, when the mode $o = \Priv$, we omit it, so that $\bigoplus_{X \sim d(E)} \varphi$, $\varphi \oplus_E \psi$, and $e \sim d(E)$ all imply that the source of randomness is private.

\subsection{Predicates and Entailment Laws}

We now define predicates and entailment laws used to manipulate outcome assertions.
The main \LogicName-specific condition is the
combination of \emph{precision} and \emph{stability}. Precision identifies the
minimal memory component described by an assertion, while stability ensures
that the scheduler count is not correlated with memory.
Together, these properties let us
move assertions across separating conjunctions without accidentally retaining
leaked count information.

To state the properties, we use two projections that separate memory
and scheduler count observations. For a
resource $\calR$, define its memory component
$\calR^\memfn
\triangleq
\langle
\calP_\calR,\,
V_\calR,\,
\memfn_\calR,\,
\cntfn_{\calR^\memfn}
\rangle$ and count component $\calR^\cntfn
\triangleq
\langle
\calP_\calR,\,
\emptyset,\,
\memfn_{\calR^\cntfn},\,
\cntfn_\calR
\rangle$, where $\cntfn_{\calR^\memfn}(\omega)\triangleq 0$ and
$\memfn_{\calR^\cntfn}(\omega)\triangleq\emptyset$ for all
$\omega\in\Omega_\calR$.
Thus $\calR^\memfn$ keeps the memory observation and erases scheduler counts,
while $\calR^\cntfn$ keeps scheduler counts and erases memory observations.

\begin{definition}[Precision]\label{def:precise}
An assertion $\varphi$ is \emph{precise}, written $\precise{\varphi}$, iff for
any $\Gamma$ under which $\varphi$ is satisfiable there exists a
minimal resource $\calR_0$ such that $\Gamma,\calR_0\vDash \varphi$ and:
\begin{mathpar}
\forall \calR.\quad
\Gamma,\calR \vDash\varphi
\quad\implies\quad
\calR_0^\memfn \preceq \calR^\memfn
\end{mathpar}
\end{definition}
\Cref{def:precise} corresponds to the intuition that $\varphi$ is
precise iff each measurable event has an exact probability.
Accordingly, the following rules generate a syntactic fragment of precise assertions.
\begin{mathpar}
    \inferrule
    {
    }
    {\precise{\sure{P}}}

    \inferrule
    {\precise{\varphi} \and \precise{\psi}}
    {\precise{\varphi\sep\psi}}

    \inferrule
    {\precise{\varphi}}
    {\textstyle
      \precise{\bigoplus^o_{X\sim d(E)}\varphi}}
\end{mathpar}
Almost-sure assertions are precise because whenever $\sure{P}$ is satisfiable under $\Gamma$,
there is a canonical memory-only resource $\calR_P$ whose
probability-one event consists exactly of supersets of $\sem{P}_\Gamma$. Any other
satisfying resource must refine $\calR_P^\memfn$. 
Separating conjunctions are precise by taking the product of the precise
witnesses for the two sides.  If
$\Gamma,\calR_i^0\vDash\varphi_i$ are the minimal witnesses, then
the independent product of two resources $\calR_1^0\otimes\calR_2^0$ satisfies
$\varphi_1\sep\varphi_2$; for any other
satisfying split $\calR_1\otimes\calR_2\preceq\calR$, precision of the two
components and monotonicity of $\otimes$ give
$(\calR_1^0\otimes\calR_2^0)^\memfn\preceq\calR^\memfn$.

Probabilistic outcome conjunction preserves precision in both observation modes
by count-erasing the branch witnesses.  Each precise branch provides a least
memory witness; after erasing the count component, their direct sum is a least memory witness for
the whole assertion.  The private side condition holds because all erased
branches have count marginal $\delta_0$. 

The precision established by these rules is only memory precision.  Unlike in
\pcol, we cannot require a least satisfying resource overall. For example, the assertion
$\sure{x\mapsto 0}$ can be satisfied by a resource where every outcome has
scheduler count $0$, and also by one where every outcome has scheduler count
$1$.  Those resources are incomparable, since refinement preserves count
observations.  Therefore, precision compares only the memory components
$\calR_0^\memfn$ and $\calR^\memfn$. Moreover, precision does not imply
secrecy of the scheduler count, for instance, $x
\sim^\Leak\bern{\tfrac12}$ is precise but leaky.

Thus, memory precision alone is not enough to define entailment laws.
Stability additionally requires an assertion to remain satisfiable after separating away a count-only component,
and is useful for producing a canonical memory and count witness.

\begin{definition}[Assertion Stability]
  An assertion $\varphi$ is trivially
  \emph{stable} in the $\strong$-mode, \ie $\Stable{\strong}{\varphi}$ always holds.
  An assertion $\varphi$ is \emph{stable} in the $\weak$-mode iff 
  whenever $\Gamma,\calR\vDash\varphi$, then $\calR$ can be decomposed into an
  independent product of memory and count components. That is,
  $\Stable{\weak}{\varphi}$ iff
   for all $\Gamma$ and $\calR$ such that $\Gamma,\calR \vDash \varphi$, there
   exist memory and count witnesses $\calR_m$ and $\calR_c$ such that:
\begin{mathpar}
  \calR_m\otimes\calR_c\preceq\calR

  \forall \omega\in\Omega_{\calR_m}.\;
        \cntfn_{\calR_m}(\omega)=0

        \forall \omega\in\Omega_{\calR_c}.\;
        \memfn_{\calR_c}(\omega)=\emptyset
\end{mathpar}
\end{definition}

A sufficient syntactic fragment for weak assertion stability is generated by
the following rules:
\begin{mathpar}
  \inferrule
  {
  }
  {\Stable{\weak}{\sure{P}}}

  \inferrule
  {\Stable{\weak}{\varphi}\and \Stable{\weak}{\psi}}
  {\Stable{\weak}{\varphi \sep \psi}}

  \inferrule
  {\Stable{\weak}{\varphi}}
  {\textstyle
    \Stable{\weak}{\bigoplus\nolimits^\Priv_{X\sim d(E)}\varphi}}
\end{mathpar}
Almost-sure assertions are stable by the almost-sure coarsening argument.  If
$\Gamma,\calR\vDash\sure{P}$, the event
$\memfn_\calR^{-1}(\sem{P}_\Gamma^{V_\calR})$ has probability one.  Thus we may
apply almost-sure coarsening to extract a memory witness
$\calR_m$ and a count witness $\calR_c$ with
$\calR_m\otimes\calR_c\preceq\calR$.  The memory witness retains the
almost-sure information needed for $P$ and has only zero scheduler counts; hence
every retained memory outcome satisfies $P$, so
$\Gamma,\calR_m\vDash\sure{P}$. 
Separating conjunctions are stable because the two decompositions can be recombined.  For a
satisfying split $\calR_\varphi\otimes\calR_\psi\preceq\calR$, its stable witness has the
form
$(\calR_{\varphi, m}\otimes\calR_{\psi, m})
\otimes
(\calR_{\varphi, c}\otimes\calR_{\psi, c})
\preceq
\calR$.

Private probabilistic choice is stable because the privacy condition in
$\bigoplus^\Priv_{X\sim d(E)}\varphi$ forces every branch to have the same
count distribution.  Thus, for
$\mu=d(\de{E}_{\LExp}(\Gamma))$, its stable witness has the form
$(\bigoplus_{v\sim\mu}\calR_{v, m})
\otimes
\calR_{c}
\preceq
\calR$
with the same count component $\calR_c$ for all branches.

Note that stability does not imply precision.
For example,
$\sure{x\mapsto 0}\lor\sure{x\mapsto 1}$ is stable because 
any satisfying resource satisfies one disjunct. Stability gives
a memory witness without counts, which still satisfies the
disjunction.
However, it is not precise, since the memory witness where $x$ is always $0$
and the memory witness where $x$ is always $1$ are incomparable.
Thus stability
removes scheduler-count dependence, but it does not identify a canonical memory
component.
As in \pcol, we also use \emph{convexity} to identify
assertions that are \emph{closed under probabilistic mixtures}.

\begin{definition}[Assertion Convexity]
An assertion $\varphi$ is \emph{convex}, written $\convex{\varphi}$, iff for
every context $\Gamma$, countable index set $I$, distribution $\mu \in \D(I)$, finite footprint
$V$, and family of resources $(\calR_i)_{i \in I}$, such that 1) $\forall i \in
\supp(\mu). V_{\calR_i} = V$, and
2) $(\Omega_{\calR_i})_{i \in \supp(\mu)}$ are pairwise disjoint, then
\[\textstyle
  \forall i \in \supp(\mu).\; \Gamma, \calR_i \vDash \varphi
  \quad\Rightarrow\quad
  \Gamma,
  \bigoplus_{i \sim \mu} \calR_i \vDash \varphi
\]
\end{definition}
Informally, convexity justifies eliminating a probabilistic mixture whose
positive-probability branches all satisfy the same assertion $\bigoplus_{X \sim
  d(E)} \varphi \Rightarrow \varphi$.
The following is a selected syntactic fragment for both properties; additional
rules appear in \Cref{app:assertions}:
\begin{mathpar}
  \inferrule
  {\;}
  {\convex{\sure{P}}}

  \inferrule
  {\convex{\varphi} \and \substack{\precise{\psi} \\ \Stable{\weak}{\psi}}}
  {\convex{\varphi\sep\psi}}

  \inferrule
  {\convex{\varphi}}
  {\textstyle
    \convex{\bigoplus^o_{X\sim d(E)}\varphi}}
%

  \inferrule
  {\convex{\varphi}}
  {\textstyle
    \convex{\bignd_{X\in E}\varphi}}
\end{mathpar}
The case for $\sep$ is delicate, as it is tempting to only require $\varphi$ and $\psi$ to be convex.
As a counterexample both conjuncts of $\varphi \triangleq
x\sim^\Leak\bern{\frac13} \sep \bignd_{Y\in\{0,1\}}\sure{y\mapsto Y}$ are
convex, but the assertion as a whole is not. To see why, consider the resources
$\calR_1$ and $\calR_2$ below, both of which satisfy $\varphi$ since in each
resource $y$ is almost surely constant.
\begin{mathpar}\arraycolsep=1pt
\calR_1 \triangleq \left[\scriptsize
  \begin{array}{lll}
    (x=1,y=1,\cntfn=1) & \mapsto\frac13
    \\
    (x=0,y=1,\cntfn=0) & \mapsto\frac23
  \end{array}
\right]
~~
\calR_2 \triangleq \left[\scriptsize
  \begin{array}{lll}
    (x=1,y=0,\cntfn=0) & \mapsto\frac13
    \\
    (x=0,y=0,\cntfn=1) & \mapsto\frac23
  \end{array}
\right]
~~
\calR_1\oplus_{\frac12}\calR_2 =\left[\scriptsize
  \begin{array}{lll}
    (x=1,y=1,\cntfn=1) & \mapsto\frac16
    \\
    (x=0,y=1,\cntfn=0) & \mapsto\frac13
    \\
    (x=1,y=0,\cntfn=0) & \mapsto\frac16
    \\
    (x=0,y=0,\cntfn=1) & \mapsto\frac13
  \end{array}
\right]
\end{mathpar}
By contrast, $\calR_1\oplus_\frac12\calR_2$ does not satisfy $\varphi$ since both $x$ and $y$ are correlated with the count, and therefore they cannot be joined by $\sep$.
The precision and stability hypotheses rule out that problem.

\Cref{fig:oplus-entailment-laws} gives some key entailment laws for manipulating outcome assertions. Notably, any assertion can distribute into an outcome conjunction, but it must be precise and stable in order to factor out. Assertions must be convex in order to collapse an outcome conjunction to a single assertion. Additional rules, including to apply consequences under other modalities are given in the appendix in \Cref{fig:separation-entailment-laws,fig:basic-entailment-laws}.

\begin{figure}
\begin{mathpar}
  \inferrule
  {}
  {\textstyle
    \bigoplus_{X \sim d(E)}^o \varphi
    \vdash
    \bignd_{X \in \supp(d(E))}\varphi}

  \inferrule
  {}
  {\textstyle
    \bigoplus_{X \sim d(E)}^\Priv \varphi
    \vdash
    \bigoplus_{X \sim d(E)}^\Leak \varphi}

  \inferrule
  {\varphi \vdash \psi}
  {\textstyle
    \bigoplus_{X \sim d(E)}^o \varphi
    \vdash
    \bigoplus_{X \sim d(E)}^o \psi}

  \inferrule
  {X \not\in \fv(\psi)}
  {\textstyle
    (\bigoplus_{X \sim d(E)}^o \varphi) \sep \psi
    \vdash
    \bigoplus_{X \sim d(E)}^o (\varphi \sep \psi)
  }

  \inferrule
  {
    X \not\in \fv(\psi)\\
    \precise{\psi}\\
    \Stable{\weak}{\psi}
  }
  {\textstyle
    \bigoplus_{X \sim d(E)}^o (\varphi \sep \psi)
    \vdash
    (\bigoplus_{X \sim d(E)}^o \varphi) \sep \psi
  }
  
  \inferrule
  {o_1\sqsubseteq o_2 \and X\neq Y \and
    X\notin\fv(E_Y) \and Y\notin\fv(E_X)}
  {\textstyle
    \bigoplus_{X\sim d_X(E_X)}^{o_1}
      \left(\bigoplus_{Y\sim d_Y(E_Y)}^{o_2}\varphi\right)
    \vdash
    \bigoplus_{Y\sim d_Y(E_Y)}^{o_2}
      \left(\bigoplus_{X\sim d_X(E_X)}^{o_1}\varphi\right)}
      
\inferrule
  {X\notin\fv(\varphi) \and \convex{\varphi}}
  {\textstyle
    \bigoplus_{X\sim d(E)}^o\varphi \vdash \varphi}
\end{mathpar}
\caption{Entailment laws for outcome conjunctions and their interaction with
separation.}
\label{fig:oplus-entailment-laws}
\end{figure}

\section{Oblivious Probabilistic Outcome Logic}
\label{sec:logic}
Just like other Hoare-style logics, Oblivious Probabilistic Outcome
Logic (\LogicName) specifications are triples $\triple{\varphi}{C}{\psi}$ where
$\varphi$ is a precondition, $C$ is a command, and $\psi$ is a
postcondition.  The pre- and postconditions $\varphi$ and $\psi$ are outcome assertions from
\Cref{sec:assertions}, which describe a distribution over program
memories.  Similar to \dol and \pcol, the postcondition applies to \emph{all} of
the nondeterministic outcomes, which is achieved by quantifying over
all schedules under which $C$ could run.

Triples are by
definition \emph{frame preserving}, meaning that the initial distribution $\mu$
is a refinement of the independent product of $\calR$ (which satisfies the
precondition), and an additional frame resource $\calR_F$, which must be
preserved by the command $C$. That frame preservation directly justifies the
soundness of the \ruleref{Frame} rule, which we saw in
\Cref{sec:overview-sep}. Similar to \pcol, triples have two modes $m\in\{\strong,
\weak\}$ with implications to the type of \ruleref{Frame} rule that is
allowed. But whereas in \pcol, \weak-mode completely disables the probabilistic
independence guarantees of the \ruleref{Frame} rule, both modes in \LogicName
indicate probabilistic independence, and the corresponding \weak-mode of
\LogicName only precludes the framing of randomized assertions that leak
scheduling information.
We write $\Stable{m}{\calR}$ for the
corresponding resource-level side condition:
\begin{mathpar}
  \Stable{\strong}{\calR}
  \triangleq \top

  \Stable{\weak}{\calR}
  \triangleq
  \forall \omega\in\Omega_\calR.\;
    \cntfn_\calR(\omega)=0 .
\end{mathpar}
Thus, strong validity imposes no restriction on the frame resource, while weak
validity requires the frame to carry no scheduler-count information.
At the level of triple validity, a probabilistic program transforms
distributions of observations, so we quantify over a possibly nonterminating
input distribution 
$\mu \in \D_\bot(\Mem \times \bbN)$. Fixing the scheduler tape $s$, we uncurry
$\de{C}$ as follows:
\begin{mathpar}
  \de{C}_s : \Mem \times \bbN \to \D_\bot(\Mem \times \bbN)

  \de{C}_s(\sigma,n) \triangleq \de{C}(\sigma)(s)(n)
\end{mathpar}
Then the $\D_\bot$ Kleisli extension $\de{C}_s^\dagger(\mu)$ denotes running
$C$ from $\mu$. Using the new notation, triple validity is defined as follows.



\begin{definition}[Validity of \LogicName Triples]\label{def:triple}
A triple $\vDash_m \triple{\varphi}{C}{\psi}$ is valid iff for
all $s \in \schedule$, $\Gamma : \LVar \rightarrow \Val$,
$\mu \in \calD_\bot(\Mem \times \bbN)$, and resources $\calR,\calR_F$:
\[
\calR \otimes \calR_F \preceq \mu
\;\land\;
\Stable{m}{\calR_F}
\;\land\;
\Gamma, \calR \vDash \varphi
\;\;\implies\;\;
\exists \calQ.\;\;
\calQ \otimes \calR_F \preceq \de{C}_s^\dagger(\mu)
\;\land\;
\Gamma, \calQ \vDash \psi
\;\land\;
V_\calQ \subseteq V_\calR
\]
\end{definition}

Note that  $\calR$ and $\calQ$ range over samples of type
$\Mem \times \bbN$, without nontermination $\bot$.  That means valid triples imply that the probability of nontermination is 0, \ie the program
\emph{almost surely terminates}, making \LogicName{} a total correctness
logic~\cite{manna1974axiomatic}. The last conjunct $V_\calQ \subseteq V_\calR$
records the ownership preservation invariant. We now present the inference rules for
reasoning about programs in \LogicName.  We write
$\vdash_m\triple{\varphi}{C}{\psi}$ to mean that a triple is derivable
according to those rules, and all of the rules are sound with respect to
\Cref{def:triple}, as shown in \Cref{app:sequential-rules,app:bounded-rank}.
\begin{theorem}[Soundness of \LogicName]
$
  \vdash_m\triple{\varphi}{C}{\psi}
  \quad\implies\quad
  \vDash_m\triple{\varphi}{C}{\psi}
$
\end{theorem}

\subsection{Rules for Sequential Commands}
\begin{figure}
\footnotesize
\centering
\begin{minipage}{\linewidth}
\begin{mathpar}
  \inferrule*[right=\rulename{Skip}]
  { }
  {\vdash_m \triple{\varphi}{\cmdskip}{\varphi}}

  \inferrule*[right=\rulename{Seq}]
  {\vdash_m\triple{\varphi}{C_1}{\vartheta} \;\; \vdash_m\triple{\vartheta}{C_2}{\psi}}
  {\vdash_m\triple{\varphi}{C_1 \cmdseqS C_2}{\psi}}

  \inferrule*[right=\rulename{Assign}]
  {
    \varphi \Rightarrow \sure{e \mapsto E} \land (\psi \sep \sure{\own(x)})
  }
  {\vdash_m \triple{\varphi}{\actassign{x}{e}}{\psi \sep \sure{x \mapsto E}}}

  \inferrule*[right=\rulename{IfT}]
  {
    \varphi \Rightarrow \sure{b \mapsto \text{true}}\;\;
    \vdash_m \triple{\varphi}{C_1}{\psi}
  }
  {
    \vdash_m \triple{\varphi}{\cmdcase{b}{C_1}{C_2}}{\psi}
  }
%

  \inferrule*[right=\rulename{Samp}]
  {
    \varphi \Rightarrow
    \sure{e \mapsto E} \land (\psi \sep \sure{\own(x)})
  }
  {\vdash_m \triple{\varphi}{x\samp d(e)}{\psi \sep (x \sim d(E))}}

  \inferrule*[right=\rulename{Prob-Safe}]
  {
    \varphi \Rightarrow \sure{e \mapsto E}\;\;
    \textsf{bits}(C_1) = \textsf{bits}(C_2) = \{k\}\\\\
    \vdash_m\triple{\varphi}{C_1}{\psi_1}\;\;
    \vdash_m\triple{\varphi}{C_2}{\psi_2}
  }
  {
    \vdash_m\triple{\varphi}{C_1 \cmdchoiceS{e} C_2}{\psi_1 \oplus^{\textsf{priv}}_E \psi_2}
  }

  \inferrule*[right=\rulename{Prob-Leak}]
  {
    \varphi \Rightarrow \sure{e \mapsto E}\\\\
    \vdash_m\triple{\varphi}{C_1}{\psi_1}\;\;
    \vdash_m\triple{\varphi}{C_2}{\psi_2}
  }
  {
    \vdash_m\triple{\varphi}{C_1 \cmdchoiceS{e} C_2}{\psi_1 \oplus_{E} \psi_2}
  }

  \inferrule*[right=\rulename{NAssign}]
  {
    X\notin\fv(\psi)\;\;
    \Stable{\weak}{\psi}\\\\
    \varphi \Rightarrow \sure{e \mapsto E} \land (\psi \sep \sure{\own(x)})
  }
  {
    \vdash_\weak
    \triple{\varphi}{x\gets e}
    {\textstyle\bignd_{X\in E}(\psi \sep \sure{x \mapsto X})}
  }
  
  \inferrule*[right=\rulename{ND}]
  {
    \Stable{\weak}{\varphi}
    \\
    \vdash_\weak\triple{\varphi}{C_1}{\psi_1}
    \\
    \vdash_\weak\triple{\varphi}{C_2}{\psi_2}
  }
  {
    \vdash_\weak\triple{\varphi}{C_1 \cmdnondetS C_2}{\psi_1 \nd \psi_2}
  }
\end{mathpar}
\end{minipage}
\caption{Rules for Sequential Commands.}
\label{fig:sequential-command-rules}
\end{figure}

\Cref{fig:sequential-command-rules} presents inference rules for sequential commands.  Many of the rules are standard for
probabilistic separation logics, including \pcol.
\ruleref{Skip} leaves the precondition unchanged; \ruleref{Seq} composes
specifications along sequential composition; \ruleref{Assign} requires the
precondition to determine the value of $e$ and to provide ownership of $x$;
\ruleref{Samp} is analogous, except that $x$ is sampled from $d(E)$ in a way
that is private from the scheduler because sampling is atomic. 

The
conditional rule
\ruleref{IfT} enters the `then' branch of an if statement once the guard
is known. A symmetric \ruleref{IfF} rule is given in \Cref{app:inference-rules}.
If the guard fact is not entailed by the precondition as a whole,
then case analysis can be done using a split rule from
\Cref{fig:structural-derived-rules}.  For example, the precondition
$\bigoplus^o_{X\sim\bern{\frac12}}\sure{x\mapsto X}$
does not entail either $\sure{x=0\mapsto\tru}$ or
$\sure{x=0\mapsto\fls}$, but after conditioning on $X$, the branch
$X=0$ implies that $x=0$ is true and $X=1$ entails that it is false.
Similarly,
$\bignd_{X\in\{0,1\}}\sure{x\mapsto X}$
requires nondeterministic branch analysis before applying \ruleref{IfT} or
\ruleref{IfF}.

The remaining rules showcase how count resources
let \LogicName prove stronger privacy properties against an oblivious
scheduler.
\ruleref{Prob-Safe} and \ruleref{Prob-Leak} are the two rules for
probabilistic choice.  In both rules, the precondition must determine the
weight expression $e$ of the probabilistic choice.  Both branch premises use
the same precondition because the command forms a
mixture of the outputs of $C_1$ and $C_2$ from the same input distribution.
For example, from a precondition such as
$x\sim\bern{\frac12}$, a probabilistic choice
$C_1\cmdchoiceS{\frac12}C_2$ does not route the $x=0$ states to $C_1$ and the
$x=1$ states to $C_2$; either branch may be selected from either input state.
Because \LogicName tracks scheduler time as a resource and assumes an
oblivious scheduler, it can derive stronger conclusions for programs whose
control flow does not leak entropy.  For instance, consider
the following commands:
\begin{mathpar}
  C_\Priv
  \triangleq
  x\actassignS 0\cmdchoiceS{\frac12}x\actassignS 1

  C_\Leak
  \triangleq
  x\actassignS 0\cmdchoiceS{\frac12}
  x \gets [1]
  
  C' \triangleq y \gets [0,1]
\end{mathpar}
As explained in \Cref{sec:overview-leak}, schedule consumption leaks information to the adversary.
In $C_\Priv$, no adversarial decisions are consumed, therefore $x$ and $y$ are probabilistically independent after executing $C_\Priv\fatsemi C'$.
But in $C_\Leak\cmdseqS C'$, the right branch
consumes one more scheduler entry than the left branch, so $C'$ reads different
scheduler bits for different values of $x$, allowing for a correlation.

Therefore, $C_1 \oplus_e C_2$ only results in a $\oplus^\Priv$ postcondition if both branches consume the same number of scheduler entries.
The function
$\mathsf{bits}(C) \in 2^\bbN$, provided in \Cref{app:bits},
records how much of the scheduling tape may be
consumed along any execution path of $C$.
Sequencing adds consumptions,
branching takes their union, and nondeterministic commands contribute one
additional scheduler choice before executing the selected branch.
For a probabilistic choice $C_1\cmdchoiceS{e}C_2$, \ruleref{Prob-Safe} admits
the $\oplus^\Priv$ postcondition only when
$\textsf{bits}(C_1)$ and $\textsf{bits}(C_2)$ are singleton sets containing the same constant.
Indeed, $\mathsf{bits}(C_\Priv) = \{0\}$, whereas $\mathsf{bits}(C_\Leak) = \{0,1\}$ so the
private postcondition $\oplus^\Priv$ is valid for $C_\Priv$, but not $C_\Leak$.

The singleton requirement avoids
branch-specific scheduler-count reasoning in assertions.  A tempting
relaxation is to require only
$\textsf{bits}(C_1) = \textsf{bits}(C_2)$, but that condition is insufficient.  For
example, let $C_1 \triangleq x \coloneqq 0 \oplus_{\frac13} x \gets [0]$ and $C_2 \triangleq x \gets [1] \oplus_{\frac13} x \coloneqq 1$, so clearly $\mathsf{bits}(C_1) = \mathsf{bits}(C_2) = \{0,1\}$.
However, in $(C_1\cmdchoiceS{\frac12}C_2)\cmdseqS z \gets [0,1]$, the choice of $z$
queries the scheduler at a position whose distribution is correlated with the
outer probabilistic choice. If $C_1$ runs, and thus $x=0$, then a tape entry is
consumed with probability $\frac23$ whereas if $C_2$ runs, and $x=1$, an entry
is consumed with only probability $\frac13$. The scheduler can use that
information to make $x=z$ with probability up to $\frac23$.

Moreover, a more permissive rule would require the two branches to have the
same scheduler-count marginal distribution on all valid inputs, which is exactly the
privacy condition for $\oplus^\Priv$. Enforcing this in the logic would
require tracking distributions of schedule consumption for each probabilistic
branch. We avoid exposing this bookkeeping: \ruleref{Prob-Safe} gives a simple
syntactic test for privacy of the outer probabilistic choice.

The nondeterminism rules are where the oblivious scheduler assumption most
clearly arises.  In \pcol, nondeterministic choice is
interpreted against an adaptive adversary: after the current memory is sampled,
the scheduler may choose a branch as a function of that memory.  Therefore,
probabilistic information in the precondition can be used to correlate the
scheduler's branch choice with the sampled state.  As a result, the
\ruleref{ND} rule in \pcol is limited to pure preconditions $\sure{P}$.

In \LogicName that restriction is relaxed by requiring the precondition to be
$\Stable{\weak}{\varphi}$ rather than pure, which is strictly more
permissive since weakly stable assertions may describe private probabilistic
resources.
As we will see in \Cref{sec:examples}, the ability to retain probabilistic information in the precondition of a nondeterministic choice is essential.
The \LogicName \ruleref{ND} rule lets us push
$\oplus^\Priv$ inside the nondeterministic postcondition, so each branch of
$\nd$ can still assert independence between a prior private random source and the
scheduler's choice.

Although \ruleref{ND} is a weak-mode rule, weak mode still allows for use of the
\ruleref{Frame} rule with
probabilistic independence guarantees, as long as the augmented information only pertains to private
random sources. That makes \LogicName's \ruleref{ND} rule more powerful than that of \dol and \pcol.
In the soundness proof, we begin by partitioning the precondition resource into $\calR = \calR_1 \oplus_p\calR_2$, where $p$ is the probability that the next scheduling bit is even. Since the frame $\calR_F$ is stable, it contributes no schedule counts, and therefore we have $\calR\otimes\calR_F = (\calR_1 \otimes \calR_F) \oplus_p (\calR_2\otimes \calR_F)$, and we still know that the left and right operands will execute $C_1$ and $C_2$, respectively. The proof is completed with a straightforward use of the induction hypothesis.




\subsection{Separation and Split}

\begin{figure}
\footnotesize
\centering
\begin{minipage}{\linewidth}
\begin{mathpar}
  \inferrule*[right=\hypertarget{rule:Frame}{\textsc{Frame}}]
  {
    \vdash_m \triple{\varphi}{C}{\psi}\;\;
    \Stable{m}{\vartheta}
  }
  {
    \vdash_m \triple{\varphi \sep \vartheta}{C}{\psi \sep \vartheta}
  }

  \inferrule*[right=\rulename{Weakening}]
  {
    \vdash_\strong \triple{\varphi}{C}{\psi}
  }
  {
    \vdash_\weak \triple{\varphi}{C}{\psi}
  }
  
  \inferrule*[right=\rulename{Disj}]{
    \vdash_m\triple{\varphi_1}{C}{\psi_1}
    \;\;
    \vdash_m\triple{\varphi_2}{C}{\psi_2}
  }{
     \vdash_m\triple{\varphi_1\vee\varphi_2}{C}{\psi_1\vee\psi_2}
  }

  \inferrule*[right=\rulename{Consequence}]
  {
    \varphi' \Rightarrow \varphi\;\;
    \vdash_m \triple{\varphi}{C}{\psi}\;\;
    \psi \Rightarrow \psi'
  }
  {
    \vdash_m \triple{\varphi'}{C}{\psi'}
  }
  
  \inferrule*[right=\rulename{Priv-Split1}]
  {\textsf{bits}(C)=\{n\}\;\;
   \vdash_m\triple{\varphi}{C}{\psi}}
 {\textstyle
   \vdash_m
     \triple
       {\bigoplus^\Priv_{X\sim d(E)}\varphi}
       {C}
       {\bigoplus^\Priv_{X\sim d(E)}\psi}}

  \inferrule*[right=\rulename{Split1}]
  {
   \vdash_m\triple{\varphi}{C}{\psi}}
 {\textstyle
   \vdash_m
     \triple
       {\bigoplus^o_{X\sim d(E)}\varphi}
       {C}
       {\bigoplus^\Leak_{X\sim d(E)}\psi}}
%

  \inferrule*[right=\rulename{ND-Split1}]
  {
   \vdash_m\triple{\varphi}{C}{\psi}}
 {\textstyle
   \vdash_m
     \triple
       {\bignd_{X\in E}\varphi}
       {C}
       {\bignd_{X\in E}\psi}}
%
%
\end{mathpar}
\end{minipage}
\caption{Structural and Derived Rules.}
\label{fig:structural-derived-rules}
\end{figure}

\Cref{fig:structural-derived-rules} presents the main structural rules of
\LogicName.  Unlike \cites{fan2025program} prior logic for oblivious adversaries,
\LogicName is based on separation logic, which enables compositional reasoning by reusing local proofs in
larger probabilistic contexts without reanalyzing the program branch
by branch \cite{psl,li2023lilac,bao2025bluebell,zilberstein2026probabilistic}.

More precisely, the \ruleref{Frame} rule extends a proof of
$\triple{\varphi}{C}{\psi}$ to a proof of
$\triple{\varphi\sep\vartheta}{C}{\psi\sep\vartheta}$ whenever the frame
$\vartheta$ is \emph{stable} in the current mode. Semantically, an input
resource satisfying $\varphi \sep \vartheta$ can be split into an active
component satisfying $\varphi$ and a framed component satisfying
$\vartheta$. Since the framed component is hidden from the local proof of $C$,
stability is needed to ensure that this hidden resource does not carry
scheduler-count information in a way that can affect the command.
To see why the restriction is essential, consider the following unsound application of the \ruleref{Frame} rule.
\[
  \inferrule{
        \vdash_\weak\triple{\sure{y \mapsto -}}{y\gets[0,1]}{(\sure{y\mapsto0}\nd\sure{y\mapsto1})}
  }{
    \vdash_\weak\triple{\sure{y \mapsto -}\sep x \overset{\Leak}{\sim} \bern{\tfrac12}}{y\gets[0,1]}{(\sure{y\mapsto0}\nd\sure{y\mapsto1})\sep x \overset{\Leak}{\sim} \bern{\tfrac12}}
  }{\ruleref{Frame}}
\]
Above, the leaky frame may correlate $x$ with the schedule count, in which case $x$ and $y$ are not independent at the end of the execution, rendering the postcondition false.
The stability premise prevents
that failure mode. In the weak case of the proof, the frame is decomposed into an independent
product of a memory fragment $\calR_\vartheta^\memfn$ and a count fragment
$\calR_\vartheta^\cntfn$; the count fragment is committed to the active
resource, while the memory fragment is retained as the count-stable
postcondition frame.  For strong triples, no analogous restriction is needed:
strong-mode reasoning
is ordinary probabilistic separation reasoning within a fixed scheduler-count
range, so the frame is never exposed to an adversarially chosen count shift.

Most of the remaining structural rules are standard.  \ruleref{Weakening} is
immediate: a strong triple satisfies the stronger frame-preservation
obligation, and therefore also satisfies the weak one.  \ruleref{Consequence}
is the usual rule for strengthening preconditions and weakening
postconditions.
The split rules are logical case-analysis principles analogous to
those in \pcol \cite{zilberstein2026probabilistic}.  \ruleref{Split1} and
\ruleref{ND-Split1} reduce a judgment over an outer probabilistic or
nondeterministic assertion to the corresponding branch judgments. In
\ruleref{Split1}, the resulting probabilistic postcondition uses the public
modality $\Leak$.
\ruleref{Priv-Split1} is the private analogue of \ruleref{Split1}. Similar to the \ruleref{Prob-Safe} rule, the
additional premise $\textsf{bits}(C)=\{n\}$ ensures deterministic scheduler
consumption by $C$, allowing the probabilistic modality in the postcondition to
remain $\Priv$.

Unlike the splitting rules in \pcol, the \LogicName variants do not require that the postcondition witnesses a partition of the sample space to satisfy disjointness in the direct sum definition, since we now use probability spaces over an index set rather than directly over memories. That means that \ruleref{Split1} in \LogicName always applies, whereas \pcol needed an additional \ruleref{Split2} rule with a collapsed convex postcondition for cases where the program state does not remain disjoint across random branches. For convenience, we provide \ruleref{Split2} as a derived rule in the appendix.

Together with the assertion language, the structural rules for separation and
split let us choose the granularity of a leakage claim. Proofs can expose exactly the
distinctions made observable by execution while preserving the private entropy
that remains.
For example, let $\code{tick} = \cmdskip \cmdnondetS \cmdskip$ be a command that consumes a schedule entry, then consider:
\begin{equation}\label{eq:partial-leak}
  C
  \triangleq
  x\samp\unif{1,3}
  \cmdseqS
  (\cmdcase{x\le 2}{\code{tick}}{\cmdskip})
  \cmdseqS
  y\gets[1,2,3]
\end{equation}
The conditional reveals whether $x\le 2$ to the scheduler, but it does not reveal the
residual choice between $1$ and $2$, which is expressed by the following weak
triple:
\[
  \vdash_\weak
  \triple{\sure{\own(x,y)}}{C}
  {
    \textstyle
    \big(
      x \overset{\Priv}{\sim} \unif{1,2}
      \sep
      \bignd_{Y=1}^{3}\sure{y\mapsto Y}
    \big)
    \oplus^\Leak_{\frac23}
    \big(
      \sure{x\mapsto3}
      \sep
      \bignd_{Y=1}^{3}\sure{y\mapsto Y}
    \big)
  }
\]
To derive that triple, we first apply the \ruleref{Consequence} $x
\sim \unif{1, 3} \Rightarrow (x \sim \unif{1,2}) \oplus_\frac23
\sure{x\mapsto 3}$ after the sampling event to separate the top-level
information about $x$ (leaked by the if statement) from
the inner information (which remains private). The remainder of the
proof proceeds mechanically using \ruleref{Split1}, \ruleref{IfT},
and \ruleref{IfF}. Ultimately, the postcondition tells us that the
scheduler could guess $x$ if its value is 3, but can only guess $x$
with probability $\frac12$ otherwise. That gives a total probability of
$\frac23\cdot\frac12+\frac13\cdot1=\frac23$ that the scheduler could
guess $x$, which is between the uninformed success probability $\frac13$
and the success probability $1$ of an adaptive scheduler.

\subsection{Loops and Almost Sure Termination}

\begin{figure}
\footnotesize
\centering
\begin{minipage}{\linewidth}
\begin{mathpar}
  \inferrule*[right=\rulename{Bounded-Rank}]
  {
    \inferrule*{}{
      \substack{
      \varphi \Rightarrow \sure{\ell \le R \le h} \\
      \varphi\sep[R = \ell] \Rightarrow \sure{e \mapsto \fls} \\
      \varphi\sep[R > \ell] \Rightarrow \sure{e \mapsto \tru}
      }
    }\;\;
    \vdash_m\triple{\varphi * \sure{R = N > \ell}}{C}{\textstyle\big(\bignd_{R=\ell}^{N-1} \varphi\big) \oplus_{\ge p}^\Leak \big(\bignd_{R=N}^h \varphi\big)}\;\;
    \inferrule*{}{
      \substack{
      0 < p \le 1 \\
      N \notin \vars(\varphi) \\
      \convex{\varphi[\ell/R]}
    }}
  }
  {
    \vdash_m\triple{\textstyle\bignd_{R=\ell}^h \varphi}{\whl eC}{\varphi[\ell/R]}
  }
 \end{mathpar}
\end{minipage}
\caption{Bounded-Rank Rule.}
\label{fig:bounded-rank-rule}
\end{figure}

\Cref{fig:bounded-rank-rule} presents the final rule of \LogicName for
analyzing while loops and proving that they almost surely terminate:
they terminate with probability $1$.
The rule is inspired by bounded-rank
principles due to \citet{mciver2005abstraction}, as well as previous
probabilistic outcome
logics \cite{zilberstein2025demonic,zilberstein2026probabilistic}, but
reimplemented for oblivious schedulers.

As in \pcol, the rule involves a loop invariant $\varphi$ and an
integer-valued rank $R$ whose bounds are entailed by the invariant, \ie
$\varphi \Rightarrow \sure{\ell \le R \le h}$.  The rank controls the guard.
The loop continues when $R>\ell$, namely
$\varphi \sep \sure{R>\ell}\Rightarrow\sure{e\mapsto\tru}$, and has a single
point of exit at $R=\ell$, namely
$\varphi \sep \sure{R=\ell}\Rightarrow\sure{e\mapsto\fls}$.
The premise of
\ruleref{Bounded-Rank} supplies the probabilistic progress argument. Given the
current rank $N>\ell$ prior to one iteration of $C$, the postcondition assigns
probability at least $p>0$ to branches satisfying $\varphi$ at some strictly
smaller rank $R\in\{\ell,\ldots,N-1\}$.  The
remaining mass may still be live at some rank $R\in\{N,\ldots,h\}$, allowing a
step to make no progress or even jump back upward.  Since the rank range is
finite, there is a uniformly positive
probability of reaching the absorbing exit rank $\ell$ in every sufficiently
long live block; equivalently, the live probability mass decays geometrically, so the loop
terminates almost surely.

The soundness proof maintains an exit/live
invariant over the loop approximants.  Given the initial active resource $\calR$, after the $i$th guard test, the proof
tracks the exit and live weights $w_i^{\exit},w_i^{\live}$ and the
corresponding normalized distributions $\mu_i^{\exit},\mu_i^{\live}$:
\begin{align*}
  \Exit(i):
  &
    w_i^{\exit}>0
    \Rightarrow
    \exists \calQ_i.\;
    \Gamma,\calQ_i\vDash\varphi[\ell/R]
    \land
    \calQ_i\otimes\calR_F\preceq\mu_i^{\exit} \land V_\calQ \subseteq V_\calR
  \\[0.4em]
  \Live(i):
  &
    w_i^{\live}>0
    \Rightarrow
    \exists I_i,
    \theta_i\in\calD(I_i),
    r_i:\supp(\theta_i)\to\{\ell+1,..,h\},
    (\calR_{i,a})_{a\in\supp(\theta_i)},
    (\mu_{i,a})_{a\in\supp(\theta_i)}.\;
  \\
  &
    \qquad{}
    \textstyle\mu_i^{\live}=\bigoplus_{a\sim\theta_i}\mu_{i,a}\\
  &
    \qquad{}
    \land
    \forall a\in\supp(\theta_i).\;
    \Gamma[R\mapsto r_i(a)],\calR_{i,a}\vDash\varphi
    \land
    \calR_{i,a}\otimes\calR_F\preceq\mu_{i,a}
    \land V_{\calR_{i, a}} \subseteq V_{\calR}
\end{align*}
Here $\Exit(i)$ describes the runs that stop at the $i$th guard test. If the
exit event has positive mass, then some resource satisfying
$\varphi[\ell/R]$, together with the preserved frame, refines the conditional
distribution $\mu_i^{\exit}$.
Meanwhile, $\Live(i)$ says that the executions inside the loop can be
put back into the same shape for the next iteration.  The live
distribution may be decomposed into a mixture of branches $a\sim\theta_i$;
each branch carries rank $r_i(a)\in\{\ell+1,\ldots,h\}$, satisfies $\varphi$
at that rank, and its active resource, together with the frame, refines the
branch distribution $\mu_{i,a}$.
The footprint bound keeps these
witnesses within a common footprint, so exit resources from different
iterations can be aggregated via direct sum.
These two clauses mirror the
premise of \ruleref{Bounded-Rank}, providing the resource witnesses needed to
continue the construction at each iteration.

The invariant drives the block-decay argument.  Let $H=h-\ell$ be the rank
height.  In every block of $H$ live iterations, there is probability at least
$p^H$ of reaching the exit rank by strictly decreasing the rank at each step, so
$w_{i+H}^{\live}\le(1-p^H)w_i^{\live}$.  Thus the live mass vanishes in the
limit, and the loop output is the countable mixture of its exit distributions.

It is worth noting that the final-postcondition premise
$\convex{\varphi[\ell/R]}$ is weaker than the corresponding precision premise
$\precise{\varphi[\ell/R]}$ in \pcol.  In \pcol, precision of the final
assertion is required after the almost-sure termination argument to identify a
unique minimal probability space for the limiting postcondition.  Our proof
does not need to identify such a canonical witness.  Instead, for each exit
time $i$, $\Exit(i)$ provides a resource $\calQ_i$ satisfying
$\varphi[\ell/R]$ and a framed refinement
$\calQ_i\otimes\calR_F\preceq\mu_i^{\exit}$.  Once the live mass vanishes, the
final resource is the countable mixture of the $\calQ_i$'s.  The frame is
preserved because $\otimes$ distributes over this mixture, and convexity of
$\varphi[\ell/R]$ is enough to show that the mixed resource still satisfies the
final postcondition.

In addition, if the loop body uses any nondeterminism, then information is leaked to the scheduler during partial executions based on the likelihood of termination after each number of steps. However, since \ruleref{Bounded-Rank} implies almost sure termination, that leaked information ultimately vanishes. Just as we saw in program (\ref{eq:partial-leak}), the entropy leaked to the adversary via the variables in the loop guard does not give it any advantage toward influencing the execution of the live branches.

\section{Case Studies}
\label{sec:examples}

We now demonstrate the use of \LogicName to verify a variety of programs, which are correct in the oblivious adversary model, but not the adaptive model. The full derivations are given in \Cref{app:examples}.

\subsection{The Monty Hall Problem}
\label{sec:monte}

The Monty Hall Problem is a probability puzzle modeled after the game show \emph{Let's Make a Deal} \cite{selvin1975letters}.
The host, Monty Hall, first randomly places a car $c$ behind door 1, 2, or 3, and places goats behind the other two
doors. Next, the player selects a door $p$. Monty then opens a door
$o$, which is neither the car nor the player's choice. Finally, the
player either accepts the prize behind the door $p$, or
switches to the other closed door ($6-p-o$).
The classic result is
that switching doors leads to a higher probability of winning the car,
which we prove in \LogicName using the program below.
\[
  c \samp \unif{1, 3} \fatsemi
  p \gets [1, 2, 3] \fatsemi
  o \gets [1, 2, 3] \setminus [c, p] \fatsemi
  p \coloneqq 6 - p - o
\]
A key property of the Monty Hall Problem is that the player's choice
is \emph{oblivious} to the random location of the car. An adaptive
player would be able to guess the location of the car
with certainty, which does not accurately model a game show.
Indeed, a similar example was proven in \dol, but there the player's
choice was fixed \cite{zilberstein2025demonic}. By contrast, the
\LogicName proof leverages the probabilistic independence guarantees
of the \ruleref{NAssign} rule. After processing the first two
commands, we get the following triple, which tells us that for any
fixed player choice $X$, the value of $c$ is uniformly random; thus,
the player wins the car ($p = c$) with probability $\frac13$.
\[\small
\inferrule*[right={Seq\phantom{XX}}]{
  \inferrule*[right=Samp,vdots=2.5em,rightskip=24em]{\;}{
    \vdash_\weak\triple{\sure{\own(c, p)}}{c\samp\unif{1,3}}{c\sim \unif{1,3}\sep\sure{\own(p)}}
  }
  \inferrule*[Right=NAssign]{\;}{
    \vdash_\weak\triple{c\sim \unif{1,3}\sep\sure{\own(p)}}{p\gets[1,2,3]}{\textstyle\bignd_{X \in \{1,2,3\}} c\sim \unif{1,3}\sep\sure{p \mapsto X}}
  }
}{
  \vdash_\weak\triple{\sure{\own(c, p)}}{c\samp\unif{1,3} \fatsemi p \gets [1,2,3]}{\textstyle\bignd_{X \in \{1,2,3\}} c\sim \unif{1,3}\sep\sure{p \mapsto X}}
}
\]
Let $S_X = \{1, 2, 3\}\setminus \{X\}$ and $d_X = \unif{S_X}$, then we can rewrite the above postcondition as follows:
\[\textstyle
  \bignd_{X \in \{1,2,3\}} \left( \sure{c \mapsto X \sep p \mapsto X} \oplus_\frac13 \bigoplus_{Y \sim d_X} \sure{c \mapsto X \sep p\mapsto Y} \right)
\]
Moving further through the program, we continue the proof by case
analysis on the outcomes above. On the left side of the
$\oplus_\frac13$, we get that $o$ is nondeterministically chosen from
$S_X$ and that $c= X$ and $p = 6-X-Z$. It is easy to see $p\neq
c$ by enumerating the possible values for $X$ and from the fact $X
\neq Z$, by definition. On the right of $\oplus_\frac13$, Monty
has only one choice for $o$ since $c \neq p$, \ie $o = 6 - X -
Y$. That means that $p = 6 - Y - o = 6 - Y - (6 - X - Y) = X$;
therefore, $c = p$. The proof is sketched below and shown fully in
\Cref{app:monte}.
\begin{mathpar}\def\arraystretch{1.25}
  \begin{array}{l}
  \ob{\sure{c\mapsto X \sep p\mapsto X \sep \own(o)}}
  \\
  \quad o\gets [1,2,3]\setminus [c,p]\fatsemi \phantom{x} \\
  \quad p \coloneqq 6-p-o
  \\
  \ob{
    \textstyle\bignd_{Z \in S_X} \sure{c \mapsto X \sep o\mapsto Z \sep p\mapsto 6-X - Z}
  }
  \\
  \ob{\sure{(c = p) \mapsto \false}}
  \end{array}
  \qquad
  \begin{array}{l}
  \ob{\bigoplus_{Y\sim d_X} \sure{c\mapsto X \sep p\mapsto Y \sep \own(o)}}
  \\
  \quad o\gets [1,2,3]\setminus [c,p]\fatsemi \phantom{x} \\
  \quad p \coloneqq 6-p-o
  \\
  \ob{
    \bigoplus_{Y\sim d_X} \sure{c \mapsto X \sep o\mapsto (6 - X - Y) \sep p\mapsto X}
  }
  \\
  \ob{\sure{(c=p)\mapsto \true}}
  \end{array}
\end{mathpar}
Putting the two cases together, we get the following postcondition, which states that the player has a $\frac23$ chance of winning after switching doors, replicating the classic result.
\[
  \textstyle\bignd_{X \in \{1,2,3\}} \left( \sure{(c = p)\mapsto \false} \oplus_\frac13 \sure{(c=p)\mapsto \true} \right)
  \qquad\Rightarrow\qquad
  (c=p) \sim \bern{\frac23}
\]

%
%

\subsection{Leader Election}
\label{sec:leader}

In distributed systems, consensus is the problem of ensuring that parallel agents agree on a data value, which is critical to database and blockchain applications. \citet{fischer1985impossibility} showed that consensus cannot be achieved deterministically in an asynchronous system with faults, but randomization provides an elegant workaround for that impossibility result \cite{ben-or1983another}. More recently,  \citet{chor1994wait-free} and \citet{aspnes2012faster,aspnes2010modular} showed that randomized consensus can be achieved more efficiently by algorithms that take advantage of an oblivious scheduler.

The trick is to use a round-based protocol, where each process probabilistically moves to the next round and all processes must adopt the leader's value if they get sufficiently far ahead. Whereas an adaptive scheduler could anticipate which process is about to advance and delay it until another process catches up, an oblivious scheduler has no power to prevent a leader from emerging \cite{chor1994wait-free}.
We distill that idea into the two-party leader election protocol below. 
\[
\vdash_\weak \ob{\sure{\own(x)}}~~
\begin{array}{l}
  x \coloneqq 0\fatsemi\phantom{x} \\
  \whl{|x| \le 1}{} \\
  \quad (x \samp \unif{x, x+1}) \nd (x \samp \unif{x-1, x})
\end{array}
~~\ob{\sure{x\mapsto 2} \nd \sure{x \mapsto -2}}
\]
In the protocol, Alice and Bob compete to become the leader. Alice is elected when $x = 2$ and Bob is elected when $x = -2$.
Each round, the adversarial scheduler picks who gets to run. When Alice runs, she increments $x$ with probability $\frac12$ and when Bob runs, he decrements $x$ with probability $\frac12$.

In an adaptive model, the program does not necessarily terminate. Since the scheduler can see the coin flips, it can always schedule the candidate who is behind until $x = 0$, so that there is no leader. However, in an oblivious model, the scheduler cannot base its choice on the value of $x$, therefore the program does terminate almost surely, which we now show. Note that the protocol does not give any guarantee of \emph{fairness}; the scheduler can determine who will be elected the leader, but cannot prevent a leader from being elected altogether.

The key intuition is that the algorithm always terminates within
three iterations with probability at least $\frac18$.  Starting from
the \emph{worst case} scenario, where $\sure{x\mapsto 0}$, one of
the candidates runs at least twice. Suppose that Alice runs at least
twice. In the worst case, Bob runs once as well. Then, Alice is
 elected if she advances twice (with probability $\frac14$), and
if Bob does not advance (with probability $\frac12$), giving a total
probability of $\frac18$ that Alice wins after three rounds.

To formalize the argument, we first need a loop invariant for the \ruleref{Bounded-Rank} rule. The progress toward termination depends not only on the value of $x$, but also the relative probabilities of $x$ taking on different values.
The invariant $\varphi$ has four parts, indexed from 0 to 3.
\begin{align*}
  \varphi_3 &\triangleq \big(\sure{x \mapsto 1} \nd \sure{x \mapsto 0} \nd \sure{x \mapsto -1}\big) \sep \sure{R=3}
  &
  \psi_3 &\triangleq \varphi_2 \nd (\varphi_0 \oplus_\frac12 \varphi_3)
  \\
  \varphi_2 &\triangleq \big(x \sim \unif{0,1} \nd x\sim\unif{-1, 0} \big) \sep \sure{R=2}
  &
  \psi_2 &\triangleq (\varphi_0 \oplus_\frac14 \varphi_3) \nd (\varphi_1 \oplus_\frac12 \varphi_3)
  \\
  \varphi_1 &\triangleq \big(\sure{x \mapsto 1} \oplus_\frac12 \sure{x \mapsto -1}\big) \sep \sure{R=1}
  &
  \psi_1 &\triangleq \varphi_0 \oplus_\frac14 \varphi_3
  \\
  \varphi_0 &\triangleq \big(\sure{x \mapsto 2} \nd \sure{x \mapsto -2}\big) \sep \sure{R=0}
  &
  \varphi &\triangleq \varphi_0 \vee \varphi_1 \vee \varphi_2 \vee \varphi_3
\end{align*}
In $\varphi_3$, which constitutes the furthest state from termination, the value of $x$ is chosen by the adversary. Next $\varphi_2$ gives slightly less leeway to the adversary by mandating that $x$ is 0 with probability exactly $\frac12$. Progressing further, $\varphi_1$ states that $x$ has equal probability of being 1 or -1. Finally, $\varphi_0$ dictates that the protocol has terminated, since $x$ is either $2$ or $-2$. We now prove that the ranks outlined in the invariant must decrease with probability at least $\frac14$ on each iteration of the loop.
We do so by case analysis on the rank. That is, for each $i \in \{1, 2, 3\}$, we prove the following triple, where $\psi_i$ is also defined above and full derivations are shown in \Cref{app:leader}.
\[
  \vdash_\weak\triple{\varphi_i}{(x \samp \unif{x, x+1}) \nd (x \samp \unif{x-1, x})}{\psi_i}
\]
The proofs themselves follow by case analysis, which we sketch
starting with the $i=3$ case. If $x = 0$, then we are guaranteed
to end in state $\varphi_2$. The cases where $x=1$ and $x=-1$ are
symmetric, so we discuss only the $x=1$ case. There are two
nondeterministic possibilities. If the adversary picks Alice, then she
has a $\frac12$ chance of winning immediately. If the adversary
instead picks Bob, then we again end up in state $\varphi_2$, since
$x=1$ and is decremented to 0 with probability $\frac12$.

When $i=2$, there are again two symmetric cases. Suppose that $x
\sim\unif{0,1}$. If Alice runs, then she has a $\frac14$ chance of winning
immediately (\ie $x=1$ with probability $\frac12$ and $x$ is incremented
with probability $\frac12$). If Bob runs, then $x$ will be 1 or $-1$
with equal probability, meaning that the scheduler has no way to prevent
a win in the next round. That is, when $i=1$, one of the players must
win with probability $\frac14$ regardless of the adversary's choice.
The $i=1$ and $i=2$ cases both take advantage of the scheduler's
obliviousness. If the scheduler were adaptive and therefore knew which
candidate were ahead, it could always run the other player. Finally,
since each $\psi_i$ implies that the rank strictly decreases with
probability at least $\frac14$, we can complete the proof.

\subsection{Online Algorithms: Paging and Caching}
\label{sec:paging}

Paging is the problem of determining what data to cache in order to reduce the number of expensive fetches from main memory. It is a classic example of an \emph{online algorithm}, since decisions must be made as adversarial requests are received.
Randomized paging algorithms perform better than any deterministic ones against an oblivious adversary \cite{fiat1991competitive,mcgeoch1991strongly}. To demonstrate that, we prove properties about the randomized paging algorithm below in \LogicName.
\[\arraycolsep=0pt
\vdash_\weak\ob{\sure{\own(c, i, m, r) \sep n \mapsto N}}
\begin{array}{l}
\quad  c \samp \bern{\tfrac12}\fatsemi i\coloneqq 0\fatsemi m \coloneqq 0\fatsemi\phantom{x} \\
\quad  \whl{i < n}{} \\
  \qquad i \coloneqq i+1\fatsemi 
   r \gets [0,1]\fatsemi\phantom{x} \\
  \qquad \ift{r \neq c}{} \\
  \quad\qquad c \samp \bern{\tfrac12} \fatsemi 
  m \coloneqq m+1
\end{array}
\ob{\textstyle\sure{m\mapsto N} \oplus_{\frac1{2^N}} \bignd_{X = 0}^{N-1} \sure{m \mapsto X}}
\]
There are two memory pages, indexed 0 and 1, and only one of them can be held in the randomly initialized cache $c$. Each round, the algorithm queries the adversary for a request $r$, which is the index of the next page to fetch. If $r=c$, then $r$ can be served from the cache, so there is nothing to do. But if $r\neq c$, then the $r$ must be fetched, and one of the pages must be evicted from the cache. The page to evict is chosen randomly. Finally, the miss count $m$ is incremented to record that a miss has occurred.
Our goal is to bound the number of cache misses that occur.

An adaptive adversary could always choose to request the page that is not cached, and therefore force the outcome $\sure{m\mapsto N}$, \ie every request is a miss. However, an oblivious adversary cannot do that, and so whether or not a miss occurs is uniformly random, assuming that the previous request was a miss. If the previous request was a hit, then choosing the same page again will also be a hit, so the chance of getting a miss again in the next round is not uniformly random. So, if the adversary alternates its request each round, it minimizes the probability of consecutive hits.

We aim to prove the postcondition above, which states that getting all misses $\sure{m\mapsto N}$ occurs with probability only $\frac1{2^N}$, and with probability $1-\frac1{2^N}$ there are strictly fewer than $N$ misses. The full proof is given in \Cref{app:paging}, but we briefly sketch the proof here. First, we need a loop invariant $\varphi$. Similar to the final postcondition, the loop invariant has two parts, which correspond to the case that all requests were cache misses ($\varphi_\mathsf{miss}$) and the case that at least one cache hit has occurred ($\varphi_\mathsf{hit}$). In addition $\varphi_\mathsf{pure}$ tracks deterministic information that ensures that the loop terminates after $N$ iterations. The rank $R$ records how many iterations there are to go until termination, and it decreases with probability 1 on each iteration.
\begin{mathpar}
  \varphi_{\textsf{miss}} \triangleq \sure{m \mapsto N-R} \sep c\sim\bern{\tfrac12}

  \varphi_{\textsf{hit}} \triangleq \textstyle\bignd_{X=0}^{N-R-1} \bignd_{Y\in\{0,1\}} \sure{m\mapsto X\sep c \mapsto Y}

  \varphi_{\textsf{pure}} \triangleq \sure{i \mapsto N-R \sep n \mapsto N \sep \own(r) \sep 0 \le R \le N}

  \varphi \triangleq \varphi_{\textsf{pure}} \sep \big( \varphi_{\textsf{miss}} \oplus_{\frac1{2^{N-R}}} \varphi_{\textsf{hit}} \big)
\end{mathpar}
The proof follows by case analysis on the random outcomes
$\varphi_{\textsf{miss}} \oplus_{\frac1{2^{N-R}}}
\varphi_{\textsf{hit}}$. In the $\varphi_\mathsf{miss}$ case,
the cache value cannot be leaked, so the adversarial
request $r \gets[0,1]$ is independent from $c$.
Hence, after each iteration, the probability of
no hits decreases by a factor of 2. In the $\varphi_\mathsf{hit}$
case, we cannot be sure that $r$ is independent from $c$. If the
prior request was a cache hit, the adversary knows precisely what is
in the cache. So, whether a miss occurs on the current round is
essentially nondeterministic, meaning that we now know that fewer than
$(N-R-1) + 1$ misses have occurred.

\section{Related Work}
\label{sec:related}

\citet{fan2025program} developed a Hoare-style logic for
\emph{concurrent} randomized programs under an oblivious scheduler.
Unlike \LogicName, their logic is not based on probabilistic independence, leading to diminished compositionality.
Their mechanism for exploiting obliviousness is the \emph{lazy coin}, which treats
a probabilistic choice $a_1 \oplus_p a_2$ as a single atomic action.
Although lazy coin precludes the scheduler from picking a strategy between the resolution of $\oplus_p$
and the execution of the associated $a_i$, it cannot
guarantee that later nondeterminism remains oblivious to
the coin flip, \eg $x$ and $y$ are independent throughout
\(
  (x \coloneqq 1 \oplus_\frac12 x \coloneqq 0) \fatsemi y\gets[0,1]
\).
Although \LogicName does not support concurrency, it provides more
expressive abstractions for reasoning about obliviousness, which are
essential to the case studies in \Cref{sec:examples}. In addition,
\citeauthor{fan2025program}'s logic cannot handle almost sure
termination and has weaker case analysis on random outcomes.
Their logic has a split rule with a collapsed convex postcondition, which
requires instrumenting the program with explicit splitting instructions.
By contrast, \LogicName supports purely logical splitting at any point
in a derivation, and its \textsc{Pub/Priv-Split1} rules retain
the shape of the outer $\bigoplus$ modality.

More broadly, \LogicName builds on work that uses
probabilistic independence as a compositional reasoning principle. Probabilistic
Separation Logic (\psl) adapted separation to express probabilistic
independence \cite{psl}.
Later logics, including
Lilac \cite{li2023lilac} and Bluebell \cite{bao2025bluebell},
refined this idea using probability spaces to reason about independence at
varying granularities.
Amaryllis \cite{lohse2026first} combines \psl-style independence with
Iris resource algebras \cite{iris,iris1} to support dynamic heap allocation
and custom state.
Most relevant for \LogicName, Amaryllis uses probability spaces over finite indices rather than program states, supporting
resources without canonical projection-based splits.
Scheduler counts are one such resource: a count of $5$ contains no information
about its provenance, since
the count could arise from $1+4$, $3+2$, and so on.
Accordingly, \Cref{sec:assertions} uses a similar indexed model. Unlike
Amaryllis, our model supports countable distributions,
while Amaryllis provides more general frame-preserving updates for reasoning
about concurrency.

Separately, Outcome Logic provides a Hoare-style framework for reasoning about
probabilistic and nondeterministic outcomes
\cite{outcome,zilberstein2024outcome,zilberstein2025demonic,zilberstein2025outcome}.
Its probabilistic-concurrent extension \pcol{} \cite{zilberstein2026probabilistic,zilberstein2025denotational} combines \psl-style independence reasoning
with Concurrent Separation Logic \cite{csl,brookes2004semantics}.
Nondeterminism rules in \pcol completely preclude the use of the \ruleref{Frame} rule
to provide probabilistic independence guarantees.
Although \LogicName{} adopts the same
outcome-assertion approach, the oblivious setting allowed us to retain independence guarantees for
private randomness across nondeterministic choices, yielding more permissive compositional rules.

\section{Conclusion}
\label{sec:conclusion}

We developed a denotational model for probabilistic programs with an oblivious
adversary and introduced Oblivious Probabilistic Outcome Logic (\LogicName) for reasoning
about programs in that model. The key idea is to treat schedule consumption as a
resource, allowing the logic
to distinguish private sources of randomness from sources that may have leaked
to the adversary. The distinction yields
expressive and compositional proof rules, including a sound
\ruleref{Frame} rule, while still recovering standard adaptive reasoning
principles when obliviousness is not needed.

Unlike prior logics, \LogicName maintains obliviousness globally rather than
locally so that facts about independence of variables can be preserved and reused
across later nondeterministic choices.
That large degree of compositionality, together with almost-sure
termination rules for reasoning about probabilistic liveness, yields a highly effective logic
that surpasses prior work.
We showed that \LogicName{} can be used to verify various realistic algorithms, whose correctness relies on an
oblivious adversary, putting them out of reach of prior logics.

A natural next direction is concurrency, where scheduler choices occur at every
step rather than only at explicit nondeterministic commands.
Concurrency will require tracking new ways that control flow leaks entropy to
the scheduler while preserving the independence reasoning that makes
\LogicName{} compositional.
The resource-based treatment of schedule consumption developed here provides a
promising foundation for such a concurrency logic.

\finalpage 
\clearpage
\bibliography{refs}


\begin{thebibliography}{50}


\ifx \showCODEN    \undefined \def \showCODEN     #1{\unskip}     \fi
\ifx \showISBNx    \undefined \def \showISBNx     #1{\unskip}     \fi
\ifx \showISBNxiii \undefined \def \showISBNxiii  #1{\unskip}     \fi
\ifx \showISSN     \undefined \def \showISSN      #1{\unskip}     \fi
\ifx \showLCCN     \undefined \def \showLCCN      #1{\unskip}     \fi
\ifx \shownote     \undefined \def \shownote      #1{#1}          \fi
\ifx \showarticletitle \undefined \def \showarticletitle #1{#1}   \fi
\ifx \showURL      \undefined \def \showURL       {\relax}        \fi
\providecommand\bibfield[2]{#2}
\providecommand\bibinfo[2]{#2}
\providecommand\natexlab[1]{#1}
\providecommand\showeprint[2][]{arXiv:#2}

\bibitem[Abramsky and Jung(1995)]%
        {abramsky1995domain}
\bibfield{author}{\bibinfo{person}{Samson Abramsky} {and}
  \bibinfo{person}{Achim Jung}.} \bibinfo{year}{1995}\natexlab{}.
\newblock \bibinfo{booktitle}{\emph{Domain Theory}}.
\newblock \bibinfo{publisher}{Oxford University Press, Inc.},
  \bibinfo{address}{USA}, \bibinfo{pages}{1–168}.
\newblock
\showISBNx{019853762X}


\bibitem[Apt and Plotkin(1986)]%
        {apt1986countable}
\bibfield{author}{\bibinfo{person}{Krzysztof Apt} {and} \bibinfo{person}{Gordon
  Plotkin}.} \bibinfo{year}{1986}\natexlab{}.
\newblock \showarticletitle{Countable nondeterminism and random assignment}.
\newblock \bibinfo{journal}{\emph{J. ACM}} \bibinfo{volume}{33},
  \bibinfo{number}{4} (\bibinfo{date}{aug} \bibinfo{year}{1986}),
  \bibinfo{pages}{724–767}.
\newblock
\showISSN{0004-5411}
\href{https://doi.org/10.1145/6490.6494}{doi:\nolinkurl{10.1145/6490.6494}}


\bibitem[Aspnes(2010)]%
        {aspnes2010modular}
\bibfield{author}{\bibinfo{person}{James Aspnes}.}
  \bibinfo{year}{2010}\natexlab{}.
\newblock \showarticletitle{A modular approach to shared-memory consensus, with
  applications to the probabilistic-write model}. In
  \bibinfo{booktitle}{\emph{Proceedings of the 29th {ACM} {SIGACT}-{SIGOPS}
  symposium on {Principles} of distributed computing}}.
  \bibinfo{publisher}{ACM}, \bibinfo{address}{Zurich Switzerland},
  \bibinfo{pages}{460--467}.
\newblock
\showISBNx{978-1-60558-888-9}
\href{https://doi.org/10.1145/1835698.1835802}{doi:\nolinkurl{10.1145/1835698.1835802}}


\bibitem[Aspnes(2012)]%
        {aspnes2012faster}
\bibfield{author}{\bibinfo{person}{James Aspnes}.}
  \bibinfo{year}{2012}\natexlab{}.
\newblock \showarticletitle{Faster randomized consensus with an oblivious
  adversary}. In \bibinfo{booktitle}{\emph{Proceedings of the 2012 ACM
  Symposium on Principles of Distributed Computing}} (Madeira, Portugal)
  \emph{(\bibinfo{series}{PODC '12})}. \bibinfo{publisher}{Association for
  Computing Machinery}, \bibinfo{address}{New York, NY, USA},
  \bibinfo{pages}{1–8}.
\newblock
\showISBNx{9781450314503}
\href{https://doi.org/10.1145/2332432.2332434}{doi:\nolinkurl{10.1145/2332432.2332434}}


\bibitem[Awodey(2006)]%
        {awodey2006category}
\bibfield{author}{\bibinfo{person}{Steve Awodey}.}
  \bibinfo{year}{2006}\natexlab{}.
\newblock \bibinfo{booktitle}{\emph{Category Theory}}.
\newblock \bibinfo{publisher}{Oxford University Press}.
\newblock
\showISBNx{9780198568612}
\href{https://doi.org/10.1093/acprof:oso/9780198568612.001.0001}{doi:\nolinkurl{10.1093/acprof:oso/9780198568612.001.0001}}


\bibitem[Bao et~al\mbox{.}(2025)]%
        {bao2025bluebell}
\bibfield{author}{\bibinfo{person}{Jialu Bao}, \bibinfo{person}{Emanuele
  D'Osualdo}, {and} \bibinfo{person}{Azadeh Farzan}.}
  \bibinfo{year}{2025}\natexlab{}.
\newblock \showarticletitle{Bluebell: An Alliance of Relational Lifting and
  Independence for Probabilistic Reasoning}.
\newblock \bibinfo{journal}{\emph{Proc. ACM Program. Lang.}}
  \bibinfo{volume}{9}, \bibinfo{number}{POPL}, Article \bibinfo{articleno}{58}
  (\bibinfo{date}{Jan.} \bibinfo{year}{2025}), \bibinfo{numpages}{31}~pages.
\newblock
\href{https://doi.org/10.1145/3704894}{doi:\nolinkurl{10.1145/3704894}}


\bibitem[Barthe et~al\mbox{.}(2020)]%
        {psl}
\bibfield{author}{\bibinfo{person}{Gilles Barthe}, \bibinfo{person}{Justin
  Hsu}, {and} \bibinfo{person}{Kevin Liao}.} \bibinfo{year}{2020}\natexlab{}.
\newblock \showarticletitle{{A Probabilistic Separation Logic}}.
\newblock \bibinfo{journal}{\emph{Proc. ACM Program. Lang.}}
  \bibinfo{volume}{4}, \bibinfo{number}{POPL}, Article \bibinfo{articleno}{55}
  (\bibinfo{date}{Jan.} \bibinfo{year}{2020}), \bibinfo{numpages}{30}~pages.
\newblock
\href{https://doi.org/10.1145/3371123}{doi:\nolinkurl{10.1145/3371123}}


\bibitem[Ben-David et~al\mbox{.}(1990)]%
        {ben-david1990power}
\bibfield{author}{\bibinfo{person}{S. Ben-David}, \bibinfo{person}{A. Borodin},
  \bibinfo{person}{R. Karp}, \bibinfo{person}{G. Tardos}, {and}
  \bibinfo{person}{A. Wigderson}.} \bibinfo{year}{1990}\natexlab{}.
\newblock \showarticletitle{On the power of randomization in online
  algorithms}. In \bibinfo{booktitle}{\emph{Proceedings of the Twenty-Second
  Annual ACM Symposium on Theory of Computing}} (Baltimore, Maryland, USA)
  \emph{(\bibinfo{series}{STOC '90})}. \bibinfo{publisher}{Association for
  Computing Machinery}, \bibinfo{address}{New York, NY, USA},
  \bibinfo{pages}{379–386}.
\newblock
\showISBNx{0897913612}
\href{https://doi.org/10.1145/100216.100268}{doi:\nolinkurl{10.1145/100216.100268}}


\bibitem[Ben-Or(1983)]%
        {ben-or1983another}
\bibfield{author}{\bibinfo{person}{Michael Ben-Or}.}
  \bibinfo{year}{1983}\natexlab{}.
\newblock \showarticletitle{Another advantage of free choice: Completely
  asynchronous agreement protocols}. In \bibinfo{booktitle}{\emph{Proceedings
  of the 2nd ACM Symposium on Principles of Distributed Computing}} (Montreal,
  Quebec, Canada) \emph{(\bibinfo{series}{PODC '83})}.
  \bibinfo{publisher}{Association for Computing Machinery},
  \bibinfo{address}{New York, NY, USA}, \bibinfo{pages}{27–30}.
\newblock
\showISBNx{0897911105}
\href{https://doi.org/10.1145/800221.806707}{doi:\nolinkurl{10.1145/800221.806707}}


\bibitem[Brookes(2004)]%
        {brookes2004semantics}
\bibfield{author}{\bibinfo{person}{Stephen Brookes}.}
  \bibinfo{year}{2004}\natexlab{}.
\newblock \showarticletitle{A Semantics for Concurrent Separation Logic}. In
  \bibinfo{booktitle}{\emph{CONCUR 2004 - Concurrency Theory}},
  \bibfield{editor}{\bibinfo{person}{Philippa Gardner} {and}
  \bibinfo{person}{Nobuko Yoshida}} (Eds.). \bibinfo{publisher}{Springer Berlin
  Heidelberg}, \bibinfo{address}{Berlin, Heidelberg}, \bibinfo{pages}{16--34}.
\newblock
\showISBNx{978-3-540-28644-8}
\href{https://doi.org/10.1007/978-3-540-28644-8_2}{doi:\nolinkurl{10.1007/978-3-540-28644-8_2}}


\bibitem[Chor et~al\mbox{.}(1994)]%
        {chor1994wait-free}
\bibfield{author}{\bibinfo{person}{Benny Chor}, \bibinfo{person}{Amos Israeli},
  {and} \bibinfo{person}{Ming Li}.} \bibinfo{year}{1994}\natexlab{}.
\newblock \showarticletitle{Wait-Free Consensus Using Asynchronous Hardware}.
\newblock \bibinfo{journal}{\emph{SIAM J. Comput.}} \bibinfo{volume}{23},
  \bibinfo{number}{4} (\bibinfo{year}{1994}), \bibinfo{pages}{701--712}.
\newblock
\href{https://doi.org/10.1137/S0097539790192635}{doi:\nolinkurl{10.1137/S0097539790192635}}


\bibitem[de~Moura et~al\mbox{.}(2015)]%
        {moura2015lean}
\bibfield{author}{\bibinfo{person}{Leonardo de Moura}, \bibinfo{person}{Soonho
  Kong}, \bibinfo{person}{Jeremy Avigad}, \bibinfo{person}{Floris van Doorn},
  {and} \bibinfo{person}{Jakob von Raumer}.} \bibinfo{year}{2015}\natexlab{}.
\newblock \showarticletitle{The Lean Theorem Prover (System Description)}. In
  \bibinfo{booktitle}{\emph{Automated Deduction - CADE-25}},
  \bibfield{editor}{\bibinfo{person}{Amy~P. Felty} {and} \bibinfo{person}{Aart
  Middeldorp}} (Eds.). \bibinfo{publisher}{Springer International Publishing},
  \bibinfo{address}{Cham}, \bibinfo{pages}{378--388}.
\newblock
\showISBNx{978-3-319-21401-6}
\href{https://doi.org/10.1007/978-3-319-21401-6_26}{doi:\nolinkurl{10.1007/978-3-319-21401-6_26}}


\bibitem[Denning(1976)]%
        {denning1976lattice}
\bibfield{author}{\bibinfo{person}{Dorothy~E. Denning}.}
  \bibinfo{year}{1976}\natexlab{}.
\newblock \showarticletitle{A lattice model of secure information flow}.
\newblock \bibinfo{journal}{\emph{Commun. ACM}} \bibinfo{volume}{19},
  \bibinfo{number}{5} (\bibinfo{date}{May} \bibinfo{year}{1976}),
  \bibinfo{pages}{236–243}.
\newblock
\showISSN{0001-0782}
\href{https://doi.org/10.1145/360051.360056}{doi:\nolinkurl{10.1145/360051.360056}}


\bibitem[Fan et~al\mbox{.}(2025)]%
        {fan2025program}
\bibfield{author}{\bibinfo{person}{Weijie Fan}, \bibinfo{person}{Hongjin
  Liang}, \bibinfo{person}{Xinyu Feng}, {and} \bibinfo{person}{Hanru Jiang}.}
  \bibinfo{year}{2025}\natexlab{}.
\newblock \showarticletitle{A Program Logic for Concurrent Randomized Programs
  in the Oblivious Adversary Model}. In \bibinfo{booktitle}{\emph{Programming
  Languages and Systems}}, \bibfield{editor}{\bibinfo{person}{Viktor
  Vafeiadis}} (Ed.). \bibinfo{publisher}{Springer Nature Switzerland},
  \bibinfo{address}{Cham}, \bibinfo{pages}{322--348}.
\newblock
\showISBNx{978-3-031-91118-7}
\href{https://doi.org/10.1007/978-3-031-91118-7_13}{doi:\nolinkurl{10.1007/978-3-031-91118-7_13}}


\bibitem[Fatourou et~al\mbox{.}(1997)]%
        {fatourou1997efficiency}
\bibfield{author}{\bibinfo{person}{Panagiota Fatourou}, \bibinfo{person}{Marios
  Mavronicolas}, {and} \bibinfo{person}{Paul Spirakis}.}
  \bibinfo{year}{1997}\natexlab{}.
\newblock \showarticletitle{Efficiency of oblivious versus non-oblivious
  schedulers for optimistic, rate-based flow control (extended abstract)}. In
  \bibinfo{booktitle}{\emph{Proceedings of the Sixteenth Annual ACM Symposium
  on Principles of Distributed Computing}} (Santa Barbara, California, USA)
  \emph{(\bibinfo{series}{PODC '97})}. \bibinfo{publisher}{Association for
  Computing Machinery}, \bibinfo{address}{New York, NY, USA},
  \bibinfo{pages}{139–148}.
\newblock
\showISBNx{0897919521}
\href{https://doi.org/10.1145/259380.259434}{doi:\nolinkurl{10.1145/259380.259434}}


\bibitem[Fiat et~al\mbox{.}(1991)]%
        {fiat1991competitive}
\bibfield{author}{\bibinfo{person}{Amos Fiat}, \bibinfo{person}{Richard~M
  Karp}, \bibinfo{person}{Michael Luby}, \bibinfo{person}{Lyle~A McGeoch},
  \bibinfo{person}{Daniel~D Sleator}, {and} \bibinfo{person}{Neal~E Young}.}
  \bibinfo{year}{1991}\natexlab{}.
\newblock \showarticletitle{Competitive paging algorithms}.
\newblock \bibinfo{journal}{\emph{Journal of Algorithms}} \bibinfo{volume}{12},
  \bibinfo{number}{4} (\bibinfo{date}{Dec.} \bibinfo{year}{1991}),
  \bibinfo{pages}{685--699}.
\newblock
\showISSN{01966774}
\href{https://doi.org/10.1016/0196-6774(91)90041-V}{doi:\nolinkurl{10.1016/0196-6774(91)90041-V}}


\bibitem[Fischer et~al\mbox{.}(1985)]%
        {fischer1985impossibility}
\bibfield{author}{\bibinfo{person}{Michael~J. Fischer},
  \bibinfo{person}{Nancy~A. Lynch}, {and} \bibinfo{person}{Michael~S.
  Paterson}.} \bibinfo{year}{1985}\natexlab{}.
\newblock \showarticletitle{Impossibility of distributed consensus with one
  faulty process}.
\newblock \bibinfo{journal}{\emph{J. ACM}} \bibinfo{volume}{32},
  \bibinfo{number}{2} (\bibinfo{date}{April} \bibinfo{year}{1985}),
  \bibinfo{pages}{374–382}.
\newblock
\showISSN{0004-5411}
\href{https://doi.org/10.1145/3149.214121}{doi:\nolinkurl{10.1145/3149.214121}}


\bibitem[Fremlin(2001)]%
        {fremlin2001measure}
\bibfield{author}{\bibinfo{person}{David Fremlin}.}
  \bibinfo{year}{2001}\natexlab{}.
\newblock \bibinfo{booktitle}{\emph{Measure Theory, Volume 2}}.
\newblock
\showISBNx{0953812928}


\bibitem[He et~al\mbox{.}(1997)]%
        {jifeng1997probabilistic}
\bibfield{author}{\bibinfo{person}{Jifeng He}, \bibinfo{person}{Karen Seidel},
  {and} \bibinfo{person}{Annabelle McIver}.} \bibinfo{year}{1997}\natexlab{}.
\newblock \showarticletitle{Probabilistic models for the guarded command
  language}.
\newblock \bibinfo{journal}{\emph{Science of Computer Programming}}
  \bibinfo{volume}{28}, \bibinfo{number}{2} (\bibinfo{year}{1997}),
  \bibinfo{pages}{171--192}.
\newblock
\shownote{Formal Specifications: Foundations, Methods, Tools and Applications}.
\newblock
\showISSN{0167-6423}
\href{https://doi.org/10.1016/S0167-6423(96)00019-6}{doi:\nolinkurl{10.1016/S0167-6423(96)00019-6}}


\bibitem[Jones and Plotkin(1989)]%
        {jones1989probabilistic}
\bibfield{author}{\bibinfo{person}{Claire Jones} {and} \bibinfo{person}{Gordon
  Plotkin}.} \bibinfo{year}{1989}\natexlab{}.
\newblock \showarticletitle{A Probabilistic Powerdomain of Evaluations}. In
  \bibinfo{booktitle}{\emph{Fourth Annual Symposium on Logic in Computer
  Science}}. \bibinfo{pages}{186--195}.
\newblock
\href{https://doi.org/10.1109/lics.1989.39173}{doi:\nolinkurl{10.1109/lics.1989.39173}}


\bibitem[Jung et~al\mbox{.}(2018)]%
        {iris}
\bibfield{author}{\bibinfo{person}{Ralf Jung}, \bibinfo{person}{Robbert
  Krebbers}, \bibinfo{person}{Jacques-Henri Jourdan}, \bibinfo{person}{Ale\v{s}
  Bizjak}, \bibinfo{person}{Lars Birkedal}, {and} \bibinfo{person}{Derek
  Dreyer}.} \bibinfo{year}{2018}\natexlab{}.
\newblock \showarticletitle{{Iris from the ground up: A modular foundation for
  higher-order concurrent separation logic}}.
\newblock \bibinfo{journal}{\emph{Journal of Functional Programming}}
  \bibinfo{volume}{28} (\bibinfo{year}{2018}).
\newblock
\href{https://doi.org/10.1017/S0956796818000151}{doi:\nolinkurl{10.1017/S0956796818000151}}


\bibitem[Jung et~al\mbox{.}(2015)]%
        {iris1}
\bibfield{author}{\bibinfo{person}{Ralf Jung}, \bibinfo{person}{David Swasey},
  \bibinfo{person}{Filip Sieczkowski}, \bibinfo{person}{Kasper Svendsen},
  \bibinfo{person}{Aaron Turon}, \bibinfo{person}{Lars Birkedal}, {and}
  \bibinfo{person}{Derek Dreyer}.} \bibinfo{year}{2015}\natexlab{}.
\newblock \showarticletitle{{Iris: Monoids and Invariants as an Orthogonal
  Basis for Concurrent Reasoning}}. In \bibinfo{booktitle}{\emph{Proceedings of
  the 42nd Annual ACM SIGPLAN-SIGACT Symposium on Principles of Programming
  Languages}} (Mumbai, India) \emph{(\bibinfo{series}{POPL '15})}.
  \bibinfo{publisher}{Association for Computing Machinery},
  \bibinfo{address}{New York, NY, USA}, \bibinfo{pages}{637–650}.
\newblock
\showISBNx{9781450333009}
\href{https://doi.org/10.1145/2676726.2676980}{doi:\nolinkurl{10.1145/2676726.2676980}}


\bibitem[Keimel and Plotkin(2017)]%
        {keimel2017mixed}
\bibfield{author}{\bibinfo{person}{Klaus Keimel} {and} \bibinfo{person}{Gordon
  Plotkin}.} \bibinfo{year}{2017}\natexlab{}.
\newblock \showarticletitle{{Mixed powerdomains for probability and
  nondeterminism}}.
\newblock \bibinfo{journal}{\emph{{Logical Methods in Computer Science}}}
  \bibinfo{volume}{{Volume 13, Issue 1}} (\bibinfo{date}{Jan.}
  \bibinfo{year}{2017}).
\newblock
\href{https://doi.org/10.23638/LMCS-13(1:2)2017}{doi:\nolinkurl{10.23638/LMCS-13(1:2)2017}}


\bibitem[Li et~al\mbox{.}(2025)]%
        {li2025total}
\bibfield{author}{\bibinfo{person}{James Li}, \bibinfo{person}{Noam
  Zilberstein}, {and} \bibinfo{person}{Alexandra Silva}.}
  \bibinfo{year}{2025}\natexlab{}.
\newblock \bibinfo{title}{Total Outcome Logic: Unified Reasoning for a Taxonomy
  of Program Logics}.
\newblock
\showeprint[arxiv]{2411.00197}~[cs.LO]
\urldef\tempurl%
\url{https://arxiv.org/abs/2411.00197}
\showURL{%
\tempurl}


\bibitem[Li et~al\mbox{.}(2023)]%
        {li2023lilac}
\bibfield{author}{\bibinfo{person}{John~M. Li}, \bibinfo{person}{Amal Ahmed},
  {and} \bibinfo{person}{Steven Holtzen}.} \bibinfo{year}{2023}\natexlab{}.
\newblock \showarticletitle{Lilac: A Modal Separation Logic for Conditional
  Probability}.
\newblock \bibinfo{journal}{\emph{Proc. ACM Program. Lang.}}
  \bibinfo{volume}{7}, \bibinfo{number}{PLDI}, Article \bibinfo{articleno}{112}
  (\bibinfo{date}{jun} \bibinfo{year}{2023}), \bibinfo{numpages}{24}~pages.
\newblock
\href{https://doi.org/10.1145/3591226}{doi:\nolinkurl{10.1145/3591226}}


\bibitem[Lohse et~al\mbox{.}(2026)]%
        {lohse2026first}
\bibfield{author}{\bibinfo{person}{Janine Lohse}, \bibinfo{person}{Tim Rohde},
  \bibinfo{person}{Jimmy Xin}, \bibinfo{person}{Niklas Mück},
  \bibinfo{person}{Iona Kuhn}, \bibinfo{person}{Derek Dreyer},
  \bibinfo{person}{Deepak Garg}, {and} \bibinfo{person}{Emanuele D’Osualdo}.}
  \bibinfo{year}{2026}\natexlab{}.
\newblock \bibinfo{booktitle}{\emph{First Steps Towards Probabilistic Iris: A
  Separation Logic with Independence, Conditioning, and Dynamic Heap
  Allocation}}.
\newblock \bibinfo{type}{{T}echnical {R}eport}. \bibinfo{institution}{Max
  Planck Institute for Software Systems (MPI-SWS)}.
\newblock
\shownote{Accessed: 2026-03-23}.
\newblock
\urldef\tempurl%
\url{https://people.mpi-sws.org/~jlohse/downloads/amaryllis.pdf}
\showURL{%
\tempurl}


\bibitem[Manna and Pnueli(1974)]%
        {manna1974axiomatic}
\bibfield{author}{\bibinfo{person}{Zohar Manna} {and} \bibinfo{person}{Amir
  Pnueli}.} \bibinfo{year}{1974}\natexlab{}.
\newblock \showarticletitle{{Axiomatic Approach to Total Correctness of
  Programs}}.
\newblock \bibinfo{journal}{\emph{Acta Inf.}} \bibinfo{volume}{3},
  \bibinfo{number}{3} (\bibinfo{date}{sep} \bibinfo{year}{1974}),
  \bibinfo{pages}{243–263}.
\newblock
\showISSN{0001-5903}
\href{https://doi.org/10.1007/BF00288637}{doi:\nolinkurl{10.1007/BF00288637}}


\bibitem[McGeoch and Sleator(1991)]%
        {mcgeoch1991strongly}
\bibfield{author}{\bibinfo{person}{Lyle~A. McGeoch} {and}
  \bibinfo{person}{Daniel~D. Sleator}.} \bibinfo{year}{1991}\natexlab{}.
\newblock \showarticletitle{A strongly competitive randomized paging
  algorithm}.
\newblock \bibinfo{journal}{\emph{Algorithmica}} \bibinfo{volume}{6},
  \bibinfo{number}{1–6} (\bibinfo{date}{June} \bibinfo{year}{1991}),
  \bibinfo{pages}{816–825}.
\newblock
\showISSN{0178-4617}
\href{https://doi.org/10.1007/BF01759073}{doi:\nolinkurl{10.1007/BF01759073}}


\bibitem[McIver and Morgan(2005)]%
        {mciver2005abstraction}
\bibfield{author}{\bibinfo{person}{Annabelle McIver} {and}
  \bibinfo{person}{Carroll Morgan}.} \bibinfo{year}{2005}\natexlab{}.
\newblock \bibinfo{booktitle}{\emph{{Abstraction, Refinement and Proof for
  Probabilistic Systems}}}.
\newblock \bibinfo{publisher}{Springer}.
\newblock
\showISBNx{9780387401157}
\showLCCN{2004057839}
\href{https://doi.org/10.1007/b138392}{doi:\nolinkurl{10.1007/b138392}}


\bibitem[Moggi(1989)]%
        {monads}
\bibfield{author}{\bibinfo{person}{E. Moggi}.} \bibinfo{year}{1989}\natexlab{}.
\newblock \showarticletitle{{Computational lambda-calculus and monads}}. In
  \bibinfo{booktitle}{\emph{[1989] Proceedings. Fourth Annual Symposium on
  Logic in Computer Science}}. \bibinfo{pages}{14--23}.
\newblock
\href{https://doi.org/10.1109/LICS.1989.39155}{doi:\nolinkurl{10.1109/LICS.1989.39155}}


\bibitem[Moggi(1991)]%
        {moggi91}
\bibfield{author}{\bibinfo{person}{Eugenio Moggi}.}
  \bibinfo{year}{1991}\natexlab{}.
\newblock \showarticletitle{{Notions of computation and monads}}.
\newblock \bibinfo{journal}{\emph{Information and Computation}}
  \bibinfo{volume}{93}, \bibinfo{number}{1} (\bibinfo{year}{1991}),
  \bibinfo{pages}{55--92}.
\newblock
\showISSN{0890-5401}
\href{https://doi.org/10.1016/0890-5401(91)90052-4}{doi:\nolinkurl{10.1016/0890-5401(91)90052-4}}


\bibitem[Motwani and Raghavan(1995)]%
        {motwani1995randomized}
\bibfield{author}{\bibinfo{person}{Rajeev Motwani} {and}
  \bibinfo{person}{Prabhakar Raghavan}.} \bibinfo{year}{1995}\natexlab{}.
\newblock \bibinfo{booktitle}{\emph{Randomized Algorithms}}.
\newblock \bibinfo{publisher}{Cambridge University Press}.
\newblock


\bibitem[Moura and Ullrich(2021)]%
        {moura2021lean}
\bibfield{author}{\bibinfo{person}{Leonardo~de Moura} {and}
  \bibinfo{person}{Sebastian Ullrich}.} \bibinfo{year}{2021}\natexlab{}.
\newblock \showarticletitle{The Lean 4 Theorem Prover and Programming
  Language}. In \bibinfo{booktitle}{\emph{Automated Deduction – CADE 28: 28th
  International Conference on Automated Deduction, Virtual Event, July 12–15,
  2021, Proceedings}}. \bibinfo{publisher}{Springer-Verlag},
  \bibinfo{address}{Berlin, Heidelberg}, \bibinfo{pages}{625–635}.
\newblock
\showISBNx{978-3-030-79875-8}
\href{https://doi.org/10.1007/978-3-030-79876-5_37}{doi:\nolinkurl{10.1007/978-3-030-79876-5_37}}


\bibitem[O'Hearn(2004)]%
        {csl}
\bibfield{author}{\bibinfo{person}{Peter~W. O'Hearn}.}
  \bibinfo{year}{2004}\natexlab{}.
\newblock \showarticletitle{{Resources, Concurrency and Local Reasoning}}. In
  \bibinfo{booktitle}{\emph{CONCUR 2004 - Concurrency Theory}}.
  \bibinfo{publisher}{Springer Berlin Heidelberg}, \bibinfo{address}{Berlin,
  Heidelberg}, \bibinfo{pages}{49--67}.
\newblock
\showISBNx{978-3-540-28644-8}
\href{https://doi.org/10.1016/j.tcs.2006.12.035}{doi:\nolinkurl{10.1016/j.tcs.2006.12.035}}


\bibitem[O'Hearn et~al\mbox{.}(2001)]%
        {localreasoning}
\bibfield{author}{\bibinfo{person}{Peter~W. O'Hearn}, \bibinfo{person}{John~C.
  Reynolds}, {and} \bibinfo{person}{Hongseok Yang}.}
  \bibinfo{year}{2001}\natexlab{}.
\newblock \showarticletitle{{Local Reasoning about Programs That Alter Data
  Structures}}. In \bibinfo{booktitle}{\emph{Proceedings of the 15th
  International Workshop on Computer Science Logic}}
  \emph{(\bibinfo{series}{CSL '01})}. \bibinfo{publisher}{Springer-Verlag},
  \bibinfo{address}{Berlin, Heidelberg}, \bibinfo{pages}{1–19}.
\newblock
\showISBNx{3540425543}
\href{https://doi.org/10.1007/3-540-44802-0_1}{doi:\nolinkurl{10.1007/3-540-44802-0_1}}


\bibitem[Raghavan and Snir(1994)]%
        {raghavan1994memory}
\bibfield{author}{\bibinfo{person}{P. Raghavan} {and} \bibinfo{person}{M.
  Snir}.} \bibinfo{year}{1994}\natexlab{}.
\newblock \showarticletitle{Memory versus randomization in on-line algorithms}.
\newblock \bibinfo{journal}{\emph{IBM J. Res. Dev.}} \bibinfo{volume}{38},
  \bibinfo{number}{6} (\bibinfo{date}{Nov.} \bibinfo{year}{1994}),
  \bibinfo{pages}{683–707}.
\newblock
\showISSN{0018-8646}
\href{https://doi.org/10.1147/rd.386.0683}{doi:\nolinkurl{10.1147/rd.386.0683}}


\bibitem[Reynolds(2002)]%
        {sl}
\bibfield{author}{\bibinfo{person}{J.C. Reynolds}.}
  \bibinfo{year}{2002}\natexlab{}.
\newblock \showarticletitle{{Separation logic: a logic for shared mutable data
  structures}}. In \bibinfo{booktitle}{\emph{Proceedings 17th Annual IEEE
  Symposium on Logic in Computer Science}}. \bibinfo{pages}{55--74}.
\newblock
\href{https://doi.org/10.1109/LICS.2002.1029817}{doi:\nolinkurl{10.1109/LICS.2002.1029817}}


\bibitem[Royden(1968)]%
        {royden1968real}
\bibfield{author}{\bibinfo{person}{H.~L. Royden}.}
  \bibinfo{year}{1968}\natexlab{}.
\newblock \bibinfo{booktitle}{\emph{Real Analysis} (\bibinfo{edition}{2d ed.}
  ed.)}.
\newblock \bibinfo{publisher}{Macmillan}, \bibinfo{address}{New York}.
\newblock
\showLCCN{68010518}


\bibitem[Scott(1970)]%
        {scott1970outline}
\bibfield{author}{\bibinfo{person}{Dana Scott}.}
  \bibinfo{year}{1970}\natexlab{}.
\newblock \bibinfo{booktitle}{\emph{{Outline of a Mathematical Theory of
  Computation}}}.
\newblock \bibinfo{type}{{T}echnical {R}eport} PRG02.
  \bibinfo{institution}{OUCL}. \bibinfo{pages}{30} pages.
\newblock


\bibitem[Scott and Strachey(1971)]%
        {scott1971towards}
\bibfield{author}{\bibinfo{person}{Dana Scott} {and} \bibinfo{person}{C.
  Strachey}.} \bibinfo{year}{1971}\natexlab{}.
\newblock \showarticletitle{{Towards a Mathematical Semantics for Computer
  Languages}}.
\newblock \bibinfo{journal}{\emph{Proceedings of the Symposium on Computers and
  Automata}}  \bibinfo{volume}{21} (\bibinfo{date}{01} \bibinfo{year}{1971}).
\newblock


\bibitem[Selvin(1975)]%
        {selvin1975letters}
\bibfield{author}{\bibinfo{person}{Steve Selvin}.}
  \bibinfo{year}{1975}\natexlab{}.
\newblock \showarticletitle{Letters to the Editor}.
\newblock \bibinfo{journal}{\emph{The American Statistician}}
  \bibinfo{volume}{29}, \bibinfo{number}{1} (\bibinfo{year}{1975}),
  \bibinfo{pages}{67--71}.
\newblock
\href{https://doi.org/10.1080/00031305.1975.10479121}{doi:\nolinkurl{10.1080/00031305.1975.10479121}}


\bibitem[Tassarotti and Harper(2019)]%
        {polaris}
\bibfield{author}{\bibinfo{person}{Joseph Tassarotti} {and}
  \bibinfo{person}{Robert Harper}.} \bibinfo{year}{2019}\natexlab{}.
\newblock \showarticletitle{{A Separation Logic for Concurrent Randomized
  Programs}}.
\newblock \bibinfo{journal}{\emph{Proc. ACM Program. Lang.}}
  \bibinfo{volume}{3}, \bibinfo{number}{POPL}, Article \bibinfo{articleno}{64}
  (\bibinfo{date}{Jan} \bibinfo{year}{2019}), \bibinfo{numpages}{30}~pages.
\newblock
\href{https://doi.org/10.1145/3290377}{doi:\nolinkurl{10.1145/3290377}}


\bibitem[Tix et~al\mbox{.}(2009)]%
        {tix2009semantic}
\bibfield{author}{\bibinfo{person}{Regina Tix}, \bibinfo{person}{Klaus Keimel},
  {and} \bibinfo{person}{Gordon Plotkin}.} \bibinfo{year}{2009}\natexlab{}.
\newblock \showarticletitle{Semantic Domains for Combining Probability and
  Non-Determinism}.
\newblock \bibinfo{journal}{\emph{Electronic Notes in Theoretical Computer
  Science}}  \bibinfo{volume}{222} (\bibinfo{year}{2009}),
  \bibinfo{pages}{3--99}.
\newblock
\showISSN{1571-0661}
\href{https://doi.org/10.1016/j.entcs.2009.01.002}{doi:\nolinkurl{10.1016/j.entcs.2009.01.002}}


\bibitem[Varacca(2002)]%
        {varacca2002powerdomain}
\bibfield{author}{\bibinfo{person}{Daniele Varacca}.}
  \bibinfo{year}{2002}\natexlab{}.
\newblock \showarticletitle{The powerdomain of indexed valuations}. In
  \bibinfo{booktitle}{\emph{Proceedings 17th Annual IEEE Symposium on Logic in
  Computer Science}}. \bibinfo{pages}{299--308}.
\newblock
\href{https://doi.org/10.1109/LICS.2002.1029838}{doi:\nolinkurl{10.1109/LICS.2002.1029838}}


\bibitem[Zilberstein(2025)]%
        {zilberstein2025outcome}
\bibfield{author}{\bibinfo{person}{Noam Zilberstein}.}
  \bibinfo{year}{2025}\natexlab{}.
\newblock \showarticletitle{Outcome Logic: A Unified Approach to the Metatheory
  of Program Logics with Branching Effects}.
\newblock \bibinfo{journal}{\emph{ACM Trans. Program. Lang. Syst.}}
  \bibinfo{volume}{47}, \bibinfo{number}{3}, Article \bibinfo{articleno}{14}
  (\bibinfo{date}{Sept.} \bibinfo{year}{2025}), \bibinfo{numpages}{71}~pages.
\newblock
\showISSN{0164-0925}
\href{https://doi.org/10.1145/3743131}{doi:\nolinkurl{10.1145/3743131}}


\bibitem[Zilberstein et~al\mbox{.}(2023)]%
        {outcome}
\bibfield{author}{\bibinfo{person}{Noam Zilberstein}, \bibinfo{person}{Derek
  Dreyer}, {and} \bibinfo{person}{Alexandra Silva}.}
  \bibinfo{year}{2023}\natexlab{}.
\newblock \showarticletitle{{Outcome Logic: A Unifying Foundation for
  Correctness and Incorrectness Reasoning}}.
\newblock \bibinfo{journal}{\emph{Proc. ACM Program. Lang.}}
  \bibinfo{volume}{7}, \bibinfo{number}{OOPSLA1}, Article
  \bibinfo{articleno}{93} (\bibinfo{date}{Apr} \bibinfo{year}{2023}),
  \bibinfo{numpages}{29}~pages.
\newblock
\href{https://doi.org/10.1145/3586045}{doi:\nolinkurl{10.1145/3586045}}


\bibitem[Zilberstein et~al\mbox{.}(2025a)]%
        {zilberstein2025denotational}
\bibfield{author}{\bibinfo{person}{Noam Zilberstein}, \bibinfo{person}{Daniele
  Gorla}, {and} \bibinfo{person}{Alexandra Silva}.}
  \bibinfo{year}{2025}\natexlab{a}.
\newblock \showarticletitle{{Denotational Semantics for Probabilistic and
  Concurrent Programs}}. In \bibinfo{booktitle}{\emph{36th International
  Conference on Concurrency Theory (CONCUR 2025)}}
  \emph{(\bibinfo{series}{Leibniz International Proceedings in Informatics
  (LIPIcs)}, Vol.~\bibinfo{volume}{348})},
  \bibfield{editor}{\bibinfo{person}{Patricia Bouyer} {and}
  \bibinfo{person}{Jaco van~de Pol}} (Eds.). \bibinfo{publisher}{Schloss
  Dagstuhl -- Leibniz-Zentrum f{"u}r Informatik}, \bibinfo{address}{Dagstuhl,
  Germany}, \bibinfo{pages}{39:1--39:24}.
\newblock
\showISBNx{978-3-95977-389-8}
\showISSN{1868-8969}
\href{https://doi.org/10.4230/LIPIcs.CONCUR.2025.39}{doi:\nolinkurl{10.4230/LIPIcs.CONCUR.2025.39}}


\bibitem[Zilberstein et~al\mbox{.}(2025b)]%
        {zilberstein2025demonic}
\bibfield{author}{\bibinfo{person}{Noam Zilberstein}, \bibinfo{person}{Dexter
  Kozen}, \bibinfo{person}{Alexandra Silva}, {and} \bibinfo{person}{Joseph
  Tassarotti}.} \bibinfo{year}{2025}\natexlab{b}.
\newblock \showarticletitle{A Demonic Outcome Logic for Randomized
  Nondeterminism}.
\newblock \bibinfo{journal}{\emph{Proc. ACM Program. Lang.}}
  \bibinfo{volume}{9}, \bibinfo{number}{POPL}, Article \bibinfo{articleno}{19}
  (\bibinfo{date}{Jan} \bibinfo{year}{2025}), \bibinfo{numpages}{30}~pages.
\newblock
\href{https://doi.org/10.1145/3704855}{doi:\nolinkurl{10.1145/3704855}}


\bibitem[Zilberstein et~al\mbox{.}(2024)]%
        {zilberstein2024outcome}
\bibfield{author}{\bibinfo{person}{Noam Zilberstein}, \bibinfo{person}{Angelina
  Saliling}, {and} \bibinfo{person}{Alexandra Silva}.}
  \bibinfo{year}{2024}\natexlab{}.
\newblock \showarticletitle{Outcome Separation Logic: Local Reasoning for
  Correctness and Incorrectness with Computational Effects}.
\newblock \bibinfo{journal}{\emph{Proc. ACM Program. Lang.}}
  \bibinfo{volume}{8}, \bibinfo{number}{OOPSLA1}, Article
  \bibinfo{articleno}{104} (\bibinfo{date}{apr} \bibinfo{year}{2024}),
  \bibinfo{numpages}{29}~pages.
\newblock
\href{https://doi.org/10.1145/3649821}{doi:\nolinkurl{10.1145/3649821}}


\bibitem[Zilberstein et~al\mbox{.}(2026)]%
        {zilberstein2026probabilistic}
\bibfield{author}{\bibinfo{person}{Noam Zilberstein},
  \bibinfo{person}{Alexandra Silva}, {and} \bibinfo{person}{Joseph
  Tassarotti}.} \bibinfo{year}{2026}\natexlab{}.
\newblock \showarticletitle{{Probabilistic Concurrent Reasoning in Outcome
  Logic: Independence, Conditioning, and Invariants}}.
\newblock \bibinfo{journal}{\emph{Proc. ACM Program. Lang.}}
  \bibinfo{volume}{10}, \bibinfo{number}{POPL} (\bibinfo{date}{Jan}
  \bibinfo{year}{2026}), \bibinfo{numpages}{30}~pages.
\newblock
\href{https://doi.org/10.1145/3776651}{doi:\nolinkurl{10.1145/3776651}}


\end{thebibliography}

\appendix
\clearpage

\crefalias{section}{appendix}
\crefalias{subsection}{appendix}   
\crefname{appendix}{appendix}{appendices}   \Crefname{appendix}{Appendix}{Appendices}

\ifx\apponly\undefined\else
\setcounter{page}{1}
\fi

{\noindent \huge\bfseries\sffamily Appendix}

\startcontents[appendix]
\printcontents[appendix]{}{1}{\setcounter{tocdepth}{2}}
\clearpage


\section{Omitted Definitions}
\noam{I just dumped a bunch of stuff here for now. We need to sort it out and organize it a bit}
This section collects auxiliary material used by the main development but
omitted from the body.

\subsection{Resource}
In this section we describe more definitions of resources useful in the
development of \LogicName.

\paragraph{Projection}
Fix a resource $\calR$ and an event $G\in\calF_\calR$ with
$\mu_\calR(G)>0$,
the conditioned
resource is defined as:\noam{I'm pretty sure this was defined earlier too, can
 we remove one of the definitions?}\gary{Well before it was conditional on
 probability space}
\[\calR\mid G
\triangleq
\langle
\calP_\calR\mid G,\,
V_\calR,\,
\memfn_{\calR\mid G},\,
\cntfn_{\calR\mid G}
\rangle
\]
where
$\memfn_{\calR\mid G}\triangleq\RestrictFn{\memfn_\calR}{G}\colon
G\to\Mem[V_\calR]$ and
$\cntfn_{\calR\mid G}\triangleq\RestrictFn{\cntfn_\calR}{G}\colon G\to\bbN$
are the restrictions of the observable maps to domain $G$.
Thus conditioning restricts the observable functions to the conditioned index
space.
Since $\cntfn_\calR$ is discretely measurable (i.e., 
$\cntfn_\calR^{-1}(B)\in\calF_\calR$ for every $B\subseteq\bbN$), for every
$B\subseteq\bbN$ with $\mu_\calR(\cntfn_\calR^{-1}(B))>0$, the scheduling
conditioning $\RestrictN{\calR}{B}$ is defined as the conditioned resource
$\RestrictN{\calR}{B}
\triangleq
\calR\mid\cntfn_\calR^{-1}(B)$.

To mirror the semantic increment of scheduler indices, we define
$\Tick(\calR)$ as the resource with the same probability space and memory
observation as $\calR$, but with count observation incremented by one; \ie
$\Tick(\calR) \triangleq \langle \calP_\calR, V_\calR, \memfn_\calR,
\cntfn_{\Tick(\calR)} \rangle$ where
$\cntfn_{\Tick(\calR)}(\omega)\triangleq \cntfn_\calR(\omega)+1$.

Fix a resource $\calR$ and a finite set $S\subseteq_\fin\Var$.  The memory projection is defined as:
\[
  \ForgetMem{\calR}{S} \triangleq
  \langle \calP_\calR, V_\calR\cap S,
  \memfn_{\ForgetMem{\calR}{S}}, \cntfn_\calR\rangle
  \quad\text{where}\quad
\memfn_{\ForgetMem{\calR}{S}}
  \triangleq
\ForgetMem{\memfn_\calR(\omega)}{S}
\]
Thus memory projection forgets only memory information.

\paragraph{Unital Resource} The unital resource is
$\UnitP
\triangleq
\langle \calP_{\UnitP}, \emptyset, \memfn_{\UnitP}, \cntfn_{\UnitP}
\rangle$
where
$\calP_{\UnitP} \triangleq
\langle \{0\}, \{\emptyset, \{0\}\}, \delta_{0}\rangle$,
and $\memfn_{\UnitP}(0)=\emptyset$, $\cntfn_{\UnitP}(0)=0$. It is the identity of the product resource, \ie $\calR \otimes \UnitP \cong \UnitP\otimes \calR \cong \calR$.

\subsection{Assertions}
\label{app:assertions}

\begin{figure}
\[
  \begin{array}{ll}
    \Gamma, \sigma \vDash \true
    & \text{always}\\
    \Gamma, \sigma \vDash \false
    & \text{never}\\
    \Gamma, \sigma \vDash P \land Q
    & \text{iff}\quad \Gamma, \sigma \vDash P \quad \text{and} \quad
      \Gamma, \sigma \vDash Q\\
    \Gamma, \sigma \vDash P \lor Q
    & \text{iff}\quad \Gamma, \sigma \vDash P \quad \text{or} \quad
      \Gamma, \sigma \vDash Q\\
    \Gamma, \sigma \vDash P * Q
    & \text{iff}\quad \exists \sigma_1, \sigma_2.\;
      \sigma_1 \Perp \sigma_2 \quad \text{and} \quad
      \sigma_1 \uplus \sigma_2 \sqsubseteq \sigma \quad \text{and} \quad
      \Gamma, \sigma_1 \vDash P \quad \text{and} \quad \Gamma, \sigma_2 \vDash Q\\
    \Gamma, \sigma \vDash \exists X.\, P
    & \text{iff}\quad \Gamma[X := v], \sigma \vDash P \; \text{for some} \; v \in \Val\\
    \Gamma, \sigma \vDash e \mapsto E
    & \text{iff}\quad \de{e}_{\Exp}(\sigma) = \de{E}_{\mathsf{LExp}}(\Gamma)\\
    \Gamma, \sigma \vDash E_1 \asymp E_2
    & \text{iff}\quad \de{E_1}(\Gamma) \asymp \de{E_2}(\Gamma)
  \end{array}
\]
\caption{Semantics of pure separation logic assertions.}
\label{fig:pure-assertion-semantics}
\end{figure}

\Cref{fig:pure-assertion-semantics} gives the interpretation of pure assertions. Most clauses
are standard; the separating conjunction $P \sep Q$ asserts a decomposition of
memory into disjoint components, the points-to assertion $e \mapsto E$
relates program expressions to logical expressions under the current memory and
logical environment, and $E_1 \asymp E_2$ compares the evaluation of two logical
expressions using the comparison operator $\asymp$.



\begin{figure}
\begin{minipage}{\linewidth}
\textbf{Precision}
\begin{mathpar}
  \inferrule
  {}
  {\precise{\sure{P}}}

  \inferrule
  {\precise{\varphi} \and \precise{\psi} \and \textsf{MutEx}(\varphi, \psi)}
  {\precise{\varphi \lor \psi}}
  
  \inferrule
  {\precise{\varphi} \and \precise{\psi}}
  {\precise{\varphi\sep\psi}}

  \inferrule
  {\precise{\varphi}}
  {\textstyle
    \precise{\bigoplus^o_{X\sim d(E)}\varphi}}
\end{mathpar}

\textbf{Weak stability}
\begin{mathpar}
  \inferrule
  {}
  {\Stable{\weak}{\sure{P}}}

  \inferrule
  {\Stable{\weak}{\varphi_1} \and \Stable{\weak}{\varphi_2}}
  {\Stable{\weak}{\varphi_1 \lor \varphi_2}}

  \inferrule
  {\Stable{\weak}{\varphi_1} \and \Stable{\weak}{\varphi_2}}
  {\Stable{\weak}{\varphi_1 \sep \varphi_2}}

  \inferrule
  {\Stable{\weak}{\varphi}}
  {\textstyle
    \Stable{\weak}{\bigoplus^\Priv_{X\sim d(E)}\varphi}}
\end{mathpar}

\textbf{Convexity}
\begin{mathpar}
  \inferrule
  {}
  {\convex{\top}}

  \inferrule
  {}
  {\convex{\bot}}

  \inferrule
  {}
  {\convex{\sure{P}}}

  \inferrule
  {\convex{\varphi} \and \convex{\psi}}
  {\convex{\varphi\land\psi}}

  \inferrule
  {\precise{\varphi} \and \Stable{\weak}{\varphi} \and \convex{\psi}}
  {\convex{\varphi\sep\psi}}

  \inferrule
  {\convex{\varphi} \and \precise{\psi} \and \Stable{\weak}{\psi}}
  {\convex{\varphi\sep\psi}}

  \inferrule
  {\convex{\varphi}}
  {\textstyle
    \convex{\bigoplus^o_{X\sim d(E)}\varphi}}

  \inferrule
  {\convex{\varphi}}
  {\textstyle
    \convex{\bignd_{X\in E}\varphi}}
\end{mathpar}
\end{minipage}
\caption{Constructor rules for precision, weak stability, and convexity.}
\label{fig:assertion-property-rules}
\end{figure}

\Cref{fig:assertion-property-rules} presents a more complete list of rules for
the three assertion properties used throughout \LogicName, i.e. precision, weak
stability, and convexity.  The side condition
$\textsf{MutEx}(\varphi,\psi)$ means that $\varphi$ and $\psi$ cannot both be
satisfiable under the same logical environment:
\[
  \textsf{MutEx}(\varphi,\psi)
  \quad\triangleq\quad
  \forall \Gamma.\;
  \neg\bigl(
    (\exists\calR.\;\Gamma,\calR\vDash\varphi)
    \land
    (\exists\calQ.\;\Gamma,\calQ\vDash\psi)
  \bigr)
\]
It is used for the precision rule for disjunction.
We include the disjunction rule because binary probabilistic choice uses
disjunction in its syntactic expansion.
In particular, the Boolean-tagged
branches in $\varphi\oplus_p\psi
\triangleq
\bigoplus^\Priv_{X\sim\bern{p}}
\bigl((\sure{X=1}\sep\varphi)\lor(\sure{X=0}\sep\psi)\bigr)$ are mutually exclusive, so the disjunction
precision rule, together with the $\bigoplus$ rule, gives precision of
$\varphi\oplus_p\psi$ from precision of $\varphi$ and $\psi$.  The corresponding
weak-stability fact follows from the weak-stability rules for disjunction and
private probabilistic choice.

The entailment laws shown in \Cref{fig:oplus-entailment-laws} compare assertions by their semantic interpretation.
We write $\varphi\vdash\psi$ when every logical context $\Gamma$
and resource $\calR$ that satisfies $\varphi$ also
satisfies $\psi$.
\begin{definition}[Semantic Entailment]
  \label{def:assertion-entailment}
  For outcome assertions $\varphi$ and $\psi$, we write
  $\varphi \vdash \psi$ when every resource satisfying $\varphi$ also satisfies
  $\psi$:
  \[
    \varphi \vdash \psi
    \qquad\text{iff}\qquad
    \forall \Gamma,\calR.\;
      \Gamma,\calR \vDash \varphi
      \Rightarrow
      \Gamma,\calR \vDash \psi
  \]
\end{definition}
We stated a more comprehensive list of entailment laws proven in \LogicName in
\Cref{fig:basic-entailment-laws,fig:separation-entailment-laws,fig:oplus-entailment-laws}
with respect to \Cref{def:assertion-entailment}.

\begin{figure}
\begin{mathpar}
  \inferrule
  {}
  {\textstyle
    \varphi \oplus^o_p \psi \dashv\vdash \psi \oplus^o_{(1 - p)} \varphi}

  \inferrule
  {pq < 1}
  {\textstyle
    (\varphi \oplus^o_p \psi) \oplus^o_q \vartheta
    \dashv\vdash
    \varphi \oplus^o_{pq}
    (\psi \oplus^o_{(1 - p)q/(1 - pq)} \vartheta)}

  \inferrule
  {X\notin\fv(\varphi) \and \convex{\varphi}}
  {\textstyle
    \bignd_{X\in E}\varphi \vdash \varphi}

  \inferrule
  {}
  {\textstyle
    \varphi[E/X] \vdash \bignd_{X\in\{E\}}\varphi}

  \inferrule
  {}
  {\textstyle
    \bignd_{X\in E}\sure{x\mapsto X}
    \vdash
    \sure{x\mapsto E}}

  \inferrule
  {E\subseteq E'}
  {\textstyle
    \bignd_{X\in E}\varphi \vdash \bignd_{X\in E'}\varphi}

  \inferrule
  {Y\notin\fv(\varphi)}
  {\textstyle
    \bigoplus_{X\sim d(E)}^o\varphi
    \dashv\vdash
    \bigoplus_{Y\sim d(E)}^o\varphi[Y/X]}

  \inferrule
  {}
  {\textstyle
    \andop{\varphi}{\psi} \dashv\vdash \andop{\psi}{\varphi}}

  \inferrule
  {}
  {\textstyle
    \andop{\varphi}{(\andop{\psi}{\vartheta})}
    \dashv\vdash
    \andop{(\andop{\varphi}{\psi})}{\vartheta}}

  \inferrule
  {}
  {\textstyle
    \andop{(\varphi \oplus^o_p \psi)}{\vartheta}
    \dashv\vdash
    (\andop{\varphi}{\vartheta}) \oplus^o_p (\andop{\psi}{\vartheta})}

  \inferrule
  {}
  {\textstyle
    \varphi\oplus_0\psi \dashv\vdash \psi}

  \inferrule
  {}
  {\textstyle
    \varphi\oplus_1\psi \dashv\vdash \varphi}
\end{mathpar}
\caption{Basic entailment laws for outcome assertions.}
\label{fig:basic-entailment-laws}
\end{figure}

\begin{figure}
\begin{mathpar}
  \inferrule
  {\varphi \vdash \varphi' \and \psi \vdash \psi'}
  {\textstyle
    \varphi \sep \psi \vdash \varphi' \sep \psi'}

  \inferrule
  {}
  {\textstyle
    \varphi \sep \psi \dashv\vdash \psi \sep \varphi}

  \inferrule
  {}
  {\textstyle
    (\varphi \sep \psi) \sep \theta
  \dashv\vdash
    \varphi \sep (\psi \sep \theta)}

  \inferrule
  {}
  {\textstyle
    \sure{P * Q} \dashv\vdash \sure{P} \sep \sure{Q}}

  \inferrule
  {}
  {\textstyle
    \sure{P} \vdash \sure{P}\sep\top}

  \inferrule
  {}
  {\textstyle
    \varphi \sep \psi \vdash \varphi}
\end{mathpar}
\caption{Entailment laws for separating conjunction.}
\label{fig:separation-entailment-laws}
\end{figure}

\subsection{Semantics}
In \Cref{sec:logic} we omitted the definition of the function
$\textsf{bits}(C)$, which is a syntactic approximation of the possible
numbers of scheduler entries consumed by $C$. The singleton condition $\textsf{bits}(C)=\{n\}$,
used by the rule \ruleref{Priv-Split1}, means
that every execution consumes exactly $n$ adversarial choices. The \ruleref{Prob-Safe} rule uses the analogous
branch-balance condition $\textsf{bits}(C_1)=\textsf{bits}(C_2)=\{n\}$.
\begin{definition}[Scheduler-consumption approximation]
  \label{app:bits}
\begin{align*}
  \textsf{bits}(\cmdskip) &= \{0\}\\
  \textsf{bits}(\actassign{x}{e}) &= \{0\}\\
  \textsf{bits}(x\samp d(\vec e)) &= \{0\}\\
  \textsf{bits}(x\gets e) &= \{1\}\\
  \textsf{bits}(\cmdseq{C_1}{C_2})
    &= \{k_1+k_2\mid k_1\in\textsf{bits}(C_1),\;
      k_2\in\textsf{bits}(C_2)\}\\
  \textsf{bits}(\cmdcase{b}{C_1}{C_2}) &= \textsf{bits}(C_1) \cup \textsf{bits}(C_2)\\
  \textsf{bits}(\cmdchoice{e}{C_1}{C_2}) &= \textsf{bits}(C_1) \cup \textsf{bits}(C_2)\\
  \textsf{bits}(\cmdnondet{C_1}{C_2})
    &= \{1+k\mid k\in\textsf{bits}(C_1)\cup\textsf{bits}(C_2)\}\\
  \textsf{bits}(\cmdwhile{e}{C})
    &= \{n\cdot k\mid n\in\bbN,\; k\in\textsf{bits}(C)\}
\end{align*}
\end{definition}
The base cases distinguish commands that do not query the scheduler from
adversarial assignment, which consumes one entry.  Sequential composition adds
consumption, branching forms take unions over possible branches, and
nondeterministic choice adds one entry for the branch selection.  The loop case
collects the consumption of any finite number of iterations.

\section{Semantics}
\subsection{Distribution}
\begin{lemma}[$\calD_\bot$ Kleisli extension is linear]
  \label{lem:index-dbot-kleisli-linearity}
  Let $X,Y$ be countable sets, let $I$ be countable, let
  $\alpha\in\calD(I)$, let $\mu_i\in\calD_\bot(X)$ for every
  $i\in\supp(\alpha)$, and let $f:X\to\calD_\bot(Y)$.  Then
  \[
    f^\dagger
    \left(\bigoplus_{i\sim\alpha}\mu_i\right)
    =
    \bigoplus_{i\sim\alpha}
    f^\dagger(\mu_i)
  \]
\end{lemma}
\begin{proof}
  Let
  $\mu_\oplus\triangleq\bigoplus_{i\sim\alpha}\mu_i$.  Write
  $f_\bot:X\cup\{\bot\}\to\calD_\bot(Y)$ for the extension of $f$ with
  $f_\bot(\bot)=\delta_\bot$.  We prove equality pointwise.  For every
  $y\in Y\cup\{\bot\}$,
  \begin{align*}
    \left(
      f^{\dagger}(\mu_\oplus)
    \right)(y)
    &=
      \sum_{x\in X\cup\{\bot\}}
      \mu_\oplus(x)\cdot f_\bot(x)(y)
      \tag{definition of Kleisli extension}\\
    &=
      \sum_{x\in X\cup\{\bot\}}
      \left(
        \sum_{i\in\supp(\alpha)}
        \alpha(i)\cdot\mu_i(x)
      \right)
      f_\bot(x)(y)
      \tag{definition of $\mu_\oplus$}\\
    &=
      \sum_{i\in\supp(\alpha)}
      \alpha(i)
      \sum_{x\in X\cup\{\bot\}}
      \mu_i(x)\cdot f_\bot(x)(y)
      \tag{exchange of nonnegative countable sums}\\
    &=
      \sum_{i\in\supp(\alpha)}
      \alpha(i)\cdot
      \left(
        f^{\dagger}(\mu_i)
      \right)(y)
      \tag{definition of Kleisli extension}\\
    &=
      \left(
        \bigoplus_{i\sim\alpha}
        f^{\dagger}(\mu_i)
      \right)(y).
      \tag{definition of convex mixture}
  \end{align*}
  Since $y$ was arbitrary, the two distributions are equal.
\end{proof}

\subsection{Denotational Semantics}

\begin{lemma}[$\calD_\bot$ Kleisli extension is Scott-continuous in the continuation]
  \label{lem:index-dbot-kleisli-continuation-scott-continuous}
  \noam{We can just cite Lemma C.10 of \citet{li2025total}}
  Let $X,Y$ be countable sets and let $\mu\in\calD_\bot(X)$.  The map
  \[
    K_\mu:(X\to\calD_\bot(Y))\to\calD_\bot(Y),
    \qquad
    K_\mu(f)\triangleq f^\dagger(\mu),
  \]
  is Scott-continuous, where $X\to\calD_\bot(Y)$ is ordered pointwise by
  $\led$.
\end{lemma}
\begin{proof}
  Monotonicity is immediate from the ordinary-coordinate formula.  If
  $f(x)\led g(x)$ for every $x\in X$, then for every ordinary output
  $y\in Y$,
  \[
    K_\mu(f)(y)
    =
    \sum_{x\in X}\mu(x)\cdot f(x)(y)
    \le
    \sum_{x\in X}\mu(x)\cdot g(x)(y)
    =
    K_\mu(g)(y),
  \]
  hence $K_\mu(f)\led K_\mu(g)$.

  Let $(f_i)_{i\in I}$ be a nonempty directed family and put
  $f^\ast\triangleq\bigsqcup_{i\in I}f_i$ pointwise.  For every ordinary
  $y\in Y$, using the standard interchange of directed suprema with
  nonnegative countable weighted sums,
  \[
    K_\mu(f^\ast)(y)
    =
    \sum_{x\in X}\mu(x)\cdot f^\ast(x)(y)\\
    =
    \sum_{x\in X}\mu(x)\cdot\sup_{i\in I}f_i(x)(y)\\
    =
    \sup_{i\in I}
    \sum_{x\in X}\mu(x)\cdot f_i(x)(y)\\
    =
    \sup_{i\in I}K_\mu(f_i)(y)
  \]
  The $\bot$ coordinate is determined by the remaining mass, so
  $K_\mu(f^\ast)=\bigsqcup_{i\in I}K_\mu(f_i)$.
\end{proof}

\begin{corollary}[Oblivious Kleisli extension is Scott-continuous in the continuation]
  \label{lem:index-oblivious-kleisli-continuation-scott-continuous}
  Let $X,Y$ be countable sets and let $o\in\calO(X)$.  Define
  \[
    B_o:(X\to\calO(Y))\to\calO(Y),
    \qquad
    B_o(f)\triangleq f^\dagger(o)
  \]
  Then $B_o$ is Scott-continuous, where $X\to\calO(Y)$ is ordered pointwise by
  $\leo$ and $\calO(Y)$ is ordered by $\leo$.
\end{corollary}
\begin{proof}
  For $h:X\to\calO(Y)$ and $s\in\schedule$, define the schedule slice
  $h_s:X\times\bbN\to\calD_\bot(Y\times\bbN)$ as 
  $h_s(x,k)\triangleq h(x)(s)(k)$.
  For every $(s,n)\in\schedule\times\bbN$, define the coordinate map
  \[
    F_{s,n}:(X\to\calO(Y))\to\calD_\bot(Y\times\bbN),
    \qquad
    F_{s,n}(f)\triangleq
    f_s^{\dagger}(o(s)(n))
  \]
  We claim that $F_{s,n}$ is Scott-continuous.  Let
  $(f_i)_{i\in I}$ be a nonempty directed family in $X\to\calO(Y)$, and put
  $f^\ast\triangleq\bigsqcup_{i\in I}f_i$.  For every $(x,k)\in X\times\bbN$,
  \[
    (f^\ast)_s(x,k)
    =
    f^\ast(x)(s)(k)
    =
    \left(\bigsqcup_{i\in I}f_i(x)\right)(s)(k)
    =
    \bigsqcup_{i\in I}(f_i)_s(x,k)
  \]
  Hence $(f^\ast)_s=\bigsqcup_{i\in I}(f_i)_s$.  Therefore, by \Cref{lem:index-dbot-kleisli-continuation-scott-continuous},
  \[
    F_{s,n}(f^\ast)
    =
    ((f^\ast)_s)^\dagger(o(s)(n))
    =
    \bigsqcup_{i\in I} ((f_i)_s)^\dagger(o(s)(n))
    =
    \bigsqcup_{i\in I}F_{s,n}(f_i)
  \]
  If $f\sqsubseteq g$ in $X\to\calO(Y)$, then $f_s\sqsubseteq g_s$ in
  $(X\times\bbN)\to\calD_\bot(Y\times\bbN)$, since for every $(x,k)$,
  \[
    f_s(x,k)
    =
    f(x)(s)(k)
    \led
    g(x)(s)(k)
    =
    g_s(x,k)
  \]
  Hence
  $F_{s,n}(f)
  =
  f_s^\dagger(o(s)(n))
  \led
  g_s^\dagger(o(s)(n))
  =
  F_{s,n}(g)$
  by \Cref{lem:index-dbot-kleisli-continuation-scott-continuous}.  Thus
  $F_{s,n}$ is Scott-continuous.

  For every $f$, $s$, and $n$,
  $F_{s,n}(f)
  =
  f_s^{\dagger}(o(s)(n))
  =
  B_o(f)(s)(n)$ by definition of Kleisli extension in $\calO$.
  Therefore the coordinate maps of $B_o$ are exactly the maps
  $(F_{s,n})_{(s,n)\in\schedule\times\bbN}$.  Since directed suprema in function
  spaces are computed pointwise \cite{abramsky1995domain}, it follows that $B_o$ is Scott-continuous.
\end{proof}

\begin{lemma}[Scott continuity of the while functional]
  \label{lem:index-while-functional-scott-continuous}
  For every command $C$ and guard $e$, the functional
  \[
    \Phi_{\langle C,e\rangle}
    :
    (\Mem\to\calO(\Mem))
    \to
    (\Mem\to\calO(\Mem))
  \]
  defined by
  \[
    \Phi_{\langle C,e\rangle}(f)(\sigma)
    \triangleq
    \begin{cases}
      f^\dagger(\de{C}(\sigma)), & \de{e}_\Exp(\sigma)=\Ttrue,\\
      \eta(\sigma), & \de{e}_\Exp(\sigma)=\Tfalse
    \end{cases}
  \]
  is Scott-continuous on $\Mem\to\calO(\Mem)$, ordered pointwise by
  $\leo$.
\end{lemma}
\begin{proof}
  First, $\Phi_{\langle C,e\rangle}$ is monotone.  Suppose
  $f,g:\Mem\to\calO(\Mem)$ and $f(\tau)\leo g(\tau)$ for every
  $\tau\in\Mem$.  Fix $\sigma\in\Mem$.  If
  $\de{e}_\Exp(\sigma)=\Tfalse$, then
  \[
    \Phi_{\langle C,e\rangle}(f)(\sigma)
    =
    \eta(\sigma)
    =
    \Phi_{\langle C,e\rangle}(g)(\sigma)
  \]
  If $\de{e}_\Exp(\sigma)=\Ttrue$, then
  \[
    \Phi_{\langle C,e\rangle}(f)(\sigma)
    =
    f^\dagger(\de{C}(\sigma))
    \leo
    g^\dagger(\de{C}(\sigma))
    =
    \Phi_{\langle C,e\rangle}(g)(\sigma)
  \]
  by \Cref{lem:index-oblivious-kleisli-continuation-scott-continuous}.  Hence
  $\Phi_{\langle C,e\rangle}(f)$ is below
  $\Phi_{\langle C,e\rangle}(g)$ in the pointwise function-space order induced
  by $\leo$.

  Now let $(f_i)_{i\in I}$ be a nonempty directed family in
  $\Mem\to\calO(\Mem)$, and put
  $f^\ast\triangleq\bigsqcup_{i\in I}f_i$, with the supremum taken pointwise.
  By the monotonicity just proved, the image family
  $(\Phi_{\langle C,e\rangle}(f_i))_{i\in I}$ is directed:
  if $i,j\in I$ and $k\in I$ satisfies $f_i\sqsubseteq f_k$ and
  $f_j\sqsubseteq f_k$, then
  $\Phi_{\langle C,e\rangle}(f_i)\sqsubseteq
  \Phi_{\langle C,e\rangle}(f_k)$ and
  $\Phi_{\langle C,e\rangle}(f_j)\sqsubseteq
  \Phi_{\langle C,e\rangle}(f_k)$.

  We prove
  $\Phi_{\langle C,e\rangle}(f^\ast)
  =
  \bigsqcup_{i\in I}\Phi_{\langle C,e\rangle}(f_i)$
  pointwise in $\sigma\in\Mem$.

  If $\de{e}_\Exp(\sigma)=\Tfalse$, then
  \[
    \Phi_{\langle C,e\rangle}(f^\ast)(\sigma)
    =
    \eta(\sigma)
    =
    \bigsqcup_{i\in I}
    \Phi_{\langle C,e\rangle}(f_i)(\sigma)
  \]
  If $\de{e}_\Exp(\sigma)=\Ttrue$, then by \Cref{lem:index-oblivious-kleisli-continuation-scott-continuous}
  \[
    \Phi_{\langle C,e\rangle}(f^\ast)(\sigma)
    =
    (f^\ast)^\dagger(\de{C}(\sigma))
    =
    \bigsqcup_{i\in I}f_i^\dagger(\de{C}(\sigma))
    =
    \bigsqcup_{i\in I}
    \Phi_{\langle C,e\rangle}(f_i)(\sigma)
  \]
  Therefore $\Phi_{\langle C,e\rangle}$ preserves directed suprema.  Together
  with monotonicity, this proves Scott continuity.
\end{proof}

\begin{lemma}[Schedule-indexed while semantic equations]
  \label{lem:index-schedule-indexed-while-semantic-equations}
  Fix a schedule $s$.  For every guard $e$ and command $C$, let
  \[
    \de{\cmdwhile{e}{C}}
    =
    \textsf{lfp}(\Phi_{\langle C,e\rangle}),
    \qquad
    \Phi_{\langle C,e\rangle}(f)(\tau)
    =
    \begin{cases}
      f^\dagger(\de{C}(\tau)),
      & \de{e}_{\Exp}(\tau)=\Ttrue,\\
      \eta(\tau),
      & \de{e}_{\Exp}(\tau)=\Tfalse .
    \end{cases}
  \]
  Let $(F_m)_{m\in\bbN}$ be its Kleene chain,
  \[
    F_0(\tau)\triangleq\bot_\calO,
    \qquad
    F_{m+1}\triangleq\Phi_{\langle C,e\rangle}(F_m)
  \]
  Write
  $F_{m,s}(\tau,j)\triangleq F_m(\tau)(s)(j)$
  for the schedule-indexed projection.  Then, for every
  $\sigma\in\Mem$ and $n\in\bbN$,
  \[
    \de{\cmdwhile{e}{C}}_s(\sigma,n)
    =
    \bigsqcup_{m\in\bbN}F_{m,s}(\sigma,n)
  \]
  The approximants satisfy
  \begin{align*}
    F_{0,s}(\sigma,n)
    &=
      \delta_\bot,\\
    F_{m+1,s}(\sigma,n)
    &=
      \begin{cases}
        \bigl((F_m)_s^\bot\bigr)^\dagger(\de{C}_s(\sigma,n)),
        & \de{e}_{\Exp}(\sigma)=\Ttrue,\\
        \delta_{(\sigma,n)},
        & \de{e}_{\Exp}(\sigma)=\Tfalse .
      \end{cases}
  \end{align*}
\end{lemma}
\begin{proof}
  By the Kleene fixed-point theorem and
  \Cref{lem:index-while-functional-scott-continuous},
  \[
    \de{\cmdwhile{e}{C}}
    =
    \bigsqcup_{m\in\bbN}F_m
  \]
  Evaluating this pointwise at $\sigma,s,n$ gives
  \[
    \de{\cmdwhile{e}{C}}_s(\sigma,n)
    =
    \bigsqcup_{m\in\bbN}F_{m,s}(\sigma,n)
  \]
  The equation for $F_{0,s}$ is immediate from
  $F_0=\bot_\calO$.  The successor equation follows by unfolding
  $F_{m+1}=\Phi_{\langle C,e\rangle}(F_m)$; in the true branch one unfolds
  Kleisli composition in $\calO$, and in the false branch
  $\eta(\sigma)(s)(n)=\delta_{(\sigma,n)}$.
\end{proof}

\begin{lemma}[Schedule-indexed while approximant unfolding on distributions]
  \label{lem:index-while-approximant-distribution-unfolding}
  Fix a schedule $s$, guard $e$, and command $C$.  Let
  $(F_m)_{m\in\bbN}$ and $F_{m,s}$ be as in
  \Cref{lem:index-schedule-indexed-while-semantic-equations}.
  For $c\in\{\tru,\fls\}$, put
  \[
    G_c\triangleq
    \{(\sigma,k)\in\Mem\times\bbN
    \mid \de{e}_{\Exp}(\sigma)=c\}
  \]
  Let $\rho\in\calD(\Mem\times\bbN)$.  For each
  $c\in\{\tru,\fls\}$, if $\rho(G_c)>0$, write $\rho\mid G_c$ for the
  conditional distribution on $G_c$; if $\rho(G_c)=0$, let $\rho\mid G_c$ be
  any fixed ordinary distribution.  Then, for every $m\in\bbN$,
  \[
    \bigl((F_{m+1})_s^\bot\bigr)^\dagger(\rho)
    =
    \rho(G_{\fls})\cdot(\rho\mid G_{\fls})
    +
    \rho(G_{\tru})\cdot
    \bigl((F_m)_s^\bot\bigr)^\dagger
    \bigl(\de{C}_s^\dagger(\rho\mid G_{\tru})\bigr)
  \]
\end{lemma}
\begin{proof}
  The sets $G_{\fls}$ and $G_{\tru}$ partition $\Mem\times\bbN$.  Hence
  \begin{align*}
    \bigl((F_{m+1})_s^\bot\bigr)^\dagger(\rho)
    &=
    \bigoplus_{(\sigma,k)\sim\rho}F_{m+1,s}(\sigma,k)
    \tag{definition of $\calD_\bot$ Kleisli extension and $\rho(\bot)=0$}\\
    &=
    \rho(G_{\fls})\cdot
    \bigoplus_{(\sigma,k)\sim\rho\mid G_{\fls}}
      F_{m+1,s}(\sigma,k)
    +
    \rho(G_{\tru})\cdot
    \bigoplus_{(\sigma,k)\sim\rho\mid G_{\tru}}
      F_{m+1,s}(\sigma,k)
    \tag{partition by $G_{\fls},G_{\tru}$}\\
    &=
    \rho(G_{\fls})\cdot
    \bigoplus_{(\sigma,k)\sim\rho\mid G_{\fls}}
      \delta_{(\sigma,k)}
    +
    \rho(G_{\tru})\cdot
    \bigoplus_{(\sigma,k)\sim\rho\mid G_{\tru}}
      \bigl((F_m)_s^\bot\bigr)^\dagger(\de{C}_s(\sigma,k))
    \tag{\Cref{lem:index-schedule-indexed-while-semantic-equations}}\\
    &=
    \rho(G_{\fls})\cdot(\rho\mid G_{\fls})
    +
    \rho(G_{\tru})\cdot
    \bigl((F_m)_s^\bot\bigr)^\dagger
    \left(
      \bigoplus_{(\sigma,k)\sim\rho\mid G_{\tru}}
        \de{C}_s(\sigma,k)
    \right)
    \tag{definition of conditioning and
    \Cref{lem:index-dbot-kleisli-linearity}}\\
    &=
    \rho(G_{\fls})\cdot(\rho\mid G_{\fls})
    +
    \rho(G_{\tru})\cdot
    \bigl((F_m)_s^\bot\bigr)^\dagger
    \bigl(\de{C}_s^\dagger(\rho\mid G_{\tru})\bigr).
    \tag{definition of $\calD_\bot$ Kleisli extension}
  \end{align*}
\end{proof}

\begin{lemma}[Schedule-indexed while approximants converge on distributions]
  \label{lem:index-schedule-indexed-while-distribution-supremum}
  Fix a schedule $s$, guard $e$, and command $C$.  Let
  $(F_m)_{m\in\bbN}$ and $F_{m,s}$ be as in
  \Cref{lem:index-schedule-indexed-while-semantic-equations}.
  For every $\mu\in\calD(\Mem\times\bbN)$, regarded as a proper element of
  $\calD_\bot(\Mem\times\bbN)$ by the standing coercion,
  \[
    \de{\cmdwhile{e}{C}}_s^\dagger(\mu)
    =
    \bigsqcup_{m\in\bbN}
    \bigl((F_m)_s^\bot\bigr)^\dagger(\mu)
  \]
  Consequently, for every ordinary output $x\in\Mem\times\bbN$,
  \[
    \left(\de{\cmdwhile{e}{C}}_s^\dagger(\mu)\right)(x)
    =
    \sup_{n\in\bbN}
    \left(\bigl((F_{n+1})_s^\bot\bigr)^\dagger(\mu)\right)(x)
  \]
\end{lemma}
\begin{proof}
  Regard $F_{m,s}$ and $\de{\cmdwhile{e}{C}}_s$ as kernels from
  $\Mem\times\bbN$ to $\calD_\bot(\Mem\times\bbN)$.
  By \Cref{lem:index-schedule-indexed-while-semantic-equations},
  \[
    \de{\cmdwhile{e}{C}}_s(\sigma,n)
    =
    \bigsqcup_{m\in\bbN}F_{m,s}(\sigma,n)
    \qquad(\sigma\in\Mem,\;n\in\bbN)
  \]
  The functional $\Phi_{\langle C,e\rangle}$ is Scott-continuous by
  \Cref{lem:index-while-functional-scott-continuous}, hence monotone.  Since
  $F_0$ is the least element and
  $F_{m+1}=\Phi_{\langle C,e\rangle}(F_m)$, induction gives
  $F_m\sqsubseteq F_{m+1}$ for every $m\in\bbN$.  Therefore
  $(F_{m,s})_{m\in\bbN}$ is directed, and
  $\de{\cmdwhile{e}{C}}_s=\bigsqcup_{m\in\bbN}F_{m,s}$ in the pointwise
  function-space order.
  Applying
  \Cref{lem:index-dbot-kleisli-continuation-scott-continuous} with input
  distribution $\mu$ gives
  \begin{align*}
    \de{\cmdwhile{e}{C}}_s^\dagger(\mu)
    &=
    \left(\bigsqcup_{m\in\bbN}F_{m,s}\right)^\dagger(\mu)
    \tag{\Cref{lem:index-schedule-indexed-while-semantic-equations}}\\
    &=
    \bigsqcup_{m\in\bbN}F_{m,s}^\dagger(\mu)
    \tag{\Cref{lem:index-dbot-kleisli-continuation-scott-continuous}}\\
    &=
    \bigsqcup_{m\in\bbN}
    \bigl((F_m)_s^\bot\bigr)^\dagger(\mu).
    \tag{$\mu(\bot)=0$}
  \end{align*}
  Thus
  \begin{equation}
    \de{\cmdwhile{e}{C}}_s^\dagger(\mu)
    =
    \bigsqcup_{m\in\bbN}
    \bigl((F_m)_s^\bot\bigr)^\dagger(\mu).
    \label{eq:index-while-distribution-approximant-supremum}
  \end{equation}
  For ordinary $x\in\Mem\times\bbN$,
  \begin{align*}
    \left(\de{\cmdwhile{e}{C}}_s^\dagger(\mu)\right)(x)
    &=
    \left(
      \bigsqcup_{m\in\bbN}
      \bigl((F_m)_s^\bot\bigr)^\dagger(\mu)
    \right)(x)
    \tag{by \eqref{eq:index-while-distribution-approximant-supremum}}\\
    &=
    \sup_{m\in\bbN}
    \left(\bigl((F_m)_s^\bot\bigr)^\dagger(\mu)\right)(x).
    \tag{ordinary-coordinate supremum in $\calD_\bot$}
  \end{align*}
  Since $F_{0,s}(\sigma,n)=\delta_\bot$ for every $(\sigma,n)$,
  $\bigl((F_0)_s^\bot\bigr)^\dagger(\mu)=\delta_\bot$, whose value at the
  ordinary output $x$ is $0$.  Hence the $m=0$ term does not change the
  supremum on the ordinary coordinate $x$, giving the stated supremum over
  $(F_{n+1})_{n\in\bbN}$.
\end{proof}

\begin{lemma}[Schedule-indexed semantic equations on distributions]
  \label{lem:index-schedule-indexed-distribution-semantic-equations}
  \noam{Again, this lemma is quite obvious}
  Fix a schedule $s$ and let
  $\mu\in\calD(\Mem\times\bbN)$, regarded as a proper element of
  $\calD_\bot(\Mem\times\bbN)$ by the standing coercion.  Let
  $\alpha_\sigma\triangleq d(\de{\vec e}_{\Exp}(\sigma))$,
  $A_\sigma\triangleq\de{e}_{\Exp}(\sigma)$,
  $p_\sigma\triangleq\de{e}_{\Exp}(\sigma)$, and let $F_{m,s}$ be as in
  \Cref{lem:index-schedule-indexed-while-semantic-equations}.  Then
  for every command $C$,
  \[
    \de{C}_s^\dagger(\mu)
    =
    \bigoplus_{(\sigma,n)\sim\mu}
    \de{C}_s(\sigma,n)
  \]
  Moreover,
  \begin{align*}
    \de{\cmdskip}_s^\dagger(\mu)
    &=
      \bigoplus_{(\sigma,n)\sim\mu}
      \delta_{(\sigma,n)}
      =
      \mu\\
    \de{\actassign{x}{e}}_s^\dagger(\mu)
    &=
      \bigoplus_{(\sigma,n)\sim\mu}
        \delta_{(\sigma[x\coloneqq\de{e}_{\Exp}(\sigma)],n)}
      \\
    \de{x\samp d(\vec e)}_s^\dagger(\mu)
    &=
      \bigoplus_{(\sigma,n)\sim\mu}
        \bigoplus_{v\sim\alpha_\sigma}
        \delta_{(\sigma[x\coloneqq v],n)}
      \\
    \de{x\gets e}_s^\dagger(\mu)
    &=
      \bigoplus_{(\sigma,n)\sim\mu}
        \delta_{
          (\sigma[x\coloneqq
          A_\sigma[(s(n)\bmod |A_\sigma|)]],n+1)
        }
      \\
    \de{\cmdseq{C_1}{C_2}}_s^\dagger(\mu)
    &=
      \bigoplus_{(\sigma,n)\sim\mu}
        \de{C_2}_s^\dagger\bigl(\de{C_1}_s(\sigma,n)\bigr)
      =
      \de{C_2}_s^\dagger\bigl(\de{C_1}_s^\dagger(\mu)\bigr)\\
    \de{\cmdcase{b}{C_1}{C_2}}_s^\dagger(\mu)
    &=
      \bigoplus_{(\sigma,n)\sim\mu}
        \begin{cases}
          \de{C_1}_s(\sigma,n), & \de{b}_{\Exp}(\sigma)=\Ttrue,\\
          \de{C_2}_s(\sigma,n), & \de{b}_{\Exp}(\sigma)=\Tfalse,
        \end{cases}\\
    \de{\cmdchoice{e}{C_1}{C_2}}_s^\dagger(\mu)
    &=
      \bigoplus_{(\sigma,n)\sim\mu}
        \de{C_1}_s(\sigma,n)
        \oplus_{p_\sigma}
        \de{C_2}_s(\sigma,n)
      \\
    \de{\cmdnondet{C_0}{C_1}}_s^\dagger(\mu)
    &=
      \bigoplus_{(\sigma,n)\sim\mu}
        \de{C_{(s(n)\bmod 2)}}_s(\sigma,n+1)
      \\
    \de{\cmdwhile{e}{C}}_s^\dagger(\mu)
    &=
      \bigoplus_{(\sigma,n)\sim\mu}
        \bigsqcup_{m\in\bbN}F_{m,s}(\sigma,n)
  \end{align*}
\end{lemma}
\begin{proof}
  For the first displayed equation, it suffices to prove equality pointwise.
  By $\mu(\bot) = 0$ and the definitions of Kleisli extension and convex mixture, for every
  $z\in(\Mem\times\bbN)\cup\{\bot\}$,
  \[
    \left(\de{C}_s^\dagger(\mu)\right)(z)
    =
    \sum_{(\sigma,n)\in\Mem\times\bbN}
    \mu(\sigma,n)\cdot
    \bigl(\de{C}_s(\sigma,n)\bigr)(z)
    =
    \left(
      \bigoplus_{(\sigma,n)\sim\mu}
      \de{C}_s(\sigma,n)
    \right)(z).
  \]

  Each constructor equation, except the final equality in the sequencing line
  and the final equality in the skip line, is obtained by unfolding the
  corresponding pointwise semantic clause, using
  \Cref{lem:index-schedule-indexed-while-semantic-equations} for while.  The
  final equality in the skip line is the right identity law.

  By \Cref{lem:index-dbot-kleisli-linearity}, for sequencing, 
  \[
    \bigoplus_{(\sigma,n)\sim\mu}
    \de{C_2}_s^\dagger\bigl(\de{C_1}_s(\sigma,n)\bigr)
    =
    \de{C_2}_s^\dagger\left(
      \bigoplus_{(\sigma,n)\sim\mu}
      \de{C_1}_s(\sigma,n)
    \right)
    =
    \de{C_2}_s^\dagger\bigl(\de{C_1}_s^\dagger(\mu)\bigr).
  \]
\end{proof}

\begin{lemma}[Schedule-indexed rule-shaped output equations]
  \label{lem:index-schedule-indexed-rule-output-equations}
  Fix a schedule $s$.
  \begin{enumerate}
    \item Let $\mu\in\calD(\Mem\times\bbN)$ and let
    $p\in[0,1]$.  Put
    \[
      G_p^e
      \triangleq
      \{(\sigma,n)\in\Mem\times\bbN
      \mid \de{e}_{\Exp}(\sigma)=p\}
    \]
    If $\mu(G_p^e)=1$, then
    \[
      \de{C_1\cmdchoiceS{e}C_2}_s^\dagger(\mu)
      =
      \de{C_1}_s^\dagger(\mu)
      \oplus_p
      \de{C_2}_s^\dagger(\mu)
    \]

    \item Let $\mu\in\calD(\Mem\times\bbN)$, put
    \[
      I_i\triangleq\{n\in\bbN\mid s(n)\bmod 2=i\},
      \qquad
      \alpha(i)\triangleq\mu(\Mem\times I_i)
      \quad(i\in\{0,1\})
    \]
    For $i\in\supp(\alpha)$, let
    $\mu_i\triangleq\RestrictN{\mu}{I_i}$.  Then \noam{Was $\Tick$ defined? (I may have deleted it from the paper)}
    \[
      \bigoplus_{i\sim\alpha}
      \de{C_i}_s^\dagger(\Tick_\calD(\mu_i))
      =
      \de{C_0\cmdnondetS C_1}_s^\dagger(\mu)
    \]
  \end{enumerate}
\end{lemma}
\begin{proof}
  For the probabilistic-choice equation, put
  $p_\sigma\triangleq\de{e}_{\Exp}(\sigma)$.  Then
  \begin{align*}
    \de{C_1\cmdchoiceS{e}C_2}_s^\dagger(\mu)
    &=
      \bigoplus_{(\sigma,n)\sim\mu}
      \bigl(
        \de{C_1}_s(\sigma,n)
        \oplus_{p_\sigma}
        \de{C_2}_s(\sigma,n)
      \bigr)
      \tag{\Cref{lem:index-schedule-indexed-distribution-semantic-equations}}\\
    &=
      \bigoplus_{(\sigma,n)\sim\mu}
      \bigl(
        \de{C_1}_s(\sigma,n)
        \oplus_p
        \de{C_2}_s(\sigma,n)
      \bigr)
      \tag{$\mu(G_p^e)=1$}\\
    &=
      \left(
        \bigoplus_{(\sigma,n)\sim\mu}
        \de{C_1}_s(\sigma,n)
      \right)
      \oplus_p
      \left(
        \bigoplus_{(\sigma,n)\sim\mu}
        \de{C_2}_s(\sigma,n)
      \right)
      \tag{definition of convex mixture}\\
    &=
      \de{C_1}_s^\dagger(\mu)
      \oplus_p
      \de{C_2}_s^\dagger(\mu).
      \tag{\Cref{lem:index-schedule-indexed-distribution-semantic-equations}}
  \end{align*}

  For the nondeterministic-choice equation, first note that for
  $i\in\supp(\alpha)$ and $(\sigma,n)\in\Mem\times\bbN$,
  \[
    \alpha(i)\cdot\mu_i(\sigma,n)
    =
    \begin{cases}
      \mu(\sigma,n), & n\in I_i,\\
      0, & n\notin I_i .
    \end{cases}
  \]
  Since $(I_i)_{i\in\{0,1\}}$ partitions $\bbN$,
  \begin{align*}
    \bigoplus_{i\sim\alpha}
    \de{C_i}_s^\dagger(\Tick_\calD(\mu_i))
    &=
      \bigoplus_{i\sim\alpha}
      \left(
        \bigoplus_{(\sigma,m)\sim\Tick_\calD(\mu_i)}
        \de{C_i}_s(\sigma,m)
      \right)
      \tag{\Cref{lem:index-schedule-indexed-distribution-semantic-equations}}\\
    &=
      \bigoplus_{i\sim\alpha}
      \left(
        \bigoplus_{(\sigma,n)\sim\mu_i}
        \de{C_i}_s(\sigma,n+1)
      \right)
      \tag{definition of $\Tick_\calD$}\\
    &=
      \bigoplus_{i\sim\alpha}
      \left(
        \bigoplus_{(\sigma,n)\sim\mu_i}
        \de{C_{(s(n)\bmod 2)}}_s(\sigma,n+1)
      \right)
      \tag{$\supp(\mu_i)\subseteq\Mem\times I_i$}\\
    &=
      \bigoplus_{(\sigma,n)\sim\mu}
      \de{C_{(s(n)\bmod 2)}}_s(\sigma,n+1)
      \tag{definition of $\mu_i$ and the partition $(I_i)_i$}\\
    &=
      \de{C_0\cmdnondetS C_1}_s^\dagger(\mu).
      \tag{\Cref{lem:index-schedule-indexed-distribution-semantic-equations}}
  \end{align*}
\end{proof}

\section{Assertion Infrastructure}
\begin{figure*}[t]
\begingroup
\footnotesize
\setlength{\abovedisplayskip}{3pt}
\setlength{\belowdisplayskip}{3pt}
\setlength{\abovedisplayshortskip}{2pt}
\setlength{\belowdisplayshortskip}{2pt}
\newcommand{\rfgroup}[3]{%
  \par\smallskip\noindent\textbf{(#1)\ \ #2.}~\emph{#3}\par\nobreak}
\begin{mdframed}[linewidth=0.4pt,roundcorner=3pt,
  innertopmargin=5pt,innerbottommargin=7pt,
  innerleftmargin=10pt,innerrightmargin=10pt,
  skipabove=3pt,skipbelow=3pt]
{\centering\normalsize\bfseries Resource Algebra and Refinement Facts\par}
\vspace{2pt}\hrule\vspace{5pt}
\noindent
\begin{minipage}[t]{0.475\textwidth}
  \rfgroup{a}{Pushforwards}{functorial; commute with products; reindexing preserves resources}
  \[\begin{gathered}
    (g\circ f)_*\calP=g_*(f_*\calP)\\
    (f,\id)_*(\calP\otimes\calQ)=f_*\calP\otimes\calQ\\
    (\id,f)_*(\calQ\otimes\calP)=\calQ\otimes f_*\calP\\
    f_*\calR\cong\calR
  \end{gathered}\]
  \rfgroup{b}{Isomorphism and order}{$\cong$ is an equivalence, $\preceq$ a preorder}
  \[\begin{gathered}
    \calR_0\preceq\calR_1\Rightarrow\Proj{\calR_0}{\bbN}=\Proj{\calR_1}{\bbN}\\
    \calR\cong\calR'\Rightarrow\calR\preceq\calR'\land\calR'\preceq\calR
  \end{gathered}\]
  \rfgroup{c}{Product laws}{commutative monoid with unit $\UnitP$; monotone; distributes over $\oplus$}
  \[\begin{gathered}
    \calR_1\otimes\calR_2\cong\calR_2\otimes\calR_1\\
    (\calR_1\otimes\calR_2)\otimes\calR_3\cong\calR_1\otimes(\calR_2\otimes\calR_3)\\
    \calR\otimes\UnitP\cong\UnitP\otimes\calR\cong\calR\\
    \otimes\text{ preserves }\cong\text{ and }\preceq\\
    \left(\bigoplus_{i\sim\alpha}\calR_i\right)\otimes\calR_F
      \cong\bigoplus_{i\sim\alpha}(\calR_i\otimes\calR_F)\\
    \Proj{\calR_1\otimes\calR_2}{\bbN}
      =\mathsf{add}_*\!\left(\Proj{\calR_1}{\bbN}\otimes\Proj{\calR_2}{\bbN}\right)
  \end{gathered}\]
  \rfgroup{d}{Direct-sum laws}{count marginals, order, branch projection, regrouping}
  \[\begin{gathered}
    \Proj{\bigoplus_{i\sim\alpha}\calR_i}{\bbN}=\bigoplus_{i\sim\alpha}\Proj{\calR_i}{\bbN}\\
    (\forall i\in\supp(\alpha).\;\calQ_i\preceq\calR_i)
      \Rightarrow\bigoplus_{i\sim\alpha}\calQ_i\preceq\bigoplus_{i\sim\alpha}\calR_i\\
    k\in\supp(\alpha)\Rightarrow\calR_k\preceq\bigoplus_{i\sim\alpha}\calR_i
      \qquad\bigoplus_{i\sim\delta_k}\calR_i\cong\calR_k\\
    \bigoplus_{i\sim\alpha}\left(\bigoplus_{j\sim\beta_i}\calR_{i,j}\right)
      \cong\bigoplus_{(i,j)\sim\gamma}\calR_{i,j} \;\text{where}\;
    \gamma(i,j)=\alpha(i)\beta_i(j)
  \end{gathered}\]
\end{minipage}\hfill
\begin{minipage}[t]{0.475\textwidth}
  \rfgroup{e}{Tick laws}{commutes with $\oplus$ and $\otimes$; preserves and reflects $\preceq$}
  \[\begin{gathered}
    \Tick\left(\bigoplus_{i\sim\alpha}\calR_i\right)=\bigoplus_{i\sim\alpha}\Tick(\calR_i)\\
    \calR_0\preceq\calR_1\Longleftrightarrow\Tick(\calR_0)\preceq\Tick(\calR_1)\\
    \Tick(\calR_1\otimes\calR_2)\cong\Tick(\calR_1)\otimes\calR_2\cong\calR_1\otimes\Tick(\calR_2)
  \end{gathered}\]
  \rfgroup{f}{Memory components}{commute with $\otimes,\oplus$; monotone; project factors}
  \[\begin{gathered}
    (\calR_1\otimes\calR_2)^\memfn\cong\calR_1^\memfn\otimes\calR_2^\memfn\\
    \left(\bigoplus_{i\sim\alpha}\calR_i\right)^\memfn\cong\bigoplus_{i\sim\alpha}\calR_i^\memfn\\
    \calR\preceq\calQ\Rightarrow\calR^\memfn\preceq\calQ^\memfn
      \qquad\calR_1^\memfn\preceq(\calR_1\otimes\calR_2)^\memfn\\
    \cntfn_\calR\equiv 0\Rightarrow\calR^\memfn\cong\calR
  \end{gathered}\]
  \rfgroup{g}{Distribution refinement}{proper target; count marginal; composes with $\preceq$}
  \[\begin{gathered}
    \calR\preceq\mu\in\calD_\bot(\Mem\times\bbN)\Rightarrow\mu(\bot)=0\\
    \calR\preceq\mu\in\calD(\Mem\times\bbN)\Rightarrow\Proj{\calR}{\bbN}=(\pi_2)_*\mu\\
    \calR_0\preceq\calR\land\calR\preceq\mu\Rightarrow\calR_0\preceq\mu
  \end{gathered}\]
  \rfgroup{h}{Refinement constructors}{stable under $\cong$, restriction, tick, $\oplus$}
  \[\begin{gathered}
    \calR\cong\calR'\land\calR\preceq\mu\Rightarrow\calR'\preceq\mu\\
    \calR\preceq\mu\Longleftrightarrow\Tick(\calR)\preceq\Tick_\calD(\mu)\\
    \calR\preceq\mu\Rightarrow\RestrictN{\calR}{E}\preceq\RestrictN{\mu}{E}\\
    (\forall i\in\supp(\alpha).\;\calQ_i\preceq\nu_i)
      \Rightarrow\bigoplus_{i\sim\alpha}\calQ_i\preceq\bigoplus_{i\sim\alpha}\nu_i\\
    \bigoplus_{i\sim\alpha}\calQ_i\preceq\nu
      \Rightarrow\exists(\nu_i).\;\nu=\bigoplus_{i\sim\alpha}\nu_i\land\calQ_i\preceq\nu_i
  \end{gathered}\]
\end{minipage}
\end{mdframed}
\endgroup
\caption{Resource algebra and distribution-refinement facts used throughout the
soundness proof, grouped by operator.}
\Description{A framed, two-column summary of resource-algebra and
distribution-refinement facts, organized under eight labeled headings (a)--(h).}
\label{fig:resource-algebra-refinement-facts}
\end{figure*}

For the assertion language and the logic \LogicName, the relevant details about the algebra for
resources $\calR=\langle \calP, V, \memfn, \cntfn \rangle$ are summarized in
\Cref{fig:resource-algebra-refinement-facts}.

\subsection{Weak Stability}
\begin{lemma}[Almost-Sure Coarsening]
  \label{lem:index-sure-coarsening}
  If $\Gamma,\calR\vDash\sure{P}$, then there exist
  $\calR^\memfn$ and $\calR^\cntfn$ such that
  \[
    \calR^\memfn\otimes\calR^\cntfn\preceq\calR,
    \qquad
    \Stable{\weak}{\calR^\memfn},
    \qquad
    \Gamma,\calR^\memfn\vDash\sure{P},
    \qquad
    \forall \omega\in\Omega_{\calR^\cntfn}.\;
    \memfn_{\calR^\cntfn}(\omega)=\emptyset
  \]
\end{lemma}
\begin{proof}
  Let $E_P\triangleq
  \memfn_\calR^{-1}(\sem{P}_\Gamma^{V_\calR})$.  Since
  $\Gamma,\calR\vDash\sure{P}$, we have $E_P\in\calF_\calR$ and
  $\mu_\calR(E_P)=1$.

Let $\calP_E$ be obtained from $\calP_\calR$ by coarsening its $\sigma$-algebra
  to the completion of the $\sigma$-algebra generated by $E_P$.  Thus $\calP_E$
  has the same underlying support and measure as $\calP_\calR$, but every
  $\calP_E$-measurable event has measure $0$ or $1$.  In particular, every
  $\calP_E$-event is independent of every $\calP_\calR$-event: for
  $A\in\calF_{\calP_E}$ and $B\in\calF_\calR$,
  $\mu_\calR(A\cap B)
  =
  \mu_\calR(A)\cdot\mu_\calR(B)$.

  Define
  \[
    \calR^\memfn
    \triangleq
    \langle \calP_E, V_\calR, \memfn_\calR, \omega\mapsto 0\rangle
    \qquad
    \calR^\cntfn
    \triangleq
    \langle \calP_\calR, \emptyset, \omega\mapsto\emptyset,
      \cntfn_\calR\rangle
  \]
  These are the memory witness and count component for weak stability.

  Now form the resource product
  $\calR_0\triangleq\calR^\memfn\otimes\calR^\cntfn$.  By the definition of
  product resources, $\calR_0$ is obtained by first taking the structural product
  on $\Omega_\calR\times\Omega_\calR$, and then reindexing along a fresh-index
  bijection
  $g:\Omega_\calR\times\Omega_\calR\to\Omega_{\calR_0}$.

  We prove $\calR_0\preceq\calR$ using the reindexed diagonal map
  $h:\Omega_\calR\to\Omega_{\calR_0}$ defined as
  $h(\omega)\triangleq g(\omega,\omega)$.

  First consider the structural product before reindexing.  Let
  $\Delta(\omega)\triangleq(\omega,\omega)$.  For measurable rectangles
  $A\times B$, with $A\in\calF_{\calP_E}$ and $B\in\calF_\calR$, we have
  $\Delta^{-1}(A\times B)=A\cap B$, and therefore
  \[
    (\mu_{\calP_E}\otimes\mu_\calR)(A\times B)
    =
    \mu_\calR(A)\cdot\mu_\calR(B)
    =
    \mu_\calR(A\cap B)
    =
    \mu_\calR(\Delta^{-1}(A\times B))
  \]
  Since rectangles generate the product $\sigma$-algebra, and the product space
  is completed, this equality extends to every event
  $S\in\calF_{\calP_E}\otimes\calF_\calR$, i.e.
  $(\mu_{\calP_E}\otimes\mu_\calR)(S)
  =
  \mu_\calR(\Delta^{-1}(S))$

  Reindexing transfers this equality to $\calR_0$.  For any
  $C\in\calF_{\calR_0}$, the definition of the product resource gives
  $g^{-1}(C)\in\calF_{\calP_E}\otimes\calF_\calR$ and
  \[
    \mu_{\calR_0}(C)
    =
    (\mu_{\calP_E}\otimes\mu_\calR)(g^{-1}(C))
    =
    \mu_\calR(\Delta^{-1}(g^{-1}(C)))
    =
    \mu_\calR(h^{-1}(C))
  \]
  Thus $h$ is measurable and measure-preserving.

  The observation equations also commute through the same reindexed diagonal.
  For every $\omega\in\Omega_\calR$,
  \[
    \memfn_{\calR_0}(h(\omega))
    =
    \memfn_{\calR^\memfn}(\omega)\uplus
    \memfn_{\calR^\cntfn}(\omega)
    =
    \memfn_\calR(\omega)\uplus\emptyset
    =
    \memfn_\calR(\omega)
  \]
  and
  \[
    \cntfn_{\calR_0}(h(\omega))
    =
    0+\cntfn_\calR(\omega)
    =
    \cntfn_\calR(\omega)
  \]
  Hence $h$ witnesses
  $\calR^\memfn\otimes\calR^\cntfn\preceq\calR$.

  Finally, $\calR^\memfn$ has constant scheduler count $0$, and
  $\calR^\cntfn$ has empty memory by construction.  Moreover, $E_P$ remains
  measurable in $\calP_E$ and still has measure $1$, so
  $\Gamma,\calR^\memfn\vDash\sure{P}$.
\end{proof}

\begin{lemma}[Weak stability constructor rules]
  \label{lem:index-weak-stability-fragment}
  The weak-stability rules displayed in
  \Cref{fig:assertion-property-rules} are valid.
\end{lemma}
\begin{proof}
  We prove the weak-stability rules in
  \Cref{fig:assertion-property-rules} by exhibiting the current memory/count
  decomposition: for every satisfying resource, we construct
  $\calR^\memfn$ and $\calR^\cntfn$ such that
  $\calR^\memfn\otimes\calR^\cntfn\preceq\calR$,
  $\Stable{\weak}{\calR^\memfn}$, $\calR^\cntfn$ has empty memory, and the
  assertion is already satisfied by $\calR^\memfn$. The case for $\sure{P}$
  follows from~\Cref{lem:index-sure-coarsening}, and the case for
  $\varphi_1\lor\varphi_2$ is immediate.

  \emph{Case $\bigoplus^\Priv_{X\sim d(E)}\varphi$.}
  Let $\mu\triangleq d(\de{E}_{\textsf{LExp}}(\Gamma))$
  and suppose
  $\Gamma,\calR\vDash\bigoplus^\Priv_{X\sim d(E)}\varphi$.  Unfolding
  satisfaction, choose branch resources
  $(\calR_v)_{v\in\supp(\mu)}$ such that
  \[
    \bigoplus_{v\sim\mu}\calR_v\preceq\calR,
    \qquad
    \Gamma[X\mapsto v],\calR_v\vDash\varphi
    \quad(v\in\supp(\mu))
  \]
  and the private side condition says that every branch has count marginal
  $(\cntfn_\calR)_*\mu_\calR$.

  By $\Stable{\weak}{\varphi}$, for each $v\in\supp(\mu)$ choose
  $\calR_v^\memfn$ and $\calR_v^\cntfn$ such that
  \begin{align*}
    &\calR_v^\memfn\otimes\calR_v^\cntfn\preceq\calR_v
      \qquad
      \Stable{\weak}{\calR_v^\memfn}
      \qquad
      \Gamma[X\mapsto v],\calR_v^\memfn\vDash\varphi\\
    &\forall\omega\in\Omega_{\calR_v^\cntfn}.\;
      \memfn_{\calR_v^\cntfn}(\omega)=\emptyset
  \end{align*}
  Since $\calR_v^\memfn$ has zero count and
  $\calR_v^\memfn\otimes\calR_v^\cntfn\preceq\calR_v$, the count marginal
  of $\calR_v^\cntfn$ is also $(\cntfn_\calR)_*\mu_\calR$.  Let
  $\calR^\cntfn$ be the canonical count-only resource on a fresh copy of
  $\bbN$, with count marginal $(\cntfn_\calR)_*\mu_\calR$, empty memory, and
  count observation the identity.  The count map of $\calR_v^\cntfn$
  witnesses
  $\calR^\cntfn\preceq \calR_v^\cntfn$
  for all $v \in \supp(\mu)$.
  Hence, by monotonicity of $\otimes$ and transitivity,
  \[
    \calR_v^\memfn\otimes\calR^\cntfn
    \preceq
    \calR_v
    \qquad(v\in\supp(\mu))
  \]

  Note that the branch memory cores cannot be summed directly: their footprints may be
  proper subsets of the mixture footprint $V$, and their supports need not be
  disjoint.
  Since $\calR_v^\memfn\otimes\calR^\cntfn\preceq\calR_v$ and
  $V_{\calR_v}=V$, we first extend each $\calR_v^\memfn$ to the common
  footprint $V$, and then reindex the resulting resources onto
  pairwise-disjoint supports.  Write $\calR_{v,\mathsf m}$ for the resulting
  resource.  Then
  \[
    \Gamma[X\mapsto v],\calR_{v,\mathsf m}\vDash\varphi
    \qquad
    \calR_{v,\mathsf m}\otimes\calR^\cntfn
    \preceq
    \calR_v
    \qquad(v\in\supp(\mu))
  \]
  Put
  $\calR^\memfn\triangleq\bigoplus_{v\sim\mu}\calR_{v,\mathsf m}$.  Then
  \[
    \calR^\memfn\otimes\calR^\cntfn
    \cong
    \bigoplus_{v\sim\mu}(\calR_{v,\mathsf m}\otimes\calR^\cntfn)
    \preceq
    \bigoplus_{v\sim\mu}\calR_v
    \preceq
    \calR
  \]
  The resource $\calR^\memfn$ is weakly stable because every branch
  $\calR_{v,\mathsf m}$ has zero count, and $\calR^\cntfn$ has empty memory by
  construction.  Finally, the branch witnesses
  $(\calR_{v,\mathsf m})_{v\in\supp(\mu)}$ prove
  $\Gamma,\calR^\memfn
  \vDash
  \bigoplus^\Priv_{X\sim d(E)}\varphi$.

  \emph{Case $\varphi_1 \sep \varphi_2$.}
  Suppose $\Gamma,\calR\vDash\varphi_1 \sep \varphi_2$.  By the satisfaction
  clause for separating conjunction, choose
  $\calR_1,\calR_2$ such that
  \[
    \calR_1\otimes\calR_2\preceq\calR
    \qquad
    \Gamma,\calR_i\vDash\varphi_i
    \quad(i\in\{1,2\})
  \]
  Applying the stability hypotheses to the two component satisfactions
  gives resources $\calR_i^\memfn$ and $\calR_i^\cntfn$, together with the
  required product witnesses, such that
  \[
    \calR_i^\memfn\otimes\calR_i^\cntfn\preceq\calR_i
    \qquad
    \Stable{\weak}{\calR_i^\memfn}
    \qquad
    \Gamma,\calR_i^\memfn\vDash\varphi_i
    \qquad
    V_{\calR_i^\cntfn}=\emptyset
    \qquad(i\in\{1,2\})
  \]
  The refinement
  $\calR_i^\memfn\otimes\calR_i^\cntfn\preceq\calR_i$ implies
  $V_{\calR_i^\memfn}\subseteq V_{\calR_i}$.  Since the original separating
  witness gives $V_{\calR_1}\cap V_{\calR_2}=\emptyset$, we have
  $V_{\calR_1^\memfn}\cap V_{\calR_2^\memfn}=\emptyset$.  The count witnesses
  have empty footprints, so the remaining products below are well formed.

  Define
  $\calR^\memfn
  \triangleq
  \calR_1^\memfn\otimes\calR_2^\memfn$ and 
  $\calR^\cntfn
  \triangleq
  \calR_1^\cntfn\otimes\calR_2^\cntfn$, we get that
  \[
    \calR^\memfn\otimes\calR^\cntfn
    \cong
    (\calR_1^\memfn\otimes\calR_1^\cntfn) 
    \otimes
    (\calR_2^\memfn\otimes\calR_2^\cntfn) 
    \preceq
    \calR_1\otimes\calR_2
    \preceq
    \calR
  \]
  The resource $\calR^\memfn$ is weakly stable, since both factors have zero
  scheduler count, and $\calR^\cntfn$ has empty footprint.  Finally, the
  satisfaction clause for separating conjunction gives
  $\Gamma,\calR^\memfn\vDash\varphi_1\sep\varphi_2$.
\end{proof}

\subsection{Memory Precision}
\begin{lemma}[Tick preservation of satisfaction]
  \label{lem:index-tick-satisfaction}
  \noam{Doesn't this follow from the next lemma? Since obviously $\calR^\memfn= (\Tick(\calR))^\memfn$, we can just apply the next lemma to conclude the goal without any induction.}
  For every assertion $\varphi$,
  \[
    \Gamma,\calR\vDash\varphi
    \quad\Longrightarrow\quad
    \Gamma,\Tick(\calR)\vDash\varphi
  \]
\end{lemma}
\begin{proof}
  By structural induction on $\varphi$. The cases for $\top$, $\bot$, $\varphi_1
  \land \varphi_2$, and $\varphi_1 \lor \varphi_2$ are immediate.
  \begin{itemize}
    \item \emph{Case $\varphi=\sure{P}$.}
    Since tick preserves the underlying probability space, footprint, and
    memory observation, the almost-sure event is the same, i.e.
    $\memfn_{\Tick(\calR)}^{-1}
    (\sem{P}_\Gamma^{V_{\Tick(\calR)}})
    =
    \memfn_{\calR}^{-1}
    (\sem{P}_\Gamma^{V_\calR})$.
    Hence $\Gamma,\Tick(\calR)\vDash\sure{P}$.

    \item \emph{Case $\varphi=\bigoplus^o_{X\sim d(E)}\psi$ and $\varphi =\bignd_{X\in E}\psi$}
    Let $\mu\triangleq d(\de{E}_{\LExp}(\Gamma))$.  Unfold the definition we get witness
    $(\calR_v)_{v\in\supp(\mu)}$ such that
    $\bigoplus_{v\sim\mu}\calR_v\preceq\calR$ and
    $\left(\Gamma[X\mapsto v],\calR_v\vDash\psi\right)_{v \in \supp(\mu)}$
    By the induction hypothesis we get that
    $\left(\Gamma[X\mapsto v],\Tick(\calR_v)\vDash\psi\right)_{v \in
      \supp(\mu)}$.
    
    The refinement obligation for the ticked branch witnesses is
    \[
      \bigoplus_{v\sim\mu}\Tick(\calR_v)
      =
      \Tick\left(\bigoplus_{v\sim\mu}\calR_v\right)
      \preceq
      \Tick(\calR)
    \]
    If $o=\Priv$, then for every $A\subseteq\bbN$ and
    $A^{-}\triangleq\{n\in\bbN\mid n+1\in A\}$,
    \begin{align*}
      \Proj{\Tick(\calR_v)}{\bbN}(A)
      &=
        \Proj{\calR_v}{\bbN}(A^{-})
        \tag{definition of $\Tick$}\\
      &=
        \Proj{\calR}{\bbN}(A^{-})
        \tag{private side condition for $\calR$}\\
      &=
        \Proj{\Tick(\calR)}{\bbN}(A).
        \tag{definition of $\Tick$}
    \end{align*}
    Thus the ticked witnesses satisfy the satisfaction clause for
    $\Gamma,\Tick(\calR)\vDash\bigoplus^o_{X\sim d(E)}\psi$.

    The case of $\varphi = \bignd_{X \in E}\psi$ follows.

    \item \emph{Case $\varphi=\varphi_1\sep\varphi_2$.}
    Choose witnesses $\calR_1,\calR_2$ such that
      $\calR_1\otimes\calR_2\preceq\calR$ and
      $\left(\Gamma,\calR_i\vDash\varphi_i\right)_{i \in \{1,2\}}$
    By the induction hypothesis,
    $\Gamma,\Tick(\calR_1)\vDash\varphi_1$.  Moreover,
    \[
      \Tick(\calR_1)\otimes\calR_2
      \cong
      \Tick(\calR_1\otimes\calR_2)
      \preceq
      \Tick(\calR)
    \]
    Thus $\Tick(\calR_1)$ and $\calR_2$ witness
    $\Gamma,\Tick(\calR)\vDash\varphi_1\sep\varphi_2$.
  \end{itemize}
\end{proof}

\begin{lemma}[Count erasure preserves satisfaction]
  \label{lem:index-count-erasure-satisfaction}
  For every assertion $\varphi$,
  \[
    \Gamma, \calR \vDash \varphi \implies \Gamma,\calR^\memfn\vDash\varphi
  \]
\end{lemma}
\begin{proof}
  By structural induction on $\varphi$. Cases $\top$, $\bot$, $\varphi_1 \land
  \varphi_2$, $\varphi_1 \lor \varphi_2$ are immediate.
  \begin{itemize}
    \item \emph{Case $\varphi=\sure{P}$.}
    Let
    $E_P\triangleq
    \memfn_\calR^{-1}(\sem{P}_\Gamma^{V_\calR})$
    be the event in which the memory observation of $\calR$ satisfies $P$.
    Since $\Gamma,\calR\vDash\sure{P}$, we have $E_P\in\calF_\calR$ and
    $\mu_\calR(E_P)=1$.
    Since $\calR^\memfn$ has the same probability space,
    footprint, and memory map as $\calR$, the corresponding $P$-event is again
    $E_P$:
    $\memfn_{\calR^\memfn}^{-1}
    (\sem{P}_\Gamma^{V_{\calR^\memfn}})
    =
    E_P$.
    Thus
    $\mu_{\calR^\memfn}(E_P)=\mu_\calR(E_P)=1$, which proves the goal.

    \item \emph{Case $\varphi=\bigoplus^o_{X\sim d(E)}\psi$ and $\varphi=\bignd_{X\in E}\psi$}
    Let $\mu\triangleq d(\de{E}_{\LExp}(\Gamma))$.  Unfolding satisfaction
    gives branch resources $(\calR_v)_{v\in\supp(\mu)}$ such that
    $\bigoplus_{v\sim\mu}\calR_v\preceq\calR$,
    $\left(\Gamma[X\mapsto v],\calR_v\vDash\psi\right)_{v \in \supp(\mu)}$.
    and $\Obs_o^\mu(\calR,(\calR_v)_{v\in\supp(\mu)})$.
    By the induction hypothesis,
    $\left(\Gamma[X\mapsto v],\calR_v^\memfn\vDash\psi\right)_{v\in\supp(\mu)}$.

    The refinement obligation for the erased branch witnesses is
    \[
      \bigoplus_{v\sim\mu}\calR_v^\memfn
      \cong
      \left(\bigoplus_{v\sim\mu}\calR_v\right)^\memfn
      \preceq
      \calR^\memfn
    \]
    If $o=\Priv$, then every memory component has zero count marginal, so
    \[
      \Proj{\calR_v^\memfn}{\bbN}
      =
      \delta_0
      =
      \Proj{\calR^\memfn}{\bbN}
      \qquad(v\in\supp(\mu))
    \]
    Thus the erased witnesses satisfy the satisfaction clause for
    $\Gamma,\calR^\memfn\vDash\bigoplus^o_{X\sim d(E)}\psi$.

    The case for $\bignd_{X \in E} \psi$ follows from the previous case for $\bigoplus^\Leak$.

    \item \emph{Case $\varphi=\varphi_1\sep\varphi_2$.}
    Choose witnesses $\calR_1,\calR_2$ such that
    $\calR_1\otimes\calR_2\preceq\calR$ and
    $\left(\Gamma,\calR_i\vDash\varphi_i\right)_{i \in \{1, 2\}}$.
    By the induction hypotheses,
    $\left(\Gamma,\calR_i^\memfn\vDash\varphi_i\right)_{i\in\{1,2\}}$.

    The refinement obligation for the erased separating witnesses is
    \[
      \calR_1^\memfn\otimes\calR_2^\memfn
      \cong
      (\calR_1\otimes\calR_2)^\memfn
      \preceq
      \calR^\memfn
    \]
    Thus $\calR_1^\memfn$ and $\calR_2^\memfn$ witness
    $\Gamma,\calR^\memfn\vDash\varphi_1\sep\varphi_2$.
  \end{itemize}
\end{proof}

\begin{lemma}[Pure assertion footprint change]
  \label{lem:index-pure-assertion-footprint}
  For every pure assertion $P$, environment $\Gamma$, finite footprints
  $\Var(P)\subseteq T\subseteq S\subseteq_\fin\Var$, and memory
  $\sigma\in\Mem[S]$,
  \[
    \sigma\in\sem{P}_\Gamma^S
    \quad\Longleftrightarrow\quad
    \ForgetMem{\sigma}{T}
    \in
    \sem{P}_\Gamma^T
  \]
\end{lemma}
\begin{proof}
  By structural induction on $P$. The cases $\true$, $\false$, conjunction,
  disjunction, and existential quantification follow directly from the
  induction hypotheses. The comparison case is independent of memory.

  For $e\mapsto E$, all program variables read by $e$ lie in
  $\Var(P)\subseteq T$, so evaluating $e$ in $\sigma$ agrees with evaluating it
  in $\ForgetMem{\sigma}{T}$; the logical expression $E$ is interpreted only by
  $\Gamma$.

  For $P_1 * P_2$, use the induction hypotheses on the two separating
  submemories. Projecting the two separating submemories to $T$ preserves
  disjointness and separating submemory inclusion, giving the forward
  direction. The reverse direction uses the same projected witnesses, since
  $\ForgetMem{\sigma}{T}\sqsubseteq\sigma$.
\end{proof}

\begin{lemma}[Assertion satisfaction is preserved by framing]
  \label{lem:index-assertion-frame-preservation} 
  \noam{This should follow from monotonicity, since $\calR \preceq \calR\otimes \calR_F$}
  \gary{Actually this is not true, $R_F$ may increment the counts}
  If $\Gamma,\calR\vDash\varphi$ and $\calR\otimes\calR_F$ is defined, then
  \[
    \Gamma,\calR\otimes\calR_F\vDash\varphi
  \]
  The case when $V_{\calR_F} = \emptyset$ is a special case since $\otimes$ is
  automatically well-defined.
\end{lemma}
\begin{proof}
  By structural induction on $\varphi$, the cases of $\top$, $\bot$, $\varphi_1
  \land \varphi_2$, and $\varphi_1 \lor \varphi_2$ are immediate.
  \begin{itemize}

  \item \emph{Case $\varphi=\sure{P}$.}
  Let $g$ be the product pairing for $\calR\otimes\calR_F$, and let
  \[
    E\triangleq
    \memfn_\calR^{-1}(\sem{P}_\Gamma^{V_\calR})
  \]
  Since $\Gamma,\calR\vDash\sure{P}$, $E\in\calF_\calR$ and
  $\mu_\calR(E)=1$.  By
  \Cref{lem:index-pure-assertion-footprint}, the almost-sure event for
  $P$ in $\calR\otimes\calR_F$ is $g(E\times\Omega_{\calR_F})$.  Hence it
  is measurable and has measure
  $\mu_\calR(E)\cdot\mu_{\calR_F}(\Omega_{\calR_F})=1$.

  \item \emph{Case $\varphi=\bigoplus^o_{X\sim d(E)}\psi$ and $\varphi =
    \bignd_{X \in E} \psi$. }
  Put $\mu\triangleq d(\de{E}_{\LExp}(\Gamma))$.  Unfold satisfaction and
  choose witnesses $(\calR_v)_{v\in\supp(\mu)}$ such that
  \[
    \bigoplus_{v\sim\mu}\calR_v\preceq\calR,
    \qquad
    \Gamma[X\mapsto v],\calR_v\vDash\psi
    \quad(v\in\supp(\mu)),
  \]
  and $\Obs_o^\mu(\calR,(\calR_v)_{v\in\supp(\mu)})$.  By the induction
  hypothesis,
  $\Gamma[X\mapsto v],\calR_v\otimes\calR_F\vDash\psi$ for every
  $v\in\supp(\mu)$.  Moreover,
  \[
    \bigoplus_{v\sim\mu}(\calR_v\otimes\calR_F)
    \cong
    \left(\bigoplus_{v\sim\mu}\calR_v\right)\otimes\calR_F
    \preceq
    \calR\otimes\calR_F
  \]
  If $o=\Leak$, the mode side condition is immediate.  If $o=\Priv$, then
  for every $v\in\supp(\mu)$,
  \begin{align*}
    \Proj{\calR_v\otimes\calR_F}{\bbN}
    &=
      (\mathsf{add}_{\bbN})_*
      (\Proj{\calR_v}{\bbN}\otimes\Proj{\calR_F}{\bbN})\\
    &=
      (\mathsf{add}_{\bbN})_*
      (\Proj{\calR}{\bbN}\otimes\Proj{\calR_F}{\bbN})\\
    &=
      \Proj{\calR\otimes\calR_F}{\bbN}
  \end{align*}
  using the private side condition for $\calR$.  Thus
  $\Gamma,\calR\otimes\calR_F\vDash\bigoplus^o_{X\sim d(E)}\psi$.

  The case $\varphi=\bignd_{X\in E}\psi$ uses the same distribution witnessing the nondeterministic choice and the
  previous case for $\bigoplus^\Leak$.

  \item \emph{Case $\varphi=\varphi_1\sep\varphi_2$.}
  Choose witnesses $\calR_1,\calR_2$ such that
  $\calR_1\otimes\calR_2\preceq\calR$ and
  $\Gamma,\calR_i\vDash\varphi_i$ for $i \in \{1, 2\}$.
  By the induction hypothesis,
  $\Gamma,\calR_1\otimes\calR_F\vDash\varphi_1$.  Also
  \[
    (\calR_1\otimes\calR_F)\otimes\calR_2
    \cong
    (\calR_1\otimes\calR_2)\otimes\calR_F
    \preceq
    \calR\otimes\calR_F
  \]
  Therefore $\calR_1\otimes\calR_F$ and $\calR_2$ witness
  $\Gamma,\calR\otimes\calR_F\vDash\varphi_1\sep\varphi_2$.
  \end{itemize}
\end{proof}

\begin{lemma}[Count-erased precise core]
  \label{lem:index-precise-core}
  If $\precise{\varphi}$ and $\Gamma,\calR\vDash\varphi$, then there exists
  a resource $\calR_0$ such that
  \[
    \Gamma,\calR_0\vDash\varphi,
    \qquad
    V_{\calR_0}\subseteq V_\calR,
    \qquad
    \cntfn_{\calR_0}(\omega)=0
    \quad\text{for all }\omega\in\Omega_{\calR_0},
  \]
  and for every $\calQ$ satisfying $\Gamma,\calQ\vDash\varphi$,
  \[
    \calR_0^\memfn\preceq\calQ^\memfn
  \]
\end{lemma}
\begin{proof}
  By precision, choose a least witness $\calQ_0$ for $\varphi$ under
  $\Gamma$.  Let $\calR_0\triangleq\calQ_0^\memfn$.  Count-erasure preserves
  satisfaction, so $\Gamma,\calR_0\vDash\varphi$.  Since
  $\calQ_0^\memfn\preceq\calR^\memfn$, the footprint of $\calR_0$ is contained
  in $V_\calR$, and its count map is constantly $0$ by construction.  The same
  precision witness gives
  $\calR_0^\memfn\cong\calQ_0^\memfn\preceq\calQ^\memfn$ for every satisfying
  $\calQ$.
\end{proof}

\begin{lemma}[Assertion precision is preserved by constructors]
  \label{lem:index-assertion-precision-constructors}
  The precision rules displayed in
  \Cref{fig:assertion-property-rules} are valid.
\end{lemma}
\begin{proof}
  We prove the precision rules in \Cref{fig:assertion-property-rules} by
  constructing a least memory witness.

  \emph{Case $\sure{P}$.}
  Fix $\Gamma$ and assume $\sure{P}$ is satisfiable. Since $\Var(P)$ is
  finite, $\Mem[\Var(P)]$ is countable, so choose a countable index set
  $\Omega_P\subseteq\bbN$ and a bijection
  $\rho_P:\Omega_P\to\Mem[\Var(P)]$. Let
  \[
    B_P\triangleq\rho_P^{-1}(\sem{P}_\Gamma^{\Var(P)})
  \]
  Since $\sure{P}$ is satisfiable, $B_P$ is nonempty.  Equip $\Omega_P$ with
  the complete two-atom probability space whose measurable sets are exactly
  those $A$ such that $B_P\subseteq A$ or $A\cap B_P=\emptyset$, assigning
  probability $1$ in the first case and $0$ in the second.  Define
  \[
    \calR_P
    \triangleq
    \langle
      \langle\Omega_P,\calF_P,\mu_P\rangle,\,
      \Var(P),\,\rho_P,\,\omega\mapsto 0
    \rangle
  \]
  The $P$-event in $\calR_P$ is exactly $B_P$, so
  $\Gamma,\calR_P\vDash\sure{P}$.

  To prove minimality, let $\Gamma,\calR\vDash\sure{P}$ and set
  $E_\calR\triangleq
  \memfn_\calR^{-1}(\sem{P}_\Gamma^{V_\calR})$.
  Define
  \[
    h_\calR(\omega)
    \triangleq
    \rho_P^{-1}\!\left(\ForgetMem{\memfn_\calR(\omega)}{\Var(P)}\right)
  \]
  By \Cref{lem:index-pure-assertion-footprint},
  $h_\calR^{-1}(B_P)=E_\calR$.  Since $E_\calR$ has probability $1$, the
  preimage of every $\calF_P$-event is either a superset of $E_\calR$ up to a
  null set or a subset of $E_\calR^c$; completeness gives measurability and
  preservation of measure.  The observation equations are
  \[
    \memfn_{\calR_P}(h_\calR(\omega))
    =
    \ForgetMem{\memfn_\calR(\omega)}{\Var(P)},
    \qquad
    \cntfn_{\calR_P}(h_\calR(\omega))=0=\cntfn_{\calR^\memfn}(\omega)
  \]
  Hence $h_\calR$ witnesses
  $\calR_P^\memfn\preceq\calR^\memfn$, proving $\precise{\sure{P}}$.

  \emph{Case $\varphi\lor\psi$.}
  Fix $\Gamma$ and a satisfying resource for $\varphi\lor\psi$.  If it
  satisfies $\varphi$, take the least $\varphi$-witness.  Every other
  $\varphi\lor\psi$ witness under the same $\Gamma$ must also satisfy
  $\varphi$, since $\textsf{MutEx}(\varphi,\psi)$ rules out satisfiable
  $\psi$-witnesses under that $\Gamma$; hence $\varphi$-minimality gives the
  required lower bound.  The $\psi$ case is symmetric.

  \emph{Case $\varphi_1\sep\varphi_2$.}
  Fix $\Gamma$ and suppose $\Gamma,\calR\vDash\varphi_1\sep\varphi_2$.
  Choose witnesses $\calR_1,\calR_2$ such that
  $\calR_1\otimes\calR_2\preceq\calR$ and
  $\Gamma,\calR_i\vDash\varphi_i$ for $i\in\{1,2\}$.  By
  \Cref{lem:index-precise-core}, choose count-erased cores $\calR_i^0$ for
  $\calR_i$.  Since $V_{\calR_i^0}\subseteq V_{\calR_i}$ and the original split
  has disjoint footprints,
  $\calR_\sep^0\triangleq\calR_1^0\otimes\calR_2^0$ is well formed and
  satisfies $\varphi_1\sep\varphi_2$.

  For minimality, let $\Gamma,\calR'\vDash\varphi_1\sep\varphi_2$ be witnessed
  by $\calQ_1\otimes\calQ_2\preceq\calR'$ with
  $\Gamma,\calQ_i\vDash\varphi_i$.  Core minimality gives
  $(\calR_i^0)^\memfn\preceq\calQ_i^\memfn$, and hence
  \[
    (\calR_\sep^0)^\memfn
    \cong
    (\calR_1^0)^\memfn\otimes(\calR_2^0)^\memfn
    \preceq
    \calQ_1^\memfn\otimes\calQ_2^\memfn
    \cong
    (\calQ_1\otimes\calQ_2)^\memfn
    \preceq
    (\calR')^\memfn
  \]
  Hence $\precise{\varphi_1\sep\varphi_2}$.

  \emph{Case $\bigoplus^o_{X\sim d(E)}\varphi$.}
  Fix $\Gamma$ and $o\in\{\Priv,\Leak\}$, and put
  $\mu\triangleq d(\de{E}_{\LExp}(\Gamma))$.  Suppose
  $\Gamma,\calR\vDash\bigoplus^o_{X\sim d(E)}\varphi$ is witnessed by branches
  $(\calR_v)_{v\in\supp(\mu)}$ with common footprint $V$.  By
  \Cref{lem:index-precise-core}, choose count-erased cores $\calQ_v^0$ for the
  branch resources $\calR_v$, and put
  \[
    W\triangleq
    \{x\in V \mid \exists v\in\supp(\mu).\;x\in V_{\calQ_v^0}\}
  \]
  Extend each $\calQ_v^0$ to the common footprint $W$ and reindex the resulting
  resources onto pairwise-disjoint supports; write the normalized branch as
  $\calQ_{v,\mathsf m}$.  Define
  \[
    \calR_\oplus^0
    \triangleq
    \bigoplus_{v\sim\mu}\calQ_{v,\mathsf m}
  \]
  Footprint extension and isomorphism transport preserve the branch assertions.
  If $o=\Leak$, there is no observation side condition;
  if $o=\Priv$, all normalized branches and their direct sum have count
  marginal $\delta_0$.  Thus
  $\Gamma,\calR_\oplus^0\vDash\bigoplus^o_{X\sim d(E)}\varphi$.

  For minimality, let $\Gamma,\calR'\vDash\bigoplus^o_{X\sim d(E)}\varphi$ be
  witnessed by branches $(\calR'_v)_{v\in\supp(\mu)}$ with common footprint
  $V'$.  Branchwise minimality gives
  $(\calQ_v^0)^\memfn\preceq(\calR'_v)^\memfn$, hence
  $V_{\calQ_v^0}\subseteq V'$ for every $v$ and therefore $W\subseteq V'$.
  After extension to $W$ and reindexing, still
  $\calQ_{v,\mathsf m}^\memfn\preceq(\calR'_v)^\memfn$ for every $v$.
  Direct-sum monotonicity and compatibility with memory components give
  \[
    (\calR_\oplus^0)^\memfn
    \preceq
    \left(\bigoplus_{v\sim\mu}\calR'_v\right)^\memfn
    \preceq
    (\calR')^\memfn
  \]
  Hence $\precise{\bigoplus^o_{X\sim d(E)}\varphi}$.
\end{proof}

\begin{lemma}[Stable core of a precise assertion]
  \label{lem:index-precise-stable-core}
  If $\precise{\varphi}$ and $\Stable{\weak}{\varphi}$, then for every
  environment $\Gamma$ under which $\varphi$ is satisfiable, there exists
  $\calQ_\varphi^0$ such that
  \[
    \Gamma,\calQ_\varphi^0\vDash\varphi,
    \qquad
    \Stable{\weak}{\calQ_\varphi^0},
  \]
  and
  \[
    \forall\calQ.\;
    \Gamma,\calQ\vDash\varphi
    \land
    \Stable{\weak}{\calQ}
    \implies
    \calQ_\varphi^0\preceq\calQ
  \]
\end{lemma}
\begin{proof}
  Fix $\Gamma$ such that $\exists\calR.\;\Gamma,\calR\vDash\varphi$.  By
  $\precise{\varphi}$ choose $\calQ_\varphi$ such that
  \[
    \Gamma,\calQ_\varphi\vDash\varphi,
    \qquad
    \forall\calQ.\;
    \Gamma,\calQ\vDash\varphi
    \implies
    \calQ_\varphi^\memfn\preceq\calQ^\memfn
  \]
  By $\Stable{\weak}{\varphi}$ applied to $\calQ_\varphi$, choose
  $\calQ_\varphi^0,\calQ_\varphi^\cntfn$ such that
  \[
    \calQ_\varphi^0\otimes\calQ_\varphi^\cntfn
    \preceq\calQ_\varphi,
    \qquad
    \Gamma,\calQ_\varphi^0\vDash\varphi,
    \qquad
    \Stable{\weak}{\calQ_\varphi^0},
    \qquad
    V_{\calQ_\varphi^\cntfn}=\emptyset
  \]
  It remains to prove the minimality condition among stable witnesses.  Fix
  $\calQ$ with $\Gamma,\calQ\vDash\varphi$ and $\Stable{\weak}{\calQ}$.  Then
  \[
    \calQ_\varphi^0
    \cong
    (\calQ_\varphi^0)^\memfn
    \preceq
    (\calQ_\varphi^0\otimes\calQ_\varphi^\cntfn)^\memfn
    \preceq
    \calQ_\varphi^\memfn
    \preceq
    \calQ^\memfn
    \cong
    \calQ
  \]
\end{proof}

\begin{lemma}[Count restrictions preserve almost-sure pure assertions]
  \label{lem:index-count-restriction-pure-assertion}
  If $\Gamma,\calR\vDash\sure{P}$, $E\subseteq\bbN$, and
  $\Proj{\calR}{\bbN}(E)>0$, then
  \[
    \Gamma,\RestrictN{\calR}{E}\vDash\sure{P}
  \]
\end{lemma}
\begin{proof}
  Consider
  $G\triangleq\cntfn_\calR^{-1}(E)$ and 
  $A\triangleq
  \memfn_\calR^{-1}(\sem{P}_\Gamma^{V_\calR})$.
  Since $\Gamma,\calR\vDash\sure{P}$, we have
  $A\in\calF_\calR$ and $\mu_\calR(A)=1$.  Also
  $\mu_\calR(G)=\Proj{\calR}{\bbN}(E)>0$, so
  $\RestrictN{\calR}{E}$ is defined.  By the definition of count restriction,
  $V_{\RestrictN{\calR}{E}}=V_\calR$ and
  $\memfn_{\RestrictN{\calR}{E}}$ is the restriction of $\memfn_\calR$ to
  $G$.  Hence
  \begin{equation}
    \memfn_{\RestrictN{\calR}{E}}^{-1}
    \left(\sem{P}_\Gamma^{V_{\RestrictN{\calR}{E}}}\right)
    =
    A\cap G .
    \label{eq:index-count-restriction-pure-event}
  \end{equation}
  Since $A,G\in\calF_\calR$, the event $A\cap G$ is measurable in the
  conditioned probability space.  Moreover,
  by the definition of count restriction and $\mu_\calR(A) = 1$, 
  \[
    \mu_{\RestrictN{\calR}{E}}(A\cap G)
    =
    \frac{\mu_\calR(A\cap G)}{\mu_\calR(G)}
    =
    \frac{\mu_\calR(G)}{\mu_\calR(G)}
    =
    1
  \]
  By \eqref{eq:index-count-restriction-pure-event} and the satisfaction
  clause for $\sure{P}$,
  $\Gamma,\RestrictN{\calR}{E}\vDash\sure{P}$.
\end{proof}

\subsection{Convexity}
\begin{lemma}[Assertion convexity is preserved by constructors]
  \label{lem:index-assertion-convexity-constructors}
  The convexity rules displayed in
  \Cref{fig:assertion-property-rules} are valid.
\end{lemma}
\begin{proof}
  We prove the convexity rules in \Cref{fig:assertion-property-rules}
  directly: fix a countable index
  set $I$, a distribution $\mu\in\D(I)$, a finite footprint $V$, and a
  pairwise-disjoint family $(\calR_i)_{i\in I}$ such that
  $V_{\calR_i}=V$ for every $i\in\supp(\mu)$.  Write
  $\calR^\oplus\triangleq\bigoplus_{i\sim\mu}\calR_i$. The cases for $\top$,
  $\bot$, and $\varphi \land \psi$ are immediate.
  \begin{itemize}
    \item \emph{Case $\sure{P}$.}
    Suppose $\Gamma,\calR_i\vDash\sure{P}$ for every $i\in\supp(\mu)$.  Let
    $E_i\triangleq
    \memfn_{\calR_i}^{-1}(\sem{P}_\Gamma^{V_{\calR_i}})$.
    Then $E_i\in\calF_{\calR_i}$ and $\mu_{\calR_i}(E_i)=1$ on every
    positive branch.  In $\calR^\oplus$, the corresponding event is the
    disjoint union of the $E_i$, and its measure is
    \[
      \sum_{i\in\supp(\mu)}
      \mu(i)\cdot\mu_{\calR_i}(E_i)
      =
      \sum_{i\in\supp(\mu)}\mu(i)
      =1
    \]
    Hence $\Gamma,\calR^\oplus\vDash\sure{P}$.

    \item \emph{Case $\varphi\sep\psi$, where $\precise{\varphi}$,
      $\Stable{\weak}{\varphi}$, and $\convex{\psi}$.}
    Suppose $\Gamma,\calR_i\vDash\varphi\sep\psi$ for every
    $i\in\supp(\mu)$.  Unfolding satisfaction, choose witnesses
    $\calR_{\varphi,i},\calR_{\psi,i}$ such that
    \[
      \calR_{\varphi,i}\otimes\calR_{\psi,i}\preceq\calR_i,
      \qquad
      \Gamma,\calR_{\varphi,i}\vDash\varphi,
      \qquad
      \Gamma,\calR_{\psi,i}\vDash\psi
      \quad(i\in\supp(\mu))
    \]
    Since $\varphi$ is satisfiable under $\Gamma$,
    \Cref{lem:index-precise-stable-core} gives a stable core
    $\calQ_\varphi^0$ satisfying $\varphi$ and minimal among stable
    $\varphi$-resources:
    \begin{equation}
      \forall\calQ.\;
      \Gamma,\calQ\vDash\varphi
      \land
      \Stable{\weak}{\calQ}
      \implies
      \calQ_\varphi^0\preceq\calQ
      \label{eq:index-sep-convex-stable-minimal}
    \end{equation}

    For each $i\in\supp(\mu)$, apply $\Stable{\weak}{\varphi}$ to
    $\calR_{\varphi,i}$ and choose
    $\calR_{\varphi,i}^\memfn,\calR_{\varphi,i}^\cntfn$ such that
    $\calR_{\varphi,i}^\memfn\otimes\calR_{\varphi,i}^\cntfn
    \preceq\calR_{\varphi,i}$,
    $\Gamma,\calR_{\varphi,i}^\memfn\vDash\varphi$, and
    $\Stable{\weak}{\calR_{\varphi,i}^\memfn}$.  The second component is
    count-only:
    $V_{\calR_{\varphi,i}^\cntfn}=\emptyset$ and
    $\memfn_{\calR_{\varphi,i}^\cntfn}(\omega)=\emptyset$ for every
    $\omega$.
    By \eqref{eq:index-sep-convex-stable-minimal}, we get that
    $\calQ_\varphi^0\preceq\calR_{\varphi,i}^\memfn$ for every
    $i\in\supp(\mu)$.

    Put $W\triangleq V\setminus V_{\calQ_\varphi^0}$.  For each
    $i\in\supp(\mu)$, the count-only frame
    $\calR_{\psi,i}\otimes\calR_{\varphi,i}^\cntfn$ still satisfies $\psi$ by
    \Cref{lem:index-assertion-frame-preservation}.  Extend these resources
    to the common footprint $W$ and relocate them to pairwise-disjoint supports;
    write the normalized resource as $\calQ_{\psi,i}$.  These transformations
    preserve satisfaction, so $\Gamma,\calQ_{\psi,i}\vDash\psi$ for every
    $i\in\supp(\mu)$.
    By $\convex{\psi}$,
    \[
      \calQ_\psi^\oplus\triangleq\bigoplus_{i\sim\mu}\calQ_{\psi,i}
      \qquad\text{satisfies}\qquad
      \Gamma,\calQ_\psi^\oplus\vDash\psi
    \]

    Finally,
    \[
      \calQ_\varphi^0\otimes\calQ_\psi^\oplus
      \cong
        \bigoplus_{i\sim\mu}
        \left(\calQ_\varphi^0\otimes\calQ_{\psi,i}\right)
      \preceq
        \bigoplus_{i\sim\mu}
        \left(
          \calR_{\varphi,i}^\memfn
          \otimes
          \calR_{\varphi,i}^\cntfn
          \otimes
          \calR_{\psi,i}
        \right)
      \preceq
        \bigoplus_{i\sim\mu}
        \left(\calR_{\varphi,i}\otimes\calR_{\psi,i}\right)
      \preceq
        \calR^\oplus
      \]
    The witnesses $\calQ_\varphi^0$ and $\calQ_\psi^\oplus$ therefore show
    $\Gamma,\calR^\oplus\vDash\varphi\sep\psi$.

    The rule with the precise stable side on the right is the same argument
    with the roles of $\varphi$ and $\psi$ swapped.

    \item \emph{Case $\bigoplus^o_{X\sim d(E)}\varphi$, for
      $o\in\{\Leak,\Priv\}$.}
    Suppose
    $\Gamma,\calR_i\vDash\bigoplus^o_{X\sim d(E)}\varphi$ for every
    $i\in\supp(\mu)$.  Put
    $\alpha\triangleq d(\de{E}_{\LExp}(\Gamma))$.  Unfolding satisfaction,
    choose, for each $i\in\supp(\mu)$, branch resources
    $(\calR_{i,v})_{v\in\supp(\alpha)}$ such that
    \begin{align}
      \bigoplus_{v\sim\alpha}\calR_{i,v}
      &\preceq
        \calR_i
        \label{eq:index-constructor-indexed-oplus-refines}\\
      \Gamma[X\mapsto v],\calR_{i,v}
      &\vDash
        \varphi
        \qquad(v\in\supp(\alpha))
        \label{eq:index-constructor-indexed-oplus-branches}\\
      \Obs_o^\alpha(\calR_i,(\calR_{i,v})_{v\in\supp(\alpha)})
      &\text{ holds}
        \label{eq:index-constructor-indexed-oplus-obs}
    \end{align}
    Extend the resources $\calR_{i,v}$ to the common outer footprint $V$ and
    relocate them to pairwise-disjoint supports indexed by pairs $(i,v)$; call
    the resulting resources $\calQ_{i,v}$.  For each $v\in\supp(\alpha)$, put
    \[
      \calQ_v^\oplus\triangleq\bigoplus_{i\sim\mu}\calQ_{i,v}
    \]
    By $\convex{\varphi}$ and the branch facts above,
    $\Gamma[X\mapsto v],\calQ_v^\oplus\vDash\varphi$ for every
    $v\in\supp(\alpha)$.  Direct-sum regrouping gives
    \begin{align}
      \bigoplus_{v\sim\alpha}\calQ_v^\oplus
      &\cong
        \bigoplus_{i\sim\mu}
        \left(\bigoplus_{v\sim\alpha}\calQ_{i,v}\right)
        \tag{regrouping of direct sums}\\
      &\preceq
        \calR^\oplus 
        \tag{by \eqref{eq:index-constructor-indexed-oplus-refines} and
        monotonicity}
    \end{align}
    If $o=\Leak$, the observation condition is immediate.  If $o=\Priv$,
    \eqref{eq:index-constructor-indexed-oplus-obs} gives
    $\Proj{\calR_{i,v}}{\bbN}
    =
    \Proj{\calR_i}{\bbN}$ for all $i \in \supp(\mu)$ and $v \in \supp(\alpha)$.
    Thus, for every $v\in\supp(\alpha)$,
    \[
      \Proj{\calQ_v^\oplus}{\bbN}
      =
      \bigoplus_{i\sim\mu}\Proj{\calR_{i,v}}{\bbN}
      =
      \bigoplus_{i\sim\mu}\Proj{\calR_i}{\bbN}
      =
      \Proj{\calR^\oplus}{\bbN}
    \]
    Hence the witnesses $(\calQ_v^\oplus)_{v\in\supp(\alpha)}$ prove
    $\Gamma,\calR^\oplus\vDash\bigoplus^o_{X\sim d(E)}\varphi$.

    \item \emph{Case $\bignd_{X\in E}\varphi$.}
    Suppose $\Gamma,\calR_i\vDash\bignd_{X\in E}\varphi$ for every
    $i\in\supp(\mu)$.  Put $A\triangleq\de{E}_{\LExp}(\Gamma)$.  For each
    $i\in\supp(\mu)$, unfold nondeterministic choice to obtain a distribution
    $\alpha_i\in\D(A)$ and branches
    $(\calR_{i,v})_{v\in\supp(\alpha_i)}$ witnessing
    $\Gamma,\calR_i\vDash\bigoplus^\Leak_{X\sim\alpha_i}\varphi$.
    Let
    \[
      \alpha(v)\triangleq
      \sum_{i\in\supp(\mu)} \mu(i)\cdot \alpha_i(v)
      \qquad(v\in A)
    \]
    For $v\in\supp(\alpha)$, let $\mu_v$ be the conditional distribution on
    outer branches, proportional to $\mu(i)\cdot\alpha_i(v)$.  After extending
    and relocating the pair-indexed resources as in the previous case, define
    \[
      \calQ_v^\oplus\triangleq
      \bigoplus_{i\sim\mu_v}\calQ_{i,v}
    \]
    By $\convex{\varphi}$, each $\calQ_v^\oplus$ satisfies
    $\varphi$ under $\Gamma[X\mapsto v]$.  The flattening/regrouping identity
    for direct sums gives
    \begin{align}
      \bigoplus_{v\sim\alpha}\calQ_v^\oplus
      &\cong
        \bigoplus_{i\sim\mu}
        \left(\bigoplus_{v\sim\alpha_i}\calQ_{i,v}\right)
        \tag{direct-sum flattening}\\
      &\preceq
        \calR^\oplus
    \end{align}
    Since leaky observation imposes no side condition, the witnesses
    $(\calQ_v^\oplus)_{v\in\supp(\alpha)}$ show
    $\Gamma,\calR^\oplus\vDash\bigoplus^\Leak_{X\sim\alpha}\varphi$.  Since
    $\alpha\in\D(A)$, this folds to
    $\Gamma,\calR^\oplus\vDash\bignd_{X\in E}\varphi$.
  \end{itemize}
\end{proof}

\begin{lemma}[Pure assertions transfer across refinement]
  \label{lem:index-pure-assertion-refinement-transfer}
  Let $P$ be a pure assertion and define
  \[
    B_P^\Gamma
    \triangleq
    \{(\sigma,k)\in\Mem\times\bbN
    \mid
    \ForgetMem{\sigma}{\Var(P)}
    \in
    \sem{P}_\Gamma^{\Var(P)}\}
  \]
  as the event in which the concrete memory satisfies $P$ on the footprint of
  $P$.  If
  \[
    \Gamma,\calR\vDash\sure{P},
    \qquad
    \calR\preceq\mu,
  \]
  then $\mu(B_P^\Gamma)=1$.
\end{lemma}
\begin{proof}
  Let $\Omega_\mu\triangleq\supp(\mu)$, and let $\lambda$ witness
  $\calR\preceq\mu$.  Define
  \[
    E_P
    \triangleq
    \{\omega\in\Omega_\calR
    \mid
    \memfn_\calR(\omega)\in\sem{P}_\Gamma^{V_\calR}\}
  \]
  as the event in which the resource memory satisfies $P$.  Unfolding
  $\Gamma,\calR\vDash\sure{P}$ gives $\mu_\calR(E_P)=1$.

  We first prove the pointwise implication from the resource event to the
  concrete event.  Suppose
  $\omega\in E_P$
  and
  $\lambda(\omega,(\sigma,k))>0$
  The support condition in $\calR\preceq\mu$ gives
  $\memfn_\calR(\omega)\sqsubseteq\ForgetMem{\sigma}{V_\calR}$.  The
  footprint requirement in the satisfaction clause for $\sure{P}$ gives
  $\Var(P)\subseteq V_\calR$.  Hence we derive
  \begin{align*}
    &\memfn_\calR(\omega)
      \in
      \sem{P}_\Gamma^{V_\calR}\\
    &\Longrightarrow
      \ForgetMem{\memfn_\calR(\omega)}{\Var(P)}
      \in
      \sem{P}_\Gamma^{\Var(P)}
      \tag{\Cref{lem:index-pure-assertion-footprint}}\\
    &\Longrightarrow
      \ForgetMem{\sigma}{\Var(P)}
      \in
      \sem{P}_\Gamma^{\Var(P)}
      \tag{pure monotonicity and $\memfn_\calR(\omega)
      \sqsubseteq\ForgetMem{\sigma}{V_\calR}$}
  \end{align*}
  Thus $(\sigma,k)\in B_P^\Gamma$, and so
  \[
    \omega\in E_P
    \;\land\;
    \lambda(\omega,(\sigma,k))>0
    \quad\Longrightarrow\quad
    (\sigma,k)\in B_P^\Gamma
  \]
  Therefore
  \begin{align*}
    \lambda(E_P\times(\Omega_\mu\setminus B_P^\Gamma))
    &=
      \sum_{\omega\in E_P}
      \sum_{(\sigma,k)\in\Omega_\mu\setminus B_P^\Gamma}
      \lambda(\omega,(\sigma,k))
      \tag{definition of product-set mass}\\
    &=
      0
      \tag{pointwise implication above}
  \end{align*}
  By the marginal equations for $\lambda$,
  \begin{align*}
    \mu(B_P^\Gamma)
    &=
      \lambda(\Omega_\calR\times(\Omega_\mu\cap B_P^\Gamma))
      \tag{right marginal of $\lambda$}\\
    &\ge
      \lambda(E_P\times(\Omega_\mu\cap B_P^\Gamma))
      \tag{monotonicity, since $E_P\subseteq\Omega_\calR$}\\
    &=
      \lambda(E_P\times\Omega_\mu)
      -
      \lambda(E_P\times(\Omega_\mu\setminus B_P^\Gamma))
      \tag{disjoint decomposition of $E_P\times\Omega_\mu$}\\
    &=
      \lambda(E_P\times\Omega_\mu)
      \tag{$\lambda(E_P\times(\Omega_\mu\setminus B_P^\Gamma))=0$}\\
    &=
      \mu_\calR(E_P)
      \tag{left marginal of $\lambda$}\\
    &=
      1 .
      \tag{satisfaction clause for $\sure{P}$}
  \end{align*}
  Since $\mu$ is a probability distribution, $\mu(B_P^\Gamma)=1$.
\end{proof}

\begin{lemma}[Guard assertions transfer across refinement]
  \label{lem:index-guard-assertion-refinement-transfer}
  Consider some expression $e$. Let $c\in\{\fls,\tru\}$ and define
  $G_c\triangleq
  \{(\sigma,k)\in\Mem\times\bbN
  \mid \de{e}_{\Exp}(\sigma)=c\}$.
  If
  \[
    \Gamma,\calR\vDash\sure{e\mapsto c},
    \qquad
    \calR\preceq\mu,
  \]
  then $\mu(G_c)=1$.
\end{lemma}
\begin{proof}
  Apply \Cref{lem:index-pure-assertion-refinement-transfer} to
  $P\triangleq e\mapsto c$.  The resulting event is
  \[
    B_{e\mapsto c}^\Gamma
    =
    \{(\sigma,k)\in\Mem\times\bbN
    \mid
    \ForgetMem{\sigma}{\Var(e)}
    \in
    \sem{e\mapsto c}_\Gamma^{\Var(e)}\}
  \]
  By the definition of $\sem{e\mapsto c}_\Gamma^{\Var(e)}$ and the fact that
  evaluation of $e$ depends only on $\Var(e)$, this event is exactly $G_c$.
  Hence $\mu(G_c)=1$.
\end{proof}

\begin{lemma}[Satisfaction Monotonicity]
  \label{lem:index-satisfaction-monotonicity}
  If $\Gamma,\calR\vDash\varphi$ and $\calR\preceq\calR'$, then
  $\Gamma,\calR'\vDash\varphi$.
\end{lemma}
\begin{proof}
  By induction on the structure of $\varphi$. The cases for $\top$, $\bot$,
  $\varphi_1 \land \varphi_2$, and $\varphi_1 \lor \varphi_2$ are immediate.
  \begin{itemize}
    \item \emph{Case $\varphi=\bigoplus^o_{X\sim d(E)}\psi$.}
    Suppose
    $\Gamma,\calR\vDash\bigoplus^o_{X\sim d(E)}\psi$ and
    $\calR\preceq\calR'$.  Let $\mu=d(\de{E}_{\LExp}(\Gamma))$.  By
    unfolding satisfaction, choose witnesses
    $(\calR_v)_{v\in\supp(\mu)}$ such that
    $\bigoplus_{v\sim\mu}\calR_v\preceq\calR$ and
    \[
      \forall v\in\supp(\mu).\;
      \Gamma[X\mapsto v],\calR_v\vDash\psi,
    \]
    and
    $\Obs_o^\mu(\calR,(\calR_v)_{v\in\supp(\mu)})$.
    By transitivity of $\preceq$,
    $\bigoplus_{v\sim\mu}\calR_v\preceq\calR'$.
    If $o=\Leak$, the mode side condition is trivial.  If $o=\Priv$, then
    $\forall v\in\supp(\mu).\;
    \Proj{\calR_v}{\bbN}=\Proj{\calR}{\bbN}$
    and $\calR\preceq\calR'$ implies
    $\Proj{\calR}{\bbN}=\Proj{\calR'}{\bbN}$ by
    \Cref{fig:resource-algebra-refinement-facts}.  Hence
    \[
      \forall v\in\supp(\mu).\;
      \Proj{\calR_v}{\bbN}=\Proj{\calR'}{\bbN}
    \]
    so $\Obs_o^\mu(\calR',(\calR_v)_{v\in\supp(\mu)})$.
    Therefore
    $\Gamma,\calR'\vDash\bigoplus^o_{X\sim d(E)}\psi$.

    \item \emph{Case $\varphi=\bignd_{X\in E}\psi$.}
    Use the same distribution witnessing the nondeterministic choice and the
    previous case for $\bigoplus^\Leak$.

    \item \emph{Case $\varphi=\sure{P}$.}
    Let $h:\Omega_{\calR'}\to\Omega_\calR$ witness
    $\calR\preceq\calR'$.  The almost-sure event for $P$ in $\calR$ pulls
    back along $h$ to a measurable event of measure one in $\calR'$, so
    $\Gamma,\calR'\vDash\sure{P}$.

    \item \emph{Case $\varphi=\varphi_1\sep\varphi_2$.}
    Choose witnesses $\calR_1,\calR_2$ such that
    \[
      \calR_1\otimes\calR_2\preceq\calR,
      \qquad
      \Gamma,\calR_1\vDash\varphi_1,
      \qquad
      \Gamma,\calR_2\vDash\varphi_2
    \]
    By transitivity of $\preceq$, the same witnesses satisfy
    $\calR_1\otimes\calR_2\preceq\calR'$, hence
    $\Gamma,\calR'\vDash\varphi_1\sep\varphi_2$.
  \end{itemize}
\end{proof}

\section{Entailment and Assertion Algebra}

\subsection{Probabilistic and Nondeterministic Assertion}

\begin{lemma}[Basic entailment laws for outcome assertions]
  \label{lem:index-basic-entailment-laws}
  \label{lem:index-binary-probabilistic-commutativity}
  \label{lem:index-binary-probabilistic-associativity}
  \label{lem:index-nondeterministic-idempotence}
  \label{lem:index-nondeterministic-singleton-introduction}
  \label{lem:index-nondeterministic-support-membership}
  \label{lem:index-nondeterministic-support-weakening}
  \label{lem:index-probabilistic-choice-alpha-renaming}
  \label{lem:index-binary-nondeterministic-commutativity}
  \label{lem:index-binary-nondeterministic-associativity}
  \label{lem:index-binary-nondeterministic-distributivity}
  \label{lem:index-binary-probabilistic-zero}
  \label{lem:index-binary-probabilistic-one}
  The entailment rules in \Cref{fig:basic-entailment-laws} are valid.
\end{lemma}
\begin{proof}
  We prove the rules in the order shown in \Cref{fig:basic-entailment-laws}.
  \begin{enumerate}[label=\textup{(\arabic*)}]
    \item
    $\inferrule
    {}
    {\textstyle
      \varphi \oplus^o_p \psi
      \dashv\vdash
      \psi \oplus^o_{1-p} \varphi}$.

    Let $\beta_p,\beta_{1-p}\in\calD(\{L,R\})$ be given by
    \[
      \beta_p(L)=p,
      \quad
      \beta_p(R)=1-p,
      \qquad
      \beta_{1-p}(L)=1-p,
      \quad
      \beta_{1-p}(R)=p
    \]
    We prove $\varphi \oplus^o_p \psi
    \vdash
    \psi \oplus^o_{1-p}\varphi$; the converse is identical with $p$ replaced
    by $1-p$.

    Fix $\Gamma,\calR$ and assume
    $\Gamma,\calR\vDash\varphi\oplus^o_p\psi$.
    Choose witnesses $\calR_L,\calR_R$ such that
    \[
      \bigoplus_{i\sim\beta_p}\calR_i\preceq\calR,\qquad
      \Gamma,\calR_L\vDash\varphi,\qquad
      \Gamma,\calR_R\vDash\psi,\qquad
      \Obs_o^{\beta_p}(\calR,(\calR_i)_{i\in\supp(\beta_p)})
    \]
    Define $\calR'_L\triangleq\calR_R$ and
    $\calR'_R\triangleq\calR_L$.  Then
    $\Gamma,\calR'_L\vDash\psi$ and
    $\Gamma,\calR'_R\vDash\varphi$.  The branch swap gives
    \[
      \bigoplus_{i\sim\beta_{1-p}}\calR'_i
      \cong
      \bigoplus_{i\sim\beta_p}\calR_i
      \preceq
      \calR
    \]
    If $o=\Leak$, the mode side condition is trivial.  If $o=\Priv$, then
    the observation condition gives
    $\Proj{\calR_L}{\bbN}
    =
    \Proj{\calR}{\bbN}
    =
    \Proj{\calR_R}{\bbN}$, and therefore
    $\Proj{\calR'_L}{\bbN}
    =
    \Proj{\calR}{\bbN}
    =
    \Proj{\calR'_R}{\bbN}$.  Thus
    \[
      \Obs_o^{\beta_{1-p}}
      (\calR,(\calR'_i)_{i\in\supp(\beta_{1-p})})
    \]
    Therefore $\Gamma,\calR\vDash\psi\oplus^o_{1-p}\varphi$.

    \item
    $\inferrule
    {pq<1}
    {\textstyle
      (\varphi \oplus^o_p \psi) \oplus^o_q \vartheta
      \dashv\vdash
      \varphi \oplus^o_{pq}
      (\psi \oplus^o_{(1-p)q/(1-pq)} \vartheta)}$.

    Assume $pq<1$ and write
    $r\triangleq\frac{(1-p)q}{1-pq}$.

    For the left-to-right entailment, fix $\Gamma,\calR$ and assume
    $\Gamma,\calR
    \vDash
    (\varphi\oplus^o_p\psi)\oplus^o_q\vartheta$.
    By unfolding the outer choice, choose
    $\calR_{\varphi\psi},\calR_\vartheta$ such that
    \[
      \calR_{\varphi\psi}\oplus_q\calR_\vartheta\preceq\calR,\qquad
      \Gamma,\calR_{\varphi\psi}\vDash\varphi\oplus^o_p\psi,\qquad
      \Gamma,\calR_\vartheta\vDash\vartheta
    \]
    By unfolding the inner choice, choose $\calR_\varphi,\calR_\psi$ such that
    \[
      \calR_\varphi\oplus_p\calR_\psi\preceq\calR_{\varphi\psi},\qquad
      \Gamma,\calR_\varphi\vDash\varphi,\qquad
      \Gamma,\calR_\psi\vDash\psi
    \]
    We derive the required refinement by regrouping the same three weighted
    branches:
    \[
      \calR_\varphi
      \oplus_{pq}
      (\calR_\psi\oplus_r\calR_\vartheta)
      \preceq
      (\calR_\varphi\oplus_p\calR_\psi)\oplus_q\calR_\vartheta
      \preceq
      \calR_{\varphi\psi}\oplus_q\calR_\vartheta
      \preceq
      \calR
    \]
    The right inner branch is witnessed by
    $\calR_\psi,\calR_\vartheta$.  By the branch obligations above and
    reflexivity of $\preceq$,
    \[
      \Gamma,\calR_\psi\vDash\psi,
      \qquad
      \Gamma,\calR_\vartheta\vDash\vartheta,
      \qquad
      \calR_\psi\oplus_r\calR_\vartheta
      \preceq
      \calR_\psi\oplus_r\calR_\vartheta
    \]
    If $o=\Leak$, it is trivial.  If $o=\Priv$, the private premises give
    \[
      \Proj{\calR_\varphi}{\bbN}
      =
      \Proj{\calR_\psi}{\bbN}
      =
      \Proj{\calR_{\varphi\psi}}{\bbN}
      =
      \Proj{\calR_\vartheta}{\bbN}
      =
      \Proj{\calR}{\bbN}
    \]
    Therefore both the inner and outer private mode side conditions on the
    right-hand side hold.  Thus
    $\Gamma,\calR
    \vDash
    \varphi\oplus^o_{pq}(\psi\oplus^o_r\vartheta)$.

    For the right-to-left entailment, fix $\Gamma,\calR$ and assume
    $\Gamma,\calR
    \vDash
    \varphi\oplus^o_{pq}(\psi\oplus^o_r\vartheta)$.
    Unfolding gives resources $\calR_\varphi$ and $\calR_{\psi\vartheta}$ with
    \[
      \calR_\varphi\oplus_{pq}\calR_{\psi\vartheta}\preceq\calR,\qquad
      \Gamma,\calR_\varphi\vDash\varphi,\qquad
      \Gamma,\calR_{\psi\vartheta}\vDash\psi\oplus^o_r\vartheta
    \]
    Unfolding the right inner choice gives $\calR_\psi,\calR_\vartheta$ with
    \[
      \calR_\psi\oplus_r\calR_\vartheta\preceq\calR_{\psi\vartheta},\qquad
      \Gamma,\calR_\psi\vDash\psi,\qquad
      \Gamma,\calR_\vartheta\vDash\vartheta
    \]
    Since $(1-pq)r=(1-p)q$ and $1-pq-(1-p)q=1-q$, the flattened three-way
    mixture has weights
    \[
      pq,
      \qquad
      (1-p)q,
      \qquad
      1-q
    \]
    Regrouping these weights gives
    \[
      (\calR_\varphi\oplus_p\calR_\psi)\oplus_q\calR_\vartheta
      \preceq
      \calR_\varphi\oplus_{pq}(\calR_\psi\oplus_r\calR_\vartheta)
      \preceq
      \calR_\varphi\oplus_{pq}\calR_{\psi\vartheta}
      \preceq
      \calR
    \]
    The left inner branch is witnessed by $\calR_\varphi,\calR_\psi$.  By the
    branch obligations above and reflexivity of $\preceq$,
    \[
      \Gamma,\calR_\varphi\vDash\varphi,
      \qquad
      \Gamma,\calR_\psi\vDash\psi,
      \qquad
      \calR_\varphi\oplus_p\calR_\psi
      \preceq
      \calR_\varphi\oplus_p\calR_\psi
    \]
    If $o=\Leak$, the mode side condition is trivial.  If $o=\Priv$, the
    private premises imply
    \[
      \Proj{\calR_\varphi}{\bbN}
      =
      \Proj{\calR_\psi}{\bbN}
      =
      \Proj{\calR_\vartheta}{\bbN}
      =
      \Proj{\calR}{\bbN},
    \]
    so the inner and outer private mode side conditions on the left-hand side
    hold. Therefore
    $\Gamma,\calR
    \vDash
    (\varphi\oplus^o_p\psi)\oplus^o_q\vartheta$.

    \item
    $\inferrule
    {X\notin\fv(\varphi) \and \convex{\varphi}}
    {\textstyle
      \bignd_{X\in E}\varphi \vdash \varphi}$.

    Assume $X\notin\fv(\varphi)$ and $\convex{\varphi}$.
    Fix $\Gamma,\calR$ and assume $\Gamma,\calR\vDash\bignd_{X\in E}\varphi$.
    Unfolding nondeterministic choice, choose
    $\mu\in\calD(\de{E}_{\LExp}(\Gamma))$ such that
    $\Gamma,\calR\vDash\bigoplus^\Leak_{X\sim\mu}\varphi$.  Unfolding this
    leaky probabilistic choice, choose witnesses $(\calR_v)_{v\in\supp(\mu)}$
    such that
    \[
      \bigoplus_{v\sim\mu}\calR_v
      \preceq
      \calR,
      \qquad
      \Gamma[X\mapsto v],\calR_v
      \vDash
      \varphi
      \quad(v\in\supp(\mu))
    \]
    Since $X\notin\fv(\varphi)$,
    $\Gamma,\calR_v\vDash\varphi$ for every $v\in\supp(\mu)$.  By
    $\convex{\varphi}$,
    $\Gamma,\bigoplus_{v\sim\mu}\calR_v\vDash\varphi$.  Since
    $\bigoplus_{v\sim\mu}\calR_v\preceq\calR$,
    \Cref{lem:index-satisfaction-monotonicity} gives
    $\Gamma,\calR\vDash\varphi$.

    \item
    $\inferrule
    {}
    {\textstyle
      \varphi[E/X] \vdash \bignd_{X\in\{E\}}\varphi}$.

    Fix $\Gamma,\calR$ and assume $\Gamma,\calR\vDash\varphi[E/X]$.  Put
    $e\triangleq\de{E}_{\LExp}(\Gamma)$.  By
    \Cref{lem:index-assertion-logical-substitution} we get
    $\Gamma[X\mapsto e],\calR\vDash\varphi$.
    In $\bignd_{X\in\{E\}}\varphi$, the support expression $\{E\}$ has exactly
    one choice under $\Gamma$, namely $e$.  Hence the point mass
    $\delta_e\in\calD(\{e\})$ witnesses the nondeterministic choice, and
    $\supp(\delta_e)=\{e\}$.  Define the branch family on this support by
    $\calR_e\triangleq\calR$.  Then
    \[
      \bigoplus_{v\sim\delta_e}\calR_v
      =
      \bigoplus_{v\sim\delta_e}\calR
      \cong
      \calR
      \preceq
      \calR,
      \qquad
      \Gamma[X\mapsto e],\calR_e\vDash\varphi
    \]
    The observation side condition
    $\Obs_\Leak^{\delta_e}(\calR,(\calR_v)_{v\in\supp(\delta_e)})$ is
    $\top$ by definition.  Therefore the satisfaction clause gives
    $\Gamma,\calR\vDash\bigoplus^\Leak_{X\sim\delta_e}\varphi$, and so
    $\Gamma,\calR\vDash\bignd_{X\in\{E\}}\varphi$.

    \item
    $\inferrule
    {}
    {\textstyle
      \bignd_{X\in E}\sure{x\mapsto X}
      \vdash
      \sure{x\mapsto E}}$.

    Fix $\Gamma,\calR$ and assume
    $\Gamma,\calR\vDash\bignd_{X\in E}\sure{x\mapsto X}$. Let
    $A\triangleq\de{E}_{\LExp}(\Gamma)$.  Unfolding nondeterministic choice,
    choose $\mu\in\calD(A)$ such that
    $\Gamma,\calR\vDash\bigoplus^\Leak_{X\sim\mu}\sure{x\mapsto X}$.
    Unfolding this probabilistic choice, choose
    $(\calR_v)_{v\in\supp(\mu)}$ such that, for
    $\calR_\oplus\triangleq\bigoplus_{v\sim\mu}\calR_v$,
    \[
      \calR_\oplus\preceq\calR,
      \qquad
      \Gamma[X\mapsto v],\calR_v\vDash\sure{x\mapsto X}
      \quad(v\in\supp(\mu))
    \]
    We first show
    $\Gamma,\calR_\oplus\vDash\sure{x\mapsto E}$.  For each
    $v\in\supp(\mu)$, define
    $A_v\triangleq
    \memfn_{\calR_v}^{-1}
    \left(\sem{x\mapsto X}_{\Gamma[X\mapsto v]}^{V_{\calR_v}}\right)$
    and
    $B_v\triangleq
    \memfn_{\calR_v}^{-1}
    \left(\sem{x\mapsto E}_\Gamma^{V_{\calR_v}}\right)$.
    The branch satisfaction gives $A_v\in\calF_{\calR_v}$ and
    $\mu_{\calR_v}(A_v)=1$.  If $\omega\in A_v$, then
    \[
      \de{x}_{\Exp}(\memfn_{\calR_v}(\omega))
      =
      \de{X}_{\LExp}(\Gamma[X\mapsto v])
      =
      v
      \in
      A
      =
      \de{E}_{\LExp}(\Gamma)
    \]
    By the support-valued reading of $x\mapsto E$, this means
    $\omega\in B_v$.  Hence $A_v\subseteq B_v$.  Since $\calR_v$ is complete
    and $B_v^c\subseteq A_v^c$ is contained in a null event,
    $B_v\in\calF_{\calR_v}$ and $\mu_{\calR_v}(B_v)=1$.

    Let
    $B_\oplus\triangleq
    \memfn_{\calR_\oplus}^{-1}
    (\sem{x\mapsto E}_\Gamma^{V_{\calR_\oplus}})$.
    The direct sum inherits memory observations branchwise, so
    $B_\oplus\cap\Omega_{\calR_v}=B_v$ for every $v\in\supp(\mu)$.  Thus
    $B_\oplus\in\calF_{\calR_\oplus}$ by the direct-sum event-space
    definition, and
    \[
      \mu_{\calR_\oplus}(B_\oplus)
      =
      \sum_{v\in\supp(\mu)}
      \mu(v)\cdot\mu_{\calR_v}(B_v)
      =
      1
    \]
    Hence $\Gamma,\calR_\oplus\vDash\sure{x\mapsto E}$.  Since
    $\calR_\oplus\preceq\calR$, \Cref{lem:index-satisfaction-monotonicity}
    gives $\Gamma,\calR\vDash\sure{x\mapsto E}$.

    \item
    $\inferrule
    {E\subseteq E'}
    {\textstyle
      \bignd_{X\in E}\varphi \vdash \bignd_{X\in E'}\varphi}$.

    The side condition means
    $\de{E}_{\LExp}(\Gamma)\subseteq\de{E'}_{\LExp}(\Gamma)$ for every
    $\Gamma$.  Fix $\Gamma,\calR$ and assume
    $\Gamma,\calR\vDash\bignd_{X\in E}\varphi$.  Unfolding nondeterministic
    choice, choose $\mu\in\calD(A)$, where
    $A\triangleq\de{E}_{\LExp}(\Gamma)$, such that
    $\Gamma,\calR\vDash\bigoplus^\Leak_{X\sim\mu}\varphi$.

    Put $A'\triangleq\de{E'}_{\LExp}(\Gamma)$. By the premise we get that
    $A\subseteq A'$.  Define the zero extension $\mu'\in\calD(A')$ by
    \[
      \mu'(v)
      \triangleq
      \begin{cases}
        \mu(v), & v\in A,\\
        0, & v\in A'\setminus A .
      \end{cases}
    \]
    Then $\supp(\mu')=\supp(\mu)$.  Unfold
    $\Gamma,\calR\vDash\bigoplus^\Leak_{X\sim\mu}\varphi$ and choose branch
    witnesses $(\calR_v)_{v\in\supp(\mu)}$ such that
    \[
      \bigoplus_{v\sim\mu}\calR_v\preceq\calR,
      \qquad
      \Gamma[X\mapsto v],\calR_v\vDash\varphi
      \quad(v\in\supp(\mu))
    \]
    Since $\supp(\mu')=\supp(\mu)$ and $\mu'$ agrees with $\mu$ on this
    support, we get that
    $\bigoplus_{v\sim\mu'}\calR_v
    =
    \bigoplus_{v\sim\mu}\calR_v
    \preceq
    \calR$.
    The branch satisfaction obligations are the same, and
    $\Obs_\Leak^{\mu'}(\calR,(\calR_v)_{v\in\supp(\mu')})=\top$ by
    definition.  Hence
    $\Gamma,\calR\vDash\bigoplus^\Leak_{X\sim\mu'}\varphi$ with
    $\mu'\in\calD(A')$, so
    $\Gamma,\calR\vDash\bignd_{X\in E'}\varphi$.

    \item
    $\inferrule
    {Y\notin\fv(\varphi)}
    {\textstyle
      \bigoplus_{X\sim d(E)}^o\varphi
      \dashv\vdash
      \bigoplus_{Y\sim d(E)}^o\varphi[Y/X]}$.

    Fix $\Gamma,\calR$ and put
    $\mu\triangleq d(\de{E}_{\LExp}(\Gamma))$.  Unfolding probabilistic choice,
    first assume $\Gamma,\calR\vDash\bigoplus_{X\sim d(E)}^o\varphi$.  Choose
    witnesses $(\calR_v)_{v\in\supp(\mu)}$ such that
    \[
      \bigoplus_{v\sim\mu}\calR_v\preceq\calR,
      \qquad
      \Gamma[X\mapsto v],\calR_v\vDash\varphi
      \quad(v\in\supp(\mu)),
    \]
    and $\Obs_o^\mu(\calR,(\calR_v)_{v\in\supp(\mu)})$.  Since
    $Y\notin\fv(\varphi)$, \Cref{lem:index-assertion-logical-substitution}
    gives
    \[
      \Gamma[Y\mapsto v],\calR_v\vDash\varphi[Y/X]
      \qquad(v\in\supp(\mu))
    \]
    The branch refinement and the mode side condition are unchanged, so the
    same witnesses prove
    $\Gamma,\calR\vDash\bigoplus_{Y\sim d(E)}^o\varphi[Y/X]$.

    Conversely, assume
    $\Gamma,\calR\vDash\bigoplus_{Y\sim d(E)}^o\varphi[Y/X]$.  Since the
    distribution expression $d(E)$ is evaluated in the outer environment in
    both assertions, the branch distribution is again $\mu$.  Choose witnesses
    $(\calR_v)_{v\in\supp(\mu)}$ such that
    \[
      \bigoplus_{v\sim\mu}\calR_v\preceq\calR,
      \qquad
      \Gamma[Y\mapsto v],\calR_v\vDash\varphi[Y/X]
      \quad(v\in\supp(\mu)),
    \]
    and $\Obs_o^\mu(\calR,(\calR_v)_{v\in\supp(\mu)})$.  The reverse
    direction of \Cref{lem:index-assertion-logical-substitution} gives
    \[
      \Gamma[X\mapsto v],\calR_v\vDash\varphi
      \qquad(v\in\supp(\mu))
    \]
    Therefore the same branch family satisfies the resource refinement, branch
    satisfaction, and mode side condition for
    $\Gamma,\calR\vDash\bigoplus_{X\sim d(E)}^o\varphi$.

    \item
    $\inferrule
    {}
    {\textstyle
      \andop{\varphi}{\psi} \dashv\vdash \andop{\psi}{\varphi}}$.

    We prove $\andop{\varphi}{\psi}\vdash\andop{\psi}{\varphi}$; the converse
    is symmetric.  Fix $\Gamma,\calR$ and assume
    $\Gamma,\calR\vDash\andop{\varphi}{\psi}$.  Choose $p\in[0,1]$ such that
    $\Gamma,\calR\vDash\varphi\oplus_p\psi$.  By the probabilistic
    commutativity case above,
    $\Gamma,\calR\vDash\psi\oplus_{1-p}\varphi$.  Folding the satisfaction
    clause for $\&$ gives $\Gamma,\calR\vDash\andop{\psi}{\varphi}$.

    \item
    $\inferrule
    {}
    {\textstyle
      \andop{\varphi}{(\andop{\psi}{\vartheta})}
      \dashv\vdash
      \andop{(\andop{\varphi}{\psi})}{\vartheta}}$.

    For the left-to-right entailment, fix $\Gamma,\calR$ and assume
    $\Gamma,\calR\vDash\andop{\varphi}{(\andop{\psi}{\vartheta})}$.  Unfolding
    the outer $\&$, choose $a\in[0,1]$ such that
    $\Gamma,\calR\vDash\varphi\oplus_a(\andop{\psi}{\vartheta})$.
    Unfolding this binary probabilistic assertion and then the right branch
    $\&$, choose $b\in[0,1]$ such that the same witnesses fold to
    $\Gamma,\calR\vDash\varphi\oplus_a(\psi\oplus_b\vartheta)$.

    Let $q\triangleq a+(1-a)b$.

    If $q=0$, then $a=0$ and $b=0$.  Hence
    \begin{align}
      \Gamma,\calR
      \vDash
      \varphi\oplus_a(\psi\oplus_b\vartheta)
      &\implies
        \Gamma,\calR\vDash\vartheta
        \tag{zero endpoint twice}\\
      &\implies
        \Gamma,\calR\vDash\andop{(\andop{\varphi}{\psi})}{\vartheta} .
        \tag{satisfaction clause for $\&$}
    \end{align}

    If $a=1$, then
    \begin{align}
      \Gamma,\calR
      \vDash
      \varphi\oplus_a(\psi\oplus_b\vartheta)
      &\implies
        \Gamma,\calR\vDash\varphi
        \tag{one endpoint}\\
      &\implies
        \Gamma,\calR\vDash\andop{(\andop{\varphi}{\psi})}{\vartheta} .
        \tag{satisfaction clause for $\&$}
    \end{align}

    It remains to consider the case $q\ne0$ and $a\ne1$.  Define
    $p\triangleq\frac{a}{q}$.
    Then $pq=a$, and therefore $pq<1$ because $a\ne1$.  Moreover,
    $1-pq=1-a$, and
    \[
      (1-p)q
      =
      q-pq
      =
      q-a
      =
      a+(1-a)b-a
      =
      (1-a)b
    \]
    Since $a\ne1$, we have $1-pq=1-a\ne0$, so
    \[
      \frac{(1-p)q}{1-pq}
      =
      \frac{(1-a)b}{1-a}
      =
      b
    \]
    By the probabilistic associativity case above,
    $\Gamma,\calR\vDash
    (\varphi\oplus_p\psi)\oplus_q\vartheta$.
    Folding the inner and outer satisfaction clauses for $\&$ gives
    $\Gamma,\calR\vDash\andop{(\andop{\varphi}{\psi})}{\vartheta}$.

    For the right-to-left entailment, fix $\Gamma,\calR$ and assume
    $\Gamma,\calR\vDash\andop{(\andop{\varphi}{\psi})}{\vartheta}$.
    Unfolding the outer $\&$, choose $q\in[0,1]$ such that
    $\Gamma,\calR\vDash(\andop{\varphi}{\psi})\oplus_q\vartheta$.
    Unfolding this binary probabilistic assertion and then the left branch
    $\&$, choose $p\in[0,1]$ such that the same witnesses fold to
    $\Gamma,\calR\vDash(\varphi\oplus_p\psi)\oplus_q\vartheta$.

    If $pq=1$, then $p=1$ and $q=1$.  Hence
    $\Gamma,\calR\vDash
    (\varphi\oplus_1\psi)\oplus_1\vartheta$.
    By two applications of the one endpoint case,
    $\Gamma,\calR\vDash\varphi$.  Applying the one endpoint case in the reverse
    direction gives
    $\Gamma,\calR\vDash\varphi\oplus_1(\andop{\psi}{\vartheta})$.
    Folding the outer satisfaction clause for $\&$ yields
    $\Gamma,\calR\vDash\andop{\varphi}{(\andop{\psi}{\vartheta})}$.

    Otherwise let
    $r\triangleq\frac{(1-p)q}{1-pq}$.  By the probabilistic associativity case
    above,
    $\Gamma,\calR\vDash
    \varphi\oplus_{pq}(\psi\oplus_r\vartheta)$.
    Folding the inner and outer satisfaction clauses for $\&$ gives
    $\Gamma,\calR\vDash\andop{\varphi}{(\andop{\psi}{\vartheta})}$.

    \item
    $\inferrule
    {}
    {\textstyle
      \andop{(\varphi \oplus^o_p \psi)}{\vartheta}
      \dashv\vdash
      (\andop{\varphi}{\vartheta}) \oplus^o_p
      (\andop{\psi}{\vartheta})}$.

    Unfold the satisfaction clause for $\andop{-}{-}$.  The two entailments
    are obtained by the same direct-sum regrouping.  In the left-to-right
    direction, if $q$ witnesses the outer $\&$, then after unfolding the
    left branch $\varphi\oplus^o_p\psi$ we have resources
    $\calR_\varphi,\calR_\psi,\calR_\vartheta$ with
    \[
      (\calR_\varphi\oplus_p\calR_\psi)\oplus_q\calR_\vartheta
      \preceq
      \calR
    \]
    By associativity of weighted direct sums, coalescing the two copies of the
    $\vartheta$ branch, this gives
    \[
      (\calR_\varphi\oplus_q\calR_\vartheta)
      \oplus_p
      (\calR_\psi\oplus_q\calR_\vartheta)
      \preceq
      \calR
    \]
    The left branch satisfies $\andop{\varphi}{\vartheta}$, witnessed by
    $\calR_\varphi,\calR_\vartheta$ and weight $q$; the right branch satisfies
    $\andop{\psi}{\vartheta}$, witnessed by
    $\calR_\psi,\calR_\vartheta$ and the same weight $q$.  Thus
    $\Gamma,\calR\vDash
    (\andop{\varphi}{\vartheta})\oplus^o_p
    (\andop{\psi}{\vartheta})$.

    Conversely, unfold the right-hand side, then unfold the two $\&$ branch
    witnesses.  Regrouping the resulting weighted direct sum in the reverse
    direction gives a witness for an outer $\&$ whose left branch is
    $\varphi\oplus^o_p\psi$ and whose right branch is $\vartheta$.  Hence
    $\Gamma,\calR\vDash
    \andop{(\varphi\oplus^o_p\psi)}{\vartheta}$.

    The observation side condition is vacuous for $o=\Leak$.  For $o=\Priv$,
    the regrouping preserves it because the branch count marginals required by
    the unfolded private choice are unchanged by associativity of weighted
    direct sums.

    \item
    $\inferrule
    {}
    {\textstyle
      \varphi\oplus_0\psi \dashv\vdash \psi}$.

    Let $\beta_0\in\calD(\{L,R\})$ be given by
    $
    \beta_0(i) \triangleq \begin{cases}
      0 & i = L\\
      1 & i = R
    \end{cases}
    $.
    Then $\supp(\beta_0)=\{R\}$, and $\varphi\oplus_0\psi$ is the binary
    instance whose $L$-branch is $\varphi$ and whose $R$-branch is $\psi$.

    We prove both entailments.
    \begin{itemize}
      \item Assume $\Gamma,\calR\vDash\varphi\oplus_0\psi$.  Unfolding
      satisfaction gives a family
      $(\calR_i)_{i\in\supp(\beta_0)}$, equivalently the single resource
      $\calR_R$, such that
      \[
        \bigoplus_{i\sim\beta_0}\calR_i\preceq\calR,
        \qquad
        \Gamma,\calR_R\vDash\psi
      \]
      Since $\supp(\beta_0)=\{R\}$, the definition of singleton direct sum gives
      $\calR_R
      \cong
      \bigoplus_{i\sim\beta_0}\calR_i
      \preceq
      \calR$.
      By \Cref{lem:index-satisfaction-monotonicity},
      $\Gamma,\calR\vDash\psi$.

      \item Assume $\Gamma,\calR\vDash\psi$.  Take the only supported branch
      witness to be $\calR_R\triangleq\calR$.  Then
      $\bigoplus_{i\sim\beta_0}\calR_i
      \cong
      \calR_R
      =
      \calR$
      and the only branch obligation is exactly
      $\Gamma,\calR_R\vDash\psi$.  Hence
      $\Gamma,\calR\vDash\varphi\oplus_0\psi$.
    \end{itemize}

    \item
    $\inferrule
    {}
    {\textstyle
      \varphi\oplus_1\psi \dashv\vdash \varphi}$.

    This is symmetric to the zero endpoint.  Replace $\beta_0,R,\psi$ by
    $\beta_1,L,\varphi$, where
    $\beta_1(i)\triangleq
    \begin{cases}
      1 & i = L\\
      0 & i = R
    \end{cases}$,
    so $\supp(\beta_1)=\{L\}$.
  \end{enumerate}
\end{proof}

\begin{lemma}[Entailment laws for separating conjunction]
  \label{lem:index-separation-entailment-laws}
  \label{lem:index-separating-assertion-monotonicity}
  \label{lem:index-separating-assertion-commutativity}
  \label{lem:index-separating-assertion-associativity}
  \label{lem:index-sure-separating-conjunction}
  \label{lem:index-sure-separating-unit-weakening}
  \label{lem:index-separating-assertion-weakening}
  The entailment rules in \Cref{fig:separation-entailment-laws} are valid.
\end{lemma}
\begin{proof}
  We prove the rules in the order shown in
  \Cref{fig:separation-entailment-laws}.
  \begin{enumerate}[label=\textup{(\arabic*)}]
    \item
    $\inferrule
    {\varphi \vdash \varphi' \and \psi \vdash \psi'}
    {\textstyle
      \varphi \sep \psi \vdash \varphi' \sep \psi'}$.

    Fix $\Gamma,\calR$ and assume $\Gamma,\calR\vDash\varphi\sep\psi$.
    Choose witnesses $\calR_\varphi,\calR_\psi$ such that
    \[
      \calR_\varphi\otimes\calR_\psi\preceq\calR,
      \qquad
      \Gamma,\calR_\varphi\vDash\varphi,
      \qquad
      \Gamma,\calR_\psi\vDash\psi
    \]
    The entailment premises give
    $\Gamma,\calR_\varphi\vDash\varphi'$ and
    $\Gamma,\calR_\psi\vDash\psi'$, so the same resources witness
    $\Gamma,\calR\vDash\varphi'\sep\psi'$.

    \item
    $\inferrule
    {}
    {\textstyle
      \varphi \sep \psi \dashv\vdash \psi \sep \varphi}$.

    This follows immediately from the commutativity of $\otimes$ and
    transitivity of $\preceq$

    \item
    $\inferrule
    {}
    {\textstyle
      (\varphi \sep \psi) \sep \theta
      \dashv\vdash
      \varphi \sep (\psi \sep \theta)}$.

    This follows immediately by the associativity of $\otimes$ and transitivity
    of $\preceq$.

    \item
    $\inferrule
    {}
    {\textstyle
      \sure{P * Q} \dashv\vdash \sure{P} \sep \sure{Q}}$.

    First assume $\Gamma,\calR\vDash\sure{P}\sep\sure{Q}$.  Choose witnesses
    $\calR_P,\calR_Q$ such that
    \[
      \calR_P\otimes\calR_Q\preceq\calR,
      \qquad
      \Gamma,\calR_P\vDash\sure{P},
      \qquad
      \Gamma,\calR_Q\vDash\sure{Q}
    \]
    Let $g$ be the product pairing for $\calR_P\otimes\calR_Q$, and define
    $E_P\triangleq
    \memfn_{\calR_P}^{-1}(\sem{P}_\Gamma^{V_{\calR_P}})$,
    $E_Q\triangleq
    \memfn_{\calR_Q}^{-1}(\sem{Q}_\Gamma^{V_{\calR_Q}})$, and
    $E_{P*Q}\triangleq
    \memfn_{\calR_P\otimes\calR_Q}^{-1}
    (\sem{P*Q}_\Gamma^{V_{\calR_P}\cup V_{\calR_Q}})$.  For
    $\omega\in\Omega_{\calR_P\otimes\calR_Q}$ with
    $g^{-1}(\omega)=(\omega_P,\omega_Q)$,
    \begin{align*}
      \omega\in E_{P*Q}
      &\Longleftrightarrow
        \memfn_{\calR_P}(\omega_P)\uplus\memfn_{\calR_Q}(\omega_Q)
        \in
        \sem{P*Q}_\Gamma^{V_{\calR_P}\cup V_{\calR_Q}}\\
      &\Longleftrightarrow
        \memfn_{\calR_P}(\omega_P)\in\sem{P}_\Gamma^{V_{\calR_P}}
        \land
        \memfn_{\calR_Q}(\omega_Q)\in\sem{Q}_\Gamma^{V_{\calR_Q}}\\
      &\Longleftrightarrow
        \omega\in g(E_P\times E_Q)
    \end{align*}
    Since $\mu_{\calR_P}(E_P)=\mu_{\calR_Q}(E_Q)=1$, it follows that
    $E_{P*Q}$ is measurable in $\calR_P\otimes\calR_Q$ and has measure $1$.
    Hence $\Gamma,\calR_P\otimes\calR_Q\vDash\sure{P*Q}$, and
    \Cref{lem:index-satisfaction-monotonicity} gives
    $\Gamma,\calR\vDash\sure{P*Q}$.

    Conversely, assume $\Gamma,\calR\vDash\sure{P*Q}$ and let
    $E_{P*Q}\triangleq
    \memfn_\calR^{-1}(\sem{P*Q}_\Gamma^{V_\calR})$.  This event is
    measurable and has measure $1$.  For $\omega\in E_{P*Q}$, choose
    separating witnesses $\sigma_P(\omega),\sigma_Q(\omega)$ for
    $\memfn_\calR(\omega)\in\sem{P*Q}_\Gamma^{V_\calR}$; outside $E_{P*Q}$,
    set both witnesses to $\emptyset$.  Put
    $V_P\triangleq\Var(P)\cap V_\calR$ and
    $V_Q\triangleq\Var(Q)\cap V_\calR$.  By the footprint and locality
    properties of pure assertions, the nonempty full-measure event $E_{P*Q}$
    forces $V_P\cap V_Q=\emptyset$, and the restrictions of the chosen
    witnesses to $V_P$ and $V_Q$ still satisfy $P$ and $Q$ on $E_{P*Q}$.

    Let $\calP_E$ be the coarsening of $\calP_\calR$ to the completion of the
    $\sigma$-algebra generated by $E_{P*Q}$, as in
    \Cref{lem:index-sure-coarsening}, and define
    \[
      \calR_P\triangleq
      \left\langle
        \calP_E, V_P,
        \omega\mapsto\ForgetMem{\sigma_P(\omega)}{V_P},
        \omega\mapsto 0
      \right\rangle,
      \qquad
      \calR_Q\triangleq
      \left\langle
        \calP_\calR, V_Q,
        \omega\mapsto\ForgetMem{\sigma_Q(\omega)}{V_Q},
        \cntfn_\calR
      \right\rangle
    \]
    The preceding paragraph gives full-measure inclusions into the
    $P$- and $Q$-satisfaction events for $\calR_P$ and $\calR_Q$; completeness
    of $\calP_E$ and $\calP_\calR$ therefore yields
    $\Gamma,\calR_P\vDash\sure{P}$ and
    $\Gamma,\calR_Q\vDash\sure{Q}$.

    Finally, let $g$ be the product pairing for
    $\calR_P\otimes\calR_Q$ and set
    $h(\omega)\triangleq g(\omega,\omega)$.  The coarsening makes this
    diagonal map measurable and measure-preserving.  On $E_{P*Q}$, its memory
    is below
    $\sigma_P(\omega)\uplus\sigma_Q(\omega)\sqsubseteq\memfn_\calR(\omega)$;
    off $E_{P*Q}$ it is empty, and its count is
    $0+\cntfn_\calR(\omega)=\cntfn_\calR(\omega)$.  Thus
    $h$ witnesses $\calR_P\otimes\calR_Q\preceq\calR$, so
    $\calR_P,\calR_Q$ witness
    $\Gamma,\calR\vDash\sure{P}\sep\sure{Q}$.

    \item
    $\inferrule
    {}
    {\textstyle
      \sure{P} \vdash \sure{P}\sep\top}$.

    Fix $\Gamma,\calR$ and assume $\Gamma,\calR\vDash\sure{P}$.  Since
    $\Gamma,\UnitP\vDash\top$ and $\calR\otimes\UnitP\cong\calR$, the
    resources $\calR,\UnitP$ witness
    $\Gamma,\calR\vDash\sure{P}\sep\top$.

    \item
    $\inferrule
    {}
    {\textstyle
      \varphi \sep \psi \vdash \varphi}$.

    Fix $\Gamma,\calR$ and assume $\Gamma,\calR\vDash\varphi\sep\psi$.
    Choose witnesses $\calR_\varphi,\calR_\psi$ such that
    \[
      \calR_\varphi\otimes\calR_\psi\preceq\calR,
      \qquad
      \Gamma,\calR_\varphi\vDash\varphi,
      \qquad
      \Gamma,\calR_\psi\vDash\psi
    \]
    By \Cref{lem:index-assertion-frame-preservation},
    $\Gamma,\calR_\varphi\otimes\calR_\psi\vDash\varphi$.  Then
    \Cref{lem:index-satisfaction-monotonicity} gives
    $\Gamma,\calR\vDash\varphi$.
  \end{enumerate}
\end{proof}

\begin{lemma}[Outcome conjunction entailment laws]
  \label{lem:index-oplus-entailment-laws}
  \label{lem:index-probabilistic-entails-nondeterministic}
  \label{lem:index-private-entails-leaky-choice}
  \label{lem:index-probabilistic-choice-monotonicity}
  \label{lem:index-probabilistic-choice-separating-right}
  \label{lem:index-probabilistic-choice-separating-left}
  \label{lem:index-probabilistic-choice-mode-commutation}
  \label{lem:index-probabilistic-choice-idempotence}
  The entailment rules in \Cref{fig:oplus-entailment-laws} are valid.
\end{lemma}
\begin{proof}
  We prove the rules in the order shown in \Cref{fig:oplus-entailment-laws}.
  \begin{enumerate}[label=\textup{(\arabic*)}]
    \item
    $\inferrule
    {}
    {\textstyle
      \bigoplus_{X \sim d(E)}^o \varphi
      \vdash
      \bignd_{X \in \supp(d(E))}\varphi}$.

    Fix $\Gamma,\calR$ and suppose
    $\Gamma,\calR\vDash\bigoplus_{X\sim d(E)}^o\varphi$.  Put
    $\mu\triangleq d(\de{E}_{\LExp}(\Gamma))$.  Unfolding satisfaction, choose
    witnesses $(\calR_v)_{v\in\supp(\mu)}$ such that
    \[
      \bigoplus_{v\sim\mu}\calR_v\preceq\calR,
      \qquad
      \Gamma[X\mapsto v],\calR_v\vDash\varphi
      \quad(v\in\supp(\mu))
    \]
    Since the leaky mode has no side condition, the same witnesses give
    $\Gamma,\calR\vDash\bigoplus^\Leak_{X\sim\mu}\varphi$.  Taking this
    $\mu\in\calD(\supp(\mu))$ in the satisfaction clause for nondeterministic
    choice gives
    $\Gamma,\calR\vDash\bignd_{X\in\supp(d(E))}\varphi$.

    \item
    $\inferrule
    {}
    {\textstyle
      \bigoplus_{X \sim d(E)}^\Priv \varphi
      \vdash
      \bigoplus_{X \sim d(E)}^\Leak \varphi}$.

    This follows immediately by definition, since $  \Obs_o^\mu(\calR,(\calR_v)_{v\in\supp(\mu)})$ is trivial for $\Obs_\Leak$.

    \item
    $\inferrule
    {\varphi \vdash \psi}
    {\textstyle
      \bigoplus_{X \sim d(E)}^o \varphi
      \vdash
      \bigoplus_{X \sim d(E)}^o \psi}$.

    Fix $\Gamma,\calR$ and suppose
    $\Gamma,\calR\vDash\bigoplus^o_{X\sim d(E)}\varphi$.  Put
    $\mu\triangleq d(\de{E}_{\LExp}(\Gamma))$.  Unfolding satisfaction, choose
    witnesses $(\calR_v)_{v\in\supp(\mu)}$ such that
    \[
      \bigoplus_{v\sim\mu}\calR_v\preceq\calR,
      \qquad
      \Gamma[X\mapsto v],\calR_v\vDash\varphi
      \quad(v\in\supp(\mu)),
    \]
    and $\Obs_o^\mu(\calR,(\calR_v)_{v\in\supp(\mu)})$.  By
    $\varphi\vdash\psi$, we get that
      $\Gamma[X\mapsto v],\calR_v\vDash\psi$ for all $v \in \supp(\mu)$.
    Hence the same witnesses and the same mode side condition prove
    $\Gamma,\calR\vDash\bigoplus^o_{X\sim d(E)}\psi$.

    \item
    $\inferrule
    {X \not\in \fv(\psi)}
    {\textstyle
      (\bigoplus_{X \sim d(E)}^o \varphi) \sep \psi
      \vdash
      \bigoplus_{X \sim d(E)}^o (\varphi \sep \psi)}$.

    Fix $\Gamma,\calR$ and suppose
    $\Gamma,\calR\vDash(\bigoplus_{X\sim d(E)}^o\varphi)\sep\psi$.  Put
    $\mu\triangleq d(\de{E}_{\LExp}(\Gamma))$.  Unfolding separating
    conjunction, choose $\calR_\oplus,\calR_\psi$ such that
    \[
      \calR_\oplus\otimes\calR_\psi\preceq\calR,
      \qquad
      \Gamma,\calR_\oplus\vDash\bigoplus_{X\sim d(E)}^o\varphi,
      \qquad
      \Gamma,\calR_\psi\vDash\psi
    \]
    Unfolding probabilistic choice, choose witnesses
    $(\calR_v)_{v\in\supp(\mu)}$ such that
    \[
      \bigoplus_{v\sim\mu}\calR_v\preceq\calR_\oplus,
      \qquad
      \Gamma[X\mapsto v],\calR_v\vDash\varphi
      \quad(v\in\supp(\mu)),
    \]
    and $\Obs_o^\mu(\calR_\oplus,(\calR_v)_{v\in\supp(\mu)})$.
    Since $X\notin\fv(\psi)$,
    $\Gamma[X\mapsto v],\calR_\psi\vDash\psi$ for all $v \in \supp(\mu)$.
    Hence, for every $v\in\supp(\mu)$,
    $\calR_v\otimes\calR_\psi$ satisfies
    $\varphi\sep\psi$ under $\Gamma[X\mapsto v]$, witnessed by
    $\calR_v$ and $\calR_\psi$.  The branch refinement is
    \[
      \bigoplus_{v\sim\mu}(\calR_v\otimes\calR_\psi)
      \cong
      \left(\bigoplus_{v\sim\mu}\calR_v\right)\otimes\calR_\psi
      \preceq
      \calR_\oplus\otimes\calR_\psi
      \preceq
      \calR
    \]
    If $o=\Leak$, the mode side condition is immediate.  If $o=\Priv$, then
    for every $v\in\supp(\mu)$,
    \begin{align*}
      \Proj{\calR_v\otimes\calR_\psi}{\bbN}
      &=
        (\mathsf{add}_{\bbN})_*
        (\Proj{\calR_v}{\bbN}\otimes\Proj{\calR_\psi}{\bbN})\\
      &=
        (\mathsf{add}_{\bbN})_*
        (\Proj{\calR_\oplus}{\bbN}\otimes\Proj{\calR_\psi}{\bbN})
        \tag{$\Obs_\Priv^\mu(\calR_\oplus,(\calR_v)_v)$}\\
      &=
        \Proj{\calR_\oplus\otimes\calR_\psi}{\bbN}\\
      &=
        \Proj{\calR}{\bbN}
        \tag{$\calR_\oplus\otimes\calR_\psi\preceq\calR$}
    \end{align*}
    by properties displayed in \Cref{fig:resource-algebra-refinement-facts}.
    Thus the witnesses
    $(\calR_v\otimes\calR_\psi)_{v\in\supp(\mu)}$ prove
    $\Gamma,\calR\vDash
    \bigoplus_{X\sim d(E)}^o(\varphi\sep\psi)$.

    \item
    $\inferrule
    {X \notin \fv(\psi) \and \precise{\psi} \and \Stable{\weak}{\psi}}
    {\textstyle
      \bigoplus_{X\sim d(E)}^o(\varphi\sep\psi)
      \vdash
      (\bigoplus_{X\sim d(E)}^o\varphi)\sep\psi}$.

    Fix $\Gamma,\calR$ and suppose
    $\Gamma,\calR\vDash\bigoplus_{X\sim d(E)}^o(\varphi\sep\psi)$.  Put
    $\mu\triangleq d(\de{E}_{\LExp}(\Gamma))$.  Unfolding probabilistic choice,
    choose branch resources $(\calR_v)_{v\in\supp(\mu)}$ such that
    \[
      \bigoplus_{v\sim\mu}\calR_v\preceq\calR,
      \qquad
      \Gamma[X\mapsto v],\calR_v\vDash\varphi\sep\psi
      \quad(v\in\supp(\mu)),
    \]
    and $\Obs_o^\mu(\calR,(\calR_v)_{v\in\supp(\mu)})$.

    For each branch $v$, we want to construct a resource $\calQ_v$ satisfying
    $\varphi$ and a single stable resource $\calQ_\psi^0$ satisfying $\psi$
    such that
    $\calQ_v\otimes\calQ_\psi^0\preceq\calR_v$.  This will let
    $\bigoplus_{v\sim\mu}\calQ_v$ witness the probabilistic choice and
    $\calQ_\psi^0$ witness the separated-out $\psi$ component.
    For each $v\in\supp(\mu)$, unfold separating conjunction.  Since
    $X\notin\fv(\psi)$, choose $\calR_{\varphi,v},\calR_{\psi,v}$ such that
    \[
      \calR_{\varphi,v}\otimes\calR_{\psi,v}\preceq\calR_v,
      \qquad
      \Gamma[X\mapsto v],\calR_{\varphi,v}\vDash\varphi,
      \qquad
      \Gamma,\calR_{\psi,v}\vDash\psi
    \]
    Applying $\Stable{\weak}{\psi}$ to $\calR_{\psi,v}$, choose
    $\calR_{\psi,v}^0,\calR_{\psi,v}^\cntfn$ such that
    \[
      \calR_{\psi,v}^0\otimes\calR_{\psi,v}^\cntfn\preceq\calR_{\psi,v},
      \qquad
      \Gamma,\calR_{\psi,v}^0\vDash\psi,
      \qquad
      \Stable{\weak}{\calR_{\psi,v}^0},
      \qquad
      V_{\calR_{\psi,v}^\cntfn}=\emptyset
    \]
    Let
    $\calQ_v
    \triangleq
    \calR_{\varphi,v}\otimes\calR_{\psi,v}^\cntfn$.
    By \Cref{lem:index-assertion-frame-preservation} we get that
    $\Gamma[X\mapsto v],\calQ_v\vDash\varphi$ for $v \in \supp(\mu)$.
    Since $\mu$ is a probability distribution, $\supp(\mu)$ is nonempty, so the
    displayed branch witnesses show that $\psi$ is satisfiable under $\Gamma$.
    Thus \Cref{lem:index-precise-stable-core} gives a stable core
    $\calQ_\psi^0$ such that
    \[
      \Gamma,\calQ_\psi^0\vDash\psi,
      \qquad
      \Stable{\weak}{\calQ_\psi^0},
    \]
    and
    \begin{equation}
      \forall\calQ.\;
      \Gamma,\calQ\vDash\psi
      \land
      \Stable{\weak}{\calQ}
      \implies
      \calQ_\psi^0\preceq\calQ .
      \label{eq:index-prob-sep-stable-core}
    \end{equation}
    Also, by \eqref{eq:index-prob-sep-stable-core},
    $\calQ_\psi^0\preceq\calR_{\psi,v}^0$ for every $v\in\supp(\mu)$.  Hence
    \[
      \calQ_v\otimes\calQ_\psi^0
      \preceq
      \calR_{\varphi,v}\otimes\calR_{\psi,v}^\cntfn
      \otimes
      \calR_{\psi,v}^0
      \cong
      \calR_{\varphi,v}\otimes
      (\calR_{\psi,v}^0\otimes\calR_{\psi,v}^\cntfn)
      \preceq
      \calR_{\varphi,v}\otimes\calR_{\psi,v}
      \preceq
      \calR_v
    \]

    It remains to verify the side condition on the scheduler observation of
    the probabilistic choice witnessed by the branches
    $(\calQ_v)_{v\in\supp(\mu)}$ with respect to modes.  The leaky case has
    no side condition.  In the private case, we must show that every branch
    $\calQ_v$ has the same count marginal as the whole choice resource.
    Define
    $\calQ_\oplus\triangleq\bigoplus_{v\sim\mu}\calQ_v$.  If $o=\Priv$, then
    $\Stable{\weak}{\calQ_\psi^0}$ gives
    $\Proj{\calQ_\psi^0}{\bbN}=\delta_0$, and therefore
    for every $v\in\supp(\mu)$,
    \begin{align*}
      \Proj{\calQ_v}{\bbN}
      &=
        \Proj{\calQ_v\otimes\calQ_\psi^0}{\bbN}
        \tag{
        $\Proj{\calQ_\psi^0}{\bbN}=\delta_0$}\\
      &=
        \Proj{\calR_v}{\bbN}\\
      &=
        \Proj{\calR}{\bbN}
        \tag{$\Obs_\Priv^\mu(\calR,(\calR_v)_v)$}
    \end{align*}
    by properties in \Cref{fig:resource-algebra-refinement-facts}
    Since $\calQ_\oplus$ is the direct sum of branches with this common count
    marginal, $\Proj{\calQ_\oplus}{\bbN}=\Proj{\calR}{\bbN}$.  Therefore
    $\Obs_o^\mu(\calQ_\oplus,(\calQ_v)_{v\in\supp(\mu)})$ holds: it is trivial
    when $o=\Leak$, and the displayed equality gives it when $o=\Priv$.  Hence
    $\Gamma,\calQ_\oplus\vDash\bigoplus_{X\sim d(E)}^o\varphi$.
    Finally,
    \[
      \calQ_\oplus\otimes\calQ_\psi^0
      \cong
      \bigoplus_{v\sim\mu}(\calQ_v\otimes\calQ_\psi^0)
      \preceq
      \bigoplus_{v\sim\mu}\calR_v
      \preceq
      \calR
    \]
    Thus $\calQ_\oplus$ and $\calQ_\psi^0$ witness
    $\Gamma,\calR\vDash(\bigoplus_{X\sim d(E)}^o\varphi)\sep\psi$.

    \item
    $\inferrule
    {o_1\sqsubseteq o_2 \and X\neq Y \and
      X\notin\fv(E_Y) \and Y\notin\fv(E_X)}
    {\textstyle
      \bigoplus_{X\sim d_X(E_X)}^{o_1}
        \left(\bigoplus_{Y\sim d_Y(E_Y)}^{o_2}\varphi\right)
      \vdash
      \bigoplus_{Y\sim d_Y(E_Y)}^{o_2}
        \left(\bigoplus_{X\sim d_X(E_X)}^{o_1}\varphi\right)}$.

    Fix $\Gamma,\calR$ and assume
    $\Gamma,\calR\vDash
    \bigoplus_{X\sim d_X(E_X)}^{o_1}
    \left(\bigoplus_{Y\sim d_Y(E_Y)}^{o_2}\varphi\right)$.
    Let
    $\mu\triangleq d_X(\de{E_X}_{\LExp}(\Gamma))$ and
    $\nu\triangleq d_Y(\de{E_Y}_{\LExp}(\Gamma))$.
    Since $X\notin\fv(E_Y)$ and $Y\notin\fv(E_X)$,
    \[
      d_Y(\de{E_Y}_{\LExp}(\Gamma[X\mapsto x]))=\nu,
      \qquad
      d_X(\de{E_X}_{\LExp}(\Gamma[Y\mapsto y]))=\mu
    \]
    for all $x\in\supp(\mu)$ and $y\in\supp(\nu)$.

    Unfolding the two choices gives witnesses
    $(\calR_x)_{x\in\supp(\mu)}$ and
    $(\calR_{x,y})_{y\in\supp(\nu)}$ such that
    \[
      \bigoplus_{x\sim\mu}\calR_x\preceq\calR,
      \qquad
      \bigoplus_{y\sim\nu}\calR_{x,y}\preceq\calR_x,
      \qquad
      \Gamma[X\mapsto x][Y\mapsto y],\calR_{x,y}\vDash\varphi,
    \]
    where the display holds for all $x\in\supp(\mu)$ and
    $y\in\supp(\nu)$, together with
    $\Obs_{o_1}^{\mu}(\calR,(\calR_x)_{x\in\supp(\mu)})$ and, for every
    $x\in\supp(\mu)$,
    $\Obs_{o_2}^{\nu}(\calR_x,(\calR_{x,y})_{y\in\supp(\nu)})$.

    For $y\in\supp(\nu)$, define
    $\calQ_y\triangleq\bigoplus_{x\sim\mu}\calR_{x,y}$.  By rectangular
    rearrangement of direct sums,
    \[
      \bigoplus_{y\sim\nu}\calQ_y
      \cong
      \bigoplus_{x\sim\mu}
        \left(\bigoplus_{y\sim\nu}\calR_{x,y}\right)
      \preceq
      \bigoplus_{x\sim\mu}\calR_x
      \preceq
      \calR
    \]

    For fixed $y$, the resources
    $(\calR_{x,y})_{x\in\supp(\mu)}$ witness the inner $o_1$-choice in
    $\calQ_y$, since $X\neq Y$ gives
    \[
      \Gamma[Y\mapsto y][X\mapsto x],\calR_{x,y}\vDash\varphi
      \quad\Longleftrightarrow\quad
      \Gamma[X\mapsto x][Y\mapsto y],\calR_{x,y}\vDash\varphi
    \]
    The $o_1$ side condition is vacuous when $o_1=\Leak$.  If $o_1=\Priv$,
    then $o_2=\Priv$, and the private side conditions give
    \[
      \Proj{\calQ_y}{\bbN}
      =
      \Proj{\calR}{\bbN}
      =
      \Proj{\calR_{x,y}}{\bbN}
    \]
    Hence
    $\Gamma[Y\mapsto y],\calQ_y
    \vDash
    \bigoplus_{X\sim d_X(E_X)}^{o_1}\varphi$ for all $y \in \supp(\nu)$.

    Finally, the outer $o_2$ side condition is vacuous when $o_2=\Leak$.  If
    $o_2=\Priv$, then the inner private side condition gives
    \[
      \Proj{\calQ_y}{\bbN}
      =
      \bigoplus_{x\sim\mu}\Proj{\calR_{x,y}}{\bbN}
      =
      \bigoplus_{x\sim\mu}\Proj{\calR_x}{\bbN}
      =
      \Proj{\calR}{\bbN}
    \]
    Therefore the satisfaction clause yields
    $\Gamma,\calR\vDash
    \bigoplus_{Y\sim d_Y(E_Y)}^{o_2}
    \left(\bigoplus_{X\sim d_X(E_X)}^{o_1}\varphi\right)$.

    \item
    $\inferrule
    {X\notin\fv(\varphi) \and \convex{\varphi}}
    {\textstyle
      \bigoplus_{X\sim d(E)}^o\varphi \vdash \varphi}$.

    Immediate by composing rule~(1),
    $\bigoplus_{X\sim d(E)}^o\varphi\vdash\bignd_{X\in\supp(d(E))}\varphi$,
    with the nondeterministic collapse
    $\bignd_{X\in E}\varphi\vdash\varphi$ of
    \Cref{fig:basic-entailment-laws} (valid since $X\notin\fv(\varphi)$ and
    $\convex{\varphi}$).
  \end{enumerate}
\end{proof}

\noam{These context invariance lemmas are pretty trivial too. I don't think we need to spell them out}

\begin{lemma}[Logical substitution for assertion satisfaction]
  \label{lem:index-assertion-logical-substitution}
  For every assertion $\varphi$ and logical expression $E$,
  \[
    \Gamma,\calR\vDash\varphi[E/X]
    \quad\Longleftrightarrow\quad
    \Gamma[X\mapsto\de{E}_{\LExp}(\Gamma)],\calR\vDash\varphi
  \]
  Moreover, if $Y\notin\fv(\varphi)$, then for every value $v$,
  \[
    \Gamma[Y\mapsto v],\calR\vDash\varphi[Y/X]
    \quad\Longleftrightarrow\quad
    \Gamma[X\mapsto v],\calR\vDash\varphi
  \]
\end{lemma}
\begin{proof}
  We prove the first equivalence by structural induction on $\varphi$.  Write
  $e\triangleq\de{E}_{\LExp}(\Gamma)$ and
  $\Gamma_X\triangleq\Gamma[X\mapsto e]$.  We use the standard substitution
  properties for logical expressions and pure assertions:
  \[
    \de{F[E/X]}_{\LExp}(\Gamma)
    =
    \de{F}_{\LExp}(\Gamma_X),
    \qquad
    \sem{P[E/X]}_\Gamma^V
    =
    \sem{P}_{\Gamma_X}^V
  \]
  The cases for $\top$, $\bot$, conjunction, and disjunction are immediate
  from the induction hypotheses.  The case $\sure{P}$ follows from the pure
  substitution identity, and the case $\varphi_1\sep\varphi_2$ follows from
  the induction hypotheses because the resource witnesses and refinement
  condition do not mention the logical environment.

  It remains to consider the binding assertions.  Let
  $\varphi=\bigoplus^o_{Z\sim d(F)}\psi$ and alpha-convert so that
  $Z\notin\{X\}\cup\fv(E)$.  Then
  \[
    \varphi[E/X]
    =
    \bigoplus^o_{Z\sim d(F[E/X])}\psi[E/X],
    \qquad
    d(\de{F[E/X]}_{\LExp}(\Gamma))
    =
    d(\de{F}_{\LExp}(\Gamma_X))
  \]
  For every branch value $v$, since $Z\notin\fv(E)$ and $Z\ne X$,
  \[
    \de{E}_{\LExp}(\Gamma[Z\mapsto v])=e,
    \qquad
    \Gamma[Z\mapsto v][X\mapsto e]
    =
    \Gamma_X[Z\mapsto v]
  \]
  Hence the induction hypothesis for $\psi$, applied under
  $\Gamma[Z\mapsto v]$, identifies the branch obligations
  \[
    \Gamma[Z\mapsto v],\calQ\vDash\psi[E/X]
    \quad\Longleftrightarrow\quad
    \Gamma_X[Z\mapsto v],\calQ\vDash\psi
  \]
  The direct-sum refinement and observation side condition are unchanged, so
  the same branch resources witness the two probabilistic-choice satisfaction
  clauses.

  The case $\varphi=\bignd_{Z\in F}\psi$ is analogous.  After the same
  alpha-conversion,
  \[
    (\bignd_{Z\in F}\psi)[E/X]
    =
    \bignd_{Z\in F[E/X]}\psi[E/X],
    \qquad
    \de{F[E/X]}_{\LExp}(\Gamma)
    =
    \de{F}_{\LExp}(\Gamma_X)
  \]
  The fixed leaky probabilistic-choice argument above gives, for every
  $\mu\in\calD(\de{F}_{\LExp}(\Gamma_X))$,
  \[
    \Gamma,\calR\vDash\bigoplus^\Leak_{Z\sim\mu}\psi[E/X]
    \quad\Longleftrightarrow\quad
    \Gamma_X,\calR\vDash\bigoplus^\Leak_{Z\sim\mu}\psi
  \]
  Quantifying over such $\mu$ gives the result for $\bignd$.

  For the second displayed equivalence, apply the first equivalence with the
  logical expression $Y$ and environment $\Gamma[Y\mapsto v]$:
  \[
    \Gamma[Y\mapsto v],\calR\vDash\varphi[Y/X]
    \Longleftrightarrow
    \Gamma[Y\mapsto v][X\mapsto v],\calR\vDash\varphi
  \]
  Since $Y\notin\fv(\varphi)$, the environments
  $\Gamma[Y\mapsto v][X\mapsto v]$ and $\Gamma[X\mapsto v]$ agree on
  $\fv(\varphi)$.  By context irrelevance,
  \[
    \Gamma[Y\mapsto v][X\mapsto v],\calR\vDash\varphi
    \Longleftrightarrow
    \Gamma[X\mapsto v],\calR\vDash\varphi
  \]
\end{proof}

\section{Soundness of the Sequential Rules}
\label{app:sequential-rules}
\subsection{Command Support Invariants}
\begin{lemma}[Singleton bits constrain count support]
  \label{lem:index-singleton-bits-count-support}
  Fix a schedule $s$.  If $\textsf{bits}(C)=\{k\}$, then for every
  $\sigma,\sigma'\in\Mem$ and $n,n'\in\bbN$,
  \[
    \bigl(\de{C}_s(\sigma,n)\bigr)(\sigma',n')>0
    \quad\Longrightarrow\quad
    n'=n+k
  \]
\end{lemma}
\begin{proof}
  By structural induction on $C$, unfolding the schedule-indexed semantics.
  We write only the
  cases where the count argument is not immediate from the displayed semantic
  equation.

  The atomic commands are immediate from their point-mass semantics.  The
  sequential case follows by factoring positive mass through an ordinary
  intermediate state and adding the two singleton counts.  The conditional and
  probabilistic-choice cases are immediate because positive ordinary output
  mass comes from a branch whose bit set is the same singleton $\{k\}$.

  \emph{Nondeterministic choice.}
  For $C=C_0\cmdnondetS C_1$, the singleton assumption on
  $(1+\textsf{bits}(C_0))\cup(1+\textsf{bits}(C_1))$ gives a $k'$ such that
  $k=k'+1$ and $\textsf{bits}(C_i)=\{k'\}$ for $i\in\{0,1\}$.  With
  $i=s(n)\bmod 2$, the semantic equation is
  \[
    \de{C_0\cmdnondetS C_1}_s(\sigma,n)
    =
    \de{C_i}_s(\sigma,n+1)
  \]
  Thus positive ordinary output mass gives, by the induction hypothesis for
  $C_i$, $n'=(n+1)+k'=n+k$.

  \emph{While loop.}
  It remains to consider $C=\cmdwhile{e}{C_0}$.  Since
  \[
    \textsf{bits}(\cmdwhile{e}{C_0})
    =
    \bigcup_{m\in\bbN}m\cdot\textsf{bits}(C_0)
    =
    \{k\},
  \]
  the $m=0$ summand gives $k=0$, and the $m=1$ summand gives
  $\textsf{bits}(C_0)=\{0\}$.  Let $F_{m,s}$ be the schedule-indexed while
  approximants.  We show by induction on $m$ that
  \[
    F_{m,s}(\sigma,n)(\sigma',n')>0
    \quad\Longrightarrow\quad
    n'=n
  \]
  The zero approximant has no ordinary output mass.  For the successor
  approximant, if the guard is false, the output is $\delta_{(\sigma,n)}$.  If
  the guard is true, positive ordinary output mass factors through an ordinary
  intermediate $(\tau,j)$ with
  \[
    \de{C_0}_s(\sigma,n)(\tau,j)>0,
    \qquad
    F_{m,s}(\tau,j)(\sigma',n')>0
  \]
  The structural induction hypothesis for $C_0$, using
  $\textsf{bits}(C_0)=\{0\}$, gives $j=n$; the approximant induction
  hypothesis then gives $n'=j=n$.

  Finally,
  $\de{\cmdwhile{e}{C_0}}_s(\sigma,n)=\bigsqcup_m F_{m,s}(\sigma,n)$.
  Hence any positive ordinary output mass of the loop appears in some finite
  approximant, so $n'=n=n+k$.
\end{proof}

\begin{lemma}[Singleton bits preserve count marginal]
  \label{lem:index-singleton-bits-count-marginal}
  If $\textsf{bits}(C)=\{k\}$, then for every
  $\mu\in\calD(\Mem\times\bbN)$ such that $\de{C}_s^\dagger(\mu)$ is proper,
  and every
  $E\subseteq\bbN$,
  \[
    \bigl((\pi_2)_*(\de{C}_s^\dagger(\mu))\bigr)(E)
    =
    ((\pi_2)_*\mu)(\{n\in\bbN\mid n+k\in E\})
  \]
  Consequently, if
  $\textsf{bits}(C_1)=\textsf{bits}(C_2)=\{k\}$,
  $\mu_i\in\calD(\Mem\times\bbN)$,
  $(\pi_2)_*\mu_1=(\pi_2)_*\mu_2$, and
  $\bigl(\de{C_i}_s^\dagger(\mu_i)\bigr)(\bot)=0$ for $i\in\{1,2\}$, then
  \[
    (\pi_2)_*\bigl(\de{C_1}_s^\dagger(\mu_1)\bigr)
    =
    (\pi_2)_*\bigl(\de{C_2}_s^\dagger(\mu_2)\bigr)
  \]
\end{lemma}
\begin{proof}
  For $E\subseteq\bbN$, write
  $E-k\triangleq\{n\in\bbN\mid n+k\in E\}$.  By
  \Cref{lem:index-singleton-bits-count-support}, every ordinary output of
  $\de{C}_s(\sigma,n)$ has count $n+k$.  Since
  $\de{C}_s^\dagger(\mu)$ is proper, no input mass contributes to $\bot$ in
  the aggregate: by nonnegativity,
  $\mu(\sigma,n)\cdot\de{C}_s(\sigma,n)(\bot)=0$ for every $(\sigma,n)$.
  Thus summing over inputs gives
  \begin{align*}
    \bigl((\pi_2)_*(\de{C}_s^\dagger(\mu))\bigr)(E)
    &=
      \bigl(\de{C}_s^\dagger(\mu)\bigr)(\Mem\times E)\\
    &=
      \sum_{(\sigma,n)\in\Mem\times\bbN}
        \mu(\sigma,n)\cdot
        \bigl(\de{C}_s(\sigma,n)\bigr)(\Mem\times E)
      \tag{definition of Kleisli extension}\\
    &=
      \sum_{\substack{(\sigma,n)\in\Mem\times\bbN\\ n+k\in E}}
        \mu(\sigma,n)
      \tag{support lemma and properness}\\
    &=
      ((\pi_2)_*\mu)(E-k)
  \end{align*}

  The consequence follows by applying the displayed equality to $C_1$ and
  $C_2$.  For every $E\subseteq\bbN$,
  \[
    \bigl((\pi_2)_*(\de{C_1}_s^\dagger(\mu_1))\bigr)(E)
    =
      ((\pi_2)_*\mu_1)(E-k)
    =
      ((\pi_2)_*\mu_2)(E-k)
    =
      \bigl((\pi_2)_*(\de{C_2}_s^\dagger(\mu_2))\bigr)(E)\]
  Hence
  $(\pi_2)_*\bigl(\de{C_1}_s^\dagger(\mu_1)\bigr)
  =
  (\pi_2)_*\bigl(\de{C_2}_s^\dagger(\mu_2)\bigr)$.
\end{proof}

\begin{lemma}[Owned-cell point update]
  \label{lem:index-owned-cell-point-update}
  Let $\mu\in\calD(\Mem\times\bbN)$ and let
  $\calR_\psi,\calR_x,\calR_F$ be resources such that
  \[
    \Gamma,\calR_\psi\vDash\psi,
    \qquad
    \Gamma,\calR_x\vDash\sure{\own(x)},
    \qquad
    (\calR_\psi\otimes\calR_x)\otimes\calR_F\preceq\mu
  \]
  For $a\in\Val$, define the update map
  $U_a^x:\Mem\times\bbN\to\Mem\times\bbN$ for updating $x$ with $a$, as
  \[
    U_a^x(\sigma,n)\triangleq(\sigma[x\coloneqq a],n),
  \]
  and define
  $\mu[a/x]\triangleq(U_a^x)_*\mu$.
  Then there exists a resource $\calR_x^a$ such that
  \[
    \Gamma,\calR_x^a\vDash\sure{x\mapsto a},
    \qquad
    \Proj{\calR_x^a}{\bbN}=\Proj{\calR_x}{\bbN},
    \qquad
    (\calR_\psi\otimes\calR_x^a)\otimes\calR_F
    \preceq
    \mu[a/x]
  \]
\end{lemma}
\begin{proof}
  Since $\Gamma,\calR_x\vDash\sure{\own(x)}$, the footprint condition gives
  $x\in V_{\calR_x}$.  Define
  \[
    \calR_x^a
    \triangleq
    \left\langle
      \calP_{\calR_x},\,
      V_{\calR_x},\,
      \omega\mapsto\memfn_{\calR_x}(\omega)[x\coloneqq a],\,
      \cntfn_{\calR_x}
    \right\rangle
  \]
  Then $\Proj{\calR_x^a}{\bbN}=\Proj{\calR_x}{\bbN}$, and every memory
  observation of $\calR_x^a$ maps $x$ to $a$, so
  $\Gamma,\calR_x^a\vDash\sure{x\mapsto a}$.

  It remains to prove the distribution refinement.  Put
  \[
    \calR_0\triangleq(\calR_\psi\otimes\calR_x)\otimes\calR_F,
    \qquad
    \calR_a\triangleq(\calR_\psi\otimes\calR_x^a)\otimes\calR_F
  \]
  Choose the same product index set and product pairing for $\calR_0$ and
  $\calR_a$.  They have the same probability space, footprint, and count map,
  and for every index $\omega$,
  $\memfn_{\calR_a}(\omega)
  =
  \memfn_{\calR_0}(\omega)[x\coloneqq a]$.

  Let $\lambda$ witness $\calR_0\preceq\mu$, and push it forward along
  $\id_{\Omega_{\calR_0}}\times U_a^x$, i.e.
  $\lambda_a\triangleq(\id_{\Omega_{\calR_0}}\times U_a^x)_*\lambda$.
  The resource marginal is unchanged, and the distribution marginal is
  $(U_a^x)_*\mu=\mu[a/x]$.  It remains only to check support.  If
  $\lambda_a(\omega,(\sigma',n'))>0$, then for some $(\sigma,n)$,
  $\lambda(\omega,(\sigma,n))>0$ and
  $U_a^x(\sigma,n)=(\sigma',n')$.  Thus
  $\sigma'=\sigma[x\coloneqq a]$ and $n'=n$.  Since $\lambda$ witnesses
  $\calR_0\preceq\mu$, by $x\in V_{\calR_0}$,
  $\sigma'=\sigma[x\coloneqq a]$, and $V_{\calR_a}=V_{\calR_0}$,
  \[
    \memfn_{\calR_a}(\omega)
    =
    \memfn_{\calR_0}(\omega)[x\coloneqq a]
    \sqsubseteq
    \ForgetMem{\sigma}{V_{\calR_0}}[x\coloneqq a]
    =
    \ForgetMem{\sigma[x\coloneqq a]}{V_{\calR_0}}
    \\
    =
    \ForgetMem{\sigma'}{V_{\calR_a}}
  \]
  \[
    \cntfn_{\calR_a}(\omega)
    =
    \cntfn_{\calR_0}(\omega)
    =
    n
    =
    n'
  \]
  Thus $\lambda_a$ witnesses $\calR_a\preceq\mu[a/x]$.
\end{proof}

\begin{lemma}[Owned-cell mixture update]
  \label{lem:index-owned-cell-mixture-update}
  Let $\mu\in\calD(\Mem\times\bbN)$ and let
  $\calR_\psi,\calR_x,\calR_F$ be resources such that
  \[
    \Gamma,\calR_\psi\vDash\psi,
    \qquad
    \Gamma,\calR_x\vDash\sure{\own(x)},
    \qquad
    (\calR_\psi\otimes\calR_x)\otimes\calR_F\preceq\mu
  \]
  For $a\in\Val$, define the update map
  $U_a^x:\Mem\times\bbN\to\Mem\times\bbN$ for updating $x$ with $a$, as
  \[
    U_a^x(\sigma,n)\triangleq(\sigma[x\coloneqq a],n),
  \]
  and define
  $\mu[a/x]\triangleq(U_a^x)_*\mu$.
  Then, for every countable $A\subseteq\Val$ and
  $\alpha\in\calD(A)$, there is a resource $\calR_x^\alpha$ such that
  \[
    \Gamma,\calR_x^\alpha
    \vDash
    \bigoplus^\Priv_{Y\sim\alpha}\sure{x\mapsto Y},
    \qquad
    (\calR_\psi\otimes\calR_x^\alpha)\otimes\calR_F
    \preceq
    \bigoplus_{a\sim\alpha}\mu[a/x]
  \]
\end{lemma}
\begin{proof}
  By \Cref{lem:index-owned-cell-point-update}, for each
  $a\in\supp(\alpha)$ choose a fresh copy of a resource $\calR_x^a$ such that
  \[
    \Gamma,\calR_x^a\vDash\sure{x\mapsto a},
    \qquad
    \Proj{\calR_x^a}{\bbN}=\Proj{\calR_x}{\bbN},
    \qquad
    (\calR_\psi\otimes\calR_x^a)\otimes\calR_F
    \preceq
    \mu[a/x]
  \]
  Thus, for every $b\in\supp(\alpha)$,
  \begin{equation}
    \Proj{\calR_x^b}{\bbN}=\Proj{\calR_x}{\bbN}.
    \label{eq:index-owned-cell-mixture-count-projection}
  \end{equation}
  Put
  $\calR_x^\alpha\triangleq\bigoplus_{a\sim\alpha}\calR_x^a$, using fresh
  copies of the branch resources.  The satisfaction clause for
  $\bigoplus^\Priv$ is witnessed by
  $(\calR_x^a)_{a\in\supp(\alpha)}$.  The refinement and branch-satisfaction
  obligations are
  \[
    \bigoplus_{a\sim\alpha}\calR_x^a
    =
    \calR_x^\alpha
    \preceq
    \calR_x^\alpha
    \tag{reflexivity}
  \]
  and, for every $a\in\supp(\alpha)$,
  \begin{align*}
    \Gamma,\calR_x^a
    &\vDash
      \sure{x\mapsto a}
      \tag{\Cref{lem:index-owned-cell-point-update}}\\
    &\Longleftrightarrow
      \Gamma,\calR_x^a
      \vDash
      \sure{x\mapsto Y}[a/Y]
      \tag{definition of assertion substitution}\\
    &\Longleftrightarrow
      \Gamma[Y\mapsto\de{a}_{\LExp}(\Gamma)],\calR_x^a
      \vDash
      \sure{x\mapsto Y}
      \tag{\Cref{lem:index-assertion-logical-substitution}
      with $\varphi=\sure{x\mapsto Y}$}\\
    &\Longleftrightarrow
      \Gamma[Y\mapsto a],\calR_x^a
      \vDash
      \sure{x\mapsto Y}.
      \tag{$\de{a}_{\LExp}(\Gamma)=a$}
  \end{align*}
  The private side condition follows from \eqref{eq:index-owned-cell-mixture-count-projection}: for every $a\in\supp(\alpha)$,
  \[
    \Proj{\calR_x^\alpha}{\bbN}
    =
    \bigoplus_{b\sim\alpha}\Proj{\calR_x^b}{\bbN}
    =
    \bigoplus_{b\sim\alpha}\Proj{\calR_x}{\bbN}
    =
    \Proj{\calR_x}{\bbN}
    =
    \Proj{\calR_x^a}{\bbN}
  \]
  The distribution-refinement claim is
  \[
    (\calR_\psi\otimes\calR_x^\alpha)\otimes\calR_F
    \cong
    \bigoplus_{a\sim\alpha}
    \bigl((\calR_\psi\otimes\calR_x^a)\otimes\calR_F\bigr)
    \preceq
    \bigoplus_{a\sim\alpha}\mu[a/x]
  \]
  via properties in \Cref{fig:resource-algebra-refinement-facts}
\end{proof}

\subsection{Inference Rules}
\label{app:inference-rules}
\begin{lemma}[Soundness of sequential inference rules]
  \label{lem:index-sequential-rules-soundness}
  Each sequential inference rule displayed below is sound: whenever an
  instance of one of these rules derives
  $\vdash_m\triple{\varphi}{C}{\psi}$ from semantically valid premises,
  the semantic triple $\vDash_m\triple{\varphi}{C}{\psi}$ holds.
\end{lemma}
\begin{proof}
  We prove the rules in order.  Each case must also establish the footprint
  bound $V_\calQ\subseteq V_\calR$ required by the validity of \LogicName{}
  triples; for the structural and base-command rules this is uniform, so we
  record the argument once here rather than repeating it in every case.  For
  \ruleref{Skip}, \ruleref{Seq}, \ruleref{IfT}, \ruleref{IfF},
  \ruleref{Weakening}, \ruleref{Disj}, and \ruleref{Inst} the output witness
  $\calQ$ is a premise witness, or a composition of premise witnesses, so the
  bound follows by reflexivity and transitivity from the premises, which carry
  the same invariant.  For \ruleref{Assign} and \ruleref{Samp} the output
  resource is built over $V_\calR$ and modifies only the owned cell $x$; since
  the precondition entails $\sure{\own(x)}$, the footprint condition gives
  $x\in V_\calR$, so no new variables are introduced.  For \ruleref{Frame} the
  witness $\calQ_\psi\otimes\calR_\vartheta$ has footprint
  $V_{\calQ_\psi}\cup V_{\calR_\vartheta}\subseteq V_\calR$.  The remaining
  rules, which recombine branch witnesses into a direct sum
  (\ruleref{ND}, \ruleref{NAssign}, \ruleref{Prob-Leak}, \ruleref{Prob-Safe},
  and the split rules), establish the bound as part of the branch construction
  in their respective cases.
  \begin{enumerate}[label=\textup{(\arabic*)}]
    \item
  \begin{mathpar}
    \inferrule*[right=Skip]
    { }
    {\vdash_m \triple{\varphi}{\cmdskip}{\varphi}}
  \end{mathpar}

  Immediate: given any semantic input
  $\Gamma,\calR\vDash\varphi$ and
  $\calR\otimes\calR_F\preceq\mu$, take $\calQ\triangleq\calR$.
  Then $\Gamma,\calQ\vDash\varphi$, and
  $\de{\cmdskip}_s^\dagger(\mu)=\mu$ by
  \Cref{lem:index-schedule-indexed-distribution-semantic-equations}.

    \item
  \begin{mathpar}
    \inferrule*[right=Seq]
    {
      \vdash_m\triple{\varphi}{C_1}{\theta}
      \and
      \vdash_m\triple{\theta}{C_2}{\psi}
    }
    {
      \vdash_m\triple{\varphi}{C_1 \cmdseqS C_2}{\psi}
    }
  \end{mathpar}

  Immediate after \Cref{lem:index-schedule-indexed-distribution-semantic-equations}.

    \item
  \begin{mathpar}
    \inferrule*[right=Assign]
    {
      \varphi \Rightarrow
      \sure{e \mapsto E} \land (\psi \sep \sure{\own(x)})
    }
    {
      \vdash_m
      \triple{\varphi}{\actassign{x}{e}}
      {\psi \sep \sure{x \mapsto E}}
    }
  \end{mathpar}

  Let $m\in\{\strong,\weak\}$, and fix
  $s,\Gamma,\mu,\calR,\calR_F$ such that
  \[
    \Gamma,\calR\vDash\varphi,
    \qquad
    \Stable{m}{\calR_F},
    \qquad
    \calR\otimes\calR_F\preceq\mu
  \]
  Since refinement against a partial distribution is defined only for proper
  targets, $\mu$ is proper; from now on we use the same symbol for its ordinary
  coercion.  Also let
  $a\triangleq\de{E}_{\LExp}(\Gamma)$.

  The entailment premise gives
  \[
    \Gamma,\calR\vDash\sure{e\mapsto E},
    \qquad
    \Gamma,\calR\vDash\psi\sep\sure{\own(x)}
  \]
  Unfold the separating conjunction and choose
  $\calR_\psi,\calR_x$ such that
  \[
    \calR_\psi\otimes\calR_x\preceq\calR,
    \qquad
    \Gamma,\calR_\psi\vDash\psi,
    \qquad
    \Gamma,\calR_x\vDash\sure{\own(x)}
  \]
  Hence
  $(\calR_\psi\otimes\calR_x)\otimes\calR_F
  \preceq
  \calR\otimes\calR_F
  \preceq
  \mu 
  $.
  Applying
  \Cref{lem:index-owned-cell-point-update} gives a resource
  $\calR_x^a$ such that
  \[
    \Gamma,\calR_x^a\vDash\sure{x\mapsto a},
    \qquad
    (\calR_\psi\otimes\calR_x^a)\otimes\calR_F
    \preceq
    \mu[a/x]
  \]

  We now identify the command output.  By
  \Cref{lem:index-assertion-frame-preservation},
  $\Gamma,\calR\otimes\calR_F\vDash\sure{e\mapsto E}$.  Applying
  \Cref{lem:index-pure-assertion-refinement-transfer} to
  $P\triangleq e\mapsto E$ and
  $\calR\otimes\calR_F\preceq\mu$, and then unfolding the event returned by
  that lemma, gives
  \begin{align*}
    1
    &=
      \mu\bigl(\{(\sigma,n)\mid
      \ForgetMem{\sigma}{\Var(e)}
      \in
      \sem{e\mapsto E}_\Gamma^{\Var(e)}\}\bigr)
      \tag{\Cref{lem:index-pure-assertion-refinement-transfer}
      and $\Var(e\mapsto E)=\Var(e)$}\\
    &=
      \mu\bigl(\{(\sigma,n)\mid
      \Gamma,\ForgetMem{\sigma}{\Var(e)}
      \vDash e\mapsto E\}\bigr)
      \tag{definition of $\sem{P}_\Gamma^S$ with $P=e\mapsto E$}\\
    &=
      \mu\bigl(\{(\sigma,n)\mid
      \de{e}_{\Exp}(\ForgetMem{\sigma}{\Var(e)})
      =
      \de{E}_{\LExp}(\Gamma)\}\bigr)
      \tag{satisfaction clause for $e\mapsto E$}\\
    &=
      \mu\bigl(\{(\sigma,n)\mid
      \de{e}_{\Exp}(\sigma)=a\}\bigr).
      \tag{evaluation of $e$ depends only on $\Var(e)$ and definition of $a$}
  \end{align*}
  Therefore
  \begin{align*}
    \de{\actassign{x}{e}}_s^\dagger(\mu)
    &=
      \bigoplus_{(\sigma,n)\sim\mu}
      \delta_{(\sigma[x\coloneqq\de{e}_{\Exp}(\sigma)],n)}
      \tag{\Cref{lem:index-schedule-indexed-distribution-semantic-equations}}\\
    &=
      \bigoplus_{(\sigma,n)\sim\mu}
      \delta_{(\sigma[x\coloneqq a],n)}
      \tag{$\mu(\{(\sigma,n)\mid\de{e}_{\Exp}(\sigma)=a\})=1$}\\
    &=
      \mu[a/x].
      \tag{definition of $\mu[a/x]$}
  \end{align*}
  Let $\calQ\triangleq\calR_\psi\otimes\calR_x^a$.  The satisfaction clause
  for $\sep$ gives
  $\Gamma,\calQ\vDash\psi\sep\sure{x\mapsto E}$, and the displayed refinement
  gives
  \[
    \calQ\otimes\calR_F
    \preceq
    \de{\actassign{x}{e}}_s^\dagger(\mu)
  \]
  Thus $\calQ$ is the required postcondition witness.

    \item
  \begin{mathpar}
    \inferrule*[right=SAMP]
    {
      \Stable{\weak}{\psi}
      \and
      \varphi \Rightarrow
      \sure{\vec e \mapsto \vec E} \land
      (\psi \sep \sure{\own(x)})
    }
    {
      \vdash_m
      \triple{\varphi}{x\samp d(\vec e)}
      {\psi \sep (x \sim d(\vec E))}
    }
  \end{mathpar}

  Let $m\in\{\strong,\weak\}$, and fix
  $s,\Gamma,\mu,\calR,\calR_F$ such that
  \[
    \Gamma,\calR\vDash\varphi,
    \qquad
    \Stable{m}{\calR_F},
    \qquad
    \calR\otimes\calR_F\preceq\mu
  \]
  Since refinement against a partial distribution is defined only for proper
  targets, $\mu$ is proper; from now on we use the same symbol for its ordinary
  coercion.  Consider the distribution
  $\alpha\triangleq d(\de{\vec E}_{\LExp}(\Gamma))$.

  The entailment premise gives
  \[
    \Gamma,\calR\vDash\sure{\vec e\mapsto\vec E},
    \qquad
    \Gamma,\calR\vDash\psi\sep\sure{\own(x)}
  \]
  Unfold the separating conjunction and choose
  $\calR_\psi,\calR_x$ such that
  \[
    \calR_\psi\otimes\calR_x\preceq\calR,
    \qquad
    \Gamma,\calR_\psi\vDash\psi,
    \qquad
    \Gamma,\calR_x\vDash\sure{\own(x)}
  \]
  Hence
  $(\calR_\psi\otimes\calR_x)\otimes\calR_F
  \preceq
  \calR\otimes\calR_F
  \preceq
  \mu
  $.
  Applying
  \Cref{lem:index-owned-cell-mixture-update} with this $\alpha$ gives a
  resource $\calR_x^\alpha$ such that
  \[
    \Gamma,\calR_x^\alpha
    \vDash
    \bigoplus^\Priv_{Y\sim\alpha}\sure{x\mapsto Y},
    \qquad
    (\calR_\psi\otimes\calR_x^\alpha)\otimes\calR_F
    \preceq
    \bigoplus_{a\sim\alpha}\mu[a/x]
  \]

  We now identify the command output.  By
  \Cref{lem:index-assertion-frame-preservation},
  $\Gamma,\calR\otimes\calR_F\vDash\sure{\vec e\mapsto\vec E}$, so
  \Cref{lem:index-pure-assertion-refinement-transfer} gives
  \[
    \mu(\{(\sigma,n)\mid
      d(\de{\vec e}_{\Exp}(\sigma))=\alpha\})=1
  \]
  Therefore
  \begin{align*}
    \de{x\samp d(\vec e)}_s^\dagger(\mu)
    &=
      \bigoplus_{(\sigma,n)\sim\mu}
      \bigoplus_{a\sim d(\de{\vec e}_{\Exp}(\sigma))}
      \delta_{(\sigma[x\coloneqq a],n)}
      \tag{\Cref{lem:index-schedule-indexed-distribution-semantic-equations}}\\
    &=
      \bigoplus_{(\sigma,n)\sim\mu}
      \bigoplus_{a\sim\alpha}
      \delta_{(\sigma[x\coloneqq a],n)}
      \tag{$\mu(\{(\sigma,n)\mid d(\de{\vec e}_{\Exp}(\sigma))=\alpha\})=1$}\\
    &=
      \bigoplus_{a\sim\alpha}
      \bigoplus_{(\sigma,n)\sim\mu}
      \delta_{(\sigma[x\coloneqq a],n)}
      \tag{regrouping convex mixtures}\\
    &=
      \bigoplus_{a\sim\alpha}\mu[a/x].
      \tag{definition of $\mu[a/x]$}
  \end{align*}
  Let $\calQ\triangleq\calR_\psi\otimes\calR_x^\alpha$.  Since
  $x\sim d(\vec E)$ abbreviates
  $\bigoplus^\Priv_{Y\sim d(\de{\vec E}_{\LExp}(\Gamma))}
  \sure{x\mapsto Y}$, the satisfaction clause for $\sep$ gives
  \[
    \Gamma,\calQ\vDash\psi\sep(x\sim d(\vec E))
  \]
  The displayed refinement gives
  \[
    \calQ\otimes\calR_F
    \preceq
    \de{x\samp d(\vec e)}_s^\dagger(\mu)
  \]
  Thus $\calQ$ is the required postcondition witness.

    \item
  \begin{mathpar}
    \inferrule*[right=NAssign]
    {
      \Stable{\weak}{\psi}
      \and
      \varphi \Rightarrow
      \sure{e \mapsto E} \land (\psi \sep \sure{\own(x)})
      \and
      X\notin\fv(\psi)
    }
    {
      \vdash_\weak
      \triple{\varphi}{x\gets e}
      {\textstyle\bignd_{X\in E}(\psi \sep \sure{x \mapsto X})}
    }
  \end{mathpar}

  Fix $s,\Gamma,\mu,\calR,\calR_F$ such that
  \[
    \Gamma,\calR\vDash\varphi,
    \qquad
    \Stable{\weak}{\calR_F},
    \qquad
    \calR\otimes\calR_F\preceq\mu
  \]
  Since refinement against a partial distribution is defined only for proper
  targets, $\mu$ is proper; from now on we use the same symbol for its ordinary
  coercion.  Consider the set
  $A\triangleq\de{E}_{\LExp}(\Gamma)$.

  The entailment premise gives
  \[
    \Gamma,\calR\vDash\sure{e\mapsto E},
    \qquad
    \Gamma,\calR\vDash\psi\sep\sure{\own(x)}
  \]
  Unfold the separating conjunction and choose
  $\calR_\psi,\calR_x$ such that
  \[
    \calR_\psi\otimes\calR_x\preceq\calR,
    \qquad
    \Gamma,\calR_\psi\vDash\psi,
    \qquad
    \Gamma,\calR_x\vDash\sure{\own(x)}
  \]
  By $\Stable{\weak}{\psi}$, choose resources
  $\calR_\psi^\memfn$ and $\calR_\psi^\cntfn$ such that
  \[
    \calR_\psi^\memfn\otimes\calR_\psi^\cntfn\preceq\calR_\psi,
    \qquad
    \Stable{\weak}{\calR_\psi^\memfn},
    \qquad
    \Gamma,\calR_\psi^\memfn\vDash\psi,
    \qquad
    V_{\calR_\psi^\cntfn}=\emptyset
  \]
  Let
  $\calR_{x,c}\triangleq\calR_x\otimes\calR_\psi^\cntfn$.
  By \Cref{lem:index-assertion-frame-preservation},
  $\Gamma,\calR_{x,c}\vDash\sure{\own(x)}$.  Moreover,
  \begin{align}
    (\calR_\psi^\memfn\otimes\calR_F)\otimes\calR_{x,c}
    &=
    (\calR_\psi^\memfn\otimes\calR_F)
    \otimes(\calR_x\otimes\calR_\psi^\cntfn)
    \tag{definition of $\calR_{x,c}$}\\
    &\cong
    ((\calR_\psi^\memfn\otimes\calR_\psi^\cntfn)
    \otimes\calR_x)\otimes\calR_F
      \notag\\
    &\preceq
    (\calR_\psi\otimes\calR_x)\otimes\calR_F
    \tag{$\otimes$ monotonicity and reflexivity}\\
    &\preceq
    \calR\otimes\calR_F
    \tag{$\otimes$ monotonicity and reflexivity}\\
    &\preceq
    \mu .
    \label{eq:index-nassign-input-refinement}
  \end{align}

  For $a\in A$, define $I_a$
  as the set of scheduler-count positions at which the fixed schedule selects
  an occurrence of $a$ from the list $A$, that is
  \[
    I_a
    \triangleq
    \{n\in\bbN
    \mid A[(s(n)\bmod |A|)]=a\},
  \]
  and define its distribution as
  $\alpha(a)\triangleq\mu(\Mem\times I_a)$.
  The family $(I_a)_{a\in A}$ partitions $\bbN$ up to empty classes, so
  $\alpha\in\calD(A)$.  Since $\calR_\psi^\memfn$ and $\calR_F$ contain only
  zero scheduler counts, \eqref{eq:index-nassign-input-refinement} gives, for
  each $a\in A$,
  \begin{align*}
    \alpha(a)
    &=
      \mu(\Mem\times I_a)
      \tag{definition of $\alpha$}\\
    &=
      ((\pi_2)_*\mu)(I_a)
      \tag{definition of $(\pi_2)_*$}\\
    &=
      \Proj{(\calR_\psi^\memfn\otimes\calR_F)\otimes\calR_{x,c}}{\bbN}(I_a)
      \tag{
      by \eqref{eq:index-nassign-input-refinement}}\\
    &=
      \left((\mathsf{add}_{\bbN})_*
      \bigl(\Proj{\calR_\psi^\memfn\otimes\calR_F}{\bbN}
      \otimes
      \Proj{\calR_{x,c}}{\bbN}\bigr)\right)(I_a)\\
    &=
      \left((\mathsf{add}_{\bbN})_*
      \bigl(\delta_0\otimes\Proj{\calR_{x,c}}{\bbN}\bigr)\right)(I_a)
      \tag{zero-count factors}\\
    &=
      \Proj{\calR_{x,c}}{\bbN}(I_a).
      \tag{definition of $\mathsf{add}_{\bbN}$}
  \end{align*}
  For each $a\in\supp(\alpha)$, let
  $\calR_{x,c}^a\triangleq\RestrictN{\calR_{x,c}}{I_a}$ and
  $\mu_a\triangleq\RestrictN{\mu}{I_a}$.
  The preceding display ensures that these restrictions are defined.  Moreover,
  \begin{align}
    (\calR_\psi^\memfn\otimes\calR_{x,c}^a)\otimes\calR_F
    &\cong
      (\calR_\psi^\memfn\otimes\calR_F)\otimes\calR_{x,c}^a
      \tag{associativity and commutativity of $\otimes$}\\
    &=
      \RestrictN{((\calR_\psi^\memfn\otimes\calR_F)
      \otimes\calR_{x,c})}{I_a}
      \tag{definition of $\calR_{x,c}^a$ and zero-counts}\\
    &\preceq
      \RestrictN{\mu}{I_a}
      \tag{by
      \eqref{eq:index-nassign-input-refinement}}\\
    &=
      \mu_a .
      \label{eq:index-nassign-restricted-input-refinement}
  \end{align}
  By \Cref{lem:index-count-restriction-pure-assertion} applied to
  $\Gamma,\calR_{x,c},\own(x)$ and $I_a$,
  $\Gamma,\calR_{x,c}^a\vDash\sure{\own(x)}$.

  Applying \Cref{lem:index-owned-cell-point-update} to
  \eqref{eq:index-nassign-restricted-input-refinement} gives, for every
  $a\in\supp(\alpha)$, a resource $\calR_{x,a}$ such that
  \[
    \Gamma,\calR_{x,a}\vDash\sure{x\mapsto a},
    \qquad
    (\calR_\psi^\memfn\otimes\calR_{x,a})\otimes\calR_F
    \preceq
    \mu_a[a/x]
  \]
  By \Cref{lem:index-tick-satisfaction},
  $\Gamma,\Tick(\calR_{x,a})\vDash\sure{x\mapsto a}$.  Also,
  \[
    (\calR_\psi^\memfn\otimes\Tick(\calR_{x,a}))\otimes\calR_F
    \cong
    \Tick\bigl((\calR_\psi^\memfn\otimes\calR_{x,a})\otimes\calR_F\bigr)\\
    \preceq
    \Tick_\calD(\mu_a[a/x])\\
  \]

  For $a\in\supp(\alpha)$, put
  $\calQ_a\triangleq
  \calR_\psi^\memfn\otimes\Tick(\calR_{x,a})$.
  Because $X\notin\fv(\psi)$, context irrelevance gives
  $\Gamma[X\mapsto a],\calR_\psi^\memfn\vDash\psi$.  Hence the satisfaction
  clause for $\sep$ gives
  $\Gamma[X\mapsto a],\calQ_a
  \vDash
  \psi\sep\sure{x\mapsto X}$.

  Let
  $\calQ\triangleq\bigoplus_{a\sim\alpha}\calQ_a$.
  The satisfaction clause for leaky probabilistic choice, with the trivial
  $\Obs_\Leak^\alpha$ side condition, gives
  $\Gamma,\calQ
  \vDash
  \textstyle\bigoplus^\Leak_{X\sim\alpha}
  (\psi\sep\sure{x\mapsto X})$.
  Since $\alpha\in\calD(A)$, the same $\alpha$ witnesses
  $\Gamma,\calQ
  \vDash
  \textstyle\bignd_{X\in E}(\psi\sep\sure{x\mapsto X})$.

  It remains to identify the command output.  By
  \Cref{lem:index-assertion-frame-preservation},
  $\Gamma,\calR\otimes\calR_F\vDash\sure{e\mapsto E}$, so
  \Cref{lem:index-pure-assertion-refinement-transfer} gives
  $\mu(\{(\sigma,n)\mid \de{e}_{\Exp}(\sigma)=A\})=1$.
  Thus
  \begin{align}
    \de{x\gets e}_s^\dagger(\mu)
    &=
      \bigoplus_{(\sigma,n)\sim\mu}
      \delta_{
        (\sigma[x\coloneqq
        \de{e}_{\Exp}(\sigma)[(s(n)\bmod |\de{e}_{\Exp}(\sigma)|)]],n+1)
      }
      \tag{\Cref{lem:index-schedule-indexed-distribution-semantic-equations}}\\
    &=
      \bigoplus_{(\sigma,n)\sim\mu}
      \delta_{
        (\sigma[x\coloneqq
        A[(s(n)\bmod |A|)]],n+1)
      }
      \tag{$\mu(\{(\sigma,n)\mid \de{e}_{\Exp}(\sigma)=A\})=1$}\\
    &=
      \bigoplus_{a\sim\alpha}
      \Tick_\calD(\mu_a[a/x]).
      \tag{definition of $\alpha$, $\mu_a$, and $I_a$}
      \label{eq:index-nassign-output-decomposition}
  \end{align}
  Therefore,
  \[
    \calQ\otimes\calR_F
    =
    \left(\bigoplus_{a\sim\alpha}\calQ_a\right)\otimes\calR_F
    \cong
    \bigoplus_{a\sim\alpha}(\calQ_a\otimes\calR_F)\\
    \preceq
    \bigoplus_{a\sim\alpha}\Tick_\calD(\mu_a[a/x])\\
    =
    \de{x\gets e}_s^\dagger(\mu)
  \]
  Thus $\calQ$ is the required postcondition witness.

    \item
  \begin{mathpar}
    \inferrule*[right=IfT]
    {
      \varphi \Rightarrow \sure{b\mapsto\tru}
      \\
      \vdash_m \triple{\varphi}{C_1}{\psi}
    }
    {
      \vdash_m \triple{\varphi}{\cmdcase{b}{C_1}{C_2}}{\psi}
    }
  \end{mathpar}

  Let $m\in\{\strong,\weak\}$, and fix
  $s,\Gamma,\mu,\calR,\calR_F$ such that
  \[
    \Gamma,\calR\vDash\varphi,
    \qquad
    \Stable{m}{\calR_F},
    \qquad
    \calR\otimes\calR_F\preceq\mu
  \]
  Since refinement against a partial distribution is defined only for proper
  targets, $\mu$ is proper; from now on we use the same symbol for its ordinary
  coercion.

  By the entailment premise, $\Gamma,\calR\vDash\sure{b\mapsto\tru}$.
  Since $\calR\otimes\calR_F$ is defined,
  \Cref{lem:index-assertion-frame-preservation} gives
  $\Gamma,\calR\otimes\calR_F\vDash\sure{b\mapsto\tru}$.
  Consider
  $G_{\tru}
  \triangleq
  \{(\sigma,n)\in\Mem\times\bbN
  \mid \de{b}_{\Exp}(\sigma)=\Ttrue\}$.
  Applying \Cref{lem:index-guard-assertion-refinement-transfer} with guard
  $b$ and $c=\tru$ to
  $\calR\otimes\calR_F\preceq\mu$ gives
  $\mu(G_{\tru})=1$.

  Applying the premise triple for $C_1$ to the same framed input gives a
  resource $\calQ$ such that
  \[
    \Gamma,\calQ\vDash\psi,
    \qquad
    \calQ\otimes\calR_F\preceq\de{C_1}_s^\dagger(\mu)
  \]
  It remains to identify the command denotation.  For every
  $y\in(\Mem\times\bbN)\cup\{\bot\}$,
  \begin{align*}
    \left(\de{\cmdcase{b}{C_1}{C_2}}_s^\dagger(\mu)\right)(y)
    &=
      \sum_{(\sigma,n)\in\Mem\times\bbN}
      \mu(\sigma,n)\cdot
      \begin{cases}
        \de{C_1}_s(\sigma,n)(y),
        & \de{b}_{\Exp}(\sigma)=\Ttrue,\\
        \de{C_2}_s(\sigma,n)(y),
        & \de{b}_{\Exp}(\sigma)=\Tfalse
      \end{cases}
      \tag{\Cref{lem:index-schedule-indexed-distribution-semantic-equations}}\\
    &=
      \sum_{(\sigma,n)\in G_{\tru}}
      \mu(\sigma,n)\cdot\de{C_1}_s(\sigma,n)(y)
      \tag{$\mu(G_{\tru})=1$}\\
    &=
      \sum_{(\sigma,n)\in\Mem\times\bbN}
      \mu(\sigma,n)\cdot\de{C_1}_s(\sigma,n)(y)
      \tag{$\mu((\Mem\times\bbN)\setminus G_{\tru})=0$}\\
    &=
      \left(\de{C_1}_s^\dagger(\mu)\right)(y).
      \tag{definition of Kleisli extension}
  \end{align*}
  Therefore
  $\de{\cmdcase{b}{C_1}{C_2}}_s^\dagger(\mu)=\de{C_1}_s^\dagger(\mu)$, and
  $\calQ$ is the required output witness.

    \item
  \begin{mathpar}
    \inferrule*[right=IfF]
    {
      \varphi \Rightarrow \sure{b\mapsto\fls}
      \\
      \vdash_m \triple{\varphi}{C_2}{\psi}
    }
    {
      \vdash_m \triple{\varphi}{\cmdcase{b}{C_1}{C_2}}{\psi}
    }
  \end{mathpar}

  Symmetric to the IfT case, exchanging
  $\tru$ with $\fls$ and $C_1$ with $C_2$.

    \item
  \begin{mathpar}
    \inferrule*[right=ND]
    {
      \Stable{\weak}{\varphi}\\
      \vdash_\weak\triple{\varphi}{C_0}{\psi_0}\\
      \vdash_\weak\triple{\varphi}{C_1}{\psi_1}
    }
    {
      \vdash_\weak\triple{\varphi}{C_0 \cmdnondetS C_1}{\andop{\psi_0}{\psi_1}}
    }
  \end{mathpar}

  Suppose
  $\Gamma,\calR\vDash\varphi$,
  $\Stable{\weak}{\calR_F}$,
  $\calR\otimes\calR_F\preceq\mu$.
  Since refinement against a partial distribution is defined only for proper
  targets, $\mu$ is proper; from now on we use the same symbol for its ordinary
  coercion.

  By $\Stable{\weak}{\varphi}$, choose a memory resource $\calR^\memfn$ and a
  count resource $\calR^\cntfn$ such that
  $\calR^\memfn\otimes\calR^\cntfn\preceq\calR$
  and
  $\Gamma,\calR^\memfn\vDash\varphi$.
  Since $\calR_F$ is weakly stable, it contains only zero scheduler counts.
  Thus
  \[
    (\calR^\memfn\otimes\calR_F)\otimes\calR^\cntfn
    \cong
    (\calR^\memfn\otimes\calR^\cntfn)\otimes\calR_F
    \preceq
    \calR\otimes\calR_F
    \preceq
    \mu
  \]

  For $i\in\{0,1\}$, let
  $I_i\triangleq\{n\in\bbN\mid s(n)\bmod 2=i\}$ and
  $\alpha(i)\triangleq(\cntfn_{\calR^\cntfn})_*\mu_{\calR^\cntfn}(I_i)$.
  Since $I_0$ and $I_1$ partition $\bbN$, $\alpha\in\calD(\{0,1\})$.  For
  each $i\in\{0,1\}$,
  \begin{align*}
    \alpha(i)
    &=
      \Proj{\calR^\cntfn}{\bbN}(I_i)
      \tag{definition of $\alpha$}\\
    &=
      \Proj{(\calR^\memfn\otimes\calR_F)\otimes\calR^\cntfn}{\bbN}(I_i)
      \tag{
      by zero-count factors}\\
    &=
      ((\pi_2)_*\mu)(I_i)\\
    &=
      \mu(\Mem\times I_i).
      \tag{definition of $(\pi_2)_*$}
  \end{align*}
  Thus, for each $i\in\{0,1\}$,
  \begin{equation}
    \alpha(i)=\mu(\Mem\times I_i).
    \label{eq:index-nd-branch-weight}
  \end{equation}
  For each $i\in\supp(\alpha)$, let
  $\calR_i^\cntfn\triangleq
  \calR^\cntfn\mid\cntfn_{\calR^\cntfn}^{-1}(I_i)$ and
  $\mu_i\triangleq\RestrictN{\mu}{I_i}$.  Thus the distributions $\mu_i$ form the
  branch decomposition of $\mu$ along the partition $(I_i)_{i\in\{0,1\}}$.
  For such $i$, the count-marginal calculation above gives
  \[
    \Proj{(\calR^\memfn\otimes\calR_F)\otimes\calR^\cntfn}{\bbN}(I_i)
    =
    \alpha(i)
    >
    0
  \]
  and 
  \begin{align*}
    (\calR^\memfn\otimes\calR_F)\otimes\calR_i^\cntfn
    &=
      \RestrictN{((\calR^\memfn\otimes\calR_F)\otimes\calR^\cntfn)}{I_i}
      \tag{(definition of $\calR_i^\cntfn$ and zero-count)}
    \\
    &\preceq
      \RestrictN{\mu}{I_i}
    \\
    &=
      \mu_i .
      \tag{(definition of $\mu_i$)}
  \end{align*}
  The input refinement for the branch premise is then
  \begin{align*}
    \Tick(\calR^\memfn\otimes\calR_i^\cntfn)\otimes\calR_F
    &\cong
      \Tick\bigl((\calR^\memfn\otimes\calR_i^\cntfn)\otimes\calR_F\bigr)\\
    &\cong
      \Tick\bigl((\calR^\memfn\otimes\calR_F)\otimes\calR_i^\cntfn\bigr)\\
    &\preceq
      \Tick_\calD(\mu_i)
  \end{align*}
  Since $\calR_i^\cntfn$ is a count restriction of the count-only resource
  $\calR^\cntfn$, it is count-only.  Hence
  \begin{align*}
    \Gamma,\calR^\memfn\vDash\varphi
    &\Longrightarrow
      \Gamma,\calR^\memfn\otimes\calR_i^\cntfn\vDash\varphi
      \tag{\Cref{lem:index-assertion-frame-preservation}
      and $\calR_i^\cntfn$ is count-only}\\
    &\Longrightarrow
      \Gamma,\Tick(\calR^\memfn\otimes\calR_i^\cntfn)\vDash\varphi .
      \tag{\Cref{lem:index-tick-satisfaction}}
  \end{align*}
  Applying the corresponding premise triple gives, for each
  $i\in\supp(\alpha)$, a resource $\calQ_i$ such that
  \[
    \Gamma,\calQ_i\vDash\psi_i
    \qquad\text{and}\qquad
    \calQ_i\otimes\calR_F
    \preceq
    \de{C_i}_s^\dagger(\Tick_\calD(\mu_i))
  \]
  Each displayed partial-distribution refinement has a proper target by
  definition.  Hence the branch mixture is proper, since finite mixtures of
  proper distributions are proper;
  in the next display we use the same symbols for the ordinary coercions of
  these branch outputs.

  Let $\calQ\triangleq\bigoplus_{i\sim\alpha}\calQ_i$.  By the distributivity of
  $\otimes$ and direct sum, and the refinement under direct sum, we get that
  \begin{align*}
    \calQ\otimes\calR_F
    &\cong
      \bigoplus_{i\sim\alpha}(\calQ_i\otimes\calR_F)\\
    &\preceq
      \bigoplus_{i\sim\alpha}
      \de{C_i}_s^\dagger(\Tick_\calD(\mu_i))\\
    &=
      \de{C_0\cmdnondetS C_1}_s^\dagger(\mu).
      \tag{\Cref{lem:index-schedule-indexed-rule-output-equations} and
      \eqref{eq:index-nd-branch-weight}}
  \end{align*}
  Finally, the satisfaction clause for $\&$ gives
  $\Gamma,\calQ\vDash\andop{\psi_0}{\psi_1}$, because each positive-weight branch
  $\calQ_i$ satisfies $\psi_i$.  Taking this $\calQ$ gives the required
  output witness.

    \item
  \begin{mathpar}
    \inferrule*[right=Weakening]
    {
      \vdash_\strong \triple{\varphi}{C}{\psi}
    }
    {
      \vdash_\weak \triple{\varphi}{C}{\psi}
    }
  \end{mathpar}

  Let $s,\Gamma,\mu,\calR,\calR_F$ satisfy the semantic precondition for the
  weak conclusion.  Since $\Stable{\strong}{\calR_F}$ holds trivially, the
  strong premise applies to the same input and yields exactly the required
  output witness.

    \item
  \begin{mathpar}
    \inferrule*[right=Disj]
    {
      \vdash_m\triple{\varphi_1}{C}{\psi_1}
      \and
      \vdash_m\triple{\varphi_2}{C}{\psi_2}
    }
    {
      \vdash_m\triple{\varphi_1 \lor \varphi_2}{C}{\psi_1 \lor \psi_2}
    }
  \end{mathpar}

  Fix a semantic input for the conclusion.  By the satisfaction clause for
  disjunction, the active resource satisfies either $\varphi_1$ or
  $\varphi_2$.  Apply the corresponding premise triple to the same framed
  input; its output witness satisfies $\psi_1$ or $\psi_2$, respectively, and
  hence satisfies $\psi_1\lor\psi_2$.

    \item
  \begin{mathpar}
    \inferrule*[right=Frame]
    {
      \vdash_m \triple{\varphi}{C}{\psi}\\
      \Stable{m}{\vartheta}
    }
    {
      \vdash_m \triple{\varphi \sep \vartheta}{C}{\psi \sep \vartheta}
    }
  \end{mathpar}

  Let $m\in\{\strong,\weak\}$, and fix a semantic input for the conclusion:
  \[
    \Gamma,\calR\vDash\varphi \sep \vartheta,
    \qquad
    \Stable{m}{\calR_F},
    \qquad
    \calR\otimes\calR_F\preceq\mu
  \]
  Unfolding the separating precondition, choose
  $\calR_\varphi$ and $\calR_\vartheta$ such that
  \[
    \calR_\varphi\otimes\calR_\vartheta\preceq\calR,
    \qquad
    \Gamma,\calR_\varphi\vDash\varphi,
    \qquad
    \Gamma,\calR_\vartheta\vDash\vartheta
  \]

  \emph{Case $m=\strong$.}
  Use $\calR_\vartheta\otimes\calR_F$ as the frame for the premise triple.
  By monotonicity, associativity, and transitivity,
  $\calR_\varphi\otimes(\calR_\vartheta\otimes\calR_F)
  \preceq
  \mu$.
  Strong resource stability is trivial, so the premise gives
  $\calQ_\psi$ such that, for $\nu\triangleq\de{C}_s^\dagger(\mu)$,
  \[
    \Gamma,\calQ_\psi\vDash\psi
    \qquad\text{and}\qquad
    \calQ_\psi\otimes(\calR_\vartheta\otimes\calR_F)
    \preceq
    \nu
  \]
  Let $\calQ\triangleq\calQ_\psi\otimes\calR_\vartheta$.  The separating
  satisfaction clause gives $\Gamma,\calQ\vDash\psi \sep \vartheta$, and
  associativity of $\otimes$ gives
  $\calQ\otimes\calR_F
  \preceq
  \nu$.
  Thus $\calQ$ is the required output witness.

  \emph{Case $m=\weak$.}
  By $\Stable{\weak}{\vartheta}$, choose
  $\calR_\vartheta^\memfn$ and $\calR_\vartheta^\cntfn$ such that
  \[
    \calR_\vartheta^\memfn\otimes\calR_\vartheta^\cntfn
    \preceq\calR_\vartheta,
    \qquad
    \Stable{\weak}{\calR_\vartheta^\memfn},
    \qquad
    \Gamma,\calR_\vartheta^\memfn\vDash\vartheta,
    \qquad
    V_{\calR_\vartheta^\cntfn}=\emptyset
  \]
  By associativity, monotonicity, and transitivity,
  $(\calR_\varphi\otimes\calR_\vartheta^\cntfn)
  \otimes
  (\calR_\vartheta^\memfn\otimes\calR_F)
  \preceq
  \mu$.
  Since $\calR_\vartheta^\memfn$ and $\calR_F$ are weakly stable, so is
  $\calR_\vartheta^\memfn\otimes\calR_F$.  Since
  $\calR_\vartheta^\cntfn$ is count-only, monotonicity of satisfaction gives
  \[
    \Gamma,\calR_\varphi\otimes\calR_\vartheta^\cntfn\vDash\varphi
  \]
  Applying the premise with active resource
  $\calR_\varphi\otimes\calR_\vartheta^\cntfn$ and frame
  $\calR_\vartheta^\memfn\otimes\calR_F$ gives $\calQ_\psi$ such that, for
  $\nu\triangleq\de{C}_s^\dagger(\mu)$,
  \[
    \Gamma,\calQ_\psi\vDash\psi
    \qquad\text{and}\qquad
    \calQ_\psi\otimes(\calR_\vartheta^\memfn\otimes\calR_F)
    \preceq
    \nu
  \]
  Let
  $\calQ\triangleq\calQ_\psi\otimes\calR_\vartheta^\memfn$.  The separating
  satisfaction clause gives $\Gamma,\calQ\vDash\psi \sep \vartheta$, and
  associativity of $\otimes$ gives
  $\calQ\otimes\calR_F
  \preceq
  \nu$.
  Thus $\calQ$ is the required output witness.

    \item
  \begin{mathpar}
    \inferrule*[right=Prob-Leak]
    {
      \varphi \Rightarrow \sure{e \mapsto p} \and
      \vdash_m\triple{\varphi}{C_1}{\psi_1} \and
      \vdash_m\triple{\varphi}{C_2}{\psi_2}
    }
    {
      \vdash_m\triple{\varphi}{C_1 \cmdchoiceS{e} C_2}{\psi_1 \oplus_{p} \psi_2}
    }
  \end{mathpar}

  Let $m\in\{\strong,\weak\}$, and fix
  $s,\Gamma,\mu,\calR,\calR_F$ such that
  \[
    \Gamma,\calR\vDash\varphi,
    \qquad
    \Stable{m}{\calR_F},
    \qquad
    \calR\otimes\calR_F\preceq\mu
  \]
  Since refinement against a partial distribution is defined only for proper
  targets, $\mu$ is proper; from now on we use the same symbol for its ordinary
  coercion.
  By $\varphi\Rightarrow\sure{e\mapsto p}$,
  $\Gamma,\calR\vDash\sure{e\mapsto p}$.  Since
  $\calR\otimes\calR_F$ is defined,
  \Cref{lem:index-assertion-frame-preservation} gives
  \[
    \Gamma,\calR\otimes\calR_F\vDash\sure{e\mapsto p}
  \]
  Define
  \[
    G_p
    \triangleq
    \{(\sigma,k)\in\Mem\times\bbN
    \mid
    \de{e}_{\Exp}(\sigma)=p\}
  \]
  as the event in which the concrete guard expression evaluates to $p$.
  Applying \Cref{lem:index-pure-assertion-refinement-transfer} to
  $P\triangleq e\mapsto p$ and the refinement
  $\calR\otimes\calR_F\preceq\mu$ gives
  $\mu(G_p)=1$
  because the event $B_{e\mapsto p}^\Gamma$ from that lemma is exactly
  $G_p$.  Thus, writing
  \[
    \nu_i\triangleq\de{C_i}_s^\dagger(\mu)
    \quad (i\in\{1,2\}),
    \qquad
    \nu\triangleq\de{C_1\cmdchoiceS{e}C_2}_s^\dagger(\mu),
  \]
  the event $G_p$ is precisely the event $G_p^e$ from
  \Cref{lem:index-schedule-indexed-rule-output-equations}.  Hence, by
  \Cref{lem:index-schedule-indexed-rule-output-equations},
  \[
    \nu=\nu_1\oplus_p\nu_2
  \]

  Applying the two premise triples to the same framed input yields resources
  $\calQ_1$ and $\calQ_2$ such that
  \[
    \Gamma,\calQ_i\vDash\psi_i,
    \qquad
    \calQ_i\otimes\calR_F\preceq\nu_i
    \qquad (i\in\{1,2\})
  \]
  Each displayed partial-distribution refinement has a proper target by
  definition, so each $\nu_i$ is proper.  Since
  $\nu=\nu_1\oplus_p\nu_2$, $\nu$ is proper as well.  In the
  following direct-sum refinement, we use the same symbols for the ordinary
  coercions of $\nu_1$ and $\nu_2$.
  Let $\beta\in\calD(\{1,2\})$ be given by
  $\beta(1)=p$ and $\beta(2)=1-p$, and put
  $\calQ\triangleq\bigoplus_{i\sim\beta}\calQ_i$.  Then
  \[
    \calQ\otimes\calR_F
    \cong
    \bigoplus_{i\sim\beta}(\calQ_i\otimes\calR_F)
    \preceq
    \bigoplus_{i\sim\beta}\nu_i
    =
    \nu
  \]

  The branch resources $\calQ_1,\calQ_2$ witness the leaky probabilistic
  choice assertion: $\bigoplus_{i\sim\beta}\calQ_i\preceq\calQ$ by
  reflexivity, $\Gamma,\calQ_i\vDash\psi_i$ for each branch, and
  $\Obs_{\Leak}^{\beta}(\calQ,(\calQ_i)_{i\in\supp(\beta)})=\top$.
  Hence the satisfaction clause gives
  $\Gamma,\calQ\vDash\psi_1\oplus_p\psi_2$.
  Thus $\calQ$ is the required output resource.

    \item
  \begin{mathpar}
    \inferrule*[right=Prob-Safe]
    {
      \varphi \implies \sure{e \mapsto p} \and
      \textsf{bits}(C_1) = \textsf{bits}(C_2) = \{k\}\\
      \vdash_m\triple{\varphi}{C_1}{\psi_1} \\
      \vdash_m\triple{\varphi}{C_2}{\psi_2}
    }
    {
      \vdash_m\triple{\varphi}{C_1 \cmdchoiceS{e} C_2}{\psi_1 \oplus^{\textsf{priv}}_p \psi_2}
    }
  \end{mathpar}

  Let $m\in\{\strong,\weak\}$, and fix
  $s,\Gamma,\mu,\calR,\calR_F$ such that
  \[
    \Gamma,\calR\vDash\varphi,
    \qquad
    \Stable{m}{\calR_F},
    \qquad
    \calR\otimes\calR_F\preceq\mu
  \]
  Since refinement against a partial distribution is defined only for proper
  targets, $\mu$ is proper; from now on we use the same symbol for its ordinary
  coercion.
  Repeat the construction from the Prob-Leak case, using
  \Cref{lem:index-schedule-indexed-rule-output-equations} for the
  probabilistic-choice output equation.  Writing
  \[
    \nu_i\triangleq\de{C_i}_s^\dagger(\mu)
    \quad (i\in\{1,2\}),
    \qquad
    \nu\triangleq\de{C_1\cmdchoiceS{e}C_2}_s^\dagger(\mu),
  \]
  let $\beta\in\calD(\{1,2\})$ be given by
  $\beta(1)=p$ and $\beta(2)=1-p$.  We obtain resources
  $\calQ_1,\calQ_2$ and
  $\calQ\triangleq\bigoplus_{i\sim\beta}\calQ_i$ such that
  \[
    \Gamma,\calQ_i\vDash\psi_i,
    \qquad
    \calQ_i\otimes\calR_F\preceq\nu_i
    \qquad(i\in\{1,2\}),
    \qquad
    \calQ\otimes\calR_F\preceq\nu
  \]
  The same construction also gives $\nu_i(\bot)=0$ for $i\in\{1,2\}$ by
  definition of partial-distribution refinement.  We use the same symbols for
  the ordinary coercions of $\nu_1$ and $\nu_2$ in the count calculation below.

  It remains to check the private mode side condition.
  Write $\mathsf{add}_{\bbN}(n_1,n_2)\triangleq n_1+n_2$.  For each
  $i\in\{1,2\}$,
  \[
    (\mathsf{add}_{\bbN})_*
    (\Proj{\calQ_i}{\bbN}\otimes\Proj{\calR_F}{\bbN})
    =
    \Proj{\calQ_i\otimes\calR_F}{\bbN}
    =
    (\pi_2)_*\nu_i
  \]
  By $\textsf{bits}(C_1)=\textsf{bits}(C_2)=\{k\}$ and
  \Cref{lem:index-singleton-bits-count-marginal}, applied with both input
  distributions equal to $\mu$,
  $(\pi_2)_*\nu_1
  =
  (\pi_2)_*\nu_2$.
  Therefore
  \[
    (\mathsf{add}_{\bbN})_*
    (\Proj{\calQ_1}{\bbN}\otimes\Proj{\calR_F}{\bbN})
    =
    (\mathsf{add}_{\bbN})_*
    (\Proj{\calQ_2}{\bbN}\otimes\Proj{\calR_F}{\bbN})
  \]
  Cancelling the common framed count marginal $\Proj{\calR_F}{\bbN}$ gives
  $\Proj{\calQ_1}{\bbN}
  =
  \Proj{\calQ_2}{\bbN}$.

  Since $\calQ=\calQ_1\oplus_p\calQ_2$, we get that
  $\Proj{\calQ}{\bbN}
  =
  p\cdot\Proj{\calQ_1}{\bbN}
  +(1-p)\cdot\Proj{\calQ_2}{\bbN}
  =
  \Proj{\calQ_1}{\bbN}
  =
  \Proj{\calQ_2}{\bbN}$.
  Therefore
  \[
    \Proj{\calQ_1}{\bbN}
    =
    \Proj{\calQ}{\bbN}
    =
    \Proj{\calQ_2}{\bbN},
  \]
  so $\Obs_\Priv^\beta(\calQ,(\calQ_i)_{i\in\supp(\beta)})$ holds.
  Together with the branch facts
  $\Gamma,\calQ_i\vDash\psi_i$ and the reflexive refinement
  $\bigoplus_{i\sim\beta}\calQ_i\preceq\calQ$, the satisfaction clause for
  private probabilistic choice gives
  $\Gamma,\calQ\vDash\psi_1\oplus^{\textsf{priv}}_p\psi_2$.
  Thus $\calQ$ is the required output resource.

    \item
  \begin{mathpar}
    \inferrule*[right=Pub-Split]
    {
      X\notin\fv(E) \\
      \forall v\in\supp(d(E)).\;
      \vdash_m\triple{\varphi[v/X]}{C}{\psi[v/X]}
    }
    {\textstyle
      \vdash_m
      \triple
      {\bigoplus^\Leak_{X\sim d(E)}\varphi}
      {C}
      {\bigoplus^\Leak_{X\sim d(E)}\psi}
    }
  \end{mathpar}

  Let $m\in\{\strong,\weak\}$, and fix
  $s,\Gamma,\mu,\calR,\calR_F$ such that
  \[
    \Gamma,\calR\vDash\bigoplus^\Leak_{X\sim d(E)}\varphi,
    \qquad
    \Stable{m}{\calR_F},
    \qquad
    \calR\otimes\calR_F\preceq\mu
  \]
  Since refinement against a partial distribution is defined only for proper
  targets, $\mu$ is proper; from now on we use the same symbol for its ordinary
  coercion.
  Put $\alpha\triangleq d(\de{E}_{\LExp}(\Gamma))$.  Unfolding the
  probabilistic choice assertion gives resources
  $(\calR_v)_{v\in\supp(\alpha)}$ such that
  \[
    \bigoplus_{v\sim\alpha}\calR_v\preceq\calR,
    \qquad
    \Gamma[X\mapsto v],\calR_v\vDash\varphi
    \quad(v\in\supp(\alpha))
  \]
  By \Cref{lem:index-assertion-logical-substitution}, for
  $v \in \supp(\alpha)$,
  $\Gamma,\calR_v\vDash\varphi[v/X]$.
  Moreover,
  \[
    \bigoplus_{v\sim\alpha}(\calR_v\otimes\calR_F)
    \cong
    \left(\bigoplus_{v\sim\alpha}\calR_v\right)\otimes\calR_F
    \preceq
    \calR\otimes\calR_F
    \preceq
    \mu
  \]
  Since the above is a refinement, we can choose
  distributions $(\mu_v)_{v\in\supp(\alpha)}$ such that
  $\mu=\bigoplus_{v\sim\alpha}\mu_v$ and
  $\calR_v\otimes\calR_F\preceq\mu_v$ for $v \in \supp(\alpha)$.

  For each $v\in\supp(\alpha)$, put
  $\nu_v\triangleq\de{C}_s^\dagger(\mu_v)$.  The $v$-premise gives a resource
  $\calQ_v$ such that
  \[
    \Gamma,\calQ_v\vDash\psi[v/X],
    \qquad
    \calQ_v\otimes\calR_F\preceq\nu_v
    \quad(v\in\supp(\alpha))
  \]
  Each displayed partial-distribution refinement has a proper target by
  definition.  Hence the branch-output mixture is proper.  We use the same
  symbols for the ordinary coercions of the $\nu_v$ in the recombination below.
  By \Cref{lem:index-assertion-logical-substitution},
  $\Gamma[X\mapsto v],\calQ_v\vDash\psi$ for $v \in \supp(\alpha)$.
  Let $\calQ\triangleq\bigoplus_{v\sim\alpha}\calQ_v$ and
  $\nu\triangleq\de{C}_s^\dagger(\mu)$.  Then, by \Cref{lem:index-dbot-kleisli-linearity}
  \[
    \calQ\otimes\calR_F
    \cong
    \bigoplus_{v\sim\alpha}(\calQ_v\otimes\calR_F)
    \preceq
    \bigoplus_{v\sim\alpha}\nu_v
    =
    \de{C}_s^\dagger\left(\bigoplus_{v\sim\alpha}\mu_v\right)
    =
    \nu
  \]

  The resources $(\calQ_v)_{v\in\supp(\alpha)}$ witness the leaky output
  choice: $\bigoplus_{v\sim\alpha}\calQ_v\preceq\calQ$ by reflexivity,
  $\Gamma[X\mapsto v],\calQ_v\vDash\psi$ for each supported branch, and
  $\Obs_\Leak^\alpha(\calQ,(\calQ_v)_{v\in\supp(\alpha)})=\top$.
  Hence the satisfaction clause gives
  $\Gamma,\calQ\vDash\bigoplus^\Leak_{X\sim d(E)}\psi$.  Since
  $\calQ\otimes\calR_F\preceq\nu$, $\calQ$ is the required output resource.

    \item
  \begin{mathpar}
    \inferrule*[right=Priv-Split]
    {
      X\notin\fv(E) \\
      \textsf{bits}(C)=\{n\} \\
      \forall v\in\supp(d(E)).\;
      \vdash_m\triple{\varphi[v/X]}{C}{\psi[v/X]}
    }
    {\textstyle
      \vdash_m
      \triple
      {\bigoplus^\Priv_{X\sim d(E)}\varphi}
      {C}
      {\bigoplus^\Priv_{X\sim d(E)}\psi}
    }
  \end{mathpar}

  Let $m\in\{\strong,\weak\}$, and fix
  $s,\Gamma,\mu,\calR,\calR_F$ such that
  \[
    \Gamma,\calR\vDash\bigoplus^\Priv_{X\sim d(E)}\varphi,
    \qquad
    \Stable{m}{\calR_F},
    \qquad
    \calR\otimes\calR_F\preceq\mu
  \]
  Since refinement against a partial distribution is defined only for proper
  targets, $\mu$ is proper; from now on we use the same symbol for its ordinary
  coercion.
  Put $\alpha\triangleq d(\de{E}_{\LExp}(\Gamma))$.  Unfolding the private
  probabilistic choice assertion gives resources
  $(\calR_v)_{v\in\supp(\alpha)}$ such that
  \[
    \bigoplus_{v\sim\alpha}\calR_v\preceq\calR,
    \qquad
    \Gamma[X\mapsto v],\calR_v\vDash\varphi,
    \qquad
    \Proj{\calR_v}{\bbN}=\Proj{\calR}{\bbN}
    \quad(v\in\supp(\alpha))
  \]
  Repeating the branch decomposition and recombination construction from
  the Pub-Split case for this same family gives
  distributions $(\mu_v)_{v\in\supp(\alpha)}$, resources
  $(\calQ_v)_{v\in\supp(\alpha)}$, and
  $\calQ\triangleq\bigoplus_{v\sim\alpha}\calQ_v$ such that, for
  $\nu_v\triangleq\de{C}_s^\dagger(\mu_v)$ and
  $\nu\triangleq\de{C}_s^\dagger(\mu)$,
  \[
    \mu=\bigoplus_{v\sim\alpha}\mu_v,
    \qquad
    \calR_v\otimes\calR_F\preceq\mu_v,
    \qquad
    \Gamma[X\mapsto v],\calQ_v\vDash\psi,
    \qquad
    \calQ_v\otimes\calR_F\preceq\nu_v
    \quad(v\in\supp(\alpha)),
  \]
  and
  $\calQ\otimes\calR_F\preceq\nu$.
  The same construction gives $\nu_v(\bot)=0$ for every
  $v\in\supp(\alpha)$ by definition of partial-distribution refinement; below
  we use the same symbols for the ordinary coercions of the $\nu_v$.

  It remains to prove the private mode side condition for the output choice.
  For any $u, v \in \supp(\alpha)$, we get that
  \[
    (\mathsf{add}_{\bbN})_*
    (\Proj{\calR_u}{\bbN}\otimes\Proj{\calR_F}{\bbN})
    =
    \Proj{\calR_u\otimes\calR_F}{\bbN}
    =
    (\pi_2)_*\mu_u
  \]
  and similarly for $v$.  Since
  $\Proj{\calR_u}{\bbN}=\Proj{\calR}{\bbN}=\Proj{\calR_v}{\bbN}$, we get
  $(\pi_2)_*\mu_u=(\pi_2)_*\mu_v$.  By
  \Cref{lem:index-singleton-bits-count-marginal}, applied to $C$ with inputs
  $\mu_u$ and $\mu_v$, and $\textsf{bits}(C)=\{n\}$, this implies
  $(\pi_2)_*\nu_u=(\pi_2)_*\nu_v$.

  Applying the count-marginal lemmas to the output refinements gives
  \[
    (\mathsf{add}_{\bbN})_*
    (\Proj{\calQ_u}{\bbN}\otimes\Proj{\calR_F}{\bbN})
    =
    (\pi_2)_*\nu_u
    =
    (\pi_2)_*\nu_v
    =
    (\mathsf{add}_{\bbN})_*
    (\Proj{\calQ_v}{\bbN}\otimes\Proj{\calR_F}{\bbN})
  \]
  Cancelling the common framed count marginal $\Proj{\calR_F}{\bbN}$ gives
  $\Proj{\calQ_u}{\bbN}=\Proj{\calQ_v}{\bbN}$ for all
  $u,v\in\supp(\alpha)$.  Since
  $\calQ=\bigoplus_{v\sim\alpha}\calQ_v$, the direct-sum definition then gives
  \[
    \Proj{\calQ_v}{\bbN}=\Proj{\calQ}{\bbN}
    \qquad(v\in\supp(\alpha))
  \]
  Thus $\Obs_\Priv^\alpha(\calQ,(\calQ_v)_{v\in\supp(\alpha)})$ holds.
  Together with $\Gamma[X\mapsto v],\calQ_v\vDash\psi$ for each supported
  branch and $\bigoplus_{v\sim\alpha}\calQ_v\preceq\calQ$ by reflexivity, the
  satisfaction clause gives
  $\Gamma,\calQ\vDash\bigoplus^\Priv_{X\sim d(E)}\psi$.  Since
  $\calQ\otimes\calR_F\preceq\nu$, $\calQ$ is the required output resource.

    \item
  \begin{mathpar}
    \inferrule*[right=ND-Split]
    {
      X\notin\fv(E) \\
      \forall v\in E.\;
      \vdash_m\triple{\varphi[v/X]}{C}{\psi[v/X]}
    }
    {\textstyle
      \vdash_m
      \triple
      {\bignd_{X\in E}\varphi}
      {C}
      {\bignd_{X\in E}\psi}
    }
  \end{mathpar}

  Let $m\in\{\strong,\weak\}$, and fix
  $s,\Gamma,\mu,\calR,\calR_F$ such that
  \[
    \Gamma,\calR\vDash\bignd_{X\in E}\varphi,
    \qquad
    \Stable{m}{\calR_F},
    \qquad
    \calR\otimes\calR_F\preceq\mu
  \]
  Since refinement against a partial distribution is defined only for proper
  targets, $\mu$ is proper; from now on we use the same symbol for its ordinary
  coercion.
  Put $A\triangleq\de{E}_{\LExp}(\Gamma)$.  Unfolding the nondeterministic
  assertion, choose $\mu_E\in\calD(A)$ such that
  $\Gamma,\calR\vDash\bigoplus^\Leak_{X\sim\mu_E}\varphi$.
  Repeating the construction from the Pub-Split case with
  this fixed distribution $\mu_E$ gives a resource
  $\calQ\triangleq\bigoplus_{v\sim\mu_E}\calQ_v$ such that, for
  $\nu\triangleq\de{C}_s^\dagger(\mu)$,
  \[
    \calQ\otimes\calR_F\preceq\nu,
    \qquad
    \Gamma,\calQ\vDash\bigoplus^\Leak_{X\sim\mu_E}\psi
  \]
  Since $\mu_E\in\calD(A)$, the same $\mu_E$ witnesses the nondeterministic
  output assertion
  $\Gamma,\calQ\vDash\bignd_{X\in E}\psi$.  Thus $\calQ$ is the required
  output resource.

    \item
  For every value $a$,
  \begin{mathpar}
    \inferrule*[right=Inst]
    {
      \vdash_m\triple{\varphi}{C}{\psi}
    }
    {
      \vdash_m\triple{\varphi[a/X]}{C}{\psi[a/X]}
    }
  \end{mathpar}

  Fix $s,\Gamma,\mu,\calR,\calR_F$ satisfying the semantic precondition for the
  conclusion triple:
  \[
    \Gamma,\calR\vDash\varphi[a/X],
    \qquad
    \Stable{m}{\calR_F},
    \qquad
    \calR\otimes\calR_F\preceq\mu
  \]
  By \Cref{lem:index-assertion-logical-substitution},
  $\Gamma[X\mapsto a],\calR\vDash\varphi$.  Applying
  $\vdash_m\triple{\varphi}{C}{\psi}$ with the logical environment
  $\Gamma[X\mapsto a]$ gives a resource $\calQ$ such that, for
  $\nu\triangleq\de{C}_s^\dagger(\mu)$,
  \[
    \Gamma[X\mapsto a],\calQ\vDash\psi,
    \qquad
    \calQ\otimes\calR_F\preceq\nu
  \]
  Another use of \Cref{lem:index-assertion-logical-substitution} gives
  $\Gamma,\calQ\vDash\psi[a/X]$, which is exactly the required semantic
  postcondition for the conclusion.

    \item
  \begin{mathpar}
    \inferrule*[right=Priv-Split1]
    {
      \textsf{bits}(C)=\{n\}\;\;
      \vdash_m\triple{\varphi}{C}{\psi}
    }
    {\textstyle
      \vdash_m
      \triple
      {\bigoplus^\Priv_{X\sim d(E)}\varphi}
      {C}
      {\bigoplus^\Priv_{X\sim d(E)}\psi}
    }
  \end{mathpar}

  The derived rule follows from the derivation
  \[
    \inferrule*[right=Priv-Split]
    {
      \inferrule*[right=Inst]
      {
        \vdash_m\triple{\varphi}{C}{\psi}
      }
      {
        \forall v\in\supp(d(E)).\;
        \vdash_m\triple{\varphi[v/X]}{C}{\psi[v/X]}
      }
    }
    {
      \vdash_m
      \triple
      {\textstyle\bigoplus^\Priv_{X\sim d(E)}\varphi}
      {C}
      {\textstyle\bigoplus^\Priv_{X\sim d(E)}\psi}
    }
  \]
  By the usual freshness convention for the bound split variable, the elided
  side condition $X\notin\fv(E)$ holds; the remaining side condition
  $\textsf{bits}(C)=\{n\}$ is the displayed premise.

    \item
  For every $o\in\{\Priv,\Leak\}$,
  \begin{mathpar}
    \inferrule*[right=Split2]
    {
      X\notin\fv(\psi) \\
      \convex{\psi}\;\;
      \vdash_m\triple{\varphi}{C}{\psi}
    }
    {\textstyle
      \vdash_m
      \triple
      {\bigoplus^o_{X\sim d(E)}\varphi}
      {C}
      {\psi}
    }
  \end{mathpar}

  The derived rule follows from the derivation
  \[
    \inferrule*[right=Consequence]
    {
      \inferrule*[right=Pub-Split]
      {
        \inferrule*[right=Inst]
        {
          \vdash_m\triple{\varphi}{C}{\psi}
        }
        {
          \forall v\in\supp(d(E)).\;
          \vdash_m\triple{\varphi[v/X]}{C}{\psi[v/X]}
        }
      }
      {
        \vdash_m
        \triple
        {\textstyle\bigoplus^\Leak_{X\sim d(E)}\varphi}
        {C}
        {\textstyle\bigoplus^\Leak_{X\sim d(E)}\psi}
      }
    }
    {
      \vdash_m
      \triple
      {\textstyle\bigoplus^o_{X\sim d(E)}\varphi}
      {C}
      {\psi}
    }
  \]
  By the usual freshness convention for the bound split variable, the elided
  side condition of Pub-Split ($X\notin\fv(E)$) holds.
  The left entailment of Consequence is
  \[
    \textstyle
    \bigoplus^o_{X\sim d(E)}\varphi
    \vdash
    \bigoplus^\Leak_{X\sim d(E)}\varphi
  \]
  by reflexivity if $o=\Leak$, and by
  \Cref{lem:index-private-entails-leaky-choice} if $o=\Priv$.  The right
  entailment is the idempotence collapse
  $\textstyle\bigoplus^\Leak_{X\sim d(E)}\psi \vdash \psi$ from
  \Cref{lem:index-probabilistic-choice-idempotence} (using $X\notin\fv(\psi)$
  and $\convex{\psi}$).

    \item
  For every $o\in\{\Priv,\Leak\}$,
  \begin{mathpar}
    \inferrule*[right=Split1]
    {
      \vdash_m\triple{\varphi}{C}{\psi}
    }
    {\textstyle
      \vdash_m
      \triple
      {\bigoplus^o_{X\sim d(E)}\varphi}
      {C}
      {\bigoplus^\Leak_{X\sim d(E)}\psi}
    }
  \end{mathpar}

  The derived rule follows from the derivation
  \[
    \inferrule*[right=Consequence]
    {
      \inferrule*[right=Pub-Split]
      {
        \inferrule*[right=Inst]
        {
          \vdash_m\triple{\varphi}{C}{\psi}
        }
        {
          \forall v\in\supp(d(E)).\;
          \vdash_m\triple{\varphi[v/X]}{C}{\psi[v/X]}
        }
      }
      {
        \vdash_m
        \triple
        {\textstyle\bigoplus^\Leak_{X\sim d(E)}\varphi}
        {C}
        {\textstyle\bigoplus^\Leak_{X\sim d(E)}\psi}
      }
    }
    {
      \vdash_m
      \triple
      {\textstyle\bigoplus^o_{X\sim d(E)}\varphi}
      {C}
      {\textstyle\bigoplus^\Leak_{X\sim d(E)}\psi}
    }
  \]
  By the usual freshness convention for the bound split variable, the elided
  side condition of Pub-Split ($X\notin\fv(E)$) holds.
  The left entailment of Consequence is reflexivity if $o=\Leak$, and
  \Cref{lem:index-private-entails-leaky-choice} if $o=\Priv$; the right
  entailment is reflexivity.

    \item
  \begin{mathpar}
    \inferrule*[right=ND-Split1]
    {
      \vdash_m\triple{\varphi}{C}{\psi}
    }
    {\textstyle
      \vdash_m
      \triple
      {\bignd_{X\in E}\varphi}
      {C}
      {\bignd_{X\in E}\psi}
    }
  \end{mathpar}

  The derived rule follows from the derivation
  \[
    \inferrule*[right=ND-Split]
    {
      \inferrule*[right=Inst]
      {
        \vdash_m\triple{\varphi}{C}{\psi}
      }
      {
        \forall v\in E.\;
        \vdash_m\triple{\varphi[v/X]}{C}{\psi[v/X]}
      }
    }
    {
      \vdash_m
      \triple
      {\textstyle\bignd_{X\in E}\varphi}
      {C}
      {\textstyle\bignd_{X\in E}\psi}
    }
  \]
  By the usual freshness convention for the bound split variable, the elided
  side condition of ND-Split ($X\notin\fv(E)$) holds.

    \item
  \begin{mathpar}
    \inferrule*[right=ND-Split2]
    {
      X\notin\fv(\psi)\;\;
      \convex{\psi} \\
      \vdash_m\triple{\varphi}{C}{\psi}
    }
    {\textstyle
      \vdash_m
      \triple
      {\bignd_{X\in E}\varphi}
      {C}
      {\psi}
    }
  \end{mathpar}

  The derived rule follows from the derivation
  \[
    \inferrule*[right=Consequence]
    {
      \inferrule*[right=ND-Split]
      {
        \inferrule*[right=Inst]
        {
          \vdash_m\triple{\varphi}{C}{\psi}
        }
        {
          \forall v\in E.\;
          \vdash_m\triple{\varphi[v/X]}{C}{\psi[v/X]}
        }
      }
      {
        \vdash_m
        \triple
        {\textstyle\bignd_{X\in E}\varphi}
        {C}
        {\textstyle\bignd_{X\in E}\psi}
      }
    }
    {
      \vdash_m
      \triple
      {\textstyle\bignd_{X\in E}\varphi}
      {C}
      {\psi}
    }
  \]
  By the usual freshness convention for the bound split variable, the elided
  side condition of ND-Split ($X\notin\fv(E)$) holds, and the elided
  entailments of Consequence are reflexivity on the
  precondition and the idempotence collapse
  $\textstyle\bignd_{X\in E}\psi \vdash \psi$ from
  \Cref{lem:index-nondeterministic-idempotence} (using $X\notin\fv(\psi)$ and
  $\convex{\psi}$).

    \item
  \begin{mathpar}
    \inferrule*[right=Consequence]
    {
      \varphi' \Rightarrow \varphi\\
      \vdash_m \triple{\varphi}{C}{\psi}\\
      \psi \Rightarrow \psi'
    }
    {
      \vdash_m \triple{\varphi'}{C}{\psi'}
    }
  \end{mathpar}

  Immediate.  Given a semantic input for the conclusion, the pre-entailment
  gives $\Gamma,\calR\vDash\varphi$.  The premise triple yields a resource
  $\calQ$ such that, for $\nu\triangleq\de{C}_s^\dagger(\mu)$,
  \[
    \Gamma,\calQ\vDash\psi,
    \qquad
    \calQ\otimes\calR_F\preceq\nu,
    \qquad
    V_\calQ\subseteq V_\calR
  \]
  The post-entailment gives $\Gamma,\calQ\vDash\psi'$, so the same $\calQ$
  witnesses the conclusion.

  \end{enumerate}
\end{proof}

\section{Bounded-Rank Soundness}
\label{app:bounded-rank}
\subsection{Foundations}
\begin{lemma}[Framed decomposition of bounded rank choice]
  \label{lem:index-framed-bounded-rank-choice-decomposition}
  For integers $L\le U$, write
  $E_{L,U}\triangleq\{L,L+1,\ldots,U\}$ and read
  $\bignd_{R=L}^{U}\psi$ as notation for
  $\bignd_{R\in E_{L,U}}\psi$.
  Let $\Gamma$ be a logical environment and let
  $\rho\in\calD_\bot(\Mem\times\bbN)$.  If
  \[
    \Gamma,\calQ\vDash\bignd_{R=L}^{U}\varphi,
    \qquad
    \calQ\otimes\calR_F\preceq\rho,
  \]
  then $\rho$ is proper and, using the same symbol for its ordinary coercion,
  there exist a distribution $\beta\in\calD(E_{L,U})$, resources
  $(\calR_m)_{m\in\supp(\beta)}$, and distributions
  $(\xi_m)_{m\in\supp(\beta)}$ such that
  $\rho
  =
  \bigoplus_{m\sim\beta}\xi_m$ and
  \[
    \Gamma[R\mapsto m],\calR_m
    \vDash
    \varphi,
    \qquad 
    \calR_m\otimes\calR_F
    \preceq
    \xi_m
    \quad(m\in\supp(\beta))
  \]
\end{lemma}
\begin{proof}
  Unfolding $\bignd$ gives a distribution
  $\beta\in\calD(E_{L,U})$ over rank values such that
  $\Gamma,\calQ\vDash\bigoplus_{R\sim\beta}\varphi$.
  Unfolding this $\bigoplus$ gives resources
  $(\calR_m)_{m\in\supp(\beta)}$ satisfying
  \[
    \bigoplus_{m\sim\beta}\calR_m\preceq\calQ,
    \qquad
    \Gamma[R\mapsto m],\calR_m\vDash\varphi
    \quad(m\in\supp(\beta))
  \]
  Push the fixed frame $\calR_F$ through this decomposition:
  \[
    \bigoplus_{m\sim\beta}(\calR_m\otimes\calR_F)
    \cong
    \left(\bigoplus_{m\sim\beta}\calR_m\right)\otimes\calR_F
    \preceq
    \calQ\otimes\calR_F
    \preceq
    \rho
  \]
  The displayed partial-distribution refinement has a proper target by
  definition, so $\rho$ is proper.
  For the rest of the proof we use the same symbol for its ordinary coercion.
  Because we have a refinement, we can choose
  distributions $(\xi_m)_{m\in\supp(\beta)}$ such that
  \[
    \rho=\bigoplus_{m\sim\beta}\xi_m,
    \qquad
    \calR_m\otimes\calR_F\preceq\xi_m
    \quad(m\in\supp(\beta))
  \]
  These are the required witnesses.
\end{proof}

\begin{lemma}[Tagged recombination of a rank-progress split]
  \label{lem:index-tagged-rank-progress-recombination}
  For integers $a\le b$, write
  $E_{a,b}\triangleq\{a,a+1,\ldots,b\}$.
  Let $\ell<n\le h$, let $p,q\in[0,1]$ with $p\le q$, let $\Gamma$ be a
  logical environment, and let
  $\rho^{<},\rho^{\ge}\in\calD(\Mem\times\bbN)$.  Define
  $\rho\triangleq q\cdot\rho^{<}+(1-q)\cdot\rho^{\ge}$.

  Suppose that
  $\beta^{<}\in\calD(E_{\ell,n-1})$,
  $\beta^{\ge}\in\calD(E_{n,h})$, resources
  $(\calR_m^{<})_{m\in\supp(\beta^{<})}$ and
  $(\calR_m^{\ge})_{m\in\supp(\beta^{\ge})}$, and distributions
  $(\xi_m^{<})_{m\in\supp(\beta^{<})}$ and
  $(\xi_m^{\ge})_{m\in\supp(\beta^{\ge})}$ satisfy for $i \in \{<, \ge\}$,
    $\rho^i
    =
    \bigoplus_{m\sim\beta^i}\xi_m^i$ and, for $m \in \supp(\beta^i)$
  \[
    \Gamma[R\mapsto m],\calR_m^i\vDash\varphi,
    \qquad
    \calR_m^i\otimes\calR_F\preceq\xi_m^i
  \]
  Then there exist a countable set $B$, a distribution $\beta\in\calD(B)$, a
  rank projection $r:B\to E_{\ell,h}$, resources
  $(\calR_b)_{b\in\supp(\beta)}$, and distributions
  $(\xi_b)_{b\in\supp(\beta)}$ such that
  $\rho
  =
  \bigoplus_{b\sim\beta}\xi_b$ and
  \begin{align*}
    \Gamma[R\mapsto r(b)],\calR_b
    &\vDash
    \varphi,
    &
    \calR_b\otimes\calR_F
    &\preceq
    \xi_b
    \quad(b\in\supp(\beta)),\\
    \beta(\{b\mid r(b)<n\})
    &\ge p,
    &
    \beta(\{b\mid \ell\le r(b)\le h\})
    &=1
  \end{align*}
  Intuitively, tags separate decreasing-rank branches from remaining-rank
  branches; each branch keeps its original witnesses, and $\beta$ gives mass
  $q\ge p$ to the decreasing ones.
\end{lemma}
\begin{proof}
  Let
  $B\triangleq
  (\{\mathsf{lt}\}\times\supp(\beta^{<}))
  \uplus
  (\{\mathsf{ge}\}\times\supp(\beta^{\ge}))$
  be the tagged disjoint union of the decreasing-rank and remaining-rank
  pieces.  Define the tagged weights by
  \[
    \beta(b)\triangleq
    \begin{cases}
      q\cdot\beta^{<}(m),
        & b=(\mathsf{lt},m),\\
      (1-q)\cdot\beta^{\ge}(m),
        & b=(\mathsf{ge},m)
    \end{cases}
  \]
  Then $\beta\in\calD(B)$ because
  $\beta(B)
  =
  q\cdot\beta^{<}(\supp(\beta^{<}))
  +(1-q)\cdot\beta^{\ge}(\supp(\beta^{\ge}))
  =
  1$.

  Define the rank projection, tagged resource, and tagged output distribution in
  the obvious way, i.e.
  \[
    r(b)\triangleq
    \begin{cases}
      m,
        & b=(\mathsf{lt},m),\\
      m,
        & b=(\mathsf{ge},m),
    \end{cases}
  \]
  \[
    \calR_{(\mathsf{lt},m)}\triangleq\calR_m^{<},
    \qquad
    \calR_{(\mathsf{ge},m)}\triangleq\calR_m^{\ge},
  \]
  \[
    \xi_{(\mathsf{lt},m)}\triangleq\xi_m^{<},
    \qquad
    \xi_{(\mathsf{ge},m)}\triangleq\xi_m^{\ge}
  \]

  Expanding the tagged direct sum gives the following calculation.
  \begin{align*}
    \bigoplus_{b\sim\beta}\xi_b
    &=
    \sum_{b\in\supp(\beta)}\beta(b)\cdot\xi_b
    =
    \sum_{b\in B}\beta(b)\cdot\xi_b
    \tag{definition of $\bigoplus$}\\
    &=
    \sum_{m\in\supp(\beta^{<})}
      \beta(\mathsf{lt},m)\cdot\xi_{(\mathsf{lt},m)}
    +
    \sum_{m\in\supp(\beta^{\ge})}
      \beta(\mathsf{ge},m)\cdot\xi_{(\mathsf{ge},m)}
    \tag{definition of $B$}\\
    &=
    \sum_{m\in\supp(\beta^{<})}
      q\cdot\beta^{<}(m)\cdot\xi_m^{<}
    +
    \sum_{m\in\supp(\beta^{\ge})}
      (1-q)\cdot\beta^{\ge}(m)\cdot\xi_m^{\ge}
    \tag{definitions of $\beta$ and $\xi$}\\
    &=
    q\cdot
    \left(
      \sum_{m\in\supp(\beta^{<})}
        \beta^{<}(m)\cdot\xi_m^{<}
    \right)
    +(1-q)\cdot
    \left(
      \sum_{m\in\supp(\beta^{\ge})}
        \beta^{\ge}(m)\cdot\xi_m^{\ge}
    \right)\\
    &=
    q\cdot
    \left(\bigoplus_{m\sim\beta^{<}}\xi_m^{<}\right)
    +(1-q)\cdot
    \left(\bigoplus_{m\sim\beta^{\ge}}\xi_m^{\ge}\right)
    \tag{definition of $\bigoplus$}\\
    &=
    q\cdot\rho^{<}
    +(1-q)\cdot\rho^{\ge}
    \tag{assumed decompositions of $\rho^{<}$ and $\rho^{\ge}$}\\
    &=
    \rho .
    \tag{definition of $\rho$}
  \end{align*}
  We now prove the branch judgment and framed refinement for each
  $b\in\supp(\beta)$.
  \begin{itemize}
    \item Suppose $b=(\mathsf{lt},m)$.  Since $b\in\supp(\beta)$, the
    definition of $\beta$ gives
    $q\cdot\beta^{<}(m)=\beta(b)>0$, hence
    $m\in\supp(\beta^{<})$.  Therefore the assumptions for the
    decreasing-rank side give
    \[
      \Gamma[R\mapsto m],\calR_m^{<}\vDash\varphi,
      \qquad
      \calR_m^{<}\otimes\calR_F\preceq\xi_m^{<}
    \]
    By the definitions of $r$, $\calR_b$, and $\xi_b$ for
    $b=(\mathsf{lt},m)$, this is exactly
    \[
      \Gamma[R\mapsto r(b)],\calR_b\vDash\varphi,
      \qquad
      \calR_b\otimes\calR_F\preceq\xi_b
    \]

    \item Suppose $b=(\mathsf{ge},m)$.  Since $b\in\supp(\beta)$, the
    definition of $\beta$ gives
    $(1-q)\cdot\beta^{\ge}(m)=\beta(b)>0$, hence
    $m\in\supp(\beta^{\ge})$.  Therefore the assumptions for the
    remaining-rank side give
    \[
      \Gamma[R\mapsto m],\calR_m^{\ge}\vDash\varphi,
      \qquad
      \calR_m^{\ge}\otimes\calR_F\preceq\xi_m^{\ge}
    \]
    By the definitions of $r$, $\calR_b$, and $\xi_b$ for
    $b=(\mathsf{ge},m)$, this is exactly
    \[
      \Gamma[R\mapsto r(b)],\calR_b\vDash\varphi,
      \qquad
      \calR_b\otimes\calR_F\preceq\xi_b
    \]
  \end{itemize}

  The $\mathsf{lt}$-tagged branches are all decreasing-rank branches.  To show
  this, fix $b$ and derive:
  \begin{align*}
    b\in\{\mathsf{lt}\}\times\supp(\beta^{<})
    &\Longrightarrow
    \exists m.\; b=(\mathsf{lt},m)
      \;\land\; m\in\supp(\beta^{<})\\
    &\Longrightarrow
    \exists m.\; b=(\mathsf{lt},m)
      \;\land\; m\in E_{\ell,n-1}
    \tag{$\beta^{<}\in\calD(E_{\ell,n-1})$}\\
    &\Longrightarrow
    \exists m.\; b=(\mathsf{lt},m)
      \;\land\; m<n\\
    &\Longrightarrow
    r(b)<n .
    \tag{definition of $r$}
  \end{align*}
  Therefore we get
  $\{\mathsf{lt}\}\times\supp(\beta^{<})
  \subseteq
  \{b\mid r(b)<n\}$. Then
  \begin{align*}
    \beta(\{b\mid r(b)<n\})
    &\ge
    \beta(\{\mathsf{lt}\}\times\supp(\beta^{<}))
    \tag{monotonicity of measure}\\
    &=
    \sum_{m\in\supp(\beta^{<})}
      \beta(\mathsf{lt},m)
    \tag{definition of mass on the tagged set}\\
    &=
    \sum_{m\in\supp(\beta^{<})}
      q\cdot\beta^{<}(m)
    \tag{definition of $\beta$}\\
    &=
    q\cdot
    \sum_{m\in\supp(\beta^{<})}\beta^{<}(m)\\
    &=
    q
    \tag{$\beta^{<}\in\calD(E_{\ell,n-1})$}\\
    &\ge p .
    \tag{$q\in[p,1]$}
  \end{align*}
  Finally, $\supp(\beta^{<})\subseteq E_{\ell,n-1}$ and
  $\supp(\beta^{\ge})\subseteq E_{n,h}$, so every tagged rank lies in
  $E_{\ell,h}$.  Thus
  $\beta(\{b\mid\ell\le r(b)\le h\})=1$.
\end{proof}

For the bounded-rank soundness proof, it is useful to view the loop as a
random walk through its guard tests: at each test, some mass exits the loop
and the remaining mass continues to the next iteration.  The following
definition packages this bookkeeping: the weights record the mass reaching
each exit/live case, while the normalized distributions are the corresponding
conditional inputs used by the next proof step.  Later, the bounded-rank
invariant will show that no live-step mass is lost to $\bot$.

\begin{definition}[Normalized exit/live decomposition]
  \label[definition]{def:index-normalized-exit-live-decomposition}
  Fix a schedule $s$, a guard $e$, and a command $C$.  For
  $c\in\{\tru,\fls\}$, define
  \[
    G_c\triangleq
    \{(\sigma,k)\in\Mem\times\bbN
    \mid \de{e}_{\Exp}(\sigma)=c\}
  \]
  as the event in which the guard $e$ evaluates to $c$.
  First define conditioning on these guard events for partial distributions.
  For $\xi\in\calD_\bot(\Mem\times\bbN)$, write
  $\xi(G_c)$ for the ordinary mass that $\xi$ assigns to $G_c$.  If
  $\xi(G_c)>0$, define the ordinary conditional distribution
  \[
    (\xi\mid G_c)(\sigma,k)
    \triangleq
    \begin{cases}
      \xi(\sigma,k)/\xi(G_c),
      &(\sigma,k)\in G_c,\\
      0,
      &(\sigma,k)\notin G_c .
    \end{cases}
  \]
  If $\xi(G_c)=0$, choose an arbitrary fixed ordinary distribution for
  $\xi\mid G_c$.

  Now fix an \emph{ordinary starting distribution}
  $\rho\in\calD(\Mem\times\bbN)$, and define
  $w_i^{\rho,\exit}$ and $w_i^{\rho,\live}$ recursively as the weight with
  which the loop exits at the $i$th guard test and remains live after $i$
  executions of $C$, respectively.
  Similarly, we define $\mu_i^{\rho,\exit}$ and
  $\mu_i^{\rho,\live}$ recursively as the normalized distributions obtained by
  conditioning on those two cases.
  \[
    w_0^{\rho,\exit}\triangleq \rho(G_{\fls}),
    \qquad
    w_0^{\rho,\live}\triangleq \rho(G_{\tru}),
    \qquad
    \mu_0^{\rho,\exit}\triangleq\rho\mid G_{\fls},
    \qquad
    \mu_0^{\rho,\live}\triangleq\rho\mid G_{\tru},
  \]
  and, recursively,
  \[
    w_{i+1}^{\rho,\exit}
    \triangleq
    w_i^{\rho,\live}\cdot
    \left(\de{C}_s^\dagger(\mu_i^{\rho,\live})\right)(G_{\fls}),
    \qquad
    w_{i+1}^{\rho,\live}
    \triangleq
    w_i^{\rho,\live}\cdot
    \left(\de{C}_s^\dagger(\mu_i^{\rho,\live})\right)(G_{\tru}),
  \]
  \[
    \mu_{i+1}^{\rho,\exit}
    \triangleq
    \de{C}_s^\dagger(\mu_i^{\rho,\live})\mid G_{\fls},
    \qquad
    \mu_{i+1}^{\rho,\live}
    \triangleq
    \de{C}_s^\dagger(\mu_i^{\rho,\live})\mid G_{\tru}
  \]
  When the starting distribution $\rho$ is clear, we omit it from the
  superscript.
\end{definition}

\begin{definition}[Live-step properness]
  \label[definition]{def:index-live-step-properness}
  With notation from
  \Cref{def:index-normalized-exit-live-decomposition}, fix an ordinary
  starting distribution $\rho\in\calD(\Mem\times\bbN)$.  We say that $C$
  satisfies \emph{live-step properness} from $\rho$ along the guard $e$ and
  schedule $s$, written $\mathsf{LiveStepProper}_{s,e}(C,\rho)$, if
  \[
    \forall i\in\bbN.\;
    w_i^{\rho,\live}>0
    \Longrightarrow
    \left(\de{C}_s^\dagger(\mu_i^{\rho,\live})\right)(\bot)=0
  \]
  When $s,e,C,\rho$ are clear, we also write
  $\mathsf{LiveStepProper}$.
\end{definition}

\begin{lemma}[Exit/live accounting facts]
  \label{lem:index-exit-live-accounting-facts}
  \label{lem:index-positive-live-weight-conditioning-event}
  \label{lem:index-normalized-weight-accounting}
  \label{lem:index-normalized-live-weight-monotonicity}
  \label{lem:index-block-decay-iteration}
  \label{lem:index-block-decay-exhausts-exit-stream}
  With notation from
  \Cref{def:index-normalized-exit-live-decomposition}, the following facts
  hold.
  \begin{enumerate}[label=\textup{(\alph*)}]
    \item More generally, let $X$ be countable, $G\subseteq X$,
    $\mu\in\calD(X)$, and suppose sequences $(w_i)_{i\in\bbN}$ of
    nonnegative reals, $(\mu_i)_{i\in\bbN}$ in $\calD(X)$, and
    $(\rho_i)_{i\in\bbN}$ in $\calD_\bot(X)$ satisfy
    \[
      w_0=\mu(G),\qquad
      \mu_0=\mu\mid G,\qquad
      w_{i+1}=w_i\cdot\rho_i(G),\qquad
      \mu_{i+1}=\rho_i\mid G
    \]
    Then $w_i>0$ implies $\mu_i(G)=1$.

    \item The live weights are nonincreasing:
    \[
      w_{i+1}^{\rho,\live}
      \le
      w_i^{\rho,\live}
    \]

    \item If $\mathsf{LiveStepProper}_{s,e}(C,\rho)$, then
    \[
      w_{i+1}^{\rho,\exit}
      +
      w_{i+1}^{\rho,\live}
      =
      w_i^{\rho,\live}
    \]

    \item If $\mathsf{LiveStepProper}_{s,e}(C,\rho)$, then, for every
    $n\in\bbN$,
    \[
      \left(\sum_{i=0}^{n}w_i^{\rho,\exit}\right)
      +
      w_n^{\rho,\live}
      =
      1
    \]

    \item Let $H\in\bbN$ and $p\in(0,1]$.  If
    $\mathsf{LiveStepProper}_{s,e}(C,\rho)$,
    \[
      H=0\Longrightarrow w_0^{\rho,\live}=0,
      \qquad
      H>0\Longrightarrow
      \forall m.\;
      w_{m+H}^{\rho,\live}
      \le
      (1-p^H)\cdot w_m^{\rho,\live},
    \]
    then
    \[
      \sum_{j\in\bbN}w_j^{\rho,\exit}=1
    \]
  \end{enumerate}
\end{lemma}
\begin{proof}
  For (a), argue by induction on $i$: if $w_0>0$, then $\mu(G)>0$; if
  $w_{i+1}>0$, then $\rho_i(G)>0$.  Thus the relevant conditioning is the
  ordinary conditional distribution on $G$, which has $G$-mass $1$.

  For (b), the recursive definition gives
  \[
    w_{i+1}^{\rho,\live}
    =
    w_i^{\rho,\live}\cdot
    \left(\de{C}_s^\dagger(\mu_i^{\rho,\live})\right)(G_{\tru}),
  \]
  and the second factor lies in $[0,1]$.  For (c), if
  $w_i^{\rho,\live}=0$ the equality is immediate; otherwise live-step
  properness and the partition $G_{\fls},G_{\tru}$ give
  \[
    w_{i+1}^{\rho,\exit}+w_{i+1}^{\rho,\live}
    =
    w_i^{\rho,\live}\cdot
    \left(1-
    \left(\de{C}_s^\dagger(\mu_i^{\rho,\live})\right)(\bot)\right)
    =
    w_i^{\rho,\live}
  \]
  Clause (d) follows by telescoping (c), with base case
  $w_0^{\rho,\exit}+w_0^{\rho,\live}=\rho(G_{\fls})+\rho(G_{\tru})=1$.

  For (e), if $H=0$, then $w_0^{\rho,\live}=0$, and (b) gives
  $w_n^{\rho,\live}=0$ for all $n$.  If $H>0$, then iterating the block-decay
  hypothesis gives
  \[
    w_{kH}^{\rho,\live}
    \le
    (1-p^H)^k\cdot w_0^{\rho,\live}
    \qquad(k\in\bbN)
  \]
  Since $0\le1-p^H<1$ and the live weights are nonincreasing by (b), this
  implies $w_n^{\rho,\live}\to0$.  Taking limits in (d) gives
  $\sum_{j\in\bbN}w_j^{\rho,\exit}=1$.
\end{proof}

\begin{lemma}[Exit/live restart and approximant facts]
  \label{lem:index-exit-live-restart-approximant-facts}
  \label{lem:index-shifted-normalized-live-sequence}
  \label{lem:index-normalized-exit-stream-approximants}
  With notation from
  \Cref{def:index-normalized-exit-live-decomposition}, the following facts
  hold.
  \begin{enumerate}[label=\textup{(\alph*)}]
    \item Suppose $w_0^{\rho,\live}>0$ and
    $\de{C}_s^\dagger(\mu_0^{\rho,\live})$ is proper.  Let
    $\bar w_j^{\exit}$, $\bar w_j^{\live}$,
    $\bar\mu_j^{\exit}$, and $\bar\mu_j^{\live}$ denote the weights and
    normalized distributions obtained by restarting the decomposition from
    $\de{C}_s^\dagger(\mu_0^{\rho,\live})$.  Then, for every $j$,
    \[
      \bar w_j^{\exit}
      =
      \frac{w_{j+1}^{\rho,\exit}}{w_0^{\rho,\live}},
      \qquad
      \bar w_j^{\live}
      =
      \frac{w_{j+1}^{\rho,\live}}{w_0^{\rho,\live}},
      \qquad 
      \bar\mu_j^{\exit}=\mu_{j+1}^{\rho,\exit},
      \qquad
      \bar\mu_j^{\live}=\mu_{j+1}^{\rho,\live}
    \]
    \item Let $(F_m)_{m\in\bbN}$ and $F_{m,s}$ be the Kleene approximants and
    schedule-indexed projections from
    \Cref{lem:index-schedule-indexed-while-semantic-equations}, instantiated
    with $s,e,C$.  If $\mathsf{LiveStepProper}_{s,e}(C,\mu)$, then, for every
    $n\in\bbN$,
    \[
      \bigl((F_{n+1})_s^\bot\bigr)^\dagger(\mu)
      =
      \left(\sum_{i=0}^{n}
        w_i^{\exit}\cdot\mu_i^{\exit}\right)
      +
      w_n^{\live}\cdot\delta_\bot
    \]
  \end{enumerate}
\end{lemma}
\begin{proof}
  Clause (a) is a direct induction on $j$, using the recursive equations for
  the restarted decomposition and for the original decomposition.  Clause (b)
  is an induction on the approximant index.  The base case follows from the
  first unfolding of the while approximant, which exits on $G_{\fls}$ and
  sends the $G_{\tru}$ mass to $\bot$.  The step uses the distribution
  unfolding of the while approximants; if $w_0^{\rho,\live}=0$, all later
  exit and live weights vanish, and otherwise clause (a) identifies the
  restarted decomposition from
  $\de{C}_s^\dagger(\mu_0^{\rho,\live})$ with the shifted original sequence.
\end{proof}

\begin{lemma}[Branchwise progress gives block decay]
  \label{lem:index-branchwise-progress-block-decay}
  Let $H>0$ and $0<p\le1$.  Let $(L_j)_{j\in\bbN}$ be events satisfying
  \[
    L_{j+1}\subseteq L_j
    \qquad(j\in\bbN),
  \]
  and let $(w_j)_{j\in\bbN}$ be nonnegative reals with
  $\PP[L_j]=w_j$.  Fix $i$ with $w_i>0$.  Let $A$ be countable,
  $\theta\in\calD(A)$, and let
  $(L_i^a)_{a\in\supp(\theta)}$ be a partition of $L_i$ such that
  \[
    \PP[L_i^a]=\theta(a)\cdot w_i
    \qquad(a\in\supp(\theta))
  \]
  Thus $L_i^a$ is the part of the live event represented by branch $a$.

  For every $a\in\supp(\theta)$, let $d_a\in\{1,\ldots,H\}$ and let
  $(D_t^a)_{0\le t\le d_a}$ be integer-valued random variables.  Define
  \[
    S_a\triangleq
    \bigcap_{t=0}^{d_a-1}
    \left(\{D_t^a=0\}\cup\{D_{t+1}^a<D_t^a\}\right)
  \]
  This is the event that branch $a$ has already exited, or strictly
  decreases, at every step in its block.

  Suppose that, for every $a\in\supp(\theta)$, the progress event implies
  exit within the block:
  \[
    L_i^a\cap S_a
    \subseteq
    L_{i+H}^{c},
  \]
  and each prefix step decreases with conditional probability at least $p$:
  \[
    \Pr\!\left[
      D_t^a=0 \;\lor\; D_{t+1}^a<D_t^a
      \,\middle|\,
      L_i^a\cap
      \bigcap_{u=0}^{t-1}
      \left(\{D_u^a=0\}\cup\{D_{u+1}^a<D_u^a\}\right)
    \right]\ge p
    \qquad(0\le t<d_a)
  \]
  Then
  \[
    w_{i+H}\le(1-p^H)\cdot w_i
  \]
  Intuitively, conditioned on being live at time $i$, the branch slices
  partition the live mass according to $\theta$.  On each slice, the
  probability of progressing until exit within the block is at least $p^H$;
  hence at most a $(1-p^H)$-fraction of the live mass can remain live.
\end{lemma}
\begin{proof}
  Fix $a\in\supp(\theta)$.  The inclusion assumption gives
  $\PP[L_{i+H}^{c}\mid L_i^a]
  \ge
  \PP[S_a\mid L_i^a]$.
  By the chain rule for conditional probabilities applied to the events
  $\{D_t^a=0\}\cup\{D_{t+1}^a<D_t^a\}$ for $0\le t<d_a$,
  and by the one-step lower bounds,
  \begin{align*}
    \PP[S_a\mid L_i^a]
    &=
    \prod_{t=0}^{d_a-1}
      \Pr\!\left[
        D_t^a=0 \;\lor\; D_{t+1}^a<D_t^a
        \,\middle|\,
        L_i^a\cap
        \bigcap_{u=0}^{t-1}
        \left(\{D_u^a=0\}\cup\{D_{u+1}^a<D_u^a\}\right)
      \right]\\
    &\ge
      \prod_{t=0}^{d_a-1}p\\
    &=
    p^{d_a}\\
    &\ge
    p^H .
      \tag{$d_a\le H$ and $0<p\le1$}
  \end{align*}
  Hence, by total probability over the partition $(L_i^a)_a$,
  \[
    \PP[L_{i+H}^{c}\mid L_i]
    =
    \sum_{a\in\supp(\theta)}
      \PP[L_i^a\mid L_i]\cdot
      \PP[L_{i+H}^{c}\mid L_i^a]\\
    \ge
    \sum_{a\in\supp(\theta)}
      \theta(a)\cdot p^H\\
    =
    p^H
    \]
  Since $L_{j+1}\subseteq L_j$ for every $j$, $L_{i+H}\subseteq L_i$.
  Therefore
  \[
    \frac{w_{i+H}}{w_i}
    =
    \frac{\PP[L_{i+H}]}{\PP[L_i]}\\
    =
    \PP[L_{i+H}\mid L_i]\\
    =
    1-\PP[L_{i+H}^{c}\mid L_i]\\
    \le
    1-p^H
  \]
  Since $w_i>0$, multiplying by $w_i$ gives
  $w_{i+H}\le(1-p^H)\cdot w_i$.
\end{proof}

\begin{definition}[Bounded-rank exit/live witnesses]
  \label[definition]{def:index-bounded-rank-exit-live-witnesses}
  With notation from
  \Cref{def:index-normalized-exit-live-decomposition}, fix a schedule $s$,
  guard $e$, command $C$, and starting distribution $\rho$, and let
  $w_i^{\exit}$, $w_i^{\live}$,
  $\mu_i^{\exit}$, and $\mu_i^{\live}$ be the
  normalized exit/live weights and distributions generated by that definition
  (omitting $\rho$ from the superscripts).
  Also fix a logical environment $\Gamma$, assertions $\varphi$ and $\chi$,
  integers $\ell\le h$, a frame $\calR_F$, and a finite footprint $V_b$.

  We define two predicates $\Exit(i)$ and $\Live(i)$.
  The predicate $\Exit(i)$ is the exit-state postcondition at the $i$th guard
  test.  When the exit weight is positive, it requires an active resource
  satisfying $\chi$ whose framed version refines the normalized exit
  distribution, with footprint contained in $V_b$.  That is, for $i\in\bbN$:
  \[
    \Exit(i) \triangleq
    w_i^{\exit}>0
    \Longrightarrow
    \exists\calQ_i.\;
    \Gamma,\calQ_i\vDash\chi
    \;\land\;
    \calQ_i\otimes\calR_F\preceq\mu_i^{\exit}
    \;\land\;
    V_{\calQ_i}\subseteq V_b
  \]

  The predicate $\Live(i)$ is the live-state invariant after $i$ executions of
  $C$.  When the live weight is positive, it requires the normalized live
  distribution to split into indexed components, each assigned a rank in
  $\{\ell+1,\ldots,h\}$, an active resource satisfying $\varphi$ at that rank,
  and a framed refinement into the corresponding component, again with active
  footprint contained in $V_b$.  That is, for $i\in\bbN$:
  \begin{align*}
    \Live(i)
    &\triangleq
    w_i^{\live}>0
    \Longrightarrow
    \exists A_i,\theta_i\in\calD(A_i),
    r_i:\supp(\theta_i)\to\{n\in\bbN\mid \ell<n\le h\},\\
    &\qquad
    (\calR_{i,a})_{a\in\supp(\theta_i)},
    (\mu_{i,a})_{a\in\supp(\theta_i)}.\\
    &\qquad
    \mu_i^{\live}
    =
    \bigoplus_{a\sim\theta_i}\mu_{i,a}\\
    &\qquad{}\land
    \forall a\in\supp(\theta_i).\;
    \Gamma[R\mapsto r_i(a)],\calR_{i,a}\vDash\varphi
    \;\land\;
    \calR_{i,a}\otimes\calR_F\preceq\mu_{i,a}
    \;\land\;
    V_{\calR_{i,a}}\subseteq V_b
  \end{align*}
\end{definition}

\begin{lemma}[Live invariant gives live-step properness]
  \label{lem:index-bounded-rank-live-step-properness}
  Fix an assertion $\varphi$, a mode $m$, a schedule $s$, a guard $e$, a
  command $C$, a logical environment $\Gamma$, a frame $\calR_F$, an input
  distribution $\mu\in\calD(\Mem\times\bbN)$, integers $\ell\le h$, a finite
  footprint $V_b$, and $p\in(0,1]$.  For integers $a\le b$, write
  $E_{a,b}\triangleq\{a,a+1,\ldots,b\}$, and let
  $\chi\triangleq\varphi[\ell/R]$.

  Instantiate the normalized exit/live decomposition of
  \Cref{def:index-normalized-exit-live-decomposition} with $(s,e,C,\mu)$,
  omitting $\mu$ from superscripts, and let $\Live$ be the live predicate from
  \Cref{def:index-bounded-rank-exit-live-witnesses} for
  $(\Gamma,\varphi,\chi,\ell,h,\calR_F,V_b)$ and the generated weights and
  distributions.  Suppose
  \[
    \vdash_m
    \triple
      {\varphi\sep\sure{R=N>\ell}}
      {C}
      {
        \left(\bignd_{R=\ell}^{N-1}\varphi\right)
        \oplus_{\ge p}
        \left(\bignd_{R=N}^{h}\varphi\right)
      },
    \qquad
    N\notin\vars(\varphi),
    \qquad
    \Stable{m}{\calR_F}
  \]
  and
  \[
    \Live(i)
    \qquad(i\in\bbN)
  \]
  Then $\mathsf{LiveStepProper}_{s,e}(C,\mu)$.
\end{lemma}
\begin{proof}
  Fix $i\in\bbN$ and assume $w_i^{\live}>0$.  By $\Live(i)$, choose
  $A_i,\theta_i \in \calD(A_i),r_i : \supp(\theta_i) \to E_{\ell+1, h},(\calR_{i,a})_{a\in\supp(\theta_i)}$, and
  $(\mu_{i,a})_{a\in\supp(\theta_i)}$ such that
  $\mu_i^{\live}
  =
  \bigoplus_{a\sim\theta_i}\mu_{i,a}$
  and, for every $a\in\supp(\theta_i)$,
  \[
    \Gamma[R\mapsto r_i(a)],\calR_{i,a}\vDash\varphi,
    \qquad
    \calR_{i,a}\otimes\calR_F\preceq\mu_{i,a},
    \qquad
    V_{\calR_{i,a}}\subseteq V_b
  \]
  We first show that each component output is proper.  Fix
  $a\in\supp(\theta_i)$ and put $n_a\triangleq r_i(a)$.  Let
  $\Gamma_{i,a}\triangleq\Gamma[N\mapsto n_a][R\mapsto n_a]$.  Since
  $N\notin\vars(\varphi)$, the environments $\Gamma[R\mapsto n_a]$ and
  $\Gamma_{i,a}$ agree on $\fv(\varphi)$.  Since $n_a\in E_{\ell+1,h}$, the
  premise precondition follows by
  \begin{align*}
    &\Gamma[R\mapsto n_a],\calR_{i,a}\vDash\varphi
    \quad\Longrightarrow\quad
    \Gamma_{i,a},\calR_{i,a}\vDash\varphi
    \tag{context irrelevance}\\
    &\Gamma_{i,a}(R)=\Gamma_{i,a}(N)=n_a
    \;\land\;
    n_a>\ell
    \quad\Longrightarrow\quad
    \Gamma_{i,a},\UnitP\vDash\sure{R=N>\ell}
    \tag{satisfaction clause for $\sure{-}$}\\
    &\Gamma_{i,a},\calR_{i,a}\vDash\varphi
    \;\land\;
    \Gamma_{i,a},\UnitP\vDash\sure{R=N>\ell}
    \;\land\;
    \calR_{i,a}\otimes\UnitP\cong\calR_{i,a}\\
    &\quad\Longrightarrow\quad
    \Gamma_{i,a},\calR_{i,a}
    \vDash
    \varphi\sep\sure{R=N>\ell}.
    \tag{definition of $\UnitP$ and satisfaction clause for $\sep$}
  \end{align*}
  Semantic soundness of the premise triple, applied with environment
  $\Gamma_{i,a}$, input resource $\calR_{i,a}$, frame $\calR_F$, and input
  distribution $\mu_{i,a}$, gives a resource $\calQ_{i,a}$ such that
  $\calQ_{i,a}\otimes\calR_F
  \preceq
  \de{C}_s^\dagger(\mu_{i,a})$.
  Therefore $\de{C}_s^\dagger(\mu_{i,a})$ is proper by definition of
  partial-distribution refinement.

  By \Cref{lem:index-dbot-kleisli-linearity}, the Kleisli extension of
  $\calD_\bot$ is linear,
  $\de{C}_s^\dagger(\mu_i^{\live})
  =
  \bigoplus_{a\sim\theta_i}\de{C}_s^\dagger(\mu_{i,a})$.
  Since every branch output in this mixture is proper, the mixture is proper,
  so $\left(\de{C}_s^\dagger(\mu_i^{\live})\right)(\bot)=0$.
  Since $i$ was arbitrary, this proves
  $\mathsf{LiveStepProper}_{s,e}(C,\mu)$.
\end{proof}

\begin{lemma}[Initial bounded-rank exit/live witnesses]
  \label{lem:index-bounded-rank-initial-exit-live-witnesses}
  Fix an assertion $\varphi$, a schedule $s$, a guard $e$, a command $C$, an
  environment $\Gamma$, resources $\calR,\calR_F$, an input distribution
  $\mu\in\calD(\Mem\times\bbN)$, and integers $\ell\le h$.

  For integers $a\le b$, write
  $E_{a,b}\triangleq\{a,a+1,\ldots,b\}$.  Let
  $\chi\triangleq\varphi[\ell/R]$.  Instantiate the
  normalized exit/live decomposition of
  \Cref{def:index-normalized-exit-live-decomposition} with $(s,e,C,\mu)$,
  omitting $\mu$ from superscripts, and let $\Exit$ and
  $\Live$ be the predicates of
  \Cref{def:index-bounded-rank-exit-live-witnesses} for
  $(\Gamma,\varphi,\chi,\ell,h,\calR_F,V_\calR)$ and the generated weights and
  distributions.  If
  \[
    \Gamma,\calR\vDash\bignd_{R=\ell}^{h}\varphi,
    \qquad
    \calR\otimes\calR_F\preceq\mu,
  \]
  and
  \[
    \varphi\sep\sure{R=\ell}\Rightarrow\sure{e\mapsto\fls},
    \qquad
    \varphi\sep\sure{R>\ell}\Rightarrow\sure{e\mapsto\tru},
  \]
  then
  \[
    \Exit(0)
    \qquad\text{and}\qquad
    \Live(0)
  \]
\end{lemma}
\begin{proof}
  Unfolding $\Gamma,\calR\vDash\bignd_{R=\ell}^{h}\varphi$ gives
  $\alpha\in\calD(E_{\ell,h})$ and resources
  $(\calR_n)_{n\in\supp(\alpha)}$ such that
  \[
    \bigoplus_{n\sim\alpha}\calR_n\preceq\calR,
    \qquad
    \Gamma[R\mapsto n],\calR_n\vDash\varphi
    \quad(n\in\supp(\alpha))
  \]
  Since $\bigoplus_{n\sim\alpha}\calR_n\preceq\calR$, the resource order gives
  $V_\oplus\subseteq V_\calR$ for the common direct-sum footprint $V_\oplus$.
  Thus $V_{\calR_n}\subseteq V_\calR$ for every $n\in\supp(\alpha)$.
  Since $\calR\otimes\calR_F\preceq\mu$,
  by refinement we can choose 
  normalized distributions $(\mu_n)_{n\in\supp(\alpha)}$ with
  \[
    \mu=\bigoplus_{n\sim\alpha}\mu_n,
    \qquad
    \calR_n\otimes\calR_F\preceq\mu_n
    \quad(n\in\supp(\alpha))
  \]
  For each $n\in\supp(\alpha)$, consider the $n$-indexed component
  $(\calR_n,\mu_n)$ of this decomposition.  We prove
  $\mu_n(G_{\fls})=1$ in the case $n=\ell$, and
  $\mu_n(G_{\tru})=1$ in the case $n>\ell$.

  If $n=\ell$, then
  \begin{align*}
    &\Gamma[R\mapsto n],\calR_n\vDash\varphi
      \land
      \Gamma[R\mapsto n],\UnitP\vDash\sure{R=\ell}
      \land
      \calR_n\otimes\UnitP\cong\calR_n\\
    &\Longrightarrow
      \Gamma[R\mapsto n],\calR_n
      \vDash
      \varphi\sep\sure{R=\ell}
      \tag{satisfaction clause for $\sep$}\\
    &\Longrightarrow
      \Gamma[R\mapsto n],\calR_n\vDash\sure{e\mapsto\fls}.
      \tag{$\varphi\sep\sure{R=\ell}\Rightarrow\sure{e\mapsto\fls}$}
  \end{align*}
  By \Cref{lem:index-assertion-frame-preservation},
  $\Gamma[R\mapsto n],\calR_n\otimes\calR_F
  \vDash\sure{e\mapsto\fls}$.  Applying
  \Cref{lem:index-guard-assertion-refinement-transfer} with
  $\calR=\calR_n\otimes\calR_F$, $\mu=\mu_n$, and $c=\fls$, using
  $\calR_n\otimes\calR_F\preceq\mu_n$, gives
  $\mu_n(G_{\fls})=1$.

  The case $n>\ell$ is identical, using $\sure{R>\ell}$ and
  $\varphi\sep\sure{R>\ell}\Rightarrow\sure{e\mapsto\tru}$, and gives
  $\mu_n(G_{\tru})=1$.

  The derivations above give, for every $n\in\supp(\alpha)$,
  \[
    n=\ell\Longrightarrow\mu_n(G_{\fls})=1,
    \qquad
    n>\ell\Longrightarrow\mu_n(G_{\tru})=1
  \]
  Since $G_{\fls}$ and $G_{\tru}$ partition $\Mem\times\bbN$, the opposite
  guard masses are zero.  Therefore
  \[
    w_0^{\exit}
    =
    \mu(G_{\fls})
    =
    \sum_{n\in\supp(\alpha)}
      \alpha(n)\cdot\mu_n(G_{\fls})
    =
    \alpha(\ell)
  \]
  and
  \[
    w_0^{\live}
    =
    \mu(G_{\tru})
    =
    \sum_{n\in\supp(\alpha)}
      \alpha(n)\cdot\mu_n(G_{\tru})
    =
    \sum_{n\in\supp(\alpha)\cap E_{\ell+1,h}}\alpha(n)
  \]

  If $w_0^{\exit}>0$, then $\alpha(\ell)>0$.  Conditioning
  $\mu=\bigoplus_{n\sim\alpha}\mu_n$ on $G_{\fls}$ keeps only the rank-$\ell$
  component, because $\mu_\ell(G_{\fls})=1$ and the other components have
  $G_{\fls}$-mass zero.  Therefore
  $\mu_0^{\exit}
  =
  \mu\mid G_{\fls}
  =
  \mu_\ell$.
  Since $\Gamma[R\mapsto\ell],\calR_\ell\vDash\varphi$ and
  $\chi=\varphi[\ell/R]$,
  \Cref{lem:index-assertion-logical-substitution} gives
  $\Gamma,\calR_\ell\vDash\chi$.
  Together with
  $\calR_\ell\otimes\calR_F\preceq\mu_\ell
  =
  \mu_0^{\exit}$ and $V_{\calR_\ell}\subseteq V_\calR$, this proves
  $\Exit(0)$, using
  $\calQ_0\triangleq\calR_\ell$.

  If $w_0^{\live}>0$, define
  $\alpha_0^{\live}\in\calD(E_{\ell,h})$ as the initial live branch
  distribution:
  \[
    \alpha_0^{\live}(n)
    \triangleq
    \begin{cases}
      \alpha(n)/w_0^{\live},
        & n\in\supp(\alpha)\cap E_{\ell+1,h},\\
      0,
        & \text{otherwise}.
    \end{cases}
  \]
  Conditioning $\mu=\bigoplus_{n\sim\alpha}\mu_n$ on $G_{\tru}$ normalizes
  precisely the components with $n>\ell$, because the rank-$\ell$ components
  have zero $G_{\tru}$-mass and the components with $n>\ell$ have
  $G_{\tru}$-mass one.  Therefore
  $\mu_0^{\live}
  =
  \mu\mid G_{\tru}
  =
  \bigoplus_{n\sim\alpha_0^{\live}}\mu_n$.
  Set
  \[
    A_0\triangleq\supp(\alpha)\cap E_{\ell+1,h},
    \qquad
    \theta_0\triangleq\alpha_0^{\live},
    \qquad
    r_0(n)\triangleq n,
    \qquad
    \calR_{0,n}\triangleq\calR_n,
    \qquad
    \mu_{0,n}\triangleq\mu_n
  \]
  Then $r_0:\supp(\theta_0)\to\{n\in\bbN\mid\ell<n\le h\}$, the displayed
  equality for $\mu_0^{\live}$ gives
  $\mu_0^{\live}
  =
  \bigoplus_{n\sim\theta_0}\mu_{0,n}$, and the branch facts give
  \[
    \Gamma[R\mapsto r_0(n)],\calR_{0,n}\vDash\varphi,
    \qquad
    \calR_{0,n}\otimes\calR_F\preceq\mu_{0,n},
    \qquad
    V_{\calR_{0,n}}\subseteq V_\calR
    \quad(n\in\supp(\theta_0))
  \]
  This proves $\Live(0)$.
\end{proof}

\begin{lemma}[One-step live component rank-progress decomposition]
  \label{lem:index-live-component-rank-progress-decomposition}
  \label{lem:index-bounded-rank-live-component-one-step-decrease}
  Fix an assertion $\varphi$, a mode $m$, a schedule $s$, a guard $e$, a
  command $C$, a logical environment $\Gamma$, a frame $\calR_F$, an input
  distribution $\mu\in\calD(\Mem\times\bbN)$, integers $\ell\le h$, a finite
  footprint $V_b$, and $p\in(0,1]$.  For integers $a\le b$, write
  $E_{a,b}\triangleq\{a,a+1,\ldots,b\}$, and let
  $\chi\triangleq\varphi[\ell/R]$.

  Instantiate the normalized exit/live decomposition of
  \Cref{def:index-normalized-exit-live-decomposition} with $(s,e,C,\mu)$,
  omitting $\mu$ from superscripts, and let $\Live$ be the live predicate from
  \Cref{def:index-bounded-rank-exit-live-witnesses} for
  $(\Gamma,\varphi,\chi,\ell,h,\calR_F,V_b)$ and the generated weights and
  distributions.  Suppose
  \[
    \vdash_m
    \triple
      {\varphi\sep\sure{R=N>\ell}}
      {C}
      {
        \left(\bignd_{R=\ell}^{N-1}\varphi\right)
        \oplus_{\ge p}
        \left(\bignd_{R=N}^{h}\varphi\right)
      },
    \qquad
    N\notin\vars(\varphi),
    \qquad
    \Stable{m}{\calR_F}
  \]
  Fix $i\in\bbN$ with $w_i^{\live}>0$, and suppose $\Live(i)$.  Choose
  witnesses
  $A_i,\theta_i,r_i,
  (\calR_{i,a})_{a\in\supp(\theta_i)}$ and
  $(\mu_{i,a})_{a\in\supp(\theta_i)}$ such that
  \[
    \theta_i\in\calD(A_i),
    \qquad
    r_i:\supp(\theta_i)\to E_{\ell+1,h},
    \qquad
    \mu_i^{\live}
    =
    \bigoplus_{a\sim\theta_i}\mu_{i,a},
  \]
  and, for every $a\in\supp(\theta_i)$,
  \[
    \Gamma[R\mapsto r_i(a)],\calR_{i,a}\vDash\varphi,
    \qquad
    \calR_{i,a}\otimes\calR_F\preceq\mu_{i,a},
    \qquad
    V_{\calR_{i,a}}\subseteq V_b
  \]
  Then, for every $a\in\supp(\theta_i)$, putting
  $n_a\triangleq r_i(a)$, there exist a countable set $B_{i,a}$, a
  distribution $\beta_{i,a}\in\calD(B_{i,a})$, a rank map
  $r'_{i,a}:B_{i,a}\to E_{\ell,h}$, resources
  $(\calR'_{i,a,b})_{b\in\supp(\beta_{i,a})}$, and distributions
  $(\xi_{i,a,b})_{b\in\supp(\beta_{i,a})}$ such that
  $\de{C}_s^\dagger(\mu_{i,a})
  =
  \bigoplus_{b\sim\beta_{i,a}}\xi_{i,a,b}$
  and
  \begin{align*}
    &\Gamma[R\mapsto r'_{i,a}(b)],\calR'_{i,a,b}\vDash\varphi,
    \qquad
    \calR'_{i,a,b}\otimes\calR_F\preceq\xi_{i,a,b}
    \quad(b\in\supp(\beta_{i,a})),\\
    &V_{\calR'_{i,a,b}}\subseteq V_{\calR_{i,a}},\\
    &\beta_{i,a}(\{b\mid r'_{i,a}(b)<n_a\})\ge p,
    \qquad
    \beta_{i,a}(\{b\mid \ell\le r'_{i,a}(b)\le h\})=1
  \end{align*}
\end{lemma}
\begin{proof}
  Fix $a\in\supp(\theta_i)$ and put $n_a\triangleq r_i(a)$.  Since
  $r_i:\supp(\theta_i)\to E_{\ell+1,h}$, we have $n_a>\ell$.
  Let
  $\Gamma_{i,a}\triangleq\Gamma[N\mapsto n_a][R\mapsto n_a]$.
  Since $N\notin\vars(\varphi)$, the environments
  $\Gamma[R\mapsto n_a]$ and $\Gamma_{i,a}$ agree on $\fv(\varphi)$.
  The premise precondition for component $a$ is obtained by the derivation
  \begin{align*}
    &\Gamma[R\mapsto n_a],\calR_{i,a}\vDash\varphi
    \quad\Longrightarrow\quad
    \Gamma_{i,a},\calR_{i,a}\vDash\varphi
    \tag{context irrelevance}\\
    &\Gamma_{i,a}(R)=\Gamma_{i,a}(N)=n_a
    \;\land\;
    n_a>\ell
    \quad\Longrightarrow\quad
    \Gamma_{i,a},\UnitP\vDash\sure{R=N>\ell}
    \tag{satisfaction clause for $\sure{-}$}\\
    &\Gamma_{i,a},\calR_{i,a}\vDash\varphi
    \;\land\;
    \Gamma_{i,a},\UnitP\vDash\sure{R=N>\ell}
    \;\land\;
    \calR_{i,a}\otimes\UnitP\cong\calR_{i,a}\\
    &\quad\Longrightarrow\quad
    \Gamma_{i,a},\calR_{i,a}
    \vDash
    \varphi\sep\sure{R=N>\ell}.
    \tag{definition of $\UnitP$ and satisfaction clause for $\sep$}
  \end{align*}

  Put
  \[
    \psi_{i,a}^{<}\triangleq\bignd_{R=\ell}^{n_a-1}\varphi,
    \qquad
    \psi_{i,a}^{\ge}\triangleq\bignd_{R=n_a}^{h}\varphi
  \]
  Apply semantic soundness of the premise triple with environment
  $\Gamma_{i,a}$, input resource $\calR_{i,a}$, frame $\calR_F$, and input
  distribution $\mu_{i,a}$.  The precondition is the judgment derived above,
  and the framed refinement is
  $\calR_{i,a}\otimes\calR_F\preceq\mu_{i,a}$.  The frame side condition is
  $\Stable{m}{\calR_F}$.  This gives a resource
  $\calQ_{i,a}$ such that
  \[
    \Gamma_{i,a},\calQ_{i,a}
    \vDash
    \psi_{i,a}^{<}\oplus_{\ge p}\psi_{i,a}^{\ge},
    \qquad
    \calQ_{i,a}\otimes\calR_F
    \preceq
    \de{C}_s^\dagger(\mu_{i,a})
  \]
  Unfolding $\oplus_{\ge p}$ gives $q_{i,a}\in[p,1]$ and resources
  $\calQ_{i,a}^{<},\calQ_{i,a}^{\ge}$ such that
  \[
    \Gamma_{i,a},\calQ_{i,a}^{<}\vDash\psi_{i,a}^{<},
    \qquad
    \Gamma_{i,a},\calQ_{i,a}^{\ge}\vDash\psi_{i,a}^{\ge},
    \qquad
    \calQ_{i,a}^{<}\oplus_{q_{i,a}}\calQ_{i,a}^{\ge}
    \preceq
    \calQ_{i,a}
  \]
  Moving the fixed frame through this probabilistic split gives
  \begin{align*}
    (\calQ_{i,a}^{<}\otimes\calR_F)
    \oplus_{q_{i,a}}
    (\calQ_{i,a}^{\ge}\otimes\calR_F)
    &\cong
    (\calQ_{i,a}^{<}\oplus_{q_{i,a}}\calQ_{i,a}^{\ge})
    \otimes\calR_F
    \tag{distributivity}\\
    &\preceq
    \calQ_{i,a}\otimes\calR_F
    \tag{monotonicity}\\
    &\preceq
    \de{C}_s^\dagger(\mu_{i,a})
  \end{align*}
  The displayed partial-distribution refinement has a proper target by
  definition, so $\de{C}_s^\dagger(\mu_{i,a})$ is proper.  We use the same
  symbol for its ordinary coercion.
  Hence by refinement we can choose
  normalized distributions $\nu_{i,a}^{<}$ and $\nu_{i,a}^{\ge}$ such
  that
  \[
    \de{C}_s^\dagger(\mu_{i,a})
    =
    q_{i,a}\cdot\nu_{i,a}^{<}
    +(1-q_{i,a})\cdot\nu_{i,a}^{\ge},
    \qquad
    \calQ_{i,a}^{<}\otimes\calR_F\preceq\nu_{i,a}^{<},
    \qquad
    \calQ_{i,a}^{\ge}\otimes\calR_F\preceq\nu_{i,a}^{\ge}
  \]

  Put $E_{i,a}^{<}\triangleq E_{\ell,n_a-1}$ and
  $E_{i,a}^{\ge}\triangleq E_{n_a,h}$.  Applying
  \Cref{lem:index-framed-bounded-rank-choice-decomposition} to
  $\psi_{i,a}^{\kappa}$ for each $\kappa\in\{<,\ge\}$ generates rank
  distributions $\beta_{i,a}^{\kappa}\in\calD(E_{i,a}^{\kappa})$, active
  resources
  $(\calR_{i,a,m}^{\kappa})_{m\in\supp(\beta_{i,a}^{\kappa})}$, and concrete
  branch distributions
  $(\xi_{i,a,m}^{\kappa})_{m\in\supp(\beta_{i,a}^{\kappa})}$ such that
  \begin{align*}
    \nu_{i,a}^{\kappa}
    &=
    \bigoplus_{m\sim\beta_{i,a}^{\kappa}}\xi_{i,a,m}^{\kappa}
    \quad(\kappa\in\{<,\ge\}),\\
    \Gamma[R\mapsto m],\calR_{i,a,m}^{\kappa}
    &\vDash\varphi,
    \qquad
    \calR_{i,a,m}^{\kappa}\otimes\calR_F
    \preceq\xi_{i,a,m}^{\kappa}
    \quad(\kappa\in\{<,\ge\},\;
    m\in\supp(\beta_{i,a}^{\kappa}))
  \end{align*}
  In both applications, using $N\notin\vars(\varphi)$ drops the irrelevant
  $N$-assignment from branch judgments.  The generated branch resources have
  footprints contained in $V_{\calR_{i,a}}$.
  Applying \Cref{lem:index-tagged-rank-progress-recombination} with
  $n=n_a$, $q=q_{i,a}$, logical environment $\Gamma$, and the two
  decompositions just obtained yields witnesses
  $B_{i,a}$,
  $\beta_{i,a}\in\calD(B_{i,a})$,
  $r'_{i,a}:B_{i,a}\to E_{\ell,h}$,
  $(\calR'_{i,a,b})_{b\in\supp(\beta_{i,a})}$, and
  $(\xi_{i,a,b})_{b\in\supp(\beta_{i,a})}$
  such that
  $\de{C}_s^\dagger(\mu_{i,a})
  =
  \bigoplus_{b\sim\beta_{i,a}}\xi_{i,a,b}$ and
  \begin{align*}
    &\Gamma[R\mapsto r'_{i,a}(b)],\calR'_{i,a,b}\vDash\varphi,
    \qquad
    \calR'_{i,a,b}\otimes\calR_F\preceq\xi_{i,a,b}
    \quad(b\in\supp(\beta_{i,a})),\\
    &V_{\calR'_{i,a,b}}\subseteq V_{\calR_{i,a}},\\
    &\beta_{i,a}(\{b\mid r'_{i,a}(b)<n_a\})\ge p,
    \qquad
    \beta_{i,a}(\{b\mid \ell\le r'_{i,a}(b)\le h\})=1
  \end{align*}
\end{proof}

\begin{lemma}[Successor bounded-rank exit/live witnesses]
  \label{lem:index-bounded-rank-exit-live-successor}
  Fix an assertion $\varphi$, a mode $m$, a schedule $s$, a guard $e$, a
  command $C$, a logical environment $\Gamma$, a frame $\calR_F$, an input
  distribution $\mu\in\calD(\Mem\times\bbN)$, integers $\ell\le h$, a finite
  footprint $V_b$, and $p\in(0,1]$.  For integers $a\le b$, write
  $E_{a,b}\triangleq\{a,a+1,\ldots,b\}$, and let
  $\chi\triangleq\varphi[\ell/R]$.

  Instantiate the normalized exit/live decomposition of
  \Cref{def:index-normalized-exit-live-decomposition} with $(s,e,C,\mu)$,
  omitting $\mu$ from superscripts, and let $\Exit$ and $\Live$ be the
  predicates of \Cref{def:index-bounded-rank-exit-live-witnesses} for
  $(\Gamma,\varphi,\chi,\ell,h,\calR_F,V_b)$ and the generated weights and
  distributions.  Suppose
  \[
    \vdash_m
    \triple
      {\varphi\sep\sure{R=N>\ell}}
      {C}
      {
        \left(\bignd_{R=\ell}^{N-1}\varphi\right)
        \oplus_{\ge p}
        \left(\bignd_{R=N}^{h}\varphi\right)
      },
    \qquad
    N\notin\vars(\varphi),
    \qquad
    \Stable{m}{\calR_F}
  \]
  \[
    \varphi\sep\sure{R=\ell}\Rightarrow\sure{e\mapsto\fls},
    \qquad
    \varphi\sep\sure{R>\ell}\Rightarrow\sure{e\mapsto\tru},
  \]
  and
  $\convex{\chi}$.
  Then, for every $i\in\bbN$,
  \[
    \Live(i)
    \Longrightarrow
    \Exit(i+1)\land\Live(i+1)
  \]
\end{lemma}
\begin{proof}
    Fix $i\in\bbN$ and assume $\Live(i)$.  We prove
    $\Exit(i+1)$ and $\Live(i+1)$.

    If $w_i^{\live}=0$, then the recursive definitions give
    $w_{i+1}^{\exit}=w_{i+1}^{\live}=0$, so both conclusions
    at $i+1$ have false antecedent.  It remains to consider the case
    $w_i^{\live}>0$.

    \emph{Setup.}
    By $\Live(i)$, choose
    $A_i,\theta_i,r_i,(\calR_{i,a})_{a\in\supp(\theta_i)}$, and
    $(\mu_{i,a})_{a\in\supp(\theta_i)}$ such that
    \[
      \theta_i\in\calD(A_i),
      \qquad
      r_i:\supp(\theta_i)\to E_{\ell+1,h},
      \qquad
      \mu_i^{\live}
      =
      \bigoplus_{a\sim\theta_i}\mu_{i,a},
    \]
    and, for every $a\in\supp(\theta_i)$,
    \[
      \Gamma[R\mapsto r_i(a)],\calR_{i,a}\vDash\varphi,
      \qquad
      \calR_{i,a}\otimes\calR_F\preceq\mu_{i,a},
      \qquad
      V_{\calR_{i,a}}\subseteq V_b
    \]
    The recursive definitions give
    $w_j^{\live}\ge0$ for every $j$.
    By \Cref{lem:index-positive-live-weight-conditioning-event}, applied
    with
    \[
      X=\Mem\times\bbN,\qquad
      G=G_{\tru},\qquad
      w_j=w_j^{\live},\qquad
      \mu_j=\mu_j^{\live},\qquad
      \rho_j=\de{C}_s^\dagger(\mu_j^{\live}),
    \]
    the assumption $w_i^{\live}>0$ gives
    \[
      \mu_i^{\live}(G_{\tru})=1
    \]
    Since $G_{\fls}$ and $G_{\tru}$ partition $\Mem\times\bbN$, this implies
    \[
      \mu_i^{\live}(G_{\fls})
      =
      1-\mu_i^{\live}(G_{\tru})
      =
      0
    \]

    By \Cref{lem:index-live-component-rank-progress-decomposition}, for
    every $a\in\supp(\theta_i)$, putting $n_a\triangleq r_i(a)$, there exist
    a countable set $B_{i,a}$, a distribution
    $\beta_{i,a}\in\calD(B_{i,a})$, a rank map
    $r'_{i,a}:B_{i,a}\to E_{\ell,h}$, resources
    $(\calR'_{i,a,b})_{b\in\supp(\beta_{i,a})}$, and distributions
    $(\xi_{i,a,b})_{b\in\supp(\beta_{i,a})}$ such that
    \begin{equation}
      \de{C}_s^\dagger(\mu_{i,a})
      =
      \bigoplus_{b\sim\beta_{i,a}}\xi_{i,a,b},
      \label{eq:index-bounded-rank-component-output}
    \end{equation}
    and
    \begin{align*}
      &\Gamma[R\mapsto r'_{i,a}(b)],\calR'_{i,a,b}\vDash\varphi,
      \qquad
      \calR'_{i,a,b}\otimes\calR_F\preceq\xi_{i,a,b}
      \quad(b\in\supp(\beta_{i,a})),\\
      &V_{\calR'_{i,a,b}}\subseteq V_{\calR_{i,a}},\\
      &\beta_{i,a}(\{b\mid r'_{i,a}(b)<n_a\})\ge p,
      \qquad
      \beta_{i,a}(\{b\mid \ell\le r'_{i,a}(b)\le h\})=1
    \end{align*}

    To classify the one-step outputs by their new ranks, keep both the live
    component label $a$ and the output branch label $b$ in a single index.
    Define the pair-index set $B_i$ and its mixture weight $\zeta_i$ by
    \[
      B_i\triangleq
      \{(a,b)\mid a\in\supp(\theta_i),\;
      b\in\supp(\beta_{i,a})\},
      \qquad
      \zeta_i(a,b)\triangleq
      \theta_i(a)\cdot\beta_{i,a}(b)
    \]
    Then $\zeta_i\in\calD(B_i)$.
    For $(a,b)\in B_i$, write
    $r'_{i}(a,b)\triangleq r'_{i,a}(b)$.
    For every $(a,b)\in B_i$, we get
    \[
    \begin{aligned}
      &\Gamma[R\mapsto r'_i(a,b)],\calR'_{i,a,b}\vDash\varphi,
      \qquad
      \calR'_{i,a,b}\otimes\calR_F\preceq\xi_{i,a,b},\\
      &V_{\calR'_{i,a,b}}\subseteq V_b .
    \end{aligned}
    \]
    The guard side conditions, frame preservation, and
    \Cref{lem:index-guard-assertion-refinement-transfer} give
    \begin{equation}
      r'_i(a,b)=\ell
      \Longrightarrow
      \xi_{i,a,b}(G_{\fls})=1,
      \qquad
      r'_i(a,b)>\ell
      \Longrightarrow
      \xi_{i,a,b}(G_{\tru})=1 .
      \label{eq:index-bounded-rank-guard-classification}
    \end{equation}

    Expanding the live mixture and then each component output gives
    \begin{align}
      \de{C}_s^\dagger(\mu_i^{\live})
      &=
      \de{C}_s^\dagger
      \left(
        \bigoplus_{a\sim\theta_i}\mu_{i,a}
      \right)
      \tag{by $\Live(i)$ that $\mu_i^{\live}
        =\bigoplus_{a\sim\theta_i}\mu_{i,a}$}\\
      &=
      \bigoplus_{a\sim\theta_i}
        \de{C}_s^\dagger(\mu_{i,a})
      \tag{\Cref{lem:index-dbot-kleisli-linearity}}\\
      &=
      \bigoplus_{a\sim\theta_i}
        \left(\bigoplus_{b\sim\beta_{i,a}}\xi_{i,a,b}\right)
      \tag{by \eqref{eq:index-bounded-rank-component-output}}\\
      &=
      \bigoplus_{(a,b)\sim\zeta_i}\xi_{i,a,b}.
      \tag{definition of $B_i$ and $\zeta_i$}
    \end{align}
    Hence
    \begin{equation}
      \de{C}_s^\dagger(\mu_i^{\live})
      =
      \bigoplus_{(a,b)\sim\zeta_i}\xi_{i,a,b}.
      \label{eq:index-bounded-rank-flattened-output}
    \end{equation}

    Define $B_i^{\exit}$ as the flattened branches whose output rank is
    $\ell$, and define $B_i^{\live}$ as the flattened branches whose
    output rank is still above $\ell$:
    \[
      B_i^{\exit}\triangleq
      \{(a,b)\in B_i\mid r'_i(a,b)=\ell\},
      \qquad
      B_i^{\live}\triangleq
      \{(a,b)\in B_i\mid r'_i(a,b)>\ell\}
    \]
    Since $r'_i(a,b)\in E_{\ell,h}$ for every $(a,b)\in B_i$,
    $B_i^{\exit}$ and $B_i^{\live}$ partition $B_i$.
    Therefore \eqref{eq:index-bounded-rank-guard-classification} gives
    \begin{equation}
      \xi_{i,a,b}(G_{\fls})
      =
      \begin{cases}
        1, &(a,b)\in B_i^{\exit},\\
        0, &(a,b)\in B_i^{\live} .
      \end{cases}
      \qquad
      \xi_{i,a,b}(G_{\tru})
      =
      \begin{cases}
        0, &(a,b)\in B_i^{\exit},\\
        1, &(a,b)\in B_i^{\live} .
      \end{cases}
      \label{eq:index-bounded-rank-branch-false-guard-mass}
    \end{equation}
    where the second equality follows because $G_{\fls}$ and $G_{\tru}$
    partition $\Mem\times\bbN$.

    \emph{Exit case.}
    If $w_{i+1}^{\exit}>0$, then the recursive definition of
    $w_{i+1}^{\exit}$ gives
    $\left(\de{C}_s^\dagger(\mu_i^{\live})\right)(G_{\fls})
    >0$.
    The false-guard mass of the one-step output is exactly the total
    $\zeta_i$-weight of the exit branches:
    \begin{align*}
      \left(\de{C}_s^\dagger(\mu_i^{\live})\right)(G_{\fls})
      &=
      \left(
        \bigoplus_{(a,b)\sim\zeta_i}\xi_{i,a,b}
      \right)(G_{\fls})
      \tag{by \eqref{eq:index-bounded-rank-flattened-output}}\\
      &=
      \sum_{(a,b)\in B_i}
        \zeta_i(a,b)\cdot\xi_{i,a,b}(G_{\fls})\\
      &=
      \sum_{(a,b)\in B_i^{\exit}}\zeta_i(a,b).
      \tag{by \eqref{eq:index-bounded-rank-branch-false-guard-mass}}
    \end{align*}
    Define a distribution
    $\zeta_i^{\exit}\in\calD(B_i^{\exit})$ normalized by
    $\left(\de{C}_s^\dagger(\mu_i^{\live})\right)(G_{\fls})$ as
    $\zeta_i^{\exit}(a,b)
    \triangleq
    \frac{\zeta_i(a,b)}
    {\left(\de{C}_s^\dagger(\mu_i^{\live})\right)(G_{\fls})}$.
    We now prove the claimed conditional distribution pointwise.  The
    conditioning denominator equals the total $\zeta_i$-weight of
    $B_i^{\exit}$, as shown above.  If $x\in G_{\fls}$, then
    \begin{align*}
      \left(
        \de{C}_s^\dagger(\mu_i^{\live})\mid G_{\fls}
      \right)(x)
      &=
      \frac{
        \left(\de{C}_s^\dagger(\mu_i^{\live})\right)(x)}
      {\left(\de{C}_s^\dagger(\mu_i^{\live})\right)(G_{\fls})}
      \tag{definition of conditioning}\\
      &=
      \frac{
        \sum_{(a,b)\in B_i}
          \zeta_i(a,b)\cdot\xi_{i,a,b}(x)}
      {\left(\de{C}_s^\dagger(\mu_i^{\live})\right)(G_{\fls})}
      \tag{by \eqref{eq:index-bounded-rank-flattened-output}}\\
      &=
      \frac{
        \sum_{(a,b)\in B_i^{\exit}}
          \zeta_i(a,b)\cdot\xi_{i,a,b}(x)}
      {\left(\de{C}_s^\dagger(\mu_i^{\live})\right)(G_{\fls})}
      \tag{by \eqref{eq:index-bounded-rank-branch-false-guard-mass}
      and $x\in G_{\fls}$}\\
      &=
      \sum_{(a,b)\in B_i^{\exit}}
        \zeta_i^{\exit}(a,b)\cdot\xi_{i,a,b}(x).
      \tag{definition of $\zeta_i^{\exit}$}
    \end{align*}
    If $x\notin G_{\fls}$, then the left-hand side is $0$ by the definition
    of conditioning, and the right-hand side is also $0$ because each
    $(a,b)\in B_i^{\exit}$ satisfies
    $\xi_{i,a,b}(G_{\fls})=1$.  Therefore
    \begin{equation}
      \mu_{i+1}^{\exit}
      =
      \de{C}_s^\dagger(\mu_i^{\live})\mid G_{\fls}
      =
      \bigoplus_{(a,b)\sim\zeta_i^{\exit}}\xi_{i,a,b}.
      \label{eq:index-bounded-rank-exit-conditional-decomposition}
    \end{equation}
    For every $(a,b)\in\supp(\zeta_i^{\exit})$,
    \begin{align*}
      (a,b)\in\supp(\zeta_i^{\exit})
      &\Longrightarrow
      (a,b)\in B_i^{\exit}
      \tag{definition of $\zeta_i^{\exit}$}\\
      &\Longrightarrow
      r'_i(a,b)=\ell
      \tag{definition of $B_i^{\exit}$}\\
      &\Longrightarrow
      \Gamma[R\mapsto\ell],\calR'_{i,a,b}\vDash\varphi
      \tag{branch judgment for $\calR'_{i,a,b}$}\\
      &\Longrightarrow
      \Gamma,\calR'_{i,a,b}\vDash\varphi[\ell/R]=\chi .
      \tag{\Cref{lem:index-assertion-logical-substitution}}
    \end{align*}
    Define $\calQ_{i+1}$ as the normalized mixture of the exit-branch
    resources:
    $\calQ_{i+1}
    \triangleq
    \bigoplus_{(a,b)\sim\zeta_i^{\exit}}
    \calR'_{i,a,b}$.
    Let $J$ be a fresh logical index whose values are pairs $(a,b)$.  By
    \Cref{lem:index-probabilistic-choice-idempotence} and $\convex{\chi}$, we
    get that
    \begin{align*}
      \forall(a,b)\in\supp(\zeta_i^{\exit}).\;
      \Gamma,\calR'_{i,a,b}\vDash\chi
      \Longrightarrow
      \Gamma,\calQ_{i+1}
      \vDash
      \bigoplus^{\Leak}_{J\sim\zeta_i^{\exit}}\chi
      \Longrightarrow
      \Gamma,\calQ_{i+1}\vDash\chi 
      \tag{$J\notin\fv(\chi)$}
    \end{align*}
    Moreover, by \eqref{eq:index-bounded-rank-exit-conditional-decomposition},
    \[
      \calQ_{i+1}\otimes\calR_F
      =
      \left(
        \bigoplus_{(a,b)\sim\zeta_i^{\exit}}
        \calR'_{i,a,b}
      \right)\otimes\calR_F
      \cong
      \bigoplus_{(a,b)\sim\zeta_i^{\exit}}
      \left(\calR'_{i,a,b}\otimes\calR_F\right)
      \preceq
      \bigoplus_{(a,b)\sim\zeta_i^{\exit}}\xi_{i,a,b}
      =
      \mu_{i+1}^{\exit}
    \]
    The same branch footprint bound gives
    $V_{\calQ_{i+1}}\subseteq V_b$.
    This proves $\Exit(i+1)$ when $w_{i+1}^{\exit}>0$.

    \emph{Live case.}
    If $w_{i+1}^{\live}>0$, then
    $\left(\de{C}_s^\dagger(\mu_i^{\live})\right)(G_{\tru})>0$.
    The goal is to prove $\Live(i+1)$, i.e. to exhibit witnesses
    satisfying
    \[
      \begin{array}{l}
      \exists A_{i+1},\theta_{i+1},r_{i+1},
      (\calR_{i+1,a})_{a\in\supp(\theta_{i+1})},
      (\mu_{i+1,a})_{a\in\supp(\theta_{i+1})}.\\[0.2em]
      \qquad
      \theta_{i+1}\in\calD(A_{i+1})
      \;\land\;
      r_{i+1}:\supp(\theta_{i+1})\to E_{\ell+1,h}\\
      \qquad
      {}\land\;
      \mu_{i+1}^{\live}
      =
      \bigoplus_{a\sim\theta_{i+1}}\mu_{i+1,a}\\
      \qquad
      {}\land\;
      \forall a\in\supp(\theta_{i+1}).\;
      \Gamma[R\mapsto r_{i+1}(a)],\calR_{i+1,a}\vDash\varphi
      \;\land\;
      \calR_{i+1,a}\otimes\calR_F\preceq\mu_{i+1,a}
      \\
      \qquad{}\land\;
      V_{\calR_{i+1,a}}\subseteq V_b.
      \end{array}
    \]
    The live witness for time $i+1$ is indexed by those pair-indices whose
    new rank is still above $\ell$, namely $B_i^{\live}$.  We normalize
    $\zeta_i$ on this set and reuse the corresponding branch resources and
    distributions:
    \[
      A_{i+1}\triangleq B_i^{\live},
      \qquad
      \theta_{i+1}(a,b)
      \triangleq
      \frac{\zeta_i(a,b)}
      {\left(\de{C}_s^\dagger(\mu_i^{\live})\right)(G_{\tru})},
      \qquad
      r_{i+1}(a,b)\triangleq r'_i(a,b),
    \]
    and set
    $\calR_{i+1,(a,b)}\triangleq\calR'_{i,a,b}$ and
    $\mu_{i+1,(a,b)}\triangleq\xi_{i,a,b}$.

    First, $\theta_{i+1}$ is a probability distribution on $A_{i+1}$.  By
    \eqref{eq:index-bounded-rank-flattened-output} and
    \eqref{eq:index-bounded-rank-branch-false-guard-mass}, its total mass is
      \[
        \sum_{(a,b)\in A_{i+1}}\theta_{i+1}(a,b)
        =
        \frac{
          \sum_{(a,b)\in B_i^{\live}}\zeta_i(a,b)}
        {\left(\de{C}_s^\dagger(\mu_i^{\live})\right)(G_{\tru})}
        =
        \frac{
          \left(\de{C}_s^\dagger(\mu_i^{\live})\right)(G_{\tru})}
        {\left(\de{C}_s^\dagger(\mu_i^{\live})\right)(G_{\tru})}
        =
        1
      \]
    Nonnegativity follows from nonnegativity of $\zeta_i$.  Hence
    $\theta_{i+1}\in\calD(A_{i+1})$.

    Second, the rank map has the required codomain.  For every
    $(a,b)\in\supp(\theta_{i+1})$,
    \begin{align*}
      (a,b)\in\supp(\theta_{i+1})
      &\Longrightarrow
      (a,b)\in A_{i+1}=B_i^{\live}
      \tag{definition of $\theta_{i+1}$}\\
      &\Longrightarrow
      r'_i(a,b)>\ell
      \tag{definition of $B_i^{\live}$}\\
      &\Longrightarrow
      r_{i+1}(a,b)=r'_i(a,b)\in E_{\ell+1,h}.
      \tag{$r'_i(a,b)\in E_{\ell,h}$ and definition of $r_{i+1}$}
    \end{align*}
    Thus $r_{i+1}:\supp(\theta_{i+1})\to E_{\ell+1,h}$.

    Third, the live distribution is the required indexed mixture.  We prove
    the equality pointwise.  If $x\in G_{\tru}$, then
    \begin{align*}
      \left(
        \de{C}_s^\dagger(\mu_i^{\live})\mid G_{\tru}
      \right)(x)
      &=
      \frac{
        \left(\de{C}_s^\dagger(\mu_i^{\live})\right)(x)}
      {\left(\de{C}_s^\dagger(\mu_i^{\live})\right)(G_{\tru})}
      \tag{definition of conditioning}\\
      &=
      \frac{
        \sum_{(a,b)\in B_i}
          \zeta_i(a,b)\cdot\xi_{i,a,b}(x)}
      {\left(\de{C}_s^\dagger(\mu_i^{\live})\right)(G_{\tru})}
      \tag{by \eqref{eq:index-bounded-rank-flattened-output}}\\
      &=
      \frac{
        \sum_{(a,b)\in B_i^{\live}}
          \zeta_i(a,b)\cdot\xi_{i,a,b}(x)}
      {\left(\de{C}_s^\dagger(\mu_i^{\live})\right)(G_{\tru})}
      \tag{by \eqref{eq:index-bounded-rank-branch-false-guard-mass}
      and $x\in G_{\tru}$}\\
      &=
      \sum_{(a,b)\in A_{i+1}}
        \theta_{i+1}(a,b)\cdot\mu_{i+1,(a,b)}(x).
      \tag{definitions of $A_{i+1}$, $\theta_{i+1}$, and
      $\mu_{i+1,(a,b)}$}
    \end{align*}
    If $x\notin G_{\tru}$, then the left-hand side is $0$ by the definition
    of conditioning, and the right-hand side is also $0$ because each
    $(a,b)\in A_{i+1}=B_i^{\live}$ satisfies
    $\xi_{i,a,b}(G_{\tru})=1$.  Therefore
    \begin{equation}
      \mu_{i+1}^{\live}
      =
      \de{C}_s^\dagger(\mu_i^{\live})\mid G_{\tru}
      =
      \bigoplus_{(a,b)\sim\theta_{i+1}}\mu_{i+1,(a,b)}.
      \label{eq:index-bounded-rank-live-conditional-decomposition}
    \end{equation}

    Fourth, each supported live branch has the required assertion judgment and
    framed refinement. By definition of $r_{i + 1}$ and $\calR_{i+1, (a, b)}$, for every $(a,b)\in\supp(\theta_{i+1})$
    \begin{align*}
      \Gamma[R\mapsto r'_i(a,b)],\calR'_{i,a,b}\vDash\varphi
      \Longrightarrow
      \Gamma[R\mapsto r_{i+1}(a,b)],
      \calR_{i+1,(a,b)}\vDash\varphi 
    \end{align*}
    and
    \[
      \calR_{i+1,(a,b)}\otimes\calR_F
      =
      \calR'_{i,a,b}\otimes\calR_F
      \preceq
      \xi_{i,a,b}
      =
      \mu_{i+1,(a,b)}
    \]
    Finally, $V_{\calR_{i+1,(a,b)}}=V_{\calR'_{i,a,b}}\subseteq V_b$.
    The five displayed checks prove $\Live(i+1)$ when $w_{i+1}^{\live}>0$.
\end{proof}

\begin{lemma}[Bounded-rank block decay]
  \label{lem:index-bounded-rank-block-decay}
  Fix an assertion $\varphi$, a mode $m$, a schedule $s$, a guard $e$, a
  command $C$, a logical environment $\Gamma$, a frame $\calR_F$, an input
  distribution $\mu\in\calD(\Mem\times\bbN)$, integers $\ell\le h$, a finite
  footprint $V_b$, and $p\in(0,1]$.  Let $H\triangleq h-\ell$.  For integers $a\le b$, write
  $E_{a,b}\triangleq\{a,a+1,\ldots,b\}$, and let
  $\chi\triangleq\varphi[\ell/R]$.

  Instantiate the normalized exit/live decomposition of
  \Cref{def:index-normalized-exit-live-decomposition} with $(s,e,C,\mu)$,
  omitting $\mu$ from superscripts, and let $\Live$ be the live predicate from
  \Cref{def:index-bounded-rank-exit-live-witnesses} for
  $(\Gamma,\varphi,\chi,\ell,h,\calR_F,V_b)$ and the generated weights and
  distributions.  Suppose
  \[
    \vdash_m
    \triple
      {\varphi\sep\sure{R=N>\ell}}
      {C}
      {
        \left(\bignd_{R=\ell}^{N-1}\varphi\right)
        \oplus_{\ge p}
        \left(\bignd_{R=N}^{h}\varphi\right)
      },
    \qquad
    N\notin\vars(\varphi),
    \qquad
    \Stable{m}{\calR_F},
  \]
  and suppose the live invariant $\Live(j)$ holds at every index $j \in \bbN$.
  Then, for every $i\in\bbN$,
  \[
    H>0\land w_i^{\live}>0
    \Longrightarrow
    w_{i+H}^{\live}
    \le
    (1-p^H)\cdot w_i^{\live}
  \]
  Intuitively, a live execution carries a rank in the interval
  $\{\ell+1,\ldots,h\}$, whose height is $H = h - \ell$.  Each iteration has probability
  at least $p$ of strictly decreasing this rank; hence, over any block of
  $H$ live iterations, the run has probability at least $p^H$ of reaching the
  exit rank $\ell$.
\end{lemma}
\begin{proof}
  Fix $i$ with $w_i^{\live}>0$ and assume $H>0$.  The goal is
  \begin{equation}
    w_{i+H}^{\live}
    \le
    (1-p^H)\cdot w_i^{\live} .
    \tag{BD}
  \end{equation}
  If $w_{i+t}^{\live}=0$ for some $t\in E_{0,H}$, then
  \Cref{lem:index-normalized-live-weight-monotonicity} gives
  $w_{i+H}^{\live}=0$, so (BD) holds.  Therefore assume
  $w_{i+t}^{\live}>0$ for $(0 \le t \le H)$.

  Let $L_j$ be the event that the run is still live after $j$ executions of
  $C$.  Then $L_{j+1}\subseteq L_j$, because a run that has exited cannot
  become live again, and $\PP[L_j]=w_j^{\live}$.  For each
  $t\in E_{0,H}$, fix the witnesses supplied by $\Live(i+t)$:
  $A_{i+t},\theta_{i+t},r_{i+t}$,
  $(\calR_{i+t,c})_{c\in\supp(\theta_{i+t})}$, and
  $(\mu_{i+t,c})_{c\in\supp(\theta_{i+t})}$ such that
  \begin{align*}
    \theta_{i+t}\in\calD(A_{i+t}),
    \qquad
    r_{i+t}:\supp(\theta_{i+t})\to E_{\ell+1,h},
    \qquad
    \mu_{i+t}^{\live}
    =
    \bigoplus_{c\sim\theta_{i+t}}\mu_{i+t,c},\\
    \Gamma[R\mapsto r_{i+t}(c)],\calR_{i+t,c}\vDash\varphi,
    \qquad
    \calR_{i+t,c}\otimes\calR_F\preceq\mu_{i+t,c},
    \qquad
    V_{\calR_{i+t,c}}\subseteq V_b
    \quad(c\in\supp(\theta_{i+t}))
  \end{align*}
  Let $(L_i^a)_{a\in\supp(\theta_i)}$ be the partition of $L_i$ induced by
  the initial live decomposition
  $\mu_i^{\live}=\bigoplus_{a\sim\theta_i}\mu_{i,a}$; hence
  $\PP[L_i^a]=\theta_i(a)\cdot w_i^{\live}$ for
  $a \in \supp(\theta_i)$

  For each $a\in\supp(\theta_i)$, put
  $d_a\triangleq r_i(a)-\ell$.  Since
  $r_i(a)\in E_{\ell+1,h}$, we have $d_a\in E_{1,H}$.  Let $D_t^a$ be the
  rank distance from $\ell$ after $t$ further executions of $C$ from time
  $i$, with $D_t^a=0$ once the run has exited.  Thus $D_0^a=d_a$ on
  $L_i^a$.  Define
  \[
    S_a\triangleq
    \bigcap_{t=0}^{d_a-1}
    \left(\{D_t^a=0\}\cup\{D_{t+1}^a<D_t^a\}\right)
  \]
  On $L_i^a$, the event $S_a$ says that the distance has either already
  reached $0$, or strictly decreases, at each of the first $d_a$ steps.  Hence
  $D_{d_a}^a=0$ on $L_i^a\cap S_a$, and because $d_a\le H$,
  \[
    L_i^a\cap S_a\subseteq L_{i+H}^{c}
  \]

  We verify the one-step lower bound required by
  \Cref{lem:index-branchwise-progress-block-decay}.  Fix
  $a\in\supp(\theta_i)$ and $0\le t<d_a$, and condition on
  \[
    L_i^a\cap
    \bigcap_{u=0}^{t-1}
    \left(\{D_u^a=0\}\cup\{D_{u+1}^a<D_u^a\}\right)
  \]
  For this displayed conditioning event, set
  \[
    q_t^a
    \triangleq
    \Pr\!\left[
      D_t^a=0
      \,\middle|\,
      L_i^a\cap
      \bigcap_{u=0}^{t-1}
      \left(\{D_u^a=0\}\cup\{D_{u+1}^a<D_u^a\}\right)
    \right]
  \]
  If $q_t^a<1$, condition further on $D_t^a>0$.  The resulting live input at
  time $i+t$ decomposes along the fixed live witness as
  $\bigoplus_{c\sim\lambda_t^a}\mu_{i+t,c}$
  for some $\lambda_t^a\in\calD(A_{i+t})$.  If $q_t^a=1$, choose any
  $\lambda_t^a\in\calD(A_{i+t})$; the coefficient $1-q_t^a$ below is then
  zero.  Applying
  \Cref{lem:index-live-component-rank-progress-decomposition} at time $i+t$
  to every $c\in\supp(\lambda_t^a)$ gives
  \[
    \beta_{i+t,c}
    \bigl(\{b\mid r'_{i+t,c}(b)<r_{i+t}(c)\}\bigr)
    \ge p
  \]
  Therefore
  \begin{align*}
    &\Pr\!\left[
      D_t^a=0 \;\lor\; D_{t+1}^a<D_t^a
      \,\middle|\,
      L_i^a\cap
      \bigcap_{u=0}^{t-1}
      \left(\{D_u^a=0\}\cup\{D_{u+1}^a<D_u^a\}\right)
    \right]\\
    &\qquad =
    q_t^a
    +
    (1-q_t^a)\cdot
    \sum_{c\in\supp(\lambda_t^a)}
    \lambda_t^a(c)\cdot
    \beta_{i+t,c}
    \bigl(\{b\mid r'_{i+t,c}(b)<r_{i+t}(c)\}\bigr)
    \tag{split on $D_t^a=0$, then on the live decomposition}\\
    &\qquad \ge
    q_t^a
    +
    (1-q_t^a)\cdot
    \sum_{c\in\supp(\lambda_t^a)}
    \lambda_t^a(c)\cdot p
    \tag{\Cref{lem:index-live-component-rank-progress-decomposition}}\\
    &\qquad =
    q_t^a+(1-q_t^a)\cdot p
    \tag{$\lambda_t^a\in\calD(A_{i+t})$}\\
    &\qquad \ge
    p .
    \tag{$0\le q_t^a\le1$ and $p\le1$}
  \end{align*}
  Thus the branchwise hypotheses of
  \Cref{lem:index-branchwise-progress-block-decay} hold for every
  $a\in\supp(\theta_i)$.  Therefore
  \Cref{lem:index-branchwise-progress-block-decay}, applied with
  $w_j=w_j^{\live}$, $A=A_i$, $\theta=\theta_i$, the branch slices $L_i^a$,
  the branch distances $D_t^a$, and the events $S_a$, gives (BD).
\end{proof}

\subsection{Bounded-Rank Soundness}
\begin{lemma}[Bounded-Rank Soundness]
\label{lem:index-bounded-rank-soundness}
\begin{mathpar}
  \inferrule*[right=\rulename{Bounded-Rank}]
  {
    \inferrule*{}{
      \substack{
      \varphi \Rightarrow \sure{\ell \le R \le h} \\
      \varphi\sep[R = \ell] \Rightarrow \sure{e \mapsto \fls} \\
      \varphi\sep[R > \ell] \Rightarrow \sure{e \mapsto \tru}
      }
    }\;\;
    \vdash_m\triple{\varphi * \sure{R = N > \ell}}{C}{\textstyle\big(\bignd_{R=\ell}^{N-1} \varphi\big) \oplus_{\ge p}^\Leak \big(\bignd_{R=N}^h \varphi\big)}\;\;
    \inferrule*{}{
      \substack{
      0 < p \le 1 \\
      N \notin \vars(\varphi) \\
      \convex{\varphi[\ell/R]}
    }}
  }
  {
    \vdash_m\triple{\textstyle\bignd_{R=\ell}^h \varphi}{\whl eC}{\varphi[\ell/R]}
  }
 \end{mathpar}
\end{lemma}
\begin{proof}
  Put $\chi\triangleq\varphi[\ell/R]$\noam{why not call it $\psi$?} and
  $H\triangleq h-\ell$.  We prove semantic validity of the conclusion in mode
  $m$.  Fix a schedule $s$, an environment $\Gamma$, resources
  $\calR,\calR_F$, and an input distribution $\mu\in\calD(\Mem\times\bbN)$ such
  that
  \[
    \Gamma,\calR\vDash\bignd_{R=\ell}^{h}\varphi,
    \qquad
    \Stable{m}{\calR_F},
    \qquad
    \calR\otimes\calR_F\preceq\mu
  \]
  It suffices to construct a resource $\calQ$ such that
  \[
    \Gamma,\calQ\vDash\chi,
    \qquad
    \calQ\otimes\calR_F\preceq\de{\whl eC}_s^\dagger(\mu)
  \]

  We use the normalized exit/live decomposition of
  \Cref{def:index-normalized-exit-live-decomposition}, instantiated with the
  fixed schedule $s$, guard $e$, command $C$, and starting distribution $\mu$;
  we omit $\mu$ from the superscript.  Here $G_c$ is the event that the guard
  evaluates to $c$, $w_i^{\exit}$ is the probability of exiting at the
  $i$th guard test, and $w_i^{\live}$ is the probability of remaining
  live after $i$ executions of $C$.  The distributions
  $\mu_i^{\exit}$ and $\mu_i^{\live}$ are the corresponding
  normalized exit and live distributions.  After this instantiation, the
  defining equations are:
  \[
    G_c
    \triangleq
      \{(\sigma,k)\in\Mem\times\bbN
      \mid \de{e}_{\Exp}(\sigma)=c\}
  \]
  \[
    w_0^{\exit}
    \triangleq \mu(G_{\fls}),
    \qquad
    w_0^{\live}
    \triangleq \mu(G_{\tru}),
    \qquad
    \mu_0^{\exit}
    \triangleq\mu\mid G_{\fls},
    \qquad
    \mu_0^{\live}
    \triangleq\mu\mid G_{\tru}
  \]
  \[
    w_{i+1}^{\exit}
    \triangleq
      w_i^{\live}\cdot
      \left(\de{C}_s^\dagger(\mu_i^{\live})\right)(G_{\fls}),
    \qquad
    w_{i+1}^{\live}
    \triangleq
      w_i^{\live}\cdot
      \left(\de{C}_s^\dagger(\mu_i^{\live})\right)(G_{\tru})
  \]
  \[
    \mu_{i+1}^{\exit}
    \triangleq
      \de{C}_s^\dagger(\mu_i^{\live})\mid G_{\fls},
    \qquad
    \mu_{i+1}^{\live}
    \triangleq
      \de{C}_s^\dagger(\mu_i^{\live})\mid G_{\tru}
  \]

  Let $(F_m)_{m\in\bbN}$ and $F_{m,s}$ be the Kleene approximants and
  schedule-indexed projections from
  \Cref{lem:index-schedule-indexed-while-semantic-equations}, instantiated
  with $s,e,C$.

  Let $\Exit$ and $\Live$ be the predicates from
  \Cref{def:index-bounded-rank-exit-live-witnesses}, instantiated with
  $(\Gamma,\varphi,\chi,\ell,h,\calR_F,V_\calR)$ and the generated exit/live weights
  and distributions.  By induction on $i$, using
  \Cref{lem:index-bounded-rank-initial-exit-live-witnesses} for $i=0$ and
  \Cref{lem:index-bounded-rank-exit-live-successor} for the successor step, for
  every $i\in\bbN$, $\Exit(i)$ and $\Live(i)$ holds.
  Hence, by \Cref{lem:index-bounded-rank-live-step-properness}, we get
  $\mathsf{LiveStepProper}_{s, e}(C, \mu)$.
  Therefore \Cref{lem:index-normalized-exit-stream-approximants} gives, for
  every $n\in\bbN$,
  \begin{equation}
    \bigl((F_{n+1})_s^\bot\bigr)^\dagger(\mu)
    =
    \left(\sum_{i=0}^{n}
      w_i^{\exit}\cdot\mu_i^{\exit}\right)
    +
    \left(w_n^{\live}\cdot\delta_\bot\right).
    \label{eq:index-bounded-rank-exit-stream}
  \end{equation}
  By the bounded-rank block-decay argument, applied with
  $\Stable{m}{\calR_F}$ and the just-proved live invariant, whenever $H>0$
  and $w_i^{\live}>0$,
  $w_{i+H}^{\live}\le(1-p^H)\cdot w_i^{\live}$.  If
  $w_i^{\live}=0$, then the recursive definition of $w_j^{\live}$ gives
  $w_{i+H}^{\live}=0$.  Hence the positive-height hypothesis needed below is
  \begin{equation}
    H>0
    \Longrightarrow
    \forall j\in\bbN.\;
    w_{j+H}^{\live}
    \le
    (1-p^H)\cdot w_j^{\live}.
    \label{eq:index-bounded-rank-block-decay}
  \end{equation}

  We now verify the hypotheses of
  \Cref{lem:index-block-decay-exhausts-exit-stream}.
  The zero-height hypothesis is
  \begin{equation}
    H=0\Longrightarrow w_0^{\live}=0.
    \label{eq:index-bounded-rank-zero-height-no-live}
  \end{equation}
  Indeed, if $H=0$, then
  $\{n\in\bbN\mid\ell<n\le h\}=\emptyset$.  Since $\Live(0)$ holds and every
  normalized distribution has nonempty support, the antecedent
  $w_0^{\live}>0$ in $\Live(0)$ is false; otherwise the consequent of
  $\Live(0)$ would require a map from a nonempty support into $\emptyset$.
  By \Cref{lem:index-block-decay-exhausts-exit-stream}, applied to the
  normalized exit/live decomposition starting from $\mu$ and the hypotheses
  \eqref{eq:index-bounded-rank-zero-height-no-live} and
  \eqref{eq:index-bounded-rank-block-decay}, together with the just-established
  $\mathsf{LiveStepProper}_{s,e}(C,\mu)$, we get
  $\sum_{j\in\bbN}w_j^{\exit}=1$.

  We now prove that the denotation of the loop is exactly the countable
  mixture of the normalized exit distributions.  First we prove equality on
  ordinary outputs.
  For every ordinary output $x\in\Mem\times\bbN$,
  \begin{align*}
    \left(\de{\whl eC}_s^\dagger(\mu)\right)(x)
    &=
    \sup_{n\in\bbN}
    \left(\bigl((F_{n+1})_s^\bot\bigr)^\dagger(\mu)\right)(x)
      \tag{\Cref{lem:index-schedule-indexed-while-distribution-supremum}}\\
    &=
    \sup_{n\in\bbN}
    \left[
      \left(\sum_{i=0}^{n}
      w_i^{\exit}\cdot\mu_i^{\exit}\right)(x)
      +
      \left(w_n^{\live}\cdot\delta_\bot\right)(x)
    \right]
      \tag{by \eqref{eq:index-bounded-rank-exit-stream}}\\
    &=
    \sup_{n\in\bbN}
    \sum_{i=0}^{n}
      w_i^{\exit}\cdot\mu_i^{\exit}(x)
      \tag{$x\ne\bot$}\\
    &=
    \left(
      \sum_{i\in\bbN}
      w_i^{\exit}\cdot\mu_i^{\exit}
    \right)(x)
      \tag{definition of countable sum of distributions}
  \end{align*}
  Since the exit weights have total mass $1$,
  $\sum_{i\in\bbN}w_i^{\exit}\cdot\mu_i^{\exit}$ is an ordinary probability
  distribution.  Moreover,
  \[
    \sum_{x\in\Mem\times\bbN}
    \left(\de{\whl eC}_s^\dagger(\mu)\right)(x)
    =
    \sum_{x\in\Mem\times\bbN}
    \left(
      \sum_{i\in\bbN}
      w_i^{\exit}\cdot\mu_i^{\exit}
    \right)(x)
    =
    1
  \]
  Since $\de{\whl eC}_s^\dagger(\mu)\in\calD_\bot(\Mem\times\bbN)$, its
  $\bot$-mass is the complement of its total ordinary mass.  Hence
  $\left(\de{\whl eC}_s^\dagger(\mu)\right)(\bot)=0$.  Therefore
  \begin{equation}
    \de{\whl eC}_s^\dagger(\mu)
    =
    \sum_{i\in\bbN}w_i^{\exit}\cdot\mu_i^{\exit} .
    \label{eq:index-bounded-rank-final-exit-mixture}
  \end{equation}
  Let $I\triangleq\{i\in\bbN\mid w_i^{\exit}>0\}$ and let
  $\beta\in\calD(I)$ be $\beta(i)=w_i^{\exit}$.  For each $i\in I$, choose
  $\calQ_i$ from $\Exit(i)$, and define the final witness as
  $\calQ\triangleq\bigoplus_{i\sim\beta}\calQ_i$.
  For every $i\in I$, the same instance of $\Exit(i)$ gives
  \begin{align*}
    \Gamma,\calQ_i
    &\vDash
    \chi,
    &
    \calQ_i\otimes\calR_F
    &\preceq
    \mu_i^{\exit},\\
    V_{\calQ_i}
    &\subseteq
    V_\calR
  \end{align*}
  Let $J\notin\fv(\chi)$ be fresh.  Then
  \begin{align*}
    \calQ
    &=
    \bigoplus_{i\sim\beta}\calQ_i,
    &
    \Gamma[J\mapsto i],\calQ_i
    &\vDash
    \chi
    \qquad(i\in I).
      \tag{definition of $\calQ$, $J\notin\fv(\chi)$}
  \end{align*}

  By the satisfaction clause for $\bigoplus^\Leak$ and
  \Cref{lem:index-probabilistic-choice-idempotence}, since $\convex{\chi}$,
  we get $\Gamma,\calQ\vDash\chi$.

  By \eqref{eq:index-bounded-rank-final-exit-mixture}, the framed refinement is
  \[
    \calQ\otimes\calR_F
    \cong
    \bigoplus_{i\sim\beta}(\calQ_i\otimes\calR_F)
    \preceq
    \de{\whl eC}_s^\dagger(\mu)
  \]
  Therefore $\calQ$ is the required postcondition witness.
\end{proof}

\newpage
\small
\clearpage
\begin{landscape}

\renewcommand{\bern}[1]{\textbf{Bern}\big( #1 \big)}

\section{Case Studies}\label{app:examples}

We now present the full derivations for the examples in \Cref{sec:examples}.

\subsection{The Monty Hall Problem}
\label{app:monte}

Recall the Monty Hall Problem from \Cref{sec:monte}.
\[
  c \samp \unif{1, 3} \fatsemi
  p \gets [1, 2, 3] \fatsemi
  o \gets [1, 2, 3] \setminus [c, p] \fatsemi
  p \coloneqq 6 - p - o
\]
We now provide the full derivation for the above program. We start from the top level, using two applications of \ruleref{Seq} along with \ruleref{Samp} and \ruleref{NAssign} to handle the first two commands, as was shown in \Cref{sec:monte}. Note that the application of \ruleref{Frame} before \ruleref{NAssign} allows us to conclude that $c$ and $p$ are independent, since the player chooses $p$ obliviously.
\[
\inferrule*[right=Seq]{
  \inferrule*[right=Samp,vdots=7em,rightskip=25em]{\;}{
    \vdash_\weak\triple{\sure{\own(c, p, o)}}{c \samp \unif{1, 3}}{c \sim \unif{1,3} \sep \sure{\own(p,o)}}
  } 
  \inferrule*[Right=Seq]{
    \inferrule*[right=Frame]{
    \inferrule*[right=NAssign]{\;}{
        \vdash_\weak\triple{\sure{\own(p)}}{p \gets [1, 2, 3]}{
        \textstyle \bignd_{X \in \{1,2,3\}}{\sure{p \mapsto X}}}
    } 
    }{
      \vdash_\weak\triple{c \sim \unif{1,3} \sep \sure{\own(p,o)}}{p \gets [1, 2, 3]}{
        \textstyle c \sim \unif{1,3} \sep \sure{\own(o)} \sep \bignd_{X \in \{1,2,3\}}\sure{p \mapsto X}}
      }
      \\
      (\ref{eq:monte-cases})
  }{
    \vdash_\weak\triple{c \sim \unif{1,3} \sep \sure{\own(p,o)}}{
      p \gets [1, 2, 3] \fatsemi
      o \gets [1, 2, 3] \setminus [c, p] \fatsemi
      p \coloneqq 6 - p - o
    }{(c = p) \sim \bern{\tfrac23}}
  } 
}{
  \vdash_\weak\triple{\sure{\own(c, p, o)}}{
    c \samp \unif{1, 3} \fatsemi
    p \gets [1, 2, 3] \fatsemi
    o \gets [1, 2, 3] \setminus [c, p] \fatsemi
    p \coloneqq 6 - p - o
  }{(c = p) \sim \bern{\tfrac23}}
} 
\]
Next, we move on to handle the last two commands. As an abbreviation, let $s = \{1,2,3\}$, let $S_X = s \setminus \{X\}$, and let $d_X = \unif{S_X}$. The first step is to apply the following consequence, which first distributes the random information into the $\bignd$ modality and then breaks the inner distribution into two parts, where $X\neq Y$ and where $X=Y$.
\begin{align*}
  \textstyle c \sim \unif{1,3} \sep \sure{\own(o)} \sep \bignd_{X \in \{1,2,3\}}\sure{p \mapsto X}
  \quad\Rightarrow&\quad
  \textstyle \bignd_{X \in \{1,2,3\}} \bigoplus_{Y \sim \unif{1,3}} \sure{c\mapsto Y \sep p \mapsto X \sep \own(o)}
  \\  \Rightarrow&\quad
  \textstyle\bignd_{X \in \{1,2,3\}} \big(  \big( \bigoplus_{Y \sim d_X} \sure{c \mapsto Y \sep p\mapsto X\sep\own(o)} \big) \oplus_\frac23 \sure{c \mapsto X \sep p \mapsto X \sep\own(o)} \big)
\end{align*}
Next, we simply use \ruleref{ND-Split2} and \textsc{Split} to do case analysis on the nondeterministic outcomes and the top-level random outcomes.
\begin{equation}\label{eq:monte-cases}
\inferrule*[right=Consequence]{
  \inferrule*[Right=ND-Split2]{
    \inferrule*[Right=Split]{
      (\ref{eq:monte-neq})
      \\
      (\ref{eq:monte-eq})
    }{
      \vdash_\weak\triple{\textstyle
      \big( \bigoplus_{Y \sim d_X} \sure{c \mapsto Y \sep p\mapsto X\sep\own(o)} \big) \oplus_\frac23 \sure{c \mapsto X \sep p \mapsto X \sep\own(o)}}{
        o \gets [1, 2, 3] \setminus [c, p] \fatsemi p \coloneqq 6 - p - o
      }{(c = p) \sim \bern{\tfrac23}}
    } 
  }{
    \vdash_\weak\triple{\textstyle
      \bignd_{X \in s} \big(  \big( \bigoplus_{Y \sim d_X} \sure{c \mapsto Y \sep p\mapsto X\sep\own(o)} \big) \oplus_\frac23 \sure{c \mapsto X \sep p \mapsto X \sep\own(o)} \big)}{
        o \gets [1, 2, 3] \setminus [c, p] \fatsemi p \coloneqq 6 - p - o
      }{(c = p) \sim \bern{\tfrac23}}
    } 
}{
  \vdash_\weak\triple{
    \textstyle c \sim \unif{1,3} \sep \sure{\own(o)} \sep \bignd_{X \in \{1,2,3\}}\sure{p \mapsto X}
  }{
        o \gets [1, 2, 3] \setminus [c, p] \fatsemi p \coloneqq 6 - p - o
  }{(c = p) \sim \bern{\tfrac23}}
} 
\end{equation}
We begin by discussing the case where $c\neq p$. To start, we simply apply \ruleref{Seq} to break up the sequential composition.
\begin{equation}\label{eq:monte-neq}
\inferrule*[right=Seq]{
  (\ref{eq:monte-neq-1})
  \\
  (\ref{eq:monte-neq-2})
}{
  \vdash_\weak\triple{\textstyle\bigoplus_{Y \sim d_X} \sure{c \mapsto Y \sep p\mapsto X \sep \own(o)}
    }{
        o \gets [1, 2, 3] \setminus [c, p] \fatsemi p \coloneqq 6 - p - o
      }{(c = p) \mapsto 1}
} 
\end{equation}
To handle the nondeterministic assignment to $o$, we first apply \textsc{Split} to open the $\bigoplus$ modality and then \ruleref{NAssign}. Since $Y \sim d_X$, we know that $Y \in \supp(d_X) = \{1,2,3\}\setminus \{X\}$, \ie $Y \neq X$. Therefore, the set $\{1,2,3\}\setminus\{Y,X\}$ must be a singleton set, and we can therefore collapse the $\bignd$ modality to the singular value $6 - X - Y$ using the rule of \ruleref{Consequence}.
\begin{equation}\label{eq:monte-neq-1}
\inferrule*[right=Consequence]{
  \inferrule*[Right=Split]{
    \inferrule*[Right=NAssign]{\;}{
      \vdash_\weak\triple{\sure{c \mapsto Y \sep p\mapsto X \sep \own(o)}
      }{
        o \gets [1, 2, 3] \setminus [c, p]
      }{\textstyle\sure{c \mapsto Y \sep p\mapsto X} \sep \bignd_{Z \in \{1,2,3\}\setminus \{Y,X\}} \sure{o \mapsto Z}}
    }
  }{
    \vdash_\weak\triple{\textstyle\bigoplus_{Y \sim d_X} \sure{c \mapsto Y \sep p\mapsto X \sep \own(o)}
    }{
        o \gets [1, 2, 3] \setminus [c, p]
      }{\textstyle \bigoplus_{Y \sim d_X} \sure{c \mapsto Y \sep p\mapsto X} \sep \bignd_{Z \in \{1,2,3\}\setminus \{Y,X\}} \sure{o \mapsto Z}}
  } 
}{
  \vdash_\weak\triple{\textstyle\bigoplus_{Y \sim d_X} \sure{c \mapsto Y \sep p\mapsto X \sep \own(o)}
    }{
        o \gets [1, 2, 3] \setminus [c, p]
      }{\textstyle \bigoplus_{Y \sim d_X} \sure{c \mapsto Y \sep p\mapsto X \sep o \mapsto 6 - X - Y}}
} 
\end{equation}
Moving on to the final command, we first apply \ruleref{Split2} to reason under the $\bigoplus$ modality and then \ruleref{Assign} to handle the deterministic assignment. That gives us that $p = 6-X-(6 - X - Y) = Y$. Given that $c=Y$ also, that means that $c = p$, so we can use a \ruleref{Consequence} to conclude that $\sure{(c=p)\mapsto 1}$ in all cases.
\begin{equation}\label{eq:monte-neq-2}
\inferrule*[right=Split2]{
  \inferrule*[Right=Consequence]{
    \inferrule*[Right=Assign]{\;}{
      \vdash_\weak\triple{\sure{c \mapsto Y \sep p\mapsto X \sep o \mapsto 6 - X - Y}}{
        p \coloneqq 6 - p - o
      }{\sure{c \mapsto Y \sep p \mapsto 6 - X - (6 - X - Y) \sep o \mapsto 6 - X - Y}}
    }
  }{
    \vdash_\weak\triple{\sure{c \mapsto Y \sep p\mapsto X \sep o \mapsto 6 - X - Y}}{
      p \coloneqq 6 - p - o
    }{\sure{(c = p) \mapsto 1}}
  } 
}{
  \vdash_\weak\triple{\textstyle \bigoplus_{Y \sim d_X} \sure{c \mapsto Y \sep p\mapsto X \sep o \mapsto 6 - X - Y}}{
    p \coloneqq 6 - p - o
  }{\sure{(c = p) \mapsto 1}}
} 
\end{equation}
Finally, we move on to the second case from (\ref{eq:monte-cases}), where $c = p$. Using mechanical applications of \ruleref{Seq}, \ruleref{NAssign}, \ruleref{ND-Split}, and \ruleref{Assign}, we get that $\bignd_{Z \in S_X} \sure{c\mapsto X\sep p\mapsto 6 - X- Z}$. Since $Z \in S_X$, then $Z \neq X$ and therefore $6 - X - Z\neq X$, so $c \neq p$, which means that we can use \ruleref{Consequence} to conclude that $\sure{(c=p)\mapsto 0}$.
\begin{equation}\label{eq:monte-eq}
\inferrule*[right=Seq]{
  \inferrule*[right=NAssign,vdots=6.5em,rightskip=30em]{\;}{
    \vdash_\weak\triple{\sure{c \mapsto X \sep p \mapsto X \sep \own(o)}}{o \gets [1, 2, 3] \setminus [c, p]}{
      \textstyle\bignd_{Z \in S_X} \sure{o \mapsto Z \sep c \mapsto X \sep p \mapsto X}}
  } 
  \inferrule*[Right=Consequence]{
    \inferrule*[Right=ND-Split]{
      \inferrule*[Right=Assign]{\;}{
        \vdash_\weak\triple{\sure{o \mapsto Z \sep c \mapsto X \sep p \mapsto X}}{p \coloneqq 6 - p - o}{\sure{c \mapsto X \sep p \mapsto 6-X-Z}}
      }
    }{
      \vdash_\weak\triple{\textstyle\bignd_{Z \in S_X}\sure{o \mapsto Z \sep c \mapsto X \sep p \mapsto X}}{p \coloneqq 6 - p - o}{\textstyle\bignd_{Z \in S_X}\sure{c \mapsto X \sep p \mapsto 6-X-Z}}
    } 
  }{
    \vdash_\weak\triple{\textstyle\bignd_{Z \in S_X} \sure{o \mapsto Z \sep c \mapsto X \sep p \mapsto X}}{p \coloneqq 6 - p - o}
      {(c = p) \mapsto 0}
  } 
}{
  \vdash_\weak\triple{\sure{c \mapsto X \sep p \mapsto X \sep \own(o)}}{
        o \gets [1, 2, 3] \setminus [c, p] \fatsemi p \coloneqq 6 - p - o
      }{(c = p) \mapsto 0}
} 
\end{equation}

\subsection{Leader Election}
\label{app:leader}

Recall the two party leader election protocol below from \Cref{sec:leader}.
\[
\begin{array}{l}
  x \coloneqq 0\fatsemi\; \\
  \whl{|x| \le 1}{} \\
  \quad \underbrace{(x \samp \unif{x, x+1}) \nd (x \samp \unif{x-1, x})}_{C_\code{body}}
\end{array}
\]
The following is the loop invariant for analyzing the program.
\begin{align*}
  \varphi_3 &\triangleq \big(\sure{x \mapsto 1} \nd \sure{x \mapsto 0} \nd \sure{x \mapsto -1}\big) \sep \sure{R=3}
  &
  \psi_3 &\triangleq \varphi_2 \nd (\varphi_0 \oplus_\frac12 \varphi_3)
  \\
  \varphi_2 &\triangleq \big(x \sim \unif{0,1} \nd x\sim\unif{-1, 0} \big) \sep \sure{R=2}
  &
  \psi_2 &\triangleq (\varphi_0 \oplus_\frac14 \varphi_3) \nd (\varphi_1 \oplus_\frac12 \varphi_3)
  \\
  \varphi_1 &\triangleq \big(\sure{x \mapsto 1} \oplus_\frac12 \sure{x \mapsto -1}\big) \sep \sure{R=1}
  &
  \psi_1 &\triangleq \varphi_0 \oplus_\frac14 \varphi_3
  \\
  \varphi_0 &\triangleq \big(\sure{x \mapsto 2} \nd \sure{x \mapsto -2}\big) \sep \sure{R=0} \\
  \varphi &\triangleq \varphi_0 \vee \varphi_1 \vee \varphi_2 \vee \varphi_3
\end{align*}
We begin by giving the top-level derivation for the program, which is quite mechanical in that it simply breaks up the sequential composition using \ruleref{Seq} and then applies the consequence $\sure{x \mapsto 0} \Rightarrow \varphi[3/R] \Rightarrow \bignd_{R=0}^3 \varphi$ so that we can then apply \ruleref{Bounded-Rank}.
\[
\inferrule*[right=Seq]{
  \inferrule*[right=Assign]{\;}{
    \vdash_\weak\triple{\sure{\own(x)}}{x \coloneqq 0}{\sure{x \mapsto 0}}
  } 
  \\
  \inferrule*[Right=Consequence]{
    \inferrule*[Right=Bounded-Rank]{
      (\ref{eq:leader-loop-body})
    }{
      \vdash_\weak\triple{\textstyle\bignd_{R=0}^3 \varphi}{\whl{|x|\le 1}{C_\code{body}}}{\sure{x\mapsto 2}\nd \sure{x \mapsto -2}}
    } 
  }{
    \vdash_\weak\triple{\sure{x \mapsto 0}}{\whl{|x|\le 1}{C_\code{body}}}{\sure{x\mapsto 2}\nd \sure{x \mapsto -2}}
  } 
}{
  \vdash_\weak\triple{\sure{\own(x)}}{x \coloneqq 0 \fatsemi \whl{|x|\le 1}{C_\code{body}}}{\sure{x\mapsto 2}\nd \sure{x \mapsto -2}}
} 
\]
To analyze the loop body, we start with the precondition of $\varphi\sep\sure{N=R>0}$. Since the rank is greater than 0, then it must be 1, 2, or 3. The proof proceeds by case analysis on the rank using the \textsc{Disj} rule. Note that for all $i \in \{1, 2, 3\}$, it holds that $\psi_i$ implies that the rank strictly decreases with probability at least $\frac14$, in other words:
\[
  \textstyle
    \psi_i \sep \sure{N=i}
    \quad\Rightarrow\quad
     \big(\bignd_{R=0}^{N-1} \varphi\big) \oplus_{\ge\frac14}\big(\bignd_{R=N}^{3} \varphi\big)
\]
That implication above justifies the use of the \ruleref{Consequence} rule.
\begin{equation}\label{eq:leader-loop-body}
  \inferrule*[right=Consequence]{
    \inferrule*[Right=Disj]{
      \inferrule*[right=Frame]{(\ref{eq:leader-R1})}{
        \vdash_\weak\triple{\varphi_1\sep\sure{N=1}}{C_\code{body}}{\psi_1\sep\sure{N=1}}
      }
      \\
      \inferrule*[Right=Disj]{
        \inferrule*[right=Frame,vdots=3em,rightskip=10em]{(\ref{eq:leader-R2})}{
          \vdash_\weak\triple{\varphi_2\sep\sure{N=2}}{C_\code{body}}{\psi_2\sep\sure{N=2}}
        }
        \inferrule*[right=Frame]{(\ref{eq:leader-R3})}{
          \vdash_\weak\triple{\varphi_3\sep\sure{N=3}}{C_\code{body}}{\psi_3\sep\sure{N=3}}
        }
      }{
        \vdash_\weak\triple{(\varphi_2 \sep\sure{N=2}) \vee (\varphi_3\sep\sure{N=3})}{C_\code{body}}{(\psi_2\sep\sure{N=2}) \vee (\psi_3\sep\sure{N=3})}
      } 
    }{
      \vdash_\weak\triple{(\varphi_1\sep\sure{N=1}) \vee (\varphi_2 \sep\sure{N=2}) \vee (\varphi_3\sep\sure{N=3})}{C_\code{body}}{(\psi_1\sep\sure{N=1}) \vee (\psi_2\sep\sure{N=2}) \vee (\psi_3\sep\sure{N=3})}
    } 
  }{
    \vdash_\weak\triple{\varphi\sep \sure{N = R > 0}}{C_\code{body}}{\textstyle\big(\bignd_{R=0}^{N-1} \varphi\big) \oplus_{\ge\frac14}\big(\bignd_{R=N}^{3} \varphi\big)}
  } 
\end{equation}
It now remains to prove the three cases. We begin with the $R=3$ case. The first step is to do case analysis on whether $x$ is 0, 1 or -1. The cases for 1 and -1 are symmetric, so we use \ruleref{ND-Split2} and obtain a single postcondition that applies to both cases.
\begin{equation}\label{eq:leader-R3}
\inferrule*[right=Consequence]{
  \inferrule*[Right=ND-Split]{
    (\ref{eq:leader-1-x0})
    \\
    \inferrule*[Right=ND-Split2]{
      (\ref{eq:leader-1-x1})
      \\
      (\ref{eq:leader-1-x-1})
    }{
      \vdash_\weak\triple{\sure{x \mapsto 1}\nd\sure{x\mapsto -1}}{(x \samp \unif{x, x+1}) \nd (x \samp \unif{x-1, x})}{(\varphi_0 \oplus_\frac12 \varphi_3) \nd\varphi_2}
    } 
  }{
    \vdash_\weak\triple{\sure{x\mapsto 0} \nd (\sure{x \mapsto 1}\nd \sure{x \mapsto -1})}{(x \samp \unif{x, x+1}) \nd (x \samp \unif{x-1, x})}{\varphi_2 \nd ((\varphi_0 \oplus_\frac12 \varphi_3) \nd\varphi_2)}
  } 
}{
  \vdash_\weak\triple{\varphi_3}{(x \samp \unif{x, x+1}) \nd (x \samp \unif{x-1, x})}{\varphi_2 \nd (\varphi_0 \oplus_\frac12 \varphi_3)}
}
\end{equation}
We start with the $x=0$ case. A straightforward application of \ruleref{ND} and \ruleref{Samp} gives us the postcondition $\varphi_2$.
\begin{equation}\label{eq:leader-1-x0}
    \inferrule*[Right=ND]{
      \inferrule*[right=Samp]{\;}{
        \vdash_\weak\triple{\sure{x \mapsto 0}}{x \samp \unif{x, x+1}}{x \sim \unif{0, 1}}
      } 
      \quad
      \inferrule*[right=Samp]{\;}{
        \vdash_\weak\triple{\sure{x \mapsto 0}}{x \samp \unif{x-1, x}}{x \sim \unif{-1, 0}}
      } 
    }{
      \vdash_\weak\triple{\sure{x \mapsto 0}}{(x \samp \unif{x, x+1}) \nd (x
        \samp \unif{x-1, x})}{x \sim \unif{0, 1} \nd x \sim \unif{-1, 0}}
    } 
\end{equation}
Now, for the $x=1$ case, we again use \ruleref{ND} and \ruleref{Samp} to obtain the postcondition $x \sim\unif{1,2} \nd x\sim\unif{0,1}$. In the left case, Alice wins ($\varphi_0$) with probability $\frac12$, otherwise we are back in state $\varphi_3$. Meanwhile, the right case implies $\varphi_2$. Therefore the final \ruleref{Consequence} below holds.
\begin{equation}\label{eq:leader-1-x1}
\inferrule*[right=Consequence]{
    \inferrule*[Right=ND]{
      \inferrule*[right=Samp]{\;}{
        \vdash_\weak\triple{\sure{x \mapsto 1}}{x \samp \unif{x, x+1}}{x \sim \unif{1, 2}}
      } 
      \\
      \inferrule*[Right=Samp]{\;}{
        \vdash_\weak\triple{\sure{x \mapsto 1}}{x \samp \unif{x-1, x}}{x \sim \unif{0, 1}}
      } 
    }{
      \vdash_\weak\triple{\sure{x \mapsto 1}}{(x \samp \unif{x, x+1}) \nd (x
        \samp \unif{x-1,x})}{x \sim \unif{1, 2} \nd x \sim \unif{0, 1}}
    } 
}{
  \vdash_\weak\triple{\sure{x \mapsto 1}}{(x \samp \unif{x, x+1}) \nd (x \samp \unif{x-1,x})}{(\varphi_0 \oplus_\frac12 \varphi_3) \nd \varphi_2}
} 
\end{equation}
The case where $x = -1$ is completely symmetric.
\begin{equation}\label{eq:leader-1-x-1}
\inferrule*[right=Consequence]{
    \inferrule*[Right=ND]{
      \inferrule*[right=Samp]{\;}{
        \vdash_\weak\triple{\sure{x \mapsto -1}}{x \samp \unif{x, x+1}}{x \sim \unif{-1, 0}}
      } 
      \quad
      \inferrule*[Right=Samp]{\;}{
        \vdash_\weak\triple{\sure{x \mapsto -1}}{x \samp \unif{x-1, x}}{x \sim \unif{-2, -1}}
      } 
    }{
      \vdash_\weak\triple{\sure{x \mapsto -1}}{(x \samp \unif{x, x+1}) \nd (x \samp \unif{x-1, x})}{x \sim \unif{-1, 0} \nd x \sim \unif{-2, -1}}
    } 
}{
  \vdash_\weak\triple{\sure{x \mapsto -1}}{(x \samp \unif{x, x+1}) \nd (x \samp
    \unif{x-1, x})}{(\varphi_0 \oplus_\frac12 \varphi_3) \nd \varphi_2}
} 
\end{equation}
Now, we move on to the case where $R=2$. Since the precondition $\varphi_2$ is composed from two nondeterministic outcomes, we must start by doing case analysis on those outcomes.
\begin{equation}\label{eq:leader-R2}
\inferrule*[right=ND-Split2]{
  (\ref{eq:leader-R2-incr})
  \\
  (\ref{eq:leader-R2-decr})
}{
  \vdash_\weak\triple{\varphi_2}{(x \samp \unif{x, x+1}) \nd (x \samp \unif{x-1,x})}{(\varphi_0 \oplus_\frac14 \varphi_3) \nd (\varphi_1 \oplus_\frac12 \varphi_3)}
}
\end{equation}
In the first case, $x \sim \unif{0,1}$. Since that precondition is stable, we can take advantage of the obliviousness of the adversary and immediately apply \ruleref{ND}. We then have two cases corresponding to whether the adversary increments or decrements $x$. In the case that $x$ is incremented, we get the postcondition $\bigoplus_{X\sim\unif{0,1}} x \sim \unif{X, X+1}$. That implies that $x=2$ with probability $\frac14$ (satisfying $\varphi_0$), and otherwise $x$ is 0 or 1 (satisfying $\varphi_3$). In the case that $x$ is decremented, we get:
\[\textstyle
  \bigoplus_{X\sim\unif{0,1}} x \sim \unif{X-1, X}
  \quad\Rightarrow\quad
  (\sure{x \mapsto 1} \oplus_\frac12 \sure{x \mapsto -1}) \oplus_\frac12 \sure{x \mapsto 0}
  \quad\Rightarrow\quad
  \varphi_1 \oplus_\frac12 \varphi_3
\]
which justifies the final \ruleref{Consequence}.
\begin{equation}\label{eq:leader-R2-incr}
    \inferrule*[right=ND\phantom{XXXXXX}]{
      \inferrule*[right=Consequence,vdots=6.5em,rightskip=15em]{
      \inferrule*[Right=Priv-Split1]{
        \vdash_\weak\triple{\sure{x \mapsto X}}{x \samp \unif{x, x+1}}{x \sim \unif{X, X+1}}
      }{
        \vdash_\weak\triple{x \sim \unif{0, 1}}{x \samp \unif{x, x+1}}{\textstyle\bigoplus_{X \sim \unif{0,1}} x \sim \unif{X, X+1}}
      } 
      }{
        \vdash_\weak\triple{x \sim \unif{0, 1}}{x \samp \unif{x, x+1}}{\varphi_0 \oplus_\frac14 \varphi_3}
      } 
      \inferrule*[Right=Consequence]{
      \inferrule*[Right=Priv-Split1]{
        \vdash_\weak\triple{\sure{x \mapsto X}}{x \samp \unif{x-1, x}}{x \sim \unif{X-1, X}}
      }{
        \vdash_\weak\triple{x \sim \unif{0, 1}}{x \samp \unif{x-1, x}}{\textstyle\bigoplus_{X \sim \unif{0,1}} x \sim \unif{X-1, X}}
      } 
      }{
        \vdash_\weak\triple{x \sim \unif{0, 1}}{x \samp \unif{x-1, x}}{\varphi_1 \oplus_\frac12 \varphi_3}
      } 
    }{
      \vdash_\weak\triple{x \sim \unif{0, 1}}{(x \samp \unif{x, x+1}) \nd (x \samp \unif{x-1,x})}{(\varphi_0 \oplus_\frac14 \varphi_3) \nd (\varphi_1 \oplus_\frac12 \varphi_3)}
    } 
\end{equation}
The case where instead $x \sim \unif{-1,0}$ is exactly symmetrical.
\begin{equation}\label{eq:leader-R2-decr}
    \inferrule*[right=ND\phantom{XXXXXX}]{
      \inferrule*[right=Consequence,vdots=6.5em,rightskip=20em]{
      \inferrule*[Right=Priv-Split1]{
        \vdash_\weak\triple{\sure{x \mapsto X}}{x \samp \unif{x, x+1}}{x \sim \unif{X, X+1}}
      }{
        \vdash_\weak\triple{x \sim \unif{-1, 0}}{x \samp \unif{x, x+1}}{\textstyle\bigoplus_{X \sim \unif{-1,0}} x \sim \unif{X, X+1}}
      } 
      }{
        \vdash_\weak\triple{x \sim \unif{-1, 0}}{x \samp \unif{x, x+1}}{\varphi_1 \oplus_\frac12 \varphi_3}
      } 
      \inferrule*[Right=Consequence]{
      \inferrule*[Right=Priv-Split1]{
        \vdash_\weak\triple{\sure{x \mapsto X}}{x \samp \unif{x-1, x}}{x \sim \unif{X-1, X}}
      }{
        \vdash_\weak\triple{x \sim \unif{-1, 0}}{x \samp \unif{x-1, x}}{\textstyle\bigoplus_{X \sim \unif{-1,0}} x \sim \unif{X-1, X}}
      } 
      }{
        \vdash_\weak\triple{x \sim \unif{-1, 0}}{x \samp \unif{x-1, x}}{\varphi_0 \oplus_\frac14 \varphi_3}
      } 
    }{
      \vdash_\weak\triple{x \sim \unif{-1, 0}}{(x \samp \unif{x, x+1}) \nd (x \samp \unif{x-1,x})}{(\varphi_0 \oplus_\frac14 \varphi_3) \nd (\varphi_1 \oplus_\frac12 \varphi_3)}
    } 
\end{equation}
Finally, we move on to the case where $R=1$. Since $\varphi_1 = \sure{x\mapsto 1}\oplus_\frac12\sure{x \mapsto -1}$ is stable, we can immediately take advantage of obliviousness and apply \ruleref{ND} to handle the two nondeterministic cases.
\begin{equation}\label{eq:leader-R1}
\inferrule*[right=Consequence]{
\inferrule*[Right=ND]{
  (\ref{eq:leader-R1-incr})
  \\
  (\ref{eq:leader-R1-decr})
}{
  \vdash_\weak\triple{\varphi_1}{(x \samp \unif{x, x+1}) \nd (x \samp \unif{x-1,x})}{(\varphi_0 \oplus_\frac14 \varphi_3) \nd (\varphi_0 \oplus_\frac14 \varphi_3)}
} 
}{
  \vdash_\weak\triple{\varphi_1}{(x \samp \unif{x, x+1}) \nd (x \samp \unif{x-1,x})}{\varphi_0 \oplus_\frac14 \varphi_3}
}
\end{equation}
In the first case where $x$ is incremented, we start by breaking up the random outcomes using \ruleref{Priv-Split1} (note that the program contains no nondeterminism, so the number of `ticks' is clearly 0 in both random branches). After applying \ruleref{Samp}, we clearly see that $x=2$ with probability $\frac14$.
\begin{equation}\label{eq:leader-R1-incr}
\inferrule*[right=Consequence]{
  \inferrule*[Right=Priv-Split1]{
    \inferrule*[right=Samp]{\;}{
      \vdash_\weak\triple{\sure{x\mapsto 1}}{x \samp\unif{x, x+1}}{x \sim \unif{1, 2}}
    }
    \\
    \inferrule*[Right=Samp]{\;}{
      \vdash_\weak\triple{\sure{x\mapsto -1}}{x \samp\unif{x, x+1}}{x \sim \unif{-1, 0}}
    }  
  }{
    \vdash_\weak\triple{\sure{x\mapsto 1}\oplus_\frac12\sure{x \mapsto -1}}{x \samp\unif{x, x+1}}{x \sim \unif{1, 2} \oplus_\frac12 x \sim \unif{-1, 0}}
  }
}{
  \vdash_\weak\triple{\sure{x\mapsto 1}\oplus_\frac12\sure{x \mapsto -1}}{x \samp\unif{x, x+1}}{\varphi_0 \oplus_\frac14 \varphi_3}
}
\end{equation}
The case where $x$ is instead decremented is completely symmetric.
\begin{equation}\label{eq:leader-R1-decr}
\inferrule*[right=Consequence]{
  \inferrule*[Right=Priv-Split1]{
    \inferrule*[right=Samp]{\;}{
      \vdash_\weak\triple{\sure{x\mapsto 1}}{x \samp\unif{x-1, x}}{x \sim \unif{0, 1}}
    }
    \\
    \inferrule*[Right=Samp]{\;}{
      \vdash_\weak\triple{\sure{x\mapsto -1}}{x \samp\unif{x-1, x}}{x \sim \unif{-2, -1}}
    }  
  }{
    \vdash_\weak\triple{\sure{x\mapsto 1}\oplus_\frac12\sure{x \mapsto -1}}{x \samp\unif{x-1, x}}{x \sim \unif{0, 1} \oplus_\frac12 x \sim \unif{-2, -1}}
  }
}{
  \vdash_\weak\triple{\sure{x\mapsto 1}\oplus_\frac12\sure{x \mapsto -1}}{x \samp\unif{x-1, x}}{\varphi_0 \oplus_\frac14 \varphi_3}
}
\end{equation}

\subsection{Online Algorithms: Paging and Caching}
\label{app:paging}

Recall the paging algorithm from \Cref{sec:paging} below. Let $C_\code{if}$ be the contents of the `then' branch of the if statement and $C_\code{body}$ be the body of the while loop:
\begin{mathpar}
\begin{array}{l}
\quad  c \samp \bern{\tfrac12}\fatsemi i\coloneqq 0\fatsemi m \coloneqq 0\fatsemi \; \\
\quad  \whl{i < n}{} \\
  \qquad i \coloneqq i+1\fatsemi\; \\
  \qquad r \gets [0,1]\fatsemi\; \\
  \qquad \ift{r \neq c}{} \\
  \quad\qquad c \samp \bern{\frac12} \fatsemi\; \\
  \quad\qquad m \coloneqq m+1
\end{array}

  C_\code{if} \triangleq c\samp\bern{\tfrac12}\fatsemi m\coloneqq m+1

  C_\code{body} \triangleq i\coloneqq i+1 \fatsemi r \gets[0,1] \fatsemi \ift{r \neq c}{C_\code{if}}
\end{mathpar}
Recall also the loop invariant below. The rank $R$ indicates the number of iterations until termination, which is given by $n-i$. Each iteration, the rank decreases with probability exactly 1, since the program terminates deterministically.
\begin{mathpar}
  \varphi_{\textsf{pure}} \triangleq \sure{i \mapsto N-R \sep n \mapsto N \sep \own(r) \sep 0 \le R \le N}

  \varphi_{\textsf{miss}} \triangleq \sure{m \mapsto N-R} \sep c\sim\bern{\tfrac12}

  \varphi_{\textsf{hit}} \triangleq \textstyle\bignd_{X=0}^{N-R-1} \bignd_{Y\in\{0,1\}} \sure{m\mapsto X\sep c \mapsto Y}

  \varphi \triangleq \varphi_{\textsf{pure}} \sep \big( \varphi_{\textsf{miss}} \oplus_{\frac1{2^{N-R}}} \varphi_{\textsf{hit}} \big)
\end{mathpar}
We now show the derivation, starting with the top level program. Analyzing the initialization sequence involves straightforward applications of \ruleref{Seq}, \ruleref{Samp}, and \ruleref{Assign}. The loop invariant specialized to $R=0$ ($\varphi[0/R]$) clearly implies the final postcondition by removing information about the variables $i$, $n$, $c$, and $r$.
\[
\inferrule*[right=Consequence]{
\inferrule*[Right=Seq]{
  \inferrule*[right=Samp,vdots=10em,rightskip=36.5em]{\;}{
    \vdash_\weak\triple{\sure{\own(c, i, m, r) \sep n \mapsto N}}{ c \samp\bern{\tfrac12}}{\sure{\own(i, m, r) \sep n \mapsto N} \sep c \sim \bern{\tfrac12}}
  } 
  \inferrule*[Right=Seq]{
    \inferrule*[right=Assign,vdots=5em,rightskip=41.5em]{\;}{
      \vdash_\weak\triple{\sure{\own(i, m, r) \sep n \mapsto N} \sep c \sim \bern{\tfrac12}}{i\coloneqq 0}{\sure{i \mapsto 0\sep \own(m, r) \sep n \mapsto N} \sep c \sim \bern{\tfrac12}}
    } 
    \inferrule*[Right=Seq]{
      \inferrule*[right=Assign]{\;}{
        \vdash_\weak\triple{\sure{i \mapsto 0\sep \own(m, r) \sep n \mapsto N} \sep c \sim \bern{\tfrac12}}{m\coloneqq 0}{\sure{i \mapsto 0\sep m\mapsto 0\sep \own(r) \sep n \mapsto N} \sep c \sim \bern{\tfrac12}}
      } 
      \quad
      (\ref{eq:paging-loop})
    }{
      \vdash_\weak\triple{\sure{i \mapsto 0\sep \own(m, r) \sep n \mapsto N} \sep c \sim \bern{\tfrac12}}{m\coloneqq 0 \fatsemi \whl{i<n}{C_\code{body}}}{\varphi[0/R]}
    }
  }{
    \vdash_\weak\triple{\sure{\own(i, m, r) \sep n \mapsto N} \sep c \sim \bern{\tfrac12}}{i\coloneqq 0\fatsemi m\coloneqq 0 \fatsemi \whl{i<n}{C_\code{body}}}{\varphi[0/R]}
  } 
}{
  \vdash_\weak\triple{\sure{\own(c, i, m, r) \sep n \mapsto N}}{ c \samp\bern{\tfrac12}\fatsemi i\coloneqq 0\fatsemi m\coloneqq 0 \fatsemi \whl{i<n}{C_\code{body}}}{\varphi[0/R]}
} 
}{
  \vdash_\weak\triple{\sure{\own(c, i, m, r) \sep n \mapsto N}}{ c \samp\bern{\tfrac12}\fatsemi i\coloneqq 0\fatsemi m\coloneqq 0 \fatsemi \whl{i<n}{C_\code{body}}}{\textstyle\sure{m\mapsto N} \oplus_{\frac1{2^N}} \bignd_{X = 0}^{N-1} \sure{m \mapsto X}}
} 
\]
To handle the while loop, we first use the following \ruleref{Consequence} so that we can apply \ruleref{Bounded-Rank}.
\[\textstyle
  \sure{i \mapsto 0\sep m\mapsto 0\sep \own(r) \sep n \mapsto N} \sep c \sim \bern{\tfrac12}
  \quad\Leftrightarrow\quad
  \varphi[N/R]
  \quad\Rightarrow\quad
  \bignd_{R=0}^N\varphi
\]
Inside the loop, we use \ruleref{Seq} and \ruleref{Assign} to handle the first command of $C_\code{body}$, which is the increment operation $i\coloneqq i+1$. The derivation for the remainder of the body is shown in (\ref{eq:paging-body}). At the end of $C_\code{body}$, we reestablish the loop invariant with the rank decreased by 1 $\varphi[R-1/R]$, which gives us:
\[\textstyle
  \varphi[R-1/R]
  \quad\Rightarrow\quad
  \bignd_{R=0}^{R-1} \varphi  
  \quad\Rightarrow\quad
  \left(  \bignd_{R=0}^{R-1} \varphi \right) \oplus_{\ge 1}\left(  \bignd_{R=R}^{N} \varphi \right)
\]
which justifies the \ruleref{Consequence} on the postcondition of the loop body.
\begin{equation}\label{eq:paging-loop}
\inferrule*[right=Consequence]{
  \inferrule*[Right=Bounded-Rank]{
    \inferrule*[Right=Consequence]{
    \inferrule*[Right=Seq]{
      \inferrule*[right=Assign]{\;}{
        \vdash_\weak\triple{\varphi \sep \sure{R > 0}}{i \coloneqq i+1}{(\varphi_\mathsf{miss}\oplus_{\frac1{2^{N-R}}} \varphi_\mathsf{hit})\sep \sure{i \mapsto N-(R-1) \sep n \mapsto N \sep \own(r) \sep  0 \le R-1 \le N}}
      } 
      \\
      (\ref{eq:paging-body})
    }{
      \vdash_\weak\triple{\varphi \sep \sure{R > 0}}{C_\code{body}}{\varphi[R-1/R]}
    } 
    }{
      \vdash_\weak\triple{\varphi \sep \sure{R > 0}}{C_\code{body}}{\textstyle\big(\bignd_{R=0}^{R-1} \varphi\big) \oplus_{\ge 1}\big(\bignd_{R=R}^{N} \varphi\big)}
    } 
  }{
    \vdash_\weak\triple{\textstyle\bignd_{R=0}^N \varphi}{\whl{i<n}{C_\code{body}}}{\varphi[0/R]}
  } 
}{
  \vdash_\weak\triple{\sure{i \mapsto 0\sep m\mapsto 0\sep \own(r) \sep n \mapsto N} \sep c \sim \bern{\tfrac12}}{\whl{i<n}{C_\code{body}}}{\varphi[0/R]}
} 
\end{equation}
Now, moving on to the remainder of the loop body, the first step is to use the \ruleref{Frame} rule to remove the static information about $i$ and $n$. The next command is the adversarial choice $r \gets[0, 1]$, but the precondition is not stable since $\varphi_\mathsf{hit}$ contains $\bignd$'s. So we must proceed by case analysis using \ruleref{Priv-Split1}. At the end, we can re-associate the \textsc{hit} and \textsc{miss} terms using \ruleref{Consequence}, and using the fact that $\frac12\cdot \frac1{2^{N-R}}=\frac1{2^{N-(R-1)}}$, we get the following.
\begin{equation}\label{eq:paging-body}
  \inferrule*[Right=Frame]{
    \inferrule*[Right=Consequence]{
      \inferrule*[Right=Priv-Split1,leftskip=3em,rightskip=5em]{
        (\ref{eq:paging-miss})
        \\
        (\ref{eq:paging-hit})
      }{
        \vdash_\weak\triple{(\varphi_\mathsf{miss}\sep\sure{\own(r)}) \oplus_{\frac1{2^{N-R}}} (\varphi_\mathsf{hit}\sep\sure{\own(r)})}{r \gets [0, 1]\fatsemi \ift{r\neq c}{C_\code{if}}}{
          (( \varphi_\mathsf{miss} \oplus_\frac12 \varphi_\mathsf{hit})[R-1/R] \sep \sure{\own(r)} )
            \oplus_{\frac1{2^{N-R}}}
            (\varphi_\mathsf{hit}[R-1/R] \sep \sure{\own(r)})}
      } 
    }{
      \vdash_\weak\triple{(\varphi_\mathsf{miss}\oplus_{\frac1{2^{N-R}}} \varphi_\mathsf{hit})\sep \sure{\own(r)}}{ r \gets [0, 1]\fatsemi \ift{r\neq c}{C_\code{if}}}{(\varphi_\mathsf{miss}[R-1/R]\oplus_{\frac1{2^{N-(R-1)}}} \varphi_\mathsf{hit}[R-1/R])\sep \sure{\own(r)}}
    } 
  }{
    \vdash_\weak\triple{(\varphi_\mathsf{miss}\oplus_{\frac1{2^{N-R}}} \varphi_\mathsf{hit})\sep \sure{i \mapsto N-(R-1) \sep n \mapsto N \sep \own(r) \sep  0 \le R-1 \le N}}{ r \gets [0, 1]\fatsemi \ift{r\neq c}{C_\code{if}}}{\varphi[R-1/R]}
  } 
\end{equation}
In the first case of the \ruleref{Priv-Split1} from (\ref{eq:paging-body}), all of the prior requests were cache misses. Since $\varphi_\mathsf{miss}$ is stable, we can take advantage of the obliviousness of the scheduler and immediately use \ruleref{NAssign} to handle $r \gets[0,1]$. Now, for any fixed value of $X$, $c\sim\bern{\tfrac12}$, which means that $c$ is equally likely to be equal to $X$ (\ie $\sure{c \mapsto \xor(X, 0)}$) or not equal to $X$ (\ie $\sure{c \mapsto \xor(X, 1)}$), which motivates the \ruleref{Consequence} below.
\begin{equation}\label{eq:paging-miss}
\inferrule*[right=Seq]{
  \inferrule*[right=Consequence]{
      \inferrule*[Right=NAssign]{\;}{
        \vdash_\weak\triple{\varphi_\mathsf{miss}\sep  \sure{\own(r)}}{r \gets [0, 1]}{ \textstyle\bignd_{X \in\{0, 1\}} \varphi_\mathsf{miss}\sep \sure{r \mapsto X}}
      } 
  }{
    \vdash_\weak\triple{\varphi_\mathsf{miss}\sep  \sure{\own(r)}}{r \gets [0, 1]}{\textstyle\bignd_{X \in\{0, 1\}} \bigoplus_{Y\sim\bern{\frac12}}\sure{r \mapsto X \sep c \mapsto \xor(X,Y) \sep m \mapsto N-R}}
  } 
  \\
  (\ref{eq:paging-miss-if})
}{
  \vdash_\weak\triple{\varphi_\mathsf{miss}\sep  \sure{\own(r)}}{r \gets [0, 1]\fatsemi \ift{r\neq c}{C_\code{if}}}{{(\varphi_\mathsf{miss} \oplus_\frac12 \varphi_\mathsf{hit})[R-1/R] \sep\sure{\own(r)}}}
}
\end{equation}
Moving on to the if statement, the two nondeterministic outcomes over $X$ will yield the same result, so we use \ruleref{ND-Split2} to do case analysis on $X$. We then use \ruleref{Priv-Split1} to do case analysis on the random outcomes, which correspond to the `then' and `else' branches of the if statement.
\begin{equation}\label{eq:paging-miss-if}
  \inferrule*[right={ND-Split2\phantom{XXXX}}]{
    \inferrule*[Right=Consequence]{
    \inferrule*[Right=Priv-Split1]{
      (\ref{eq:paging-miss-if-then})
      \\
      (\ref{eq:paging-miss-if-else})
    }{
      \vdash_\weak\triple{\textstyle\bigoplus_{Y\sim\bern{\frac12}}\sure{r \mapsto X \sep c \mapsto \xor(X,Y) \sep m \mapsto N-R}}{\ift{r\neq c}{C_\code{if}}}{(\varphi_\mathsf{miss}[R-1/R] \sep\sure{\own(r)}) \oplus_\frac12 (\varphi_\mathsf{hit}[R-1/R] \sep\sure{\own(r)})}
    } 
    }{
      \vdash_\weak\triple{\textstyle\bigoplus_{Y\sim\bern{\frac12}}\sure{r \mapsto X \sep c \mapsto \xor(X,Y) \sep m \mapsto N-R}}{\ift{r\neq c}{C_\code{if}}}{(\varphi_\mathsf{miss} \oplus_\frac12 \varphi_\mathsf{hit})[R-1/R] \sep\sure{\own(r)}}
    } 
  }{
    \vdash_\weak\triple{
      \textstyle\bignd_{X \in\{0, 1\}} \bigoplus_{Y\sim\bern{\frac12}}\sure{r \mapsto X \sep c \mapsto \xor(X,Y) \sep m \mapsto N-R}
    }{
      \ift{r\neq c}{C_\code{if}}
    }{
      {(\varphi_\mathsf{miss} \oplus_\frac12 \varphi_\mathsf{hit})[R-1/R] \sep\sure{\own(r)}}
    }  
  } 
\end{equation}
For the `then' branch, we first use \ruleref{IfT} to move into the if statement. A simple application of \ruleref{Samp} and \ruleref{Assign} gets us the postcondition:
\[
  \sure{r \mapsto X \sep m \mapsto N-(R-1)} \sep c\sim\bern{\tfrac12}
  \quad\Rightarrow\quad
  \left(\sure{m \mapsto N-(R-1)} \sep c\sim\bern{\tfrac12} \right) \sep \sure{\own(r)}
  \quad=\quad
  \varphi_\mathsf{miss}[R-1/R] \sep\sure{\own(r)}
\]
Therefore, we can use the \ruleref{Consequence} below.
\begin{equation}\label{eq:paging-miss-if-then}
\inferrule*[right=IfT]{
  \inferrule*[Right=Seq]{
    \inferrule*[right=Samp,vdots=5em,rightskip=35em]{\;}{
      \vdash_\weak\triple{\sure{r \mapsto X \sep c \mapsto \xor(X,1) \sep m \mapsto N-R}}{c \samp \bern{\tfrac12}}{\sure{r \mapsto X \sep m \mapsto N-R} \sep c\sim\bern{\tfrac12}}
    } 
    \inferrule*[Right=Consequence]{
      \inferrule*[Right=Assign]{\;}{
        \vdash_\weak\triple{\sure{r \mapsto X \sep m \mapsto N-R} \sep c\sim\bern{\tfrac12}}{m \coloneqq m+1}{\sure{r \mapsto X \sep m \mapsto N-(R - 1)} \sep c\sim\bern{\tfrac12}}
      } 
    }{
      \vdash_\weak\triple{\sure{r \mapsto X \sep m \mapsto N-R} \sep c\sim\bern{\tfrac12}}{m \coloneqq m+1}{\varphi_\mathsf{miss}[R-1/R] \sep \sure{\own(r)}}
    } 
  }{
    \vdash_\weak\triple{\sure{r \mapsto X \sep c \mapsto \xor(X,1) \sep m \mapsto N-R}}{c \samp \bern{\tfrac12} \fatsemi m \coloneqq m+1}{\varphi_\mathsf{miss}[R-1/R] \sep \sure{\own(r)}}
  } 
}{
  \vdash_\weak\triple{\sure{r \mapsto X \sep c \mapsto \xor(X,1) \sep m \mapsto N-R}}{\iftf{r\neq c}{C_\code{if}}\skp}{\varphi_\mathsf{miss}[R-1/R] \sep \sure{\own(r)}}
} 
\end{equation}
The `else' branch is simply a $\skp$ command, so using the following \ruleref{Consequence}:
\[
  \sure{r \mapsto X \sep c \mapsto \xor(X,0) \sep m \mapsto N-R}
  \quad\Rightarrow\quad
  \textstyle\left(\bignd_{X=0}^{N-(R-1)-1} \bignd_{Y\in\{0,1\}} \sure{m\mapsto X\sep c \mapsto Y}\right) \sep \sure{\own(r)}
  \quad=\quad
  \varphi_\mathsf{hit}[R-1/R] \sep  \sure{\own(r)}
\]
We can complete the proof as follows using \ruleref{IfF} and \ruleref{Skip}.
\begin{equation}\label{eq:paging-miss-if-else}
\inferrule*[right=Consequence]{
  \inferrule*[Right=IfF]{
    \inferrule*[Right=Skip]{\;}{
      \vdash_\weak\triple{\sure{r \mapsto X \sep c \mapsto \xor(X,0) \sep m \mapsto N-R}}{\skp}{\sure{r \mapsto X \sep c \mapsto \xor(X,0) \sep m \mapsto N-R}}
    } 
  }{
    \vdash_\weak\triple{\sure{r \mapsto X \sep c \mapsto \xor(X,0) \sep m \mapsto N-R}}{\iftf{r\neq c}{C_\code{if}}\skp}{\sure{r \mapsto X \sep c \mapsto \xor(X,0) \sep m \mapsto N-R}}
  } 
}{
  \vdash_\weak\triple{\sure{r \mapsto X \sep c \mapsto \xor(X,0) \sep m \mapsto N-R}}{\iftf{r\neq c}{C_\code{if}}\skp}{\varphi_\mathsf{hit}[R-1/R] \sep \sure{\own(r)}}
} 
\end{equation}
We now proceed with the second case of the \ruleref{Priv-Split1} from (\ref{eq:paging-body}), where there has been at least one prior cache hit. Analyzing the adversarial choice $r \gets[0,1]$ is a bit harder because the precondition is not stable. We show that derivation in (\ref{eq:paging-hit-nd}), which essentially results in $r$ being nondeterministically chosen to be either equal or not equal to $c$. We then do case analysis with \ruleref{ND-Split2} to analyze the two branches of the if statement.
\begin{equation}\label{eq:paging-hit}
\inferrule*[right=Seq]{
  (\ref{eq:paging-hit-nd})
  \\
  \inferrule*[Right=ND-Split2]{
      (\ref{eq:paging-hit-if-then})
      \\
      (\ref{eq:paging-hit-if-else})
  }{
    \vdash_\weak\triple{
      \textstyle \bignd_{Z \in \{0,1\}} \bignd_{X=0}^{N-R-1} \bignd_{Y \in \{0, 1\}} \sure{m \mapsto X \sep c \mapsto Y \sep r \mapsto \xor(Y,Z)}
    }{
      \ift{r\neq c}{C_\code{if}}
    }{
      {\varphi_\mathsf{hit}[R-1/R] \sep\sure{\own(r)}}
    }
  } 
}{
  \vdash_\weak\triple{\varphi_\mathsf{hit}\sep  \sure{\own(r)}}{r \gets [0, 1]\fatsemi \ift{r\neq c}{C_\code{if}}}{{\varphi_\mathsf{hit}[R-1/R] \sep\sure{\own(r)}}}
} 
\end{equation}
To analyze the $r\gets[0,1]$ command, we use \ruleref{ND-Split} twice to get to a pure assertion, which is stable, so that we can apply \ruleref{NAssign}. After reapplying the $\bignd$ modalities, we use a \ruleref{Consequence} to rearrange them and to rewrite the assertion about $r$ to be $\sure{r\mapsto \xor(Y, Z)}$, similar to what we did in (\ref{eq:paging-miss}).
\begin{equation}\label{eq:paging-hit-nd}
\inferrule*[right=Consequence]{
  \inferrule*[Right=ND-Split]{
    \inferrule*[Right=ND-Split]{
      \inferrule*[Right=NAssign]{\;}{
        \vdash_\weak\triple{\sure{m \mapsto X \sep c \mapsto Y \sep \own(r)}}{r \gets [0, 1]}{
           \textstyle\bignd_{Z \in \{0,1\}} \sure{m \mapsto X \sep c \mapsto Y \sep r \mapsto Z}
        }
      }
    }{
      \vdash_\weak\triple{
        \textstyle\bignd_{Y \in \{0, 1\}} \sure{m \mapsto X \sep c \mapsto Y \sep \own(r)}
      }{
        r \gets [0, 1]
      }{
        \textstyle\bignd_{Y \in \{0, 1\}} \bignd_{Z \in \{0,1\}} \sure{m \mapsto X \sep c \mapsto Y \sep r \mapsto Z}
      }
    } 
  }{
    \vdash_\weak\triple{
      \textstyle\bignd_{X=0}^{N-R-1} \bignd_{Y \in \{0, 1\}} \sure{m \mapsto X \sep c \mapsto Y \sep \own(r)}
    }{
      r \gets [0, 1]
    }{
      \textstyle\bignd_{X=0}^{N-R-1} \bignd_{Y \in \{0, 1\}}\bignd_{Z \in \{0,1\}} \sure{m \mapsto X \sep c \mapsto Y \sep r \mapsto Z}
    }
  } 
}{
  \vdash_\weak\triple{\varphi_\mathsf{hit}\sep  \sure{\own(r)}}{r \gets [0, 1]}{
    \textstyle \bignd_{Z \in \{0,1\}} \bignd_{X=0}^{N-R-1} \bignd_{Y \in \{0, 1\}} \sure{m \mapsto X \sep c \mapsto Y \sep r \mapsto \xor(Y,Z)}
  }
} 
\end{equation}
Now, for the `then' branch of the if statement from (\ref{eq:paging-hit}), we first apply \ruleref{IfT}, \ruleref{Seq}, and \ruleref{Samp} to move inside the if statement and dispatch the cache eviction command $c \samp \bern{\tfrac12}$. We then apply \ruleref{ND-Split} and \ruleref{Assign} to handle the cache miss increment. Finally, since $X$ is at most $N-R-1$, then $X+1$ is at most $N- (R-1) -1$, so we can use a \ruleref{Consequence} to conclude that $\varphi_\mathsf{hit}\sep \sure{\own(r)}$ holds at the end.
\begin{equation}\label{eq:paging-hit-if-then}
\inferrule*[right=IfT]{
  \inferrule*[Right=Seq]{
    \inferrule*[right=Samp,vdots=6.5em,rightskip=45em]{\;}{
      \vdash_\weak\triple{
        \textstyle \bignd_{X=0}^{N-R-1} \bignd_{Y \in \{0, 1\}} \sure{m \mapsto X \sep c \mapsto Y \sep r \mapsto \xor(Y,1)}
      }{c \samp \bern{\tfrac12}}{
        \textstyle (\bignd_{X=0}^{N-R-1} \sure{m \mapsto X \sep \own(r)}) \sep c\sim\bern{\tfrac12}}
    } 
    \inferrule*[Right=Consequence]{
      \inferrule*[Right=ND-Split]{
        \inferrule*[Right=Assign]{\;}{
          \vdash_\weak\triple{\sure{m \mapsto X \sep \own(r)} \sep c\sim\bern{\tfrac12}}{m \coloneqq m+1}{\sure{m \mapsto X+1 \sep \own(r)} \sep c\sim\bern{\tfrac12}}
        } 
      }{
        \vdash_\weak\triple{
          \textstyle \bignd_{X=0}^{N-R-1} \sure{m \mapsto X \sep \own(r)} \sep c\sim\bern{\tfrac12}
        }{m \coloneqq m+1}{
          \textstyle \bignd_{X=0}^{N-R-1} \sure{m \mapsto X+1 \sep \own(r)} \sep c\sim\bern{\tfrac12}
        }
      } 
    }{
        \vdash_\weak\triple{
          \textstyle \bignd_{X=0}^{N-R-1} (\sure{m \mapsto X \sep \own(r)}) \sep c\sim\bern{\tfrac12}
        }{m \coloneqq m+1}{\varphi_\mathsf{hit}[R-1/R] \sep \sure{\own(r)}}    
    } 
  }{
    \vdash_\weak\triple{
      \textstyle \bignd_{X=0}^{N-R-1} \bignd_{Y \in \{0, 1\}} \sure{m \mapsto X \sep c \mapsto Y \sep r \mapsto \xor(Y,1)}
    }{c \samp \bern{\tfrac12} \fatsemi m \coloneqq m+1}{\varphi_\mathsf{hit}[R-1/R] \sep \sure{\own(r)}}
  } 
}{
  \vdash_\weak\triple{
    \textstyle \bignd_{X=0}^{N-R-1} \bignd_{Y \in \{0, 1\}} \sure{m \mapsto X \sep c \mapsto Y \sep r \mapsto \xor(Y,1)}
  }{\iftf{r\neq c}{C_\code{if}}\skp}{\varphi_\mathsf{hit}[R-1/R] \sep \sure{\own(r)}}
} 
\end{equation}
The `false' branch is simple because the command being executed is just $\skp$.
\begin{equation}\label{eq:paging-hit-if-else}
\inferrule*[right=IfF]{
  \inferrule*[Right=Consequence]{
    \inferrule*[Right=Skip]{\;}{
      \vdash_\weak\triple{
        \varphi_\mathsf{hit}[R-1/R] \sep \sure{\own(r)}
      }{\skp}{\varphi_\mathsf{hit}[R-1/R] \sep \sure{\own(r)}}
    } 
  }{
    \vdash_\weak\triple{
    \textstyle \bignd_{X=0}^{N-R-1} \bignd_{Y \in \{0, 1\}} \sure{m \mapsto X \sep c \mapsto Y \sep r \mapsto \xor(Y,0)}
  }{\skp}{\varphi_\mathsf{hit}[R-1/R] \sep \sure{\own(r)}}
  } 
}{
  \vdash_\weak\triple{
    \textstyle \bignd_{X=0}^{N-R-1} \bignd_{Y \in \{0, 1\}} \sure{m \mapsto X \sep c \mapsto Y \sep r \mapsto \xor(Y,0)}
  }{\iftf{r\neq c}{C_\code{if}}\skp}{\varphi_\mathsf{hit}[R-1/R] \sep \sure{\own(r)}}
} 
\end{equation}

\end{landscape}

\end{document}